\documentclass[11pt]{book}

\usepackage{amsmath}
\usepackage{amsthm}
\usepackage{endnotes}
\usepackage{amssymb}
\usepackage{url}
\usepackage[colorlinks=true,linkcolor=blue,citecolor=blue]{hyperref}

\usepackage{algorithm}
\usepackage{xspace}
\usepackage{color}
\usepackage{mathtools}
\usepackage{enumerate}
\usepackage{algorithmicx}
\usepackage{algpseudocode}
\usepackage{natbib}
\usepackage{fancyhdr}
\usepackage{cleveref}

\newtheorem{theorem}{Theorem}[chapter]
\newtheorem{definition}[theorem]{Definition}
\newtheorem{example}[theorem]{Example}

\newtheorem{corollary}[theorem]{Corollary}
\newtheorem{lemma}[theorem]{Lemma}

\newtheorem{fact}[theorem]{Fact}
\newtheorem{assumption}[theorem]{Assumption}

\usepackage{booktabs}

\usepackage{listings}

\usepackage[framemethod=tikz]{mdframed}
\usetikzlibrary{shadows}
\newmdenv[shadow=true,shadowcolor=black,font=\sffamily,rightmargin=8pt]{shadedbox}
\newmdenv[shadow=true,shadowcolor=blue,font=\sffamily,rightmargin=8pt]{shadedbox_definition_concepts}

\usepackage{fancyhdr}
\usepackage{mathtools}
\usepackage{bm}
\usepackage{braket}
\usepackage{lineno}

\usepackage[most]{tcolorbox}
\usepackage{mleftright}
\mleftright

\usepackage{multirow}
\usepackage{makecell}
\usepackage{longtable}

\newcommand{\ie}{\emph{i.e.}}

\def\Gsoft{G^{\mathrm{soft}}}

\def\Ures{U_{\mathrm{res}}}
\def\Uln{U_{\mathrm{LN}}}
\def\relu{\mathrm{ReLU}}
\def\Sinput{(\alpha_s,a_s)}
\def\Qinput{(\alpha_q,a_q)}
\def\Kinput{(\alpha_k,a_k)}
\def\Vinput{(\alpha_v,a_v)}

\def\Winput{(\alpha_w,a_w)}
\def\QKinput{(\alpha_{0},a_{0})}

\DeclareMathOperator*{\argmax}{\arg\max}
\DeclareMathOperator{\diag}{diag}
\DeclareMathOperator{\ERM}{ERM}
\DeclareMathOperator{\Gene}{Gene}
\DeclareMathOperator{\Var}{Var}
\DeclareMathOperator{\Circ}{circ}
\DeclareMathOperator{\Tanh}{tanh}

\def\L{{\mathcal L}}

\def\cO{{\mathcal{O}}}

\def\norm#1{\mathopen\| #1 \mathclose\|} 
\def\Gsoft{G^{\mathrm{soft}}}

\def\Ures{U_{\mathrm{res}}}
\def\Uln{U_{\mathrm{LN}}}
\def\relu{\mathrm{ReLU}}

\newcommand{\XGate}{\mathop{\text{X}}}
\newcommand{\YGate}{\mathop{\text{Y}}}
\newcommand{\ZGate}{\mathop{\text{Z}}}
\newcommand{\CZ}{\mathop{\text{CZ}}}
\newcommand{\SWAP}{\mathop{\text{SWAP}}}
\newcommand{\Tr}{\mathop{\text{Tr}}}
\newcommand{\Ber}{\mathop{\text{Ber}}}
\newcommand{\Pro}{\mathop{\text{Pr}}}
\newcommand{\RX}{\mathop{\text{RX}}}
\newcommand{\RY}{\mathop{\text{RY}}}
\newcommand{\RZ}{\mathop{\text{RZ}}}
\newcommand{\CNOT}{\mathop{\text{CNOT}}}
\newcommand{\Hada}{\mathop{\text{H}}}
\def\>{\rangle}
\def\<{\langle}

\DeclareMathOperator{\PAC}{\mathsf{PAC}}

\DeclareMathOperator{\BQP}{\mathsf{BQP}}
\DeclareMathOperator{\PPP}{\mathsf{P}}

\let\footnote=\endnote

\title{\textbf{Quantum Machine Learning}\\ A Hands-on Tutorial for Machine Learning Practitioners and Researchers  } 
\date{}
\author{
    Yuxuan Du$^{1,*,\dagger}$ \and 
    Xinbiao Wang$^{1,*}$ \and 
    Naixu Guo$^{2,*}$\and
    Zhan Yu$^{2,*}$\and
    Yang Qian$^{1,*}$\and
    Kaining Zhang$^{1,*}$\and
    Min-Hsiu Hsieh$^{3,\upharpoonleft}$\and
    Patrick Rebentrost$^{2,\perp}$\and
    Dacheng Tao$^{1,\ddagger}$
}

\date{
    $^{1}$ \small{College of Computing and Data Science, Nanyang Technological University, 639798, Singapore} \\
    $^{2}$ \small{Centre for Quantum Technologies, National University of Singapore, 117543, Singapore} \\
    $^{3}$ \small{Hon Hai Research Institute, Taipei, 114, Taiwan} \\
    \bigskip
    $^{\dagger}$ \small{duyuxuan123@gmail.com}\\
    $^{\upharpoonleft}$  \small{min-hsiu.hsieh@foxconn.com}\\
    $^{\perp}$ \small{cqtfpr@nus.edu.sg}\\
    $^{\ddagger}$ \small{dacheng.tao@ntu.edu.sg}\\
    \bigskip
    $^{*}$ \small{Equal contributions}\\
}

\begin{document}

\frontmatter  % title page, contents, catalog information
\maketitle

\chapter*{Abstract}
This tutorial intends to introduce readers with a background in AI to quantum machine learning (QML)---a rapidly evolving field that seeks to leverage the power of quantum computers to reshape the landscape of machine learning. For self-consistency, this tutorial covers foundational principles, representative QML algorithms, their potential applications, and critical aspects such as trainability, generalization, and computational complexity. In addition, practical code demonstrations are provided in \url{https://qml-tutorial.github.io/} to illustrate real-world implementations and facilitate hands-on learning. Together, these elements offer readers a comprehensive overview of the latest advancements in QML. By bridging the gap between classical machine learning and quantum computing, this tutorial serves as a valuable resource for those looking to engage with QML and explore the forefront of AI in the quantum era.

\tableofcontents

\chapter*{Preface}
     \addcontentsline{toc}{chapter}{Preface}
     \markboth{\sffamily\slshape Preface}
       {\sffamily\slshape Preface}

Quantum computers, as next-generation computational devices, harness the quantum principles of superposition and entanglement to process information in ways fundamentally different from classical computers. These unique properties enable quantum computers to address many practical problems that are intractable for classical computers. Although quantum computing is still in its early stages, we have entered an era since 2019 where quantum supremacy has been experimentally demonstrated by several research groups and industrial organizations, underscoring the immense potential of quantum technologies to transform various aspects of everyday life.

Machine learning (ML) is widely regarded as one of the most promising and impactful applications of quantum computing. The ability of quantum computing to accelerate advancements in foundational models, such as generative pre-trained transformers (GPTs), and even pave the way toward artificial general intelligence (AGI), is particularly compelling. Recent progress in both theories and experiments has exhibited the power of quantum machine learning (QML), where the integration of quantum computing with machine learning may lead to novel approaches that outperform classical algorithms by offering faster runtimes, better performance, and reduced data requirements in areas such as computer vision, natural language processing, drug discovery, finance, and fundamental science.

As an interdisciplinary field, the development of QML requires close collaboration between leading scientists and engineers in both quantum computing and artificial intelligence (AI). At the same time, as QML advances alongside the continuous progress of quantum hardware, there is a growing need for expertise from the AI community to drive this emerging field forward. However, the distinct conceptual frameworks and terminologies of quantum and classical computing present significant barriers for researchers and practitioners with a classical machine learning background in understanding the mechanisms behind QML algorithms and the benefits they may offer. Reducing this barrier to entry remains a major challenge within the community. 

To overcome this challenge, we have written this tutorial to deliver a comprehensive introduction to the latest developments in QML, specifically designed for readers with expertise in machine learning. Whether you are an AI researcher, a machine learning practitioner, or a computer science student, this resource will equip you with a solid foundation in the principles and techniques of QML. By bridging the gap between classical ML and quantum computing, this tutorial could serve as a useful resource for those looking to engage with quantum machine learning and explore the forefront of AI in the quantum era.

\bigskip
\noindent The Authors\\
Feb, 2025

\mainmatter

\chapter{Introduction} \label{c-intro}

The advancement of computational power has always been a driving force behind every major industrial revolution. The invention of the modern computer, followed by the central processing unit (CPU), led to the ``digital revolution'', transforming industries through process automation and the rise of information technology. More recently, the development of graphical processing units (GPUs) has powered the era of artificial intelligence (AI) and big data, enabling breakthroughs in areas such as intelligent transportation systems, autonomous driving, scientific simulations, and complex data analysis. However, as we approach the physical and practical limits of Moore's law—the principle that the number of transistors on a chip doubles approximately every two years—traditional computing devices like CPUs and GPUs are nearing the end of their scaling potential. The ever-growing demand for computational power, driven by the exponential increase in data and the complexity of modern applications, necessitates new computing paradigms. Among the leading candidates to meet these challenges are \textbf{quantum computers} \citep{feynman2017quantum}, which leverage the unique principles of quantum mechanics such as superposition and entanglement to process information in ways that classical systems cannot, with the potential to revolutionize diverse aspects of daily life.

One of the most concrete and direct ways to understand the potential of quantum computers is through the framework of complexity theory \citep{watrous2008quantum}. Theoretical computer scientists have demonstrated that quantum computers can efficiently solve problems within the $\BQP$ (Bounded-Error Quantum Polynomial Time) complexity class, meaning these problems can be solved in polynomial time by a quantum computer. In contrast, classical computers are limited to efficiently solving problems within the $\PPP$ (Polynomial Time) complexity class. While it is widely believed, though not proven, that $\PPP \subseteq \BQP$, this suggests that quantum computers can provide exponential speedups for certain problems in $\BQP$ that are intractable for classical machines.

A prominent example of such a problem is large-number factorization, which forms the basis of RSA cryptography. Shor’s algorithm \citep{shor1999polynomial}, a quantum algorithm, can factor large numbers in polynomial time, while the most efficient known classical factoring algorithm requires super-polynomial time. For instance, breaking an RSA-2048 bit encryption key would take a classical computer approximately 300 trillion years, whereas an \textit{ideal} quantum computer could complete the task in around 10 seconds. However, constructing  `ideal' quantum computers remains a significant challenge. As will be discussed in later chapters, based on current fabrication techniques, this task could potentially be completed in approximately 8 hours using a noisy quantum computer with a sufficient number of \textbf{qubits}—the fundamental units of quantum computation \citep{gidney2021factor}.
 
The convergence of the computational power offered by quantum machines and the limitations faced by AI models has led to the rapid emergence of the field: \textbf{quantum machine learning} (QML) \citep{biamonte2017quantum}. In particular, the challenges in modern AI stem from the neural scaling law \citep{kaplan2020scaling}, which posits that ``bigger is often better.'' Since 2020, this principle has driven the development of increasingly colossal models, featuring more complex architectures and an ever-growing number of parameters. However, this progress comes at an immense cost. For instance, training a model like ChatGPT on a single GPU would take approximately 355 years, while the cloud computing costs for training such large models can reach tens of thousands of dollars.  

These staggering costs present a critical barrier to the future growth of AI. Quantum computing, celebrated for its extraordinary computational capabilities, holds the potential to overcome these limitations. It offers the possibility of advancing models like generative pre-trained transformers (GPTs) and accelerating progress toward artificial general intelligence (AGI). Quantum computing, and more specifically QML,  represents a paradigm shift, moving from the classical ``it from bit'' framework to the quantum ``it from qubit'' era, with the potential to reshape the landscape of AI and computational science.

\section{A First Glimpse of Quantum Machine Learning}

So, what exactly is quantum machine learning (QML)? In its \textit{simplest} terms, the focus of this tutorial on QML can be summarized as follows (see Chapter~\ref{chapter1:explored-task-QML} for the systematic overview).
\begin{tcolorbox}[colback=blue!5!white,colframe=blue!75!black,title=Quantum machine learning (informal)]
 QML explores learning algorithms that can be executed on \underline{quantum computers} to accomplish \underline{specified tasks} with \underline{potential advantages} over classical implementations.
\end{tcolorbox}

The three key elements in the above interpretation are: \textit{quantum processors}, \textit{specified tasks}, and \textit{advantages}. In what follows, let us elucidate the specific meaning of each of these terms, providing the necessary foundation for a deeper understanding of the mechanisms and potential of QML.

\subsection{Quantum computers} \label{chapt1-sec1-1-1}

The origins of quantum computing can be traced back to 1980 when Paul Benioff introduced the \textit{quantum Turing machine} \citep{benioff1982quantum}, a quantum analog of the classical Turing machine designed to describe computational models through the lens of quantum theory. Since then, several models of quantum computation have emerged, including \textit{circuit-based quantum computation} \citep{nielsen2010quantum}, \textit{one-way quantum computation} \citep{raussendorf2001one}, \textit{adiabatic quantum computation} \citep{albash2018adiabatic}, and \textit{topological quantum computation} \citep{kitaev2003fault}. All of them have been shown to be computationally equivalent to the quantum Turing machine, meaning that a perfect implementation of any one of these models can simulate the others with no more than polynomial overhead. Given the prominence of \textbf{circuit-based quantum computers} in both the research and industrial communities and their rapid advancement, this tutorial will focus primarily on this model of quantum computing.

Quantum computing gained further momentum in the early 1980s when physicists faced an exponential increase in computational overhead while simulating quantum dynamics, particularly as the number of particles in a system grew. This ``curse of dimensionality'' prompted Yuri Manin and Richard Feynman to independently propose leveraging quantum phenomena to build quantum computers, arguing that such devices would be far more efficient for simulating quantum systems than classical computers. 

However, as a universal computing device, the potential of quantum computers extends well beyond quantum simulations. In the 1990s, \citet{shor1999polynomial} developed a groundbreaking quantum algorithm for large-number factorization, posing a serious threat to widely used encryption protocols such as RSA and Diffie–Hellman. In 1996, Grover's algorithm demonstrated a quadratic speedup for unstructured search problems \citep{grover1996fast}, a task with broad applications. Since then, the influence of quantum computing has expanded into a wide range of fields, with new quantum algorithms being developed to achieve runtime speedups in areas such as finance \citep{herman2023quantum}, drug design \citep{santagati2024drug}, optimization \citep{abbas2024challenges}, and, most relevant to this tutorial, machine learning.

\begin{figure}[t!]
\centering
\includegraphics[width=0.9\textwidth]{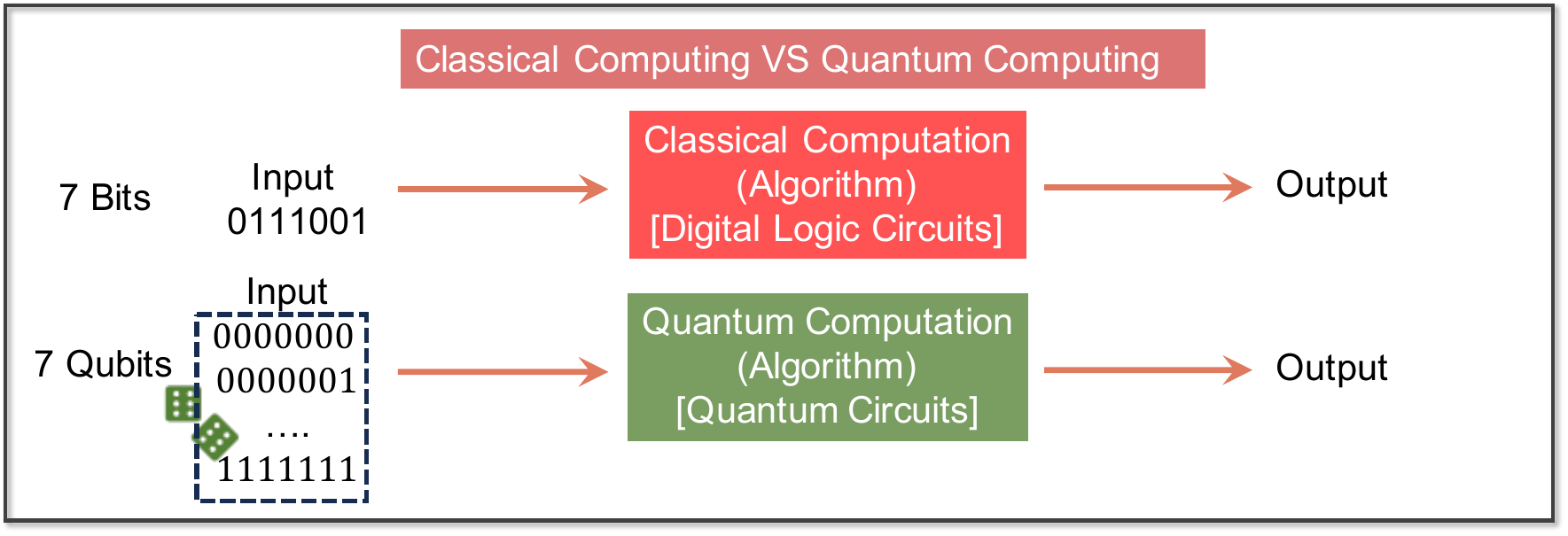}
  \caption{\textbf{The paradigm between classical and quantum computing.} The mechanisms between classical and quantum computing are very similar, where both of them involve input, computation, and output. In classical computing, the input refers to a bit-string, the computation part refers to the digital logic circuits, and the output also refers to a bit-string. In quantum computing, the input is a single- or multi-qubit state.  The computation involves quantum circuits. And the output of quantum computers requires quantum measurement, which aims to extract information from the quantum world to the classical world.  
}
  \label{fig:chap1:comparison-clc-quantum}
\end{figure}

An intuitive way to understand why quantum computers can outperform classical computers is by comparing their fundamental components.  As illustrated in Figure~\ref{fig:chap1:comparison-clc-quantum}, both types of computers consist of three fundamental components: input, the computational process, and the output. The implementation of these three components in classical and quantum computing is summarized in Table \ref{tab:my-table}.

\begin{table}[h!]
\centering 
\caption{\textbf{Comparison between classical and quantum computing.}}
\label{tab:my-table}
\footnotesize
\begin{tabular}{l|l|l}
\toprule 
\multicolumn{1}{c}{} & \multicolumn{1}{|c|}{Classical} & \multicolumn{1}{c}{Quantum} \\ \midrule
Input       & Binary bits & Quantum bits \\  
Computation & Digital logical circuits                 &  Quantum circuits \\  
Output      & Retrieve solution                         &  Quantum measurements \\ 
\end{tabular}
\end{table}

The advantages of quantum computers stem primarily from the key distinctions between classical bits and quantum bits (\textbf{qubits}), as well as between digital logic circuits and \textbf{quantum circuits}, as outlined below:
\begin{itemize}
  \item \textit{Bits versus Qubits}. A classical bit is a binary unit that takes on a value of either $0$ or $1$. In contrast, a quantum bit, or qubit, can exist in a superposition of both $0$ and $1$ simultaneously, represented by a \textit{two-dimensional vector} where the entries correspond to the \textit{probabilities} of the qubit being in each state. 
    
    Furthermore, while classical bits follow the Cartesian product rule, qubits adhere to the tensor product rule. This distinction implies that an $N$-qubit system is described by a $2^N$-dimensional vector, allowing quantum systems to encode information exponentially with $N$---far surpassing the capacity of classical bits. Table~\ref{tab:chapt1-math-single-multi-qubits} summarizes the mathematical expressions of classical and quantum bits. 
  \begin{table}[h!]
  \centering 
\caption{\textbf{Mathematical representations of $N$-(quantum) bits in classical and quantum computers}. Here the symbols `$\dagger$' and $\mathbb{C}$ denote the transpose conjugation and complex space, respectively.}
\label{tab:chapt1-math-single-multi-qubits}
\footnotesize
\begin{tabular}{l|c|c}
\toprule 
\multicolumn{1}{c|}{} & Classical           & Quantum                                                                                                           \\ 
\midrule  
Single bit ($N=1$)               & $\bm{x}\in\{0,1\}$  & \begin{tabular}[c]{@{}c@{}}$[a_1, a_2]^{\dagger}\in \mathbb{C}^2$\\ s.t. $a_1^2 + a_2^2=1$\end{tabular}                                \\  \midrule 
Multiple bits ($N>1$)           & $\bm{x}\in \{0,1\}^N$ & \begin{tabular}[c]{@{}c@{}} s.t.$[a_1, a_2, ... ,a_{2^N}]^{\dagger}\in \mathbb{C}^{2^N}$\\ s.t. $a_1^2 + a_2^2+...+a_{2^N}^2=1$\end{tabular} \\ \hline
\end{tabular}
\end{table}

  \item \textit{Digital logic circuits versus quantum circuits}. Classical computers rely on digital logic circuits composed of logic gates that perform operations on bits in a deterministic manner, as illustrated in Figure~\ref{fig:chap1:comparison-clc-quantum}. In contrast, quantum circuits consist of \textbf{quantum gates}, which act on single or multiple qubits to modify their states—the probability amplitudes $a_1,..., a_{2^N}$, as shown in Table~\ref{tab:chapt1-math-single-multi-qubits}. Owing to the universality of quantum gates, for any given input qubit state, there always exists a specific quantum circuit capable of transforming the input state into one corresponding to the target solution—a particular probability distribution. For certain probability distributions, a quantum computer can use a polynomial number of quantum gates relative to the qubit count $N$ to generate the distribution, whereas classical computers require an exponential number of gates with $N$ to achieve the same result. This difference underpins the quantum advantage.
  \item The readout process in quantum computing differs fundamentally from that in classical computing, as it involves \textbf{quantum measurements}, which extract information from a quantum system and translate it into a form that can be interpreted by classical systems. For problems in quantum physics and chemistry, quantum measurements can reveal far more useful information than classical simulations of the same systems, enabling significant runtime speedups in obtaining the desired physical properties.
\end{itemize}

The formal definitions of quantum computing are presented in Chapter~\ref{Chapter2:basics-QC}. As we will see, the power of quantum computers is primarily determined by two factors: the number of qubits and the quantum gates, as well as their respective qualities. The term ``qualities'' refers to the fact that fabricating quantum computers is highly challenging, as both qubits and quantum gates are prone to errors. These qualities are measured using various physical metrics. One commonly used metric is \textit{quantum volume} $V_Q$ \citep{cross2019validating}, which quantifies a quantum computer's capabilities by accounting for both its error rates and overall performance. Mathematically, the quantum volume represents the maximum size of square quantum circuits that the computer can successfully implement to achieve the \textit{heavy output generation problem}, i.e.,
\begin{equation}
  \log_2(V_Q) = \arg\max_{m} \min(m, d(m)),
\end{equation}
where $m\leq N$ is a number of qubits selected from the given $N$-qubit quantum computer, and $d(m)$ is the number of qubits in the largest square circuits for which we can reliably sample heavy outputs with probability greater than $2/3$. The heavy output generation problem discussed here stems from proposals aimed at demonstrating quantum advantage. That is, if a quantum computer is of sufficiently high quality, we should expect to observe heavy outputs frequently across a range of random quantum circuit families. For illustration, Table~\ref{chapter1-tab:quantum_progress} summarizes the progress of quantum computers as of 2024.

\begin{table}[]
\caption{\textbf{Progress of quantum computers up to December 2024.}}
\label{chapter1-tab:quantum_progress}
\centering
\footnotesize
\begin{tabular}{l|c|c|c|c}
\toprule 
 
\multicolumn{1}{c|}{\textbf{Date}} & \textbf{$\log_2(V_Q)$} & \textbf{$N$} & \textbf{Manufacturer} & \textbf{System Name} \\ 
\midrule
Dec, 2024    & -  & 105 & Google     & Willow         \\  
Aug, 2024    & 21 & 56  & Quantinuum & H2-1      \\ 
Jul, 2024      & 9  & 156 & IBM        & Heron     \\  
Jun, 2023      & 19 & 20  & Quantinuum & H1-1      \\  
Sep, 2022 & 13 & 20  & Quantinuum & H1-1      \\  
Apr, 2022     & 12 & 12  & Quantinuum & H1-2      \\  
Jul, 2021      & 10 & 10  & Honeywell  & H1        \\  
Nov, 2020  & 7  & 10  & Honeywell  & H1        \\  
Aug, 2020    & 6  & 27  & IBM        & Falcon r4 \\  
\end{tabular}
\end{table}

\begingroup
\allowdisplaybreaks
\begin{tcolorbox}[enhanced, 
    breakable,colback=gray!5!white,colframe=gray!75!black,title=Remark]
 Note that quantum volume is not the unique metric for evaluating the performance of quantum computers. There are several other metrics that assess the power of quantum processors from different perspectives. For instance, \textit{Circuit Layer Operations Per Second} (CLOPS) \citep{wack2021quality} measures the computing speed of quantum computers, reflecting the feasibility of running practical calculations that involve a large number of quantum circuits. Additionally, \textit{effective quantum volume} \citep{kechedzhi2024effective} provides a more nuanced comparison between noisy quantum processors and classical computers, considering factors such as error rates and noise levels. These metrics, among others, offer a more comprehensive understanding of the strengths and limitations of quantum computers across various applications.  
\end{tcolorbox}
\endgroup

\subsection{Different measures of quantum advantages}

\label{chapt1:subsec:diff_quantum_adv}

What do we mean when we refer to quantum advantage? Broadly, quantum advantage is demonstrated when quantum computers can solve a problem \textit{more efficiently} than classical computers. However, the notion of ``efficiency'' in this context is not uniquely defined.

The most common measure of efficiency is \textit{runtime complexity}. By harnessing quantum effects, certain computations can be accelerated significantly—sometimes even exponentially—enabling tasks that are otherwise infeasible for classical computers. A prominent example is Shor’s algorithm, which achieves an exponential speedup in large-number factorization relative to the best classical algorithms. In terms of runtime complexity, the quantum advantage is achieved when the upper bound of a quantum algorithm’s runtime for a given task is \textit{lower than} the theoretical lower bound of all possible classical algorithms for the same task.

In quantum learning theory \citep{arunachalam2017guest}, efficiency can also be measured by \textit{sample complexity}, particularly within the Probably Approximately Correct ($\PAC$) learning framework, which is central to this tutorial. In this context, sample complexity is defined as the number of interactions (e.g., quires of target quantum systems or measurements) required for a learner to achieve a desired prediction accuracy below a specified threshold. Here, the quantum advantage is realized when the upper bound on the sample complexity of a quantum learning algorithm for a given task is  \textit{lower than} the lower bound of all classical learning algorithms. While low sample complexity is a necessary condition for efficient learning, it does not guarantee practical efficiency alone; for example, identifying useful training examples within a small sample size may still require substantial computational time.

\begin{tcolorbox}[colback=gray!5!white,colframe=gray!75!black,title=Remark]
(Difference of sample complexity in classical and quantum ML). In classical ML, sample complexity typically refers to the number of training examples required for a model to generalize effectively, such as the number of labeled images needed to train an image classifier. In quantum ML, however, sample complexity can take on varied meanings depending on the context, as shown below.  
\begin{itemize}
  \item Quantum state tomography (see Chapter~\ref{chapter2:Sec-2.3.2-readout}). Here the sample complexity refers to the number of measurements required to accurately reconstruct the quantum state of a system. 
  \item Evaluation of the generalization ability of quantum neural networks (see Chapter~\ref{chapt5:sec:qnn_theory}). Here the sample complexity refers to the number of input-output pairs needed to train the network to approximate a target function, similar to classical ML.
  \item Quantum system learning. Here the sample complexity often refers to the number of queries to interact with the target quantum system, such as the number of times a system must be probed to learn its Hamiltonian dynamics. 
\end{itemize}
\end{tcolorbox}

In addition to sample complexity, another commonly used measure in quantum learning theory is \textit{quantum query complexity}, particularly within the frameworks of quantum statistical learning and quantum exact learning. As these frameworks are not the primary focus of this tutorial, interested readers are referred to \citep{anshu2024survey} for a more detailed discussion.

Quantum advantage can be pursued through \textit{two main approaches}. The first involves identifying problems with quantum circuits that demonstrate provable advantages over classical counterparts in the aforementioned measures \citep{harrow2017quantum}. Such findings deepen our understanding of quantum computing's potential and expand its range of applications. However, these quantum circuits often require substantial quantum resources, which are currently beyond the reach of near-term quantum computers. Additionally, for many tasks, analytically determining the upper bound of classical algorithm complexities is challenging. 

These challenges have motivated a second approach: demonstrating that current quantum devices can perform accurate computations on a scale that exceeds brute-force classical simulations—a milestone known as ``quantum utility.'' Quantum utility refers to quantum computations that yield reliable, accurate solutions to problems beyond the reach of brute-force classical methods and otherwise accessible only through classical approximation techniques \citep{kim2023evidence}. This approach represents a step toward practical computational advantage with noise-limited quantum circuits. Reaching the era of quantum utility signifies that quantum computers have attained a level of scale and reliability enabling researchers to use them as effective tools for scientific exploration, potentially leading to groundbreaking new insights.

\subsection{Explored tasks in quantum machine learning}\label{chapter1:explored-task-QML}

\begin{figure}
\centering
\includegraphics[width=0.9\textwidth]{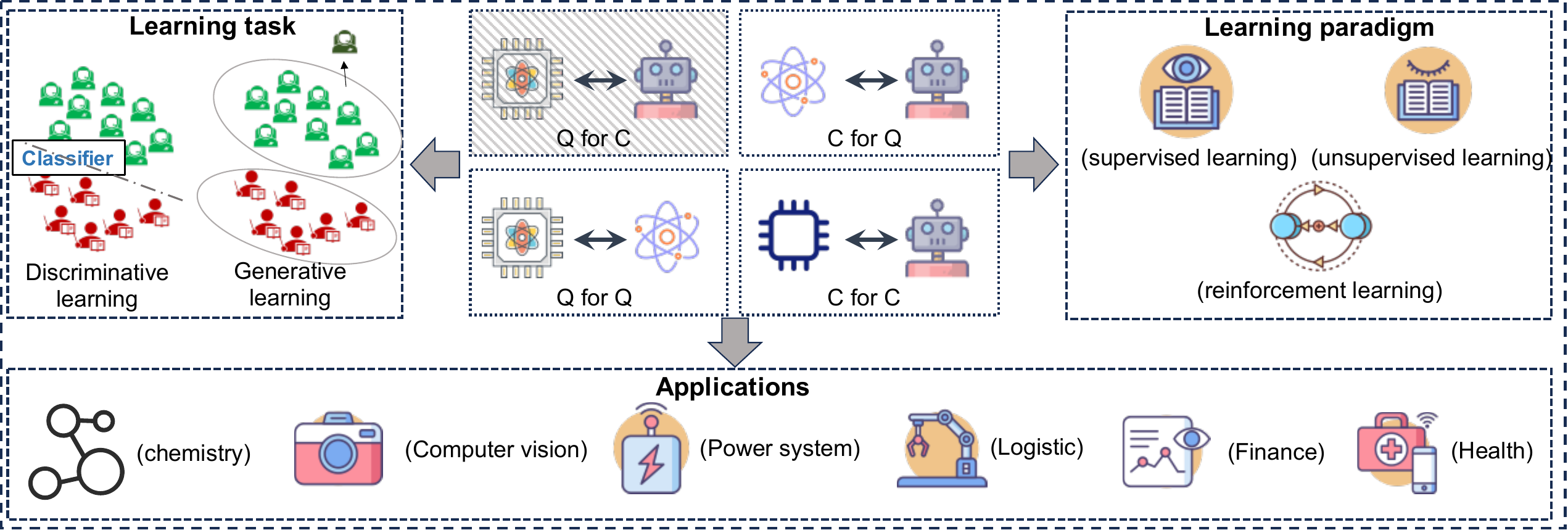}
  \caption{\textbf{Different research directions in QML}. QML can be categorized into four types based on the interplay of quantum (\textbf{\textsf{Q}}) and classical (\textbf{\textsf{C}}) systems: \textbf{\textsf{Q}} for \textbf{\textsf{Q}} (quantum algorithms for quantum data), \textbf{\textsf{Q}} for \textbf{\textsf{C}} (quantum algorithms for classical data), \textbf{\textsf{C}} for \textbf{\textsf{Q}} (classical algorithms for quantum data), and \textbf{\textsf{C}} for \textbf{\textsf{C}} (classical algorithms for classical data). This tutorial primarily focuses on the \textbf{\textsf{Q}} for \textbf{\textsf{C}} category. Beyond the role of the learner and system, QML can also be classified by learning tasks (discriminative and generative learning), learning paradigms (supervised, unsupervised, and reinforcement learning), and diverse applications such as chemistry, computer vision, power systems, logistics, finance, and healthcare.}
  \label{fig:chapt1-tasks-QML}
\end{figure}

What are the main areas of focus in QML? QML research is extensive and can be broadly divided into four primary sectors, each defined by the nature of the computing resources (whether the computing device is quantum (\textbf{\textsf{Q}}) or classical (\textbf{\textsf{C}})) and the type of data involved (whether generated by a quantum (\textbf{\textsf{Q}}) or classical (\textbf{\textsf{C}}) system). The explanations of these four sectors are as follows.

\begin{itemize}
  \item[] \textbf{\textsf{CC} Sector}. The \textbf{\textsf{CC}} sector refers to classical data processed on classical systems, representing traditional machine learning. Here, classical ML algorithms run on classical processors (e.g., CPUs and GPUs) and are applied to classical datasets. A typical example is using neural networks to classify images of cats and dogs.
  \item[] \textbf{\textsf{CQ} Sector}: The \textbf{\textsf{CQ}} sector involves using classical ML algorithms on classical processors to analyze quantum data collected from quantum systems. Typical examples include applying classical neural networks to classify quantum states, estimating properties of quantum systems from measurement data, and employing classical regression models to predict outcomes of quantum experiments.
  \item[] \textbf{\textsf{QC} Sector}. The \textbf{\textsf{QC}} sector involves developing QML algorithms that run on quantum processors (QPUs) to process classical data. In this context, quantum computing resources are leveraged to enhance or accelerate the analysis of classical datasets. Typical examples include applying QML algorithms, such as quantum neural networks and quantum kernels, to improve pattern recognition in image analysis.
  \item[] \textbf{\textsf{QQ} Sector}. The \textbf{\textsf{QQ}} sector involves developing QML algorithms executed on QPUs to process quantum data. In this context, quantum computing resources are leveraged to reduce the computational cost of analyzing and understanding complex quantum systems. Typical examples include using quantum neural networks for quantum state classification and applying quantum-enhanced algorithms to simulate quantum many-body systems.
\end{itemize}

The classification above is not exhaustive. As illustrated in Figure~\ref{fig:chapt1-tasks-QML}, each sector can be further subdivided based on various learning paradigms, such as discriminative vs. generative learning or supervised, unsupervised, and semi-supervised learning. Additionally, each sector can be further categorized according to different application domains, such as finance, healthcare, and logistics.

\begin{tcolorbox}[colback=gray!5!white,colframe=gray!75!black,title=Remark]
The primary focus of this tutorial is on the \textbf{\textsf{QC}} and \textbf{\textsf{QQ}} sectors. For more details on \textbf{\textsf{CQ}}, interested readers can refer to \citep{schuld2015introduction,dunjko2018machine,carleo2019machine}.
 \end{tcolorbox}

\section{Progress of Quantum Machine Learning} 

Huge efforts have been made to the \textbf{\textsf{QC}} and \textbf{\textsf{QQ}} sectors to determine which tasks and conditions allow QML to offer computational advantages over classical machine learning. In this regard, to provide a clearer understanding of QML's progress, it is essential to first review recent advancements in quantum computers, the foundational substrate for quantum algorithms.

\subsection{Progress of quantum computers}

The novelty and inherent challenges of utilizing quantum physics for computation have driven the development of various computational architectures, giving rise to the formalized concept of circuit-based quantum computers, as discussed in Chapter~\ref{chapt1-sec1-1-1}. In pursuit of this goal, numerous companies and organizations are striving to establish their architecture as the leading approach and to be the first to demonstrate practical utility or quantum advantage on a large-scale quantum device. 

Common architectures currently include superconducting qubits (employed by IBM and Google), ion-trap systems (pioneered by IonQ), and Rydberg atom systems (developed by QuEra), each offering distinct advantages \citep{cheng2023noisy}. Specifically, superconducting qubits excel in scalability and fast gate operations \citep{huang2020superconducting}, while ion-trap systems are known for their high coherence times, precise control over individual qubits, and full connectivity of all qubits \citep{bruzewicz2019trapped}.  Moreover, Rydberg atom systems enable flexible qubit connectivity through highly controllable interactions \citep{morgado2021quantum}. Besides these architectures, integrated photonic quantum computers are emerging as promising alternatives for robust and scalable quantum computation.

\begin{figure}
\centering
\includegraphics[width=0.98\textwidth]{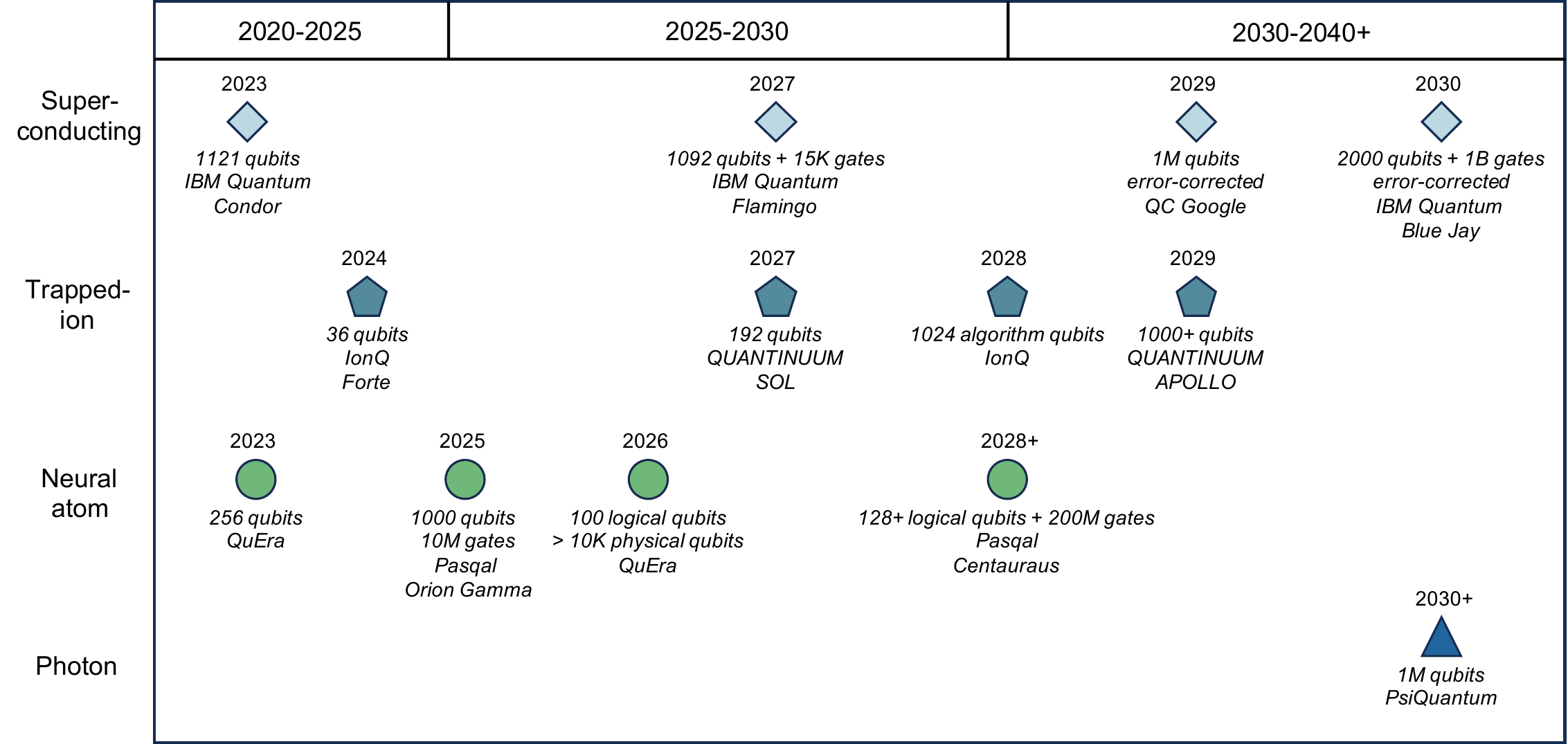}
  \caption{\textbf{Common quantum architectures and roadmaps from different quantum companies.}}
\end{figure}

Despite recent advances, today’s quantum computers remain highly sensitive to environmental noise and prone to quantum decoherence, lacking the stability needed for fault-tolerant operation. This results in qubits, quantum gates, and quantum measurements that are inherently imperfect, introducing errors that can lead to incorrect outputs. To capture this stage in quantum computing, John Preskill coined the term ``noisy intermediate-scale quantum'' (\textsf{NISQ}) era \citep{preskill2018quantum}, which describes the current generation of quantum processors. These processors feature up to thousands of qubits, but their capabilities are restricted with error-prone gates and limited coherence times.

In the \textsf{NISQ} era, notable achievements have been made alongside new challenges. Industrial and academic teams, such as those at Google and USTC, have demonstrated quantum advantages on specific sampling tasks, where the noisy quantum computers they fabricated outperform classical computers in computational efficiency \citep{arute2019quantum,wu2021strong}. However, most quantum algorithms that theoretically offer substantial runtime speedups depend on fault-tolerant, error-free quantum systems—capabilities that remain beyond the reach of current technology.

At this pivotal stage, the path forward in quantum computing calls for progress on both hardware and algorithmic fronts. 

On the hardware side, it is essential to continuously improve qubit count, coherence times, gate fidelities, and the accuracy of quantum measurements across various quantum architectures. Once the number and quality of qubits surpass certain thresholds, \textit{quantum error correction} codes can be implemented \citep{nielsen2010quantum}, paving the way for fault-tolerant quantum computing (\textsf{FTQC}). Broadly, quantum error correction uses redundancy and entanglement to detect and correct errors without directly measuring the quantum state, thus preserving coherence. Advancements in quantum processors will enable a progression from the \textsf{NISQ} era to the early \textsf{FTQC} era, ultimately reaching the fully \textsf{FTQC} era.

On the algorithmic side, two key questions must be addressed: 
\begin{itemize}
  \item [(Q1)] How can \textsf{NISQ} devices be utilized to perform meaningful computations with practical utility?
  \item [(Q2)] What types of quantum algorithms can be executed on early fault-tolerant and fully fault-tolerant quantum computers to realize the potential of quantum computing in real-world applications? 
\end{itemize}
Progress on either question could have broad implications. A positive answer to (Q1) would suggest that \textsf{NISQ} quantum computers have immediate practical applicability, while advancements in (Q2) would expand the scope and impact of quantum computing as more robust, fault-tolerant systems become feasible. In the following two sections, we examine recent progress in QML concerning these two questions.

\subsection{Progress of quantum machine learning under FTQC}\label{chapt1-sec-progress-FTQC}

A key milestone in \textsf{FTQC}-based QML algorithms is the quantum linear equations solver introduced by \citet{harrow2009quantum}. Many machine learning models rely on solving linear equations, a computationally intensive task that often dominates the overall runtime due to the polynomial scaling of complexity with matrix size. The HHL algorithm provides a breakthrough by reducing runtime complexity to poly-logarithmic scaling with matrix size, given that the matrix is well-conditioned and sparse. This advancement is highly significant for AI, where datasets frequently reach sizes in the millions or even billions.

The exponential runtime speedup achieved by the HHL algorithm has garnered significant attention from the research community, highlighting the potential of quantum computing in AI. Following this milestone, a body of work has emerged that employs the quantum matrix inversion techniques developed in HHL (or its variants) as subroutines in the design of various \textsf{FTQC}-based QML algorithms, offering runtime speedups over their classical counterparts \citep{montanaro2016quantum,dalzell2023quantum}. Notable examples include quantum principal component analysis \citep{lloyd2014quantum} and quantum support vector machines \citep{rebentrost2014quantum}.

Another milestone in \textsf{FTQC}-based QML algorithms is the quantum singular value transformation (QSVT), proposed by \citet{gilyenquantum2019}. QSVT enables polynomial transformations of the singular values of a linear operator embedded within a unitary matrix, offering a unifying framework for various quantum algorithms. It has connected and enhanced a broad range of quantum techniques, including amplitude amplification, quantum linear system solvers, and quantum simulation methods. Compared to the HHL algorithm for solving linear equations, QSVT provides improved scaling factors, making it a more efficient tool for addressing these problems in the context of QML.

In addition to advancements in linear equation solving, another promising line of research in \textsf{FTQC}-based QML focuses on leveraging quantum computing to enhance deep neural networks (DNNs) rather than traditional machine learning models. This research track has two main areas of focus. The first is the acceleration of DNN optimization, with notable examples including the development of efficient quantum algorithms for dissipative differential equations to expedite (stochastic) gradient descent, as well as Quantum Langevin dynamics for optimization \citep{chen2023quantum,liu2024towards}. The second area centers on advancing Transformers using quantum computing. In Chapter~\ref{Chapter5:Transformer}, we will discuss in detail how quantum computing can be employed to accelerate Transformers during the inference stage.

\begin{tcolorbox}[enhanced, 
    breakable,colback=gray!5!white,colframe=gray!75!black,title=Remark]
However, there are several critical caveats of the HHL-based QML algorithms. First, the assumption of efficiently preparing the quantum states corresponding to classical data runtime is very strong and may be impractical in the dense setting. Second, the obtained result $\bm{x}$ is still in the quantum form $\ket{\bm{x}}$. Note that extracting one entry of $\ket{\bm{x}}$ into the classical form requires $O(\sqrt{N})$ runtime, which collapses the claimed exponential speedups. The above two issues amount to the read-in and read-out bottlenecks in QML \citep{aaronson2015read}. The last caveat is that the employed strong quantum input model such as quantum random access memory (QRAM) \citep{giovannetti2008quantum} leads to an inconclusive comparison. Through exploiting a classical analog of QRAM as the input model, there exist efficient classical algorithms to solve recommendation systems in poly-logarithmic time in the size of input data.  
\end{tcolorbox}

\subsection{Progress of quantum machine learning under NISQ}\label{chapt1-sec-progress-NISQ}

The work conducted by \citet{havlivcek2019supervised} marked a pivotal moment for QML in the \textsf{NISQ} era. This study demonstrated the implementation of quantum kernel methods and quantum neural networks (QNNs) on a 5-qubit superconducting quantum computer, highlighting potential quantum advantages from the perspective of complexity theory. Unlike the aforementioned \textsf{FTQC} algorithms, quantum kernel methods and QNNs are flexible and can be effectively adapted to the limited quantum resources available in the \textsf{NISQ} era. These demonstrations, along with advancements in quantum hardware, sparked significant interest in exploring QML applications using \textsf{NISQ} quantum devices. We will delve into quantum kernel methods and QNNs in Chapter~\ref{Chapter3:kernel} and Chapter~\ref{cha5:qnn}, respectively.

\begin{tcolorbox}[enhanced, breakable,colback=blue!5!white,colframe=blue!75!black,title=Quantum neural networks (informal)]
A quantum neural network (QNN) is a hybrid model that leverages quantum computers to implement trainable models similar to classical neural networks, while using classical optimizers to complete the training process.
\end{tcolorbox}

As shown in Figure~\ref{fig:chapt1-DNN-vs-QNN}, the mechanisms of QNNs and deep neural networks (DNNs) are almost the same, whereas the only difference is the way of implementing the trainable model. This difference gives the potential of quantum learning models to solve complex problems beyond the reach of classical neural networks, opening new frontiers in many fields. Roughly speaking, research in QNNs and quantum kernel methods has primarily focused on three key areas: (I) \textit{quantum learning models and applications}, (II) \textit{the adaptation of advanced AI topics to QML}, and (III) \textit{theoretical foundations of quantum learning models}. A brief overview of each category is provided below.

\begin{figure}
\centering
\includegraphics[width=0.99\textwidth]{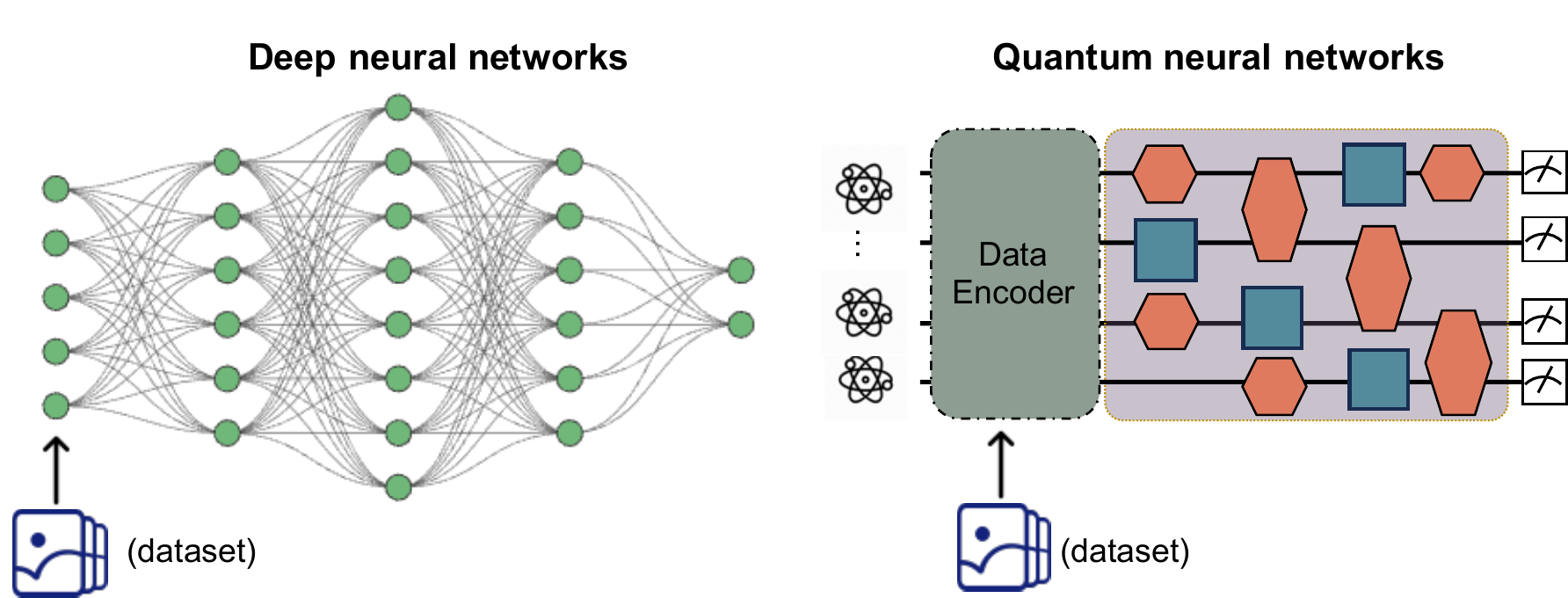}
\caption{\textbf{Mechanisms of DNNs and QNNs}. Both DNNs and QNNs follow an iterative approach. At each iteration, they take input data, process it through multiple layers, and produce an output prediction. The key difference between DNNs and QNNs is the way of implementing their learning models. }
  \label{fig:chapt1-DNN-vs-QNN}
\end{figure}

\smallskip
(I) \textsc{Quantum learning models and applications}. This category focuses on implementing various DNNs on \textsf{NISQ} quantum computers to tackle a wide range of tasks.
 
From a model architecture perspective, quantum analogs of popular classical machine learning models have been developed, including quantum versions of multilayer perceptrons (MLPs), autoencoders, convolutional neural networks (CNNs), recurrent neural networks (RNNs), extreme learning machines, generative adversarial networks (GANs), diffusion models, and Transformers. Some of these QNN structures have even been validated on real quantum platforms, demonstrating the feasibility of applying quantum algorithms to tasks traditionally dominated by classical deep learning \citep{cerezo2021variational,li2022recent,tian2023recent}.

From an application perspective, QML models implemented on \textsf{NISQ} devices have been explored across diverse fields, including fundamental science, image classification, image generation, financial time series prediction, combinatorial optimization, healthcare, logistics, and recommendation systems. These applications demonstrate the broad potential of QML in the \textsf{NISQ} era, though achieving full quantum advantage in these areas remains an ongoing challenge \citep{bharti2022noisy,cerezo2022challenges}.

\smallskip
(II) \textsc{Adaptation of advanced AI topics to QML}. Beyond model design, advanced topics from AI have been extended to QML, aiming to enhance the performance and robustness of different QML models. Examples include quantum architecture search \citep{du2022quantum} (the quantum equivalent of neural architecture search), advanced optimization techniques \citep{stokes2020quantum}, and pruning methods to reduce the complexity of quantum models \citep{sim2021adaptive,wang2023symmetric}. Other areas of active research include adversarial learning \citep{Lu2020Quantum}, continual learning \citep{Jiang_2022Quantum}, differential privacy \citep{du2021quantum,Watkins2023Quantum}, distributed learning \citep{du2022distributed}, federated learning \citep{ren2023towards}, and interpretability within the context of QML \citep{pira2024interpretability}. These techniques have the potential to significantly improve the efficiency and effectiveness of QML models, addressing some of the current limitations of \textsf{NISQ} devices.

\smallskip
(III) \textsc{Theoretical foundations}.  Quantum learning theory \citep{banchi2023statistical} has garnered increasing attention, aiming to compare the capabilities of different QML models and to identify the theoretical advantages of QML over classical machine learning models. As shown in Figure~\ref{fig:chapter1-QLT-overview}, the learnability of QML models can be evaluated across three key dimensions: expressivity, trainability, and generalization capabilities. Below, we provide a brief overview of each measure.

\begin{figure}
\centering
\includegraphics[width=0.8\textwidth]{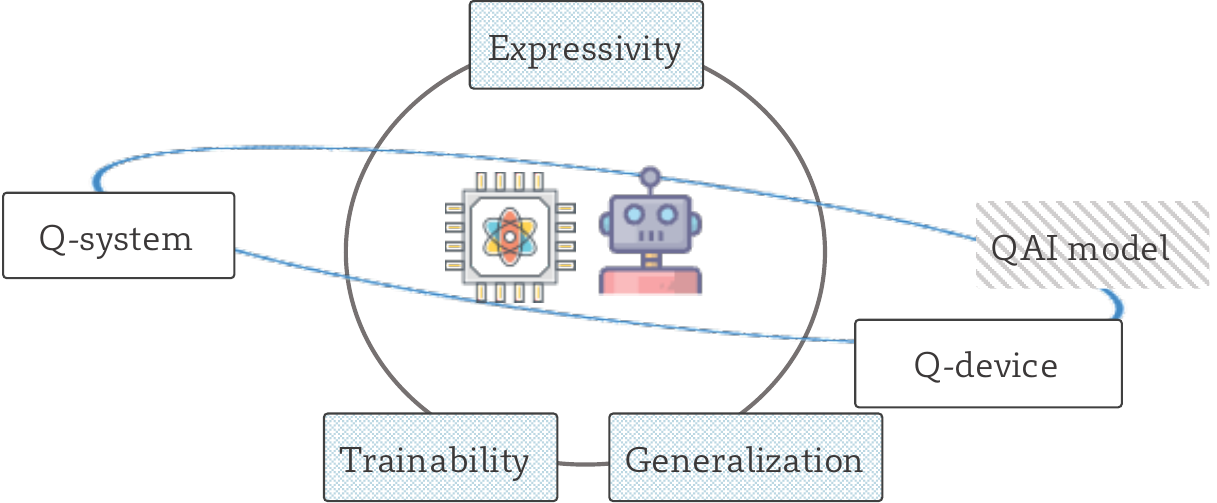}
  \caption{\textbf{The learnability of quantum machine learning models}.}
  \label{fig:chapter1-QLT-overview}
\end{figure}

\begin{itemize}
\item Trainability. This area examines how the design of QNNs influences their convergence properties, including the impact of system noise and measurement errors on the ability to converge to local or global minima.
  \item Expressivity. Researchers investigate how the number of parameters and the structure of QNNs affect the size of the hypothesis space they can represent. A central question is whether QNNs and quantum kernels can efficiently represent functions or patterns that classical neural networks cannot, thereby offering potential quantum advantage.
  \item Generalization. This focuses on understanding how the gap between training and test error evolves with the size of the dataset, the structure of QNNs or quantum kernels, and the number of parameters. The goal is to determine whether QML models can generalize more effectively than classical models, particularly in the presence of noisy data or when training data is limited.
\end{itemize}

The combination of advancements in model design, application domains, and theoretical understanding is driving the progress of QML in the NISQ era. Although the field is still in its early stages, the progress achieved thus far provides promising insights into the potential of quantum computing to enhance conventional AI. As quantum hardware continues to evolve, further breakthroughs are expected, potentially unlocking new possibilities for practical QML applications.

\begin{tcolorbox}[enhanced, 
    breakable,colback=gray!5!white,colframe=gray!75!black,title=Remark]
It is important to note that QNNs and quantum kernel methods can also be considered \textsf{FTQC} algorithms when executed on fully fault-tolerant quantum computers. The reason these algorithms are discussed in the context of \textsf{NISQ} devices is their flexibility and robustness, making them well-suited to the limitations of current quantum hardware.
\end{tcolorbox}

\subsection{A brief review of quantum machine learning}
Unlike quantum hardware, where the number of qubits has rapidly scaled from zero to thousands, the development of QML algorithms—and quantum algorithms more broadly—has taken an inverse trajectory, transitioning from \textsf{FTQC} to \textsf{NISQ} devices. This shift reflects the move from idealized theoretical frameworks to practical implementations. The convergence of quantum hardware and QML algorithms, where the quantum resources required by these algorithms become attainable on real quantum computers, enables researchers to experimentally evaluate the power and limitations of various quantum algorithms.

Based on the minimum quantum resources required to complete learning tasks, we distinguish between \textsf{FTQC} algorithms, discussed in Chapter~\ref{chapt1-sec-progress-FTQC}, and \textsf{NISQ} algorithms, including QNNs and quantum kernel methods, in Chapter~\ref{chapt1-sec-progress-NISQ}. \textsf{FTQC}-based QML algorithms necessitate error-corrected quantum computers with tens of billions of qubits—an achievement that remains far from realization. In contrast, QNNs and quantum kernels are more flexible and can be executed on both \textsf{NISQ} and \textsf{FTQC} devices, depending on the available resources.

As quantum hardware continues to progress, the development of QML algorithms must evolve in tandem. A promising direction is to integrate \textsf{FTQC} algorithms with QNNs and quantum kernel methods, creating new QML algorithms that can be run on current quantum processors while offering enhanced quantum advantages across various tasks.

\section{Organization of This Tutorial}

To encourage and enable computer scientists to engage with the rapidly growing field of quantum AI, we provide this hands-on tutorial that revisits QML algorithms from a \textit{computer science perspective}. With this aim, the tutorial is designed to balance theory, practical implementations, and applications, making it suitable for both researchers and practitioners with a background in classical machine learning. The tutorial is divided into the following chapters:

Chapter \ref{Chapter2:basics-QC}: \textsc{Basics of Quantum Computing}. Before delving into QML, this chapter lays the groundwork by introducing the fundamental concepts of quantum computing. It covers the transition from classical bits to quantum bits, explains quantum circuit models, illustrates how quantum systems interface with classical systems through quantum read-in and read-out mechanisms, and presents some fundamental concepts of quantum linear algebra. By the end of this chapter, you will understand that a solid grasp of linear algebra is all you need to comprehend the basics of quantum computing.

 Chapters \ref{Chapter3:kernel}, \ref{cha5:qnn} and \ref{Chapter5:Transformer}: \textsc{Classical ML Models Extended to Quantum Frameworks}. Each of these chapters follows a consistent structure, starting with a review of the classical model and progressing to its quantum extension—Quantum kernel methods in Chapter~\ref{Chapter3:kernel},  Quantum neural networks in Chapter~\ref{cha5:qnn}, and Quantum Transformers in Chapter~\ref{Chapter5:Transformer}. This unified structure enables readers to clearly understand how classical machine learning models can be translated into quantum implementations and how quantum computers may offer computational advantages.

Appendix. The Appendix serves as a supplementary resource, providing a summary of notations and essential mathematical tools that are omitted from the main text for brevity. In particular, it includes basic introduction of  concentration inequalities, the Haar measure, and other foundational concepts relevant to the tutorial. 

To provide a clear and comprehensive learning experience, each chapter is composed of the following parts:
\begin{enumerate}
  \item Classical foundations and quantum model construction. Each chapter begins with a review of the classical version of the model, ensuring that readers are well-acquainted with the foundational concepts before exploring their quantum adaptations. After this review, we introduce quantum versions of the models, focusing on implementations based on \textsf{NISQ}, \textsf{FTQC}, or both.
  \item  Theoretical analysis. There is nothing more practical than a good theory. In this tutorial, each chapter provides a theoretical analysis of the learnability of QML models, focusing on key aspects such as expressivity, trainability, and generalization capabilities. 

To ensure a balance between depth and self-consistency, this tutorial provides proof for the most significant theoretical results, as highlighted by  \textbf{Theorems and Lemmas}. For results that are less central to the main content of this tutorial, we present them as \textbf{Facts} and include appropriate references, allowing readers to easily locate the complete proofs if desired.

  \item Code implementation. ``Talk is cheap, show me the code.'' To provide a practical, hands-on learning experience, each chapter includes code implementations using real-world datasets (to be specified). This section walks readers through the process of implementing quantum models on simulated or real quantum hardware. All numerical examples illustrated in this tutorial are available in \url{https://qml-tutorial.github.io/}, accompanied by Jupyter Notebooks. 
    
    Instead of building everything from scratch, we leverage the well-established PennyLane library for implementation \citep{bergholm2018pennylane}. This choice does not imply any specific preference. Other quantum computing libraries, such as Qiskit \citep{javadi2024quantum}, Cirq \citep{cirq_developers_2024_11398048}, and TensorFlow Quantum \citep{broughton2020tensorflow}, can also be used, offering similar capabilities and flexibility.

\item Frontier topics and future directions. Each chapter concludes with an exploration of cutting-edge topics and emerging challenges in the field. This part highlights open research problems, ongoing developments, and potential future directions for the quantum versions of each model, providing insights into where the field may be headed.
\end{enumerate}

\chapter{Basics of Quantum Computing}
\label{Chapter2:basics-QC}

In this chapter, we introduce the fundamental concepts of quantum computation, such as quantum states, quantum circuits, and quantum measurements, along with key topics in quantum machine learning, including quantum read-in, quantum read-out, and quantum linear algebra. These foundational elements are essential for understanding quantum machine learning algorithms and will be repeatedly referenced throughout the subsequent chapters. 

This chapter is organized as follows: Chapter~\ref{cha3:sec:qubit} introduces quantum bits and their mathematical representations; Chapter~\ref{cha3:sec:circuit} covers quantum circuits, including quantum gates, quantum channels, and quantum measurements; Chapter~\ref{cha3:sec:in-out} discusses how to encode classical data into quantum systems and extract classical information from quantum states; Chapter~\ref{chap2:preliminary-sec:linearAlgebra} delves into quantum linear algebra; Chapter~\ref{Chapter2:preliminary-code} provides practical coding exercises to reinforce these concepts; and finally, Chapter~\ref{chap-preliminary-sec:Remark} presents recent advancements in efficient quantum read-in and read-out techniques for further exploration.

\section{From Classical Bits to Quantum Bits}\label{cha3:sec:qubit}

In this section, we define quantum bits (qubits) and present the mathematical tools used to describe quantum states. We begin by discussing classical bits and then transition to their quantum counterparts. We recommend interested readers consult the textbook \citep{nielsen2010quantum}  for the detailed explanations. 
   
\subsection{Classical bits} 
In classical computing, a bit is the basic unit of information, which can exist in one of two distinct states: $0$ or $1$. Each bit holds a definite value at any given time. When multiple classical bits are used together, they can represent more complex information. For instance, a set of three bits can represent $2^3=8$ distinct states, ranging from $000$ to $111$.

\subsection{Quantum bits (Qubits)}\label{subsec:qubit}

Analogous to the role of `bit' in classical computation, the basic element in quantum computation is the quantum bit (\textit{qubit}). We start by introducing the representation of single-qubit states and then extend this to two-qubit and multi-qubit states.

\noindent \textbf{Single-qubit state}. A single-qubit state can be represented by a two-dimensional vector with unit length. Mathematically, a qubit state can be written as 
\begin{equation}
\bm{a} = 
\begin{bmatrix}
\bm{a}_1 \\
\bm{a}_2
\end{bmatrix}\in \mathbb{C}^2~,
\end{equation} 
where $|\bm{a}_1|^2+|\bm{a}_2|^2=1$ satisfies the \textit{normalization constraint}. Following conventions in quantum theory, we use Dirac notation to represent vectors \citep{nielsen2010quantum}, i.e., $\bm{a}$ is denoted by $\ket{\bm{a}}$ (named `\textit{ket}') with
\begin{equation}\label{eqn:sec-intro-1}
  \ket{\bm{a}} =\bm{a}_1\ket{0}+\bm{a}_2\ket{1}~,
\end{equation}
where $\ket{0}\equiv \bm{e}_0\equiv \begin{bmatrix}
  1 \\ 0
\end{bmatrix}$ and $\ket{1}\equiv \bm{e}_1\equiv \begin{bmatrix}
  0 \\ 1
\end{bmatrix}$ are two computational (unit) basis states. In this representation, the coefficients $\bm{a}_1$ and $\bm{a}_2$ are referred to as \textit{amplitudes}. The probabilities of obtaining the outcomes $0$ or $1$ upon measurement of the qubit are given by $|\bm{a}_1|^2$ and $|\bm{a}_2|^2$, respectively. The normalization constraint ensures that these probabilities always sum to one, as required by the probabilistic nature of quantum mechanics. In addition, the conjugated transpose of $\bm{a}$, i.e. $\bm{a}^{\dagger}$, is denoted by $\bra{\bm{a}}$ (named `\textit{bra}') with
\begin{equation}
  \bra{\bm{a}} = \bm{a}_1^{*}\bra{0}+\bm{a}_2^{*}\bra{1} \in \mathbb{C}^2~,
\end{equation}
where $\bra{0}\equiv \bm{e}_0^\top \equiv [1, 0]$, $\bra{1}\equiv \bm{e}_1^\top\equiv [0, 1]$, and the symbol `$\top$' denotes the transpose operation.

The physical interpretation of coefficients $\{\bm{a}_i\}$ is \textit{probability amplitudes}. Namely, when we intend to extract information from the qubit state $\ket{\bm{a}}$ into the classical form, quantum measurements are applied to this state, where the probability of sampling the basis $\ket{0}$ ($\ket{1}$) is $|\bm{a}_1|^2$ $(|\bm{a}_2|^2)$. Recall that the classical bit only permits the deterministic status with `$0$' or `$1$', while the qubit state in Eqn.~(\ref{eqn:sec-intro-1}) is the \textit{superposition} of the two status `$\ket{0}$' and `$\ket{1}$'.
    
\begin{tcolorbox}[enhanced, 
    breakable,colback=gray!5!white,colframe=gray!75!black,title=Remark]
The \textit{quantum superposition} leads to a distinct power between quantum and classical computation, where the former can accomplish certain tasks with provable advantages.
\end{tcolorbox}

\noindent \textbf{Two-qubit state}. The two qubits obey the tensor product rule, i.e.,
\begin{equation}
    \left[\begin{matrix}
   \bm{x}_1 \\ \bm{x}_2
\end{matrix} \right] \otimes \left[\begin{matrix}
   \bm{y}_1 \\ \bm{y}_2
\end{matrix} \right] 
= \left[\begin{matrix}
   \bm{x}_1 \left[\begin{matrix}
   \bm{y}_1 \\ \bm{y}_2
\end{matrix} \right] \\ \bm{x}_2 \left[\begin{matrix}
   \bm{y}_1 \\ \bm{y}_2
\end{matrix} \right]
\end{matrix} \right]    
= \left[\begin{matrix}
  \bm{x}_1 \bm{y}_1  \\ \bm{x}_1 \bm{y}_2  \\ \bm{x}_2 \bm{y}_1  \\ \bm{y}_2 \bm{y}_2 
\end{matrix} \right],
\end{equation}
which differs from the classical bits yielding the Cartesian product rule. 

For instance, let the first qubit follow Eqn.~(\ref{eqn:sec-intro-1}) and the second qubit state be $\ket{\bm{b}}=\bm{b}_1\ket{0}+\bm{b}_2\ket{1}$ with $|\bm{b}_1|^2+|\bm{b}_2|^2=1$. 
The two-qubit state formed by $\ket{\bm{a}}$ and  $\ket{\bm{b}}$  is defined as 
\begin{equation}\label{eqn:sec-intro-2}
  \ket{\bm{a}}\otimes \ket{\bm{b}} = \bm{a}_1\bm{b}_1 \ket{0}\otimes\ket{0} + \bm{a}_1\bm{b}_2 \ket{0}\otimes\ket{1}  + \bm{a}_2\bm{b}_1 \ket{1}\otimes\ket{0}  + \bm{a}_2\bm{b}_2 \ket{1}\otimes\ket{1} \in \mathbb{C}^4~,
\end{equation}
where the computational basis follows $\ket{0}\otimes\ket{0}\equiv \left[\begin{smallmatrix}
  1 \\ 0 \\ 0 \\ 0
\end{smallmatrix} \right]$, $\ket{0}\otimes\ket{1}\equiv \left[\begin{smallmatrix}
  0 \\ 1 \\ 0 \\ 0
\end{smallmatrix} \right]$ , $\ket{1}\otimes\ket{0}\equiv \left[\begin{smallmatrix}
  0 \\ 0 \\ 1 \\ 0
\end{smallmatrix} \right]$ , $\ket{1}\otimes\ket{1}\equiv \left[\begin{smallmatrix}
  0 \\ 0 \\ 0 \\ 1
\end{smallmatrix} \right]$, and the coefficients satisfy $\sum_{i=1}^2\sum_{j=1}^2 |\bm{a}_i\bm{b}_j|^2 = 1$. 

\begin{tcolorbox}[enhanced, 
    breakable,colback=gray!5!white,colframe=gray!75!black,title=Remark]
For ease of notations, the state $\ket{\bm{a}}\otimes \ket{\bm{b}}$ can be simplified as $\ket{\bm{a}\bm{b}}$, $\ket{\bm{a},\bm{b}}$, or $\ket{\bm{a}}\ket{\bm{b}}$. We will \textit{interchangeably} use these notations throughout the tutorial.
\end{tcolorbox}

\begin{shadedbox}
\begin{example}\label{example:chapt2-bell-state}
A typical example of a two-qubit state is the \textit{Bell state}, which represents a maximally entangled quantum state of two qubits. There are four types of Bell states, expressed as:
\begin{align}
  \ket{\phi^+} &= \frac{1}{\sqrt{2}} \left( \ket{00} + \ket{11} \right), \nonumber \\
  \ket{\phi^-} &= \frac{1}{\sqrt{2}} \left( \ket{00} - \ket{11} \right), \nonumber \\
  \ket{\psi^+} &= \frac{1}{\sqrt{2}} \left( \ket{01} + \ket{10} \right), \nonumber \\
  \ket{\psi^-} &= \frac{1}{\sqrt{2}} \left( \ket{01} - \ket{10} \right).
\end{align}
Each Bell state is a superposition of two computational basis states in the four-dimensional Hilbert space. 
\end{example}
\end{shadedbox}

\medskip
\noindent \textbf{Multi-qubit state}. We now generalize the above two-qubit case to the $N$-qubit case with $N>2$.  In particular, an $N$-qubit state $\ket{\psi}$ is a $2^N$-dimensional vector with
\begin{equation}\label{eqn:sec-intro-3}
  \ket{\psi} = \sum_{i=1}^{2^N} \bm{c}_i \ket{i}\in \mathbb{C}^{2^N} ~,
\end{equation} 
where the coefficients satisfy the normalization constraint $\sum_{i=1}^{2^N} |\bm{c}_i|^2 = 1$ and the symbol `$i$' of the computational basis $\ket{i}$ refers to a bit-string with $i\in \{0, 1\}^{N}$. As with the single-qubit case, the physical interpretation of coefficients $\{\bm{c}_i\}$ is probability amplitudes, where the probability to sample the bit-string `$i$' is $|\bm{c}_i|^2$. When the number of nonzero entries in $\bm{c}=[\bm{c}_1,...,\bm{c}_i,...,\bm{c}_{2^N}]^\top$ is larger than one, which implies that different bit-strings are coexisting coherently, the state $\ket{\psi}$ is called \textit{in superposition}.

\begin{tcolorbox}[enhanced, 
    breakable,colback=gray!5!white,colframe=gray!75!black,title=Remark]
 In quantum computing, a basis state $\ket{i}$ refers to a computational basis state in the Hilbert space of a quantum system. For an $N$-qubit system, the computational basis states are represented as $\ket{i}\in \{\ket{0\cdots 0},\ket{0\cdots 1},...,\ket{1\cdots 1}\}$, where $i$ is the binary representation of the state index. These states form an orthonormal basis of the $2^N$-dimensional Hilbert space, satisfying
    \begin{equation}
        \braket{i|j}=\delta_{ij}, \forall i,j\in [2^N].
    \end{equation}
    These basis states are fundamental for representing and analyzing quantum states, as any arbitrary quantum state can be expressed as a linear combination of these basis states.

Moreover, the size of $\bm{c}$ exponentially scales with the number of qubits $N$, attributed to the tensor product rule. This exponential dependence is an indispensable factor to achieve quantum supremacy \citep{arute2019quantum}, since it is extremely expensive and even intractable to record all information of $\bm{c}$ by classical devices for the modest number of qubits, e.g., $N>100$.
\end{tcolorbox}

\noindent \textbf{Entangled multi-qubit state}. A fundamental phenomenon in multi-qubit quantum systems is \textit{entanglement}, which represents a non-classical correlation between quantum systems that cannot be explained by classical physics. As proved by \citet{jozsa2003role}, quantum entanglement is an indispensable component to offer an exponential speed-up over classical computation. A representative example is Shor's algorithm, which utilizes entanglement to attain an exponential speed-up over any classical factoring algorithm. In an entangled quantum state, the state of one qubit cannot be fully described independently of the other qubits, even if they are spatially separated. The formal definition of entanglement for states in Dirac notation is as follows:

\begin{definition}(Entanglement for States in Dirac Notation)\label{def:entanglement-pure}
    An $N$-qubit state $\ket{\psi}\in \mathbb{C}^{2^N}$ is \textit{entangled} if it cannot be expressed as the tensor product of states of its subsystems $A$ and $B$:
    \begin{equation}
        \ket{\psi} \neq \ket{\psi_a} \otimes \ket{\psi_b}, \quad \forall \ket{\psi_a} \in \mathbb{C}^{2^{N_A}}, \ket{\psi_b} \in \mathbb{C}^{2^{N_B}}, N_A+N_B=N.
    \end{equation}
    If the state can be expressed in this form, it is referred to as \textit{seperable}.
\end{definition}

\begin{shadedbox}
\begin{example}\label{example:ghz} (GHZ state). 
    A typical example of an entangled $N$-qubit state is the Greenberger-Horne-Zeilinger (GHZ) state \citep{greenberger1989going}, which is a generalization of the two-qubit Bell state (see Example~\ref{example:chapt2-bell-state}) to a maximally entangled $N$-qubit state. The general form of an \(N\)-qubit GHZ state is:
    \begin{align}
        \ket{\text{GHZ}_N} = \frac{1}{\sqrt{2}} \left( \ket{0}^{\otimes N} + \ket{1}^{\otimes N} \right).
    \end{align}
    For $N = 3$, the GHZ state is:
    \begin{equation}
        \ket{\text{GHZ}_3} = \frac{1}{\sqrt{2}} \left( \ket{000} + \ket{111} \right).
    \end{equation}
\end{example}
\end{shadedbox}
 A key property of the entangled states (e.g., Bell states and GHZ states) is that measuring one qubit determines the outcome of measuring the other qubit, reflecting their strong quantum correlation.

\subsection{Density matrix}\label{cha2:density_mat}
Another description of quantum states is through \textit{density matrix} or \textit{density operators}. The reason for establishing density operators instead of Dirac notations arises from the imperfection of physical systems. Specifically, Dirac notations introduced in Chapter~\ref{subsec:qubit} are used to describe `ideal' quantum states (i.e., \textit{pure states}), where the operated qubits are isolated from the environment. Alternatively, when the operated qubits interact with the environment unavoidably, the density operators are employed to describe the behavior of quantum states living in this open system. As such, density operators describe more general quantum states.  

Mathematically, an $N$-qubit density operator, denoted by $\rho\in\mathbb{C}^{2^N\times 2^N}$,  presents a mixture of $m$ quantum pure states $\ket{\psi_i}\in\mathbb{C}^{2^N}$ with probability $p_i\in [0,1]$ and $\sum_{i=1}^m p_i =1$, i.e.,
\begin{equation}
  \rho = \sum_{i=1}^m p_i\rho_i~,
\end{equation}
where $\rho_i =\ket{\psi_i}\bra{\psi_i}\in\mathbb{C}^{2^N\times 2^N}$ is the outer product of the pure state $\ket{\psi_i}$. The outer product of two vectors $\ket{u},\ket{v}\in\mathbb{C}^n$ is expressed as
\begin{equation} 
    \ket{u}\bra{v} = \begin{bmatrix} 
        u_1 \\
        u_2 \\
        \vdots \\ 
        u_n \end{bmatrix} \begin{bmatrix} v_1^* & v_2^* & \cdots & v_n^* \end{bmatrix} = \begin{bmatrix} u_1 v_1^* & u_1 v_2^* & \cdots & u_1 v_n^* \\
        u_2 v_1^* & u_2 v_2^* & \cdots & u_2 v_n^* \\
        \vdots & \vdots & \ddots & \vdots \\
        u_n v_1^* & u_n v_2^* & \cdots & u_n v_n^* \end{bmatrix},
\end{equation}
where $u_i$ and $v_i^*$ are the element of $\ket{u}$ and the conjugate transpose $\bra{v}$, respectively.

From the perspective of computer science, the density operator $\rho$ is just a \textit{positive semi-definite matrix} with trace-preserving, i.e., $\bm{0}\preceq \rho$ and $\Tr(\rho)=1$. 
\begin{definition}(Positive semi-definite matrix)\label{def:psd}
    A matrix $A \in \mathbb{C}^{n \times n}$ is positive semi-definite (PSD) if it satisfies the following conditions:
    \begin{enumerate}
        \item $A$ is Hermitian: $A = A^\dagger$
        \item For any nonzero vector $\ket{v} \in \mathbb{C}^n$, $\bra{v}A\ket{v} \geq 0$, where $\bra{v}A\ket{v}$ represents the quadratic form of $A$ with respect to $\ket{v}$.
    \end{enumerate}
\end{definition}

When $m=1$, the density operator $\rho$ amounts to a pure state with $\rho = \ket{\psi_1}\bra{\psi_1}$. When $m>1$, the density operator $\rho$ describes a `mixed' quantum state, where the rank of $\rho$ is larger than $1$. A simple criterion to discriminate the pure states with the mixed states is as follows: the pure state with $m=1$ yields $\Tr(\rho^n)=\Tr(\rho)=1$ for any $n\in \mathbb{N}_+$; the mixed state with $m>1$ satisfies $\Tr(\rho^n)<\Tr(\rho)=1$ for any $n\in \mathbb{N}_+\setminus \{1\}$. Similar to the Definition~\ref{def:entanglement-pure} for entanglement of pure states, we can define the entanglement of mixed states.

\begin{definition}(Entanglement for Mixed States)\label{def:entanglement-mixed}
    Let $\rho$ be a density operator acting on a composite Hilbert space $\mathcal{H}_A \otimes \mathcal{H}_B$. The state $\rho$ is said to be \textit{entangled} if it cannot be expressed as:
    \begin{equation}
    \rho = \sum_{i} p_i \, \rho_A^{(i)} \otimes \rho_B^{(i)},
    \end{equation}
    where $p_i \geq 0$, $\sum_{i} p_i = 1$, and $\rho_A^{(i)}$ and $\rho_B^{(i)}$ are density operators on $\mathcal{H}_A$ and $\mathcal{H}_B$, respectively. If $\rho$ can be written in this form, it is called \textit{separable}.
\end{definition}

\begin{shadedbox}
\begin{example}
(Density matrix representations).\\
(i). Consider the single-qubit pure state \(\ket{\psi} = \frac{1}{\sqrt{2}}(\ket{0} + \ket{1})\). The corresponding density operator is:
        \[
        \rho = \ket{\psi}\bra{\psi} = \frac{1}{2} \begin{bmatrix} 
        1 & 1 \\ 
        1 & 1 
        \end{bmatrix}.
        \]
        Here, \(\Tr(\rho^2) = \Tr(\rho) = 1\), confirming that it is a pure state.\\
(ii).  Consider the classical probabilistic mixture of \(\ket{0}\) and \(\ket{1}\), each with equal probability \(p = 0.5\). The density operator is:
        \[
        \rho = 0.5 \ket{0}\bra{0} + 0.5 \ket{1}\bra{1} = \frac{1}{2} \begin{bmatrix} 
        1 & 0 \\ 
        0 & 1 
        \end{bmatrix}.
        \]
        In this case, \(\Tr(\rho^2) = 0.5 < \Tr(\rho) = 1\), indicating it is a mixed state.
\end{example}
\end{shadedbox}

\section{From Digital Logical Circuit to Quantum Circuit Model}\label{cha3:sec:circuit}

To process quantum states, we need to introduce quantum computation, a fundamental model of which is the quantum circuit model. In this section, we will begin with classical computation in Chapter~\ref{cha2:sec2:classical} and transit to details about the quantum circuit model in Chapter~\ref{cha2:sec:qcir}, including quantum gates, quantum channel, and quantum measurements.
       
\subsection{Classical digital logical circuit}\label{cha2:sec2:classical}

Digital logic circuits are the foundational building blocks of classical computing systems. They process classical bits by performing logical operations through logic gates. In this subsection, we introduce the essential components of digital logic circuits and their functionality, followed by a discussion of how these classical circuits relate to quantum circuits.

\subsubsection{Logic gates}

Logic gates are the basic components of a digital circuit. They take binary inputs, represented as \(0\) or \(1\), and produce a binary output based on a predefined logical operation. The most common logic gates include:

\begin{enumerate}
    \item \textbf{NOT Gate}: This gate inverts the input bit, i.e., it produces \(1\) if the input is \(0\), and vice versa. Its truth table is shown in Table~\ref{tab:not};
    
    \begin{table}[h!]
        \centering 
        \caption{\textbf{ Input-output mapping of the NOT gate.}}
        \label{tab:not}
        \footnotesize
        \begin{tabular}{c|c}
        \toprule 
        \multicolumn{1}{c}{Input (A)} & \multicolumn{1}{|c}{Output (NOT A)} \\ \midrule
        $0$       & $1$ \\  
        $1$ & $0$ \\  
        \end{tabular}
    \end{table}

    \item \textbf{AND Gate}: Produces an output of \(1\) only if both input bits are \(1\); otherwise, it outputs \(0\). The truth table is shown in Table~\ref{tab:and};
    \begin{table}[h!]
        \centering 
        \caption{\textbf{Input-output mapping of the AND gate.}}
        \label{tab:and}
        \footnotesize
        \begin{tabular}{c|c|c}
        \toprule 
        \multicolumn{1}{c}{Input (A)} & \multicolumn{1}{|c}{Input (B)} & Output (A AND B) \\ \midrule
        $0$ & $0$ & $0$ \\  
        $0$ & $1$ & $0$ \\
        $1$ & $0$ & $0$ \\
        $1$ & $1$ & $1$ \\
        \end{tabular}
    \end{table}

    \item \textbf{OR Gate}: Outputs \(1\) if at least one input is \(1\). The truth table is shown in Table~\ref{tab:or};
    \begin{table}[h!]
        \centering 
         \caption{\textbf{Input-output mapping of the OR gate.}}
        \label{tab:or}
        \footnotesize
        \begin{tabular}{c|c|c}
        \toprule 
        \multicolumn{1}{c}{Input (A)} & \multicolumn{1}{|c}{Input (B)} & Output (A OR B) \\ \midrule
        $0$ & $0$ & $0$ \\  
        $0$ & $1$ & $1$ \\
        $1$ & $0$ & $1$ \\
        $1$ & $1$ & $1$ \\
        \end{tabular}
    \end{table}

    \item \textbf{XOR Gate}: Produces an output of \(1\) if the inputs are different, and \(0\) otherwise. The truth table is shown in Table~\ref{tab:xor}.
    \begin{table}[h!]
        \centering 
         \caption{\textbf{Input-output mapping of the XOR gate.}}
        \label{tab:xor}
        \footnotesize
        \begin{tabular}{c|c|c}
        \toprule 
        \multicolumn{1}{c}{Input (A)} & \multicolumn{1}{|c}{Input (B)} & Output (A XOR B) \\ \midrule
        $0$ & $0$ & $0$ \\  
        $0$ & $1$ & $1$ \\
        $1$ & $0$ & $1$ \\
        $1$ & $1$ & $0$ \\
        \end{tabular}
    \end{table}
\end{enumerate}

These logic gates can be combined in various configurations to build more complex circuits capable of performing arbitrary arithmetic operations.

\subsubsection{Circuit design and universality}

A classical digital logic circuit is composed of interconnected gates designed to perform specific tasks, such as addition or multiplication. A key property of these circuits is \textbf{universality}, meaning any logical function can be implemented using a finite set of gates. For example, the \textbf{NAND Gate} (NOT AND) and \textbf{NOR Gate} (NOT OR) are universal gates. Any other logical operation can be constructed using only NAND or NOR gates \citep{leach1994digital}.

\subsection{Quantum circuit}\label{cha2:sec:qcir}
 
Classical digital logical circuits provide the essential framework for understanding computation. While classical circuits operate on bits and perform deterministic operations, quantum circuits manipulate qubits and involve probabilistic behavior. The concepts of logic gates, circuit design, and universality lay the groundwork for transitioning to quantum circuits introduced in this subsection.

\subsubsection{Quantum gate}\label{cha2:quantum-gate}       
Recall that the computational toolkit for classical computers is logic gates, e.g., NOT, AND, OR, and XOR,  which are applied to the single bit or multiple bits to accomplish computation. Similarly, the computational toolkit for quantum computers (or quantum circuits) is \textbf{quantum gate}, which \textit{operates on qubits} introduced in Chapter~\ref{subsec:qubit} to complete the computation. In the following, we will introduce both single-qubit and multi-qubit gates.

\medskip

\begin{figure*}
  \centering
  \includegraphics[width=0.9\textwidth]{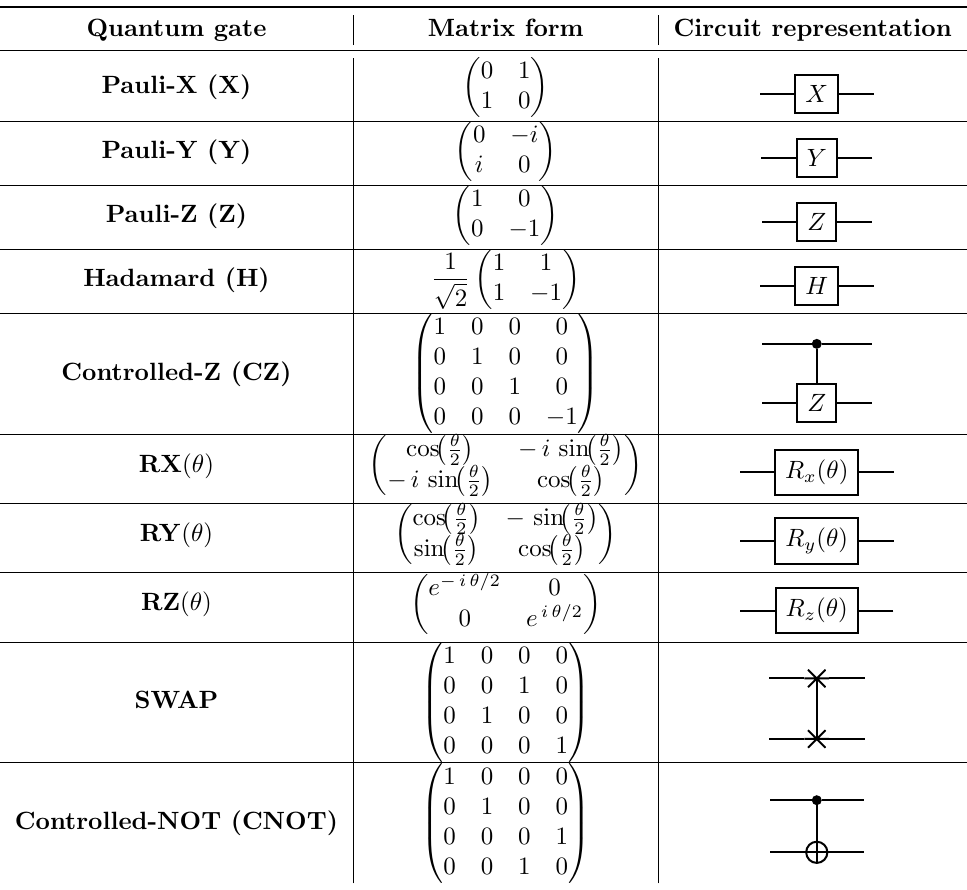}
  \caption{\textbf{The summarization of quantum gates.} The table contains the abbreviation, the mathematical form, and the graph representation of a set of universal quantum gates. $i$ represents the imaginary unit.}
\label{tab:Q-gates}
\end{figure*}

\noindent \textbf{Single-qubit gates}. Single-qubit gates control the evolution of the single-qubit state $\ket{\bm{a}}$. Due to the law of quantum mechanics, the evolved state should satisfy the normalization constraint. The implication of this constraint is that the evolution must be a unitary operation. Concretely, denoted $U\in \mathbb{C}^{2\times 2}$ as a linear operator and the evolved state as 
\begin{equation}\label{eqn:sec-intro-4}
  \ket{\hat{\bm{a}}}:=U\ket{\bm{a}}=\hat{\bm{a}}_1\ket{0}+\hat{\bm{a}}_2\ket{1}\in \mathbb{C}^{2}~,
\end{equation}
the summation of coefficients $|\hat{\bm{a}}_1|^2+|\hat{\bm{a}}_2|^2=\braket{\hat{\bm{a}}|\hat{\bm{a}}} = \braket{ \bm{a}| U^{\dagger} U |\bm{a} }$ is equal to $1$ if and only if $U$ is unitary with $U^{\dagger} U =  U U^{\dagger} = \mathbb{I}_2$. The symbol `$\dagger$' denotes the conjugate transpose operation. Under the density operator representation, the evolution of $\ket{\bm{a}}$  yields 
\begin{equation}\label{eqn:sec-intro-5}
  \hat{\rho} = U\rho U^{\dagger}~,
\end{equation}  
where $\hat{\rho} = \ket{\hat{\bm{a}}}  \bra{\hat{\bm{a}}}$ and  $\rho=\ket{\bm{a}}\bra{\bm{a}}$. 

Several common single-qubit gates, including Pauli-X, Pauli-Y, Pauli-Z, Hadamard, and rotational single-qubit gates about the X, Y, and Z axes ($\RX, \RY, \RZ$), are illustrated in Figure~\ref{tab:Q-gates}. According to Theorem~4.1 in \citep{nielsen2010quantum}, any unitary operation on a single qubit can be decomposed into a sequence of rotations as:
\begin{equation} U = \RZ(\alpha)\RY(\beta)\RZ(\gamma), \end{equation} where $\alpha, \beta, \gamma \in [0, 2\pi)$, up to a global phase shift.

\begin{figure}[h!]
\centering
\includegraphics[width=0.98\textwidth]{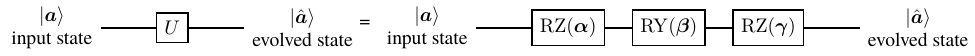}
\caption{\textbf{The evolution of the single-qubit state decomposed into the quantum gates.}}
\label{fig:single-qubit-evolve} 
\end{figure}

The evolution from $\ket{\bm{a}}$ to $\ket{\hat{\bm{a}}}$ can be visualized using a quantum circuit diagram, as illustrated in Figure \ref{fig:single-qubit-evolve}. The wire in the circuit represents a qubit, which evolves from the initial state $\ket{\bm{a}}$ on the left to the final state $\ket{\hat{\bm{a}}}$ on the right. Gates are applied sequentially from left to right along the wire.

\begin{tcolorbox}[enhanced, 
    breakable,colback=gray!5!white,colframe=gray!75!black,title=Remark] The circuit model serves as a foundational framework for describing quantum computation due to its \textit{intuitive} and \textit{modular} nature, making it accessible for researchers and practitioners transitioning from classical to quantum computing. First, the circuit model provides a standardized graphical language to represent complex quantum algorithms, enabling clear visualization of the computational flow and interactions among qubits. Second, the modularity of the circuit model allows quantum operations to be easily decomposed into a pre-defined gate set, ensuring compatibility across different quantum hardware architectures.
\end{tcolorbox}

\noindent \textbf{Multi-qubit gates}. The evolution of the $N$-qubit quantum state can be effectively generalized by the single-qubit case. That is, the unitary operator $U\in \mathbb{C}^{2^N\times 2^N}$ evolves an $N$-qubit state $\ket{\psi}$ in Eqn.~(\ref{eqn:sec-intro-3}) as
\begin{equation}\label{eqn:sec-intro-6}
\ket{\widehat{\psi}} = U\ket{\psi}\in \mathbb{C}^{2^N} ~. 
\end{equation} 
The evolution of $\ket{\psi}$ under the density operator representation is denoted by $\hat{\rho} = U\rho U^{\dagger}$, where $\hat{\rho} =\ket{\widehat{\psi}}\bra{\widehat{\psi}}$ and $\rho=\ket{\psi}\bra{\psi}$. 
 
\begin{tcolorbox}[enhanced, 
    breakable,colback=gray!5!white,colframe=gray!75!black,title=Remark] In the view of computer science, the quantum (logic) gates in Figure \ref{tab:Q-gates} are well-designed matrices with the following properties. First, all quantum gates are unitary (e.g., $\XGate  \XGate^{\dagger} = \mathbb{I}_2$). Second, $\XGate$, $\YGate$, $\ZGate$, $\Hada$ gates have the fixed form with size $2\times 2$; $\CNOT$, $\CZ$, and $\SWAP$ gates have the fixed form with size $4\times 4$. Third, $\RX(\theta)$, $\RY(\theta)$, $\RZ(\theta)$ gates are matrices controlled by a single variable $\theta$.       
\end{tcolorbox}

Figure~\ref{tab:Q-gates} includes two significant multi-qubit gates: the controlled-Z (CZ) gate and the controlled-NOT (CNOT) gate. For instance, the CNOT gate operates on two qubits: a \textit{control} qubit (top line) and a \textit{target} qubit (bottom line). If the control qubit is $0$, the target qubit remains unchanged; if the control qubit is $1$, the target qubit is flipped.

\begin{shadedbox}
    \begin{example}
        (State evolved by multi-qubit gates). Figure~\ref{fig:two-qubit-evolve} illustrates the evolution of a 3-qubit state $\ket{\psi}$ under a multi-qubit circuit consisting of multi-qubit gates. Each wire represents a qubit, and the evolution occurs from left to right. Starting with the initial state $\ket{\psi} = \ket{000}$, a Hadamard gate is applied to the first qubit, followed by two CNOT gates: one acting on the first and second qubits, and the other acting on the second and third qubits. The final evolved state, shown on the right, is the GHZ state introduced in Example~\ref{example:ghz}, i.e., $\ket{\widehat{\psi}} = U\ket{\psi} = \frac{1}{\sqrt{2}}(\ket{000} + \ket{111})$. The entire unitary operation can be represented as:
        \begin{equation}\label{eqn:sec-intro-7}
        U = (\Hada \otimes \mathbb{I}_4) (\CNOT \otimes \mathbb{I}_2) (\mathbb{I}_2 \otimes \CNOT).
        \end{equation}
    \end{example}    
\end{shadedbox}

\begin{figure}[h!]
    \centering
    \includegraphics[width=0.8\textwidth]{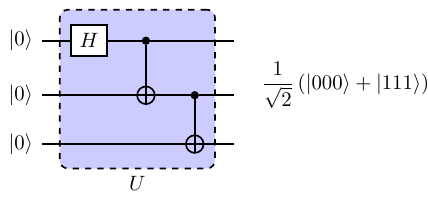}
    \caption{\textbf{The decomposition of the multi-qubit circuit $U$ in the case of $N=3$.}}
    \label{fig:two-qubit-evolve}  
\end{figure}

\begin{tcolorbox}[enhanced, 
    breakable,colback=gray!5!white,colframe=gray!75!black,title=Remark] The CNOT gate plays a pivotal role in quantum computing due to its unique ability to generate entangled states, such as the Bell states and GHZ states presented in Examples~\ref{example:chapt2-bell-state}\&\ref{example:ghz}. Besides, the CNOT gate is one of the most commonly implemented gates on quantum hardware. Its design and optimization directly impact the fidelity and scalability of quantum systems.
\end{tcolorbox}

\noindent\textbf{A universal quantum gate set}. While many single and multi-qubit gates exist, it is sufficient to use a \textit{universal set of gates} to construct any unitary operation. As proved in Chapter~4.5.2 of Ref.~\citep{nielsen2010quantum}, any unitary operator $U$ in Eqn.~(\ref{eqn:sec-intro-6}) can be decomposed into the single-qubit and two-qubit gates with a certain arrangement. 
\begin{fact}[Solovay-Kitaev theorem, \citep{dawson2005solovay}]
    Suppose we are given a fixed universal gate set $\mathcal{G}$, which generates a dense group ${\rm SU}(d)$. Then any unitary operator $U\in {\rm SU}(d)$ can be approximated to an arbitrary precision $\epsilon>0$ by a finite sequence of gates from $\mathcal{G}$. Formally, there exists a decomposition such that
    \begin{equation}
        \left\|U-\prod_{l=1}^{L} G_l\right\|_{op} \leq \epsilon, \quad G_l \in \mathcal{G}, \quad L \in \mathbb{N},
    \end{equation}
    where $\left\|\cdot\right\|_{op}$ is the operator norm which is the largest singular value of a matrix, and $L$ is the required number of gates that scales as:
    \begin{equation}
        L = O\left(\log^c(1/\epsilon)\right),
    \end{equation}
    with $ c \approx 4 $.
\end{fact}
A commonly used universal gate set includes single-qubit rotations \(\RX(\theta)\), \(\RY(\theta)\), \(\RZ(\theta)\), and two-qubit gates such as the \(\CNOT\) gate. As illustrated in Figure~\ref{fig:chap1:comparison-clc-quantum}, any ideal quantum computation can be represented by a unitary operator. This universal gate set provides a practical and foundational toolkit for implementing arbitrary quantum algorithms.

\subsubsection{Quantum channels}
Analogous to the unitary operation describing the evolution of quantum states in the closed system, the quantum channel formalizes the evolution of quantum states in the open system. Refer to the textbook \citep{wilde2011classical} for more details.

Mathematically, every quantum channel $\mathcal{N}(\cdot)$ can be treated as a linear, completely positive, and trace-preserving map (CPTP map). 
\begin{definition}[CPTP map]
  Denote $\mathcal{L}(\mathcal{H})$ as the space of square linear operators acting on the Hilbert space $\mathcal{H}$. We say $\mathcal{N}(\cdot)$ is a CPTP map if the following conditions are satisfied:
\begin{itemize}
  \item  The `linearity' requires for any $X_A, Y_A \in \mathcal{L}(\mathcal{H}_A)$ and $a,b\in\mathbb{C}$, $\mathcal{N}(a X_A + b Y_A)=a\mathcal{N}(X_A) + b \mathcal{N}(Y_A)$.
  \item The definition of completely positive is as follows.  A linear map $\mathcal{N}:\mathcal{L}(\mathcal{H}_A)\rightarrow \mathcal{L}(\mathcal{H}_B)$ is a positive map if $\mathcal{N}(X_A)$ is positive semi-definite for all positive semi-definite operators $X_A \in \mathcal{L}(\mathcal{H}_A)$. Moreover, a linear map $\mathcal{N}:\mathcal{L}(\mathcal{H}_A)\rightarrow \mathcal{L}(\mathcal{H}_B)$ is completely positive if $\mathbb{I}_R\otimes \mathcal{N}$ is a positive map for any size of $R$.
  \item  The trace preservation requires $\Tr(\mathcal{N}(X_A))=\Tr(X_A)$ for any $X_A  \in \mathcal{L}(\mathcal{H}_A)$.  
\end{itemize}  
\end{definition}

A quantum channel can be represented by the Choi-Kraus decomposition \citep{nielsen2010quantum}. Mathematically, let $\mathcal{L}(\mathcal{H}_A, \mathcal{H}_B)$ denote the space of linear operators taking $\mathcal{H}_A$ to $\mathcal{H}_B$. The Choi–Kraus decomposition of the quantum channel $\mathcal{N}(\cdot): \mathcal{L}(\mathcal{H}_A)\rightarrow \mathcal{L}(\mathcal{H}_B)$ is
\begin{equation}\label{eqn:def-kraus}
  \mathcal{N}(X_A)= \sum_{a=1}^d \mathbf{M}_a X_A \mathbf{M}_a^{\dagger}
\end{equation} 
where $X_A\in \mathcal{L}(\mathcal{H}_A)$, $M_a\in \mathcal{L}(\mathcal{H}_A, \mathcal{H}_B)$,$\sum_{a=1}^d \mathbf{M}_a^{\dagger}\mathbf{M}_a=\mathbb{I}_{\text{dim}(\mathcal{H}_A)}$, and $d\leq \text{dim}(\mathcal{H}_A)\text{dim}(\mathcal{H}_B)$. Here $\text{dim}(\mathcal{H}_*)$ refers to the dimension of the space $\mathcal{H}_*$.

We next introduce two common types of quantum channels, which will be broadly employed in the following context to simulate the noise of quantum devices. 

The first type is the \textit{depolarizing channel}, which considers the scenario such that the information of the input state can be entirely lost with some probability. 
\begin{definition}[Depolarization channel]\label{def:depolar-channel}
  Given an $N$-qubit quantum state $\rho\in\mathbb{C}^{2^N\times 2^N}$, the depolarization channel $\mathcal{N}_p$ acts on a $2^N$-dimensional Hilbert space follows 
  \begin{equation}
    \mathcal{N}_p(\rho) = (1-p)\rho + p\frac{\mathbb{I}_{2^N}}{2^N}~,
  \end{equation}
   where $\mathbb{I}_{2^N}/{2^N}$ refers to the maximally mixed state and $p$ is a scalar representing the depolarization rate.
\end{definition}

\begin{shadedbox}
\begin{example}(Single-qubit state with depolarization channel). 
    Consider a single-qubit pure state $\rho=\ket{0}\bra{0}$ with the density matrix
\begin{equation}
        \rho = \ket{0}\bra{0} = \begin{bmatrix}
            1 & 0 \\
            0 & 0
            \end{bmatrix}.
    \end{equation}
    When the depolarizing channel $\mathcal{N}_p$ acts on this state, the output is given by   \begin{equation}
        \mathcal{N}_p(\rho) = (1-p) \begin{bmatrix}
            1 & 0 \\
            0 & 0
            \end{bmatrix}
            + \frac{p}{2} \begin{bmatrix}
            1 & 0 \\
            0 & 1
            \end{bmatrix} = \begin{bmatrix}
                1 - \frac{p}{2} & 0 \\
                0 & \frac{p}{2}
                \end{bmatrix}.                        
    \end{equation}
    Therefore, the purity is inferred as:
    \begin{equation}
        \Tr(\mathcal{N}_p^2(\rho)) = 1 - p + \frac{p^2}{2}.
    \end{equation}
    When $p=0$, the state remains pure and unchanged. When $0<p\leq 1$, the state becomes a mixture of states $\ket{0}$ and $\ket{1}$ with $\Tr(\mathcal{N}_p^2(\rho))<1$. When $p=1$, the state evolves into the maximally mixed state.
\end{example}
\end{shadedbox}

The second type is the \textit{Pauli channel}, which serves as a dominant noise source in many computing architectures and as a practical model for analyzing error correction \citep{flammia2020efficient}.  
\begin{definition}[Single-qubit Pauli channel]\label{def:pauli-channel}
  Given a quantum state $\rho\in\mathbb{C}^{2\times 2}$, the single-qubit Pauli channel $\mathcal{N}_{\vec{p}}$ acts on this state follows 
  \begin{equation}
    \mathcal{N}_{\vec{p}}(\rho) = p_I \rho + p_X \XGate \rho \XGate +  p_Y \YGate \rho \YGate  + p_Z \ZGate \rho \ZGate ~,
  \end{equation}
   where $\vec{p}= (p_I, p_X, p_Y, p_Z)$ and  $p_I+p_X+p_Y+p_Z=1$.
\end{definition}
Note that for a single-qubit system, the depolarization channel $\mathcal{N}_p$ is a special Pauli channel with setting $p_X=p_Y=p_Z=p$. 

\begin{shadedbox}
\begin{example} (Single-qubit state with Pauli channel). 
    Consider a single-qubit pure state $\rho=\ket{0}\bra{0}$ with the density matrix:
    \begin{equation}
        \rho = \ket{0}\bra{0} = \begin{bmatrix}
            1 & 0 \\
            0 & 0
            \end{bmatrix}.
    \end{equation}
    When the Pauli channel $\mathcal{N}_{\vec{p}}$ acts on this state, the output is given by
    \begin{align}
        \mathcal{N}_{\vec{p}}(\rho) &= p_I \begin{bmatrix}
            1 & 0 \\
            0 & 0
        \end{bmatrix}
        + p_X \begin{bmatrix}
            0 & 0 \\
            0 & 1
        \end{bmatrix}
        + p_Y \begin{bmatrix}
            0 & 0 \\
            0 & 1
        \end{bmatrix}
        + p_Z \begin{bmatrix}
            1 & 0 \\
            0 & 0
        \end{bmatrix} \\
        &= \begin{bmatrix}
            p_I + p_Z & 0 \\
            0 & p_X + p_Y
        \end{bmatrix}.
    \end{align}
    Let us analyze three special cases for the probability vector $\vec{p} = (p_I, p_X, p_Y, p_Z)$:
    \begin{itemize}
        \item \textbf{Case 1}: If $p_X = p_Y = p_Z = p$, the Pauli channel reduces to the depolarization channel and the prepared state becomes:
        \begin{equation}
            \mathcal{N}_{\vec{p}}(\rho) = \begin{bmatrix}
                1 - 2p & 0 \\
                0 & 2p
            \end{bmatrix}.
        \end{equation}
        \item \textbf{Case 2}: If $p_Y = p_Z = 0$, the prepared state becomes:
        \begin{equation}
            \mathcal{N}_{\vec{p}}(\rho) = \begin{bmatrix}
                1 - p_X & 0 \\
                0 & p_X
            \end{bmatrix}.
        \end{equation}
        In this scenario, the Pauli channel reduces to the \textit{bit-flip channel}.
        \item \textbf{Case 3}: For other values of $\vec{p}$, the effect of the Pauli channel on the pure state $\ket{0}$ can be interpreted as a combination of the depolarizing channel and the bit-flip channel.
    \end{itemize}
\end{example}
\end{shadedbox}

To generalize the single-qubit Pauli channel to a multi-qubit Pauli channel, we extend the definition to account for the action of Pauli operators on multiple qubits.
\begin{definition}[Multi-qubit Pauli channel]\label{def:multi-qubit-pauli-channel}
  Given a quantum state $\rho\in\mathbb{C}^{2^N\times 2^N}$ for an $N$-qubit system, the multi-qubit Pauli channel $\mathcal{N}_{\vec{p}}$ acts as 
  \begin{equation}
    \mathcal{N}_{\vec{p}}(\rho) = \sum_{P\in \mathcal{P}_N}p_PP\rho P^\dagger,
  \end{equation}
   where $ \mathcal{P}_N=\{I,X,Y,Z\}^{\otimes N}$ denotes the set of all tensor products of the $N$ single-qubit Pauli operators, and $p_P$ is the probability of applying the Pauli operator $P$ with $\sum_{P\in \mathcal{P}_N}p_P=1$.
\end{definition}

\begin{tcolorbox}[enhanced, 
    breakable,colback=gray!5!white,colframe=gray!75!black,title=Remark] The multi-qubit Pauli channel considers the existence of correlated Pauli noise on different qubits. If each qubit only experiences independent single-qubit Pauli noise, the multi-qubit channel can be written as the tensor product of sing-qubit Pauli channels:
    \begin{equation}
        \mathcal{N}_{\vec{p}}(\rho) = \otimes_{i=1}^N \mathcal{N}_{\vec{p}_i}(\rho),
    \end{equation}
    where $\mathcal{N}_{\vec{p}_i}$ is the single-qubit Pauli channel acting on the $i$-th qubit with probabilities $\vec{p}_i=(p_I, p_X, p_Y, p_Z)$.
\end{tcolorbox}

Having acknowledged the motivation and definition of quantum channels, it is natural to ask \textit{what is the relation between quantum channels and quantum gates?} A straightforward observation is that a quantum gate is a special case of a quantum channel. Conversely, the evolution of the quantum state can be built from the unitary operation via the isometric extension \citep{wilde2011classical}. The following theorem shows that any quantum channel arises from a unitary evolution on a larger Hilbert space.

\begin{tcolorbox}[enhanced, 
    breakable,colback=gray!5!white,colframe=gray!75!black,title=Remark] According to the Choi-Kraus decomposition, the unitary operator is a special case of a quantum channel. Specifically, when $d=1$, the quantum channel reduces to:
    \begin{equation}
        \mathcal{N}(X_A)= \mathbf{M}_1 X_A \mathbf{M}_1^{\dagger},
    \end{equation}
    where $\mathbf{M}_1$ is a unitary operator satisfying $\mathbf{M}_1^{\dagger}\mathbf{M}_1=\mathbb{I}$. This highlights that all unitary operators are quantum channels, but not all quantum channels are unitary.
\end{tcolorbox}

\begin{theorem}\citep{wilde2011classical}\label{thm:isometry}
  Let $\mathcal{N}(\cdot):\mathcal{L}(\mathcal{H}_A)\rightarrow \mathcal{L}(\mathcal{H}_B)$ be a quantum channel defined in Eqn.~(\ref{eqn:def-kraus}). Let $\mathcal{H}_E$ be the Hilbert space of an auxiliary system. Denote the input state as $\rho$ (i.e., a density operator $\rho \in \mathbb{C}^{\text{dim}(\mathcal{H}_A)\times \text{dim}(\mathcal{H}_A)}$). Then there exists a unitary $U: \mathcal{L}(\mathcal{H}_A \otimes \mathcal{H}_E)\rightarrow \mathcal{L}(\mathcal{H}_B\otimes \mathcal{H}_E)$ and a normalized vector (i.e., a pure state) $\ket{\varphi} \in \mathbb{C}^{\text{dim}(\mathcal{H}_E)}$ such that 
  \begin{equation}\label{eq:isometry}
  \mathcal{N}(\mathcal{\rho}) = \text{Tr}_E\left(U(\rho \otimes \ket{\varphi}\bra{\varphi} )U^\dagger \right)~,
  \end{equation}
  where $\text{Tr}_E(\cdot)$ denotes the partial trace over the ancillary Hilbert space $\mathcal{H}_E$, and the dimension of $\mathcal{H}_E$ depends on the rank of the Kraus representation of $\mathcal{N}$.
\end{theorem}

\begin{proof}[Proof sketch of Theorem~\ref{thm:isometry}]
We extend the system to include an ancillary Hilbert space $\mathcal{H}_E$, representing the environment. The combined space $\mathcal{H}_A \otimes \mathcal{H}_E$ forms a closed physical system, where the evolution of the quantum state can be described by a unitary operator $U$ acting on $\mathcal{H}_B \otimes \mathcal{H}_E$.

To find a feasible unitary $U$, we express the quantum channel $\mathcal{N}$ using its isometric extension \citep{wilde2011classical}, i.e.,
\begin{equation}\label{eq:isometric-ext}
    \mathcal{N}(\rho) = \text{Tr}_E\left(V \rho V^\dagger\right),
\end{equation}
where $V: \mathcal{H}_A \to \mathcal{H}_B \otimes \mathcal{H}_E$ is an isometry operator embedding the input state into the larger Hilbert space. For simplicity, assume $\mathcal{H}_A = \mathcal{H}_B$. The isometry operator $V$ can always be embedded into a unitary operator $U$ acting on $\mathcal{H}_B \otimes \mathcal{H}_E$, ensuring that $U$ captures the reversible evolution of the extended system.

Next, we augment the input state $\rho$ by introducing an ancillary state $\ket{\varphi} \in \mathcal{H}_E$, yielding the combined state $\rho \otimes \ket{\varphi}\bra{\varphi}$. Substituting this augmented state and the unitary operator $U$ into the isometric extension in Eqn.~(\ref{eq:isometric-ext}) gives Eqn.~(\ref{eq:isometry}). Theorem~\ref{thm:isometry} is thereby proven.
\end{proof}

The translation between the unitary operation and the quantum channels described by Theorem~\ref{thm:isometry} can be visually explained, as shown in Figure~\ref{fig:isometry}. In this diagram, the first wire corresponds to the original input state $\rho$, while the second wire represents the initial state $\ket{\varphi}$ of the environment. To determine the output of the quantum channel $\mathcal{N}$ applied to $\rho$, an $\mathcal{N}$-induced unitary operation $U$ is performed on the combined system, followed by a partial trace over the environment to discard its information.

\begin{figure}[h!]
  \centering
  \includegraphics[width = 0.6\textwidth]{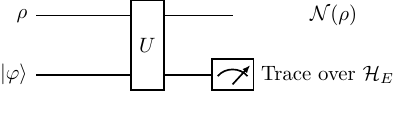}
  \caption{\textbf{The evolution of quantum states based on Theorem \ref{thm:isometry}}. }
  \label{fig:isometry}
\end{figure} 
 
\subsubsection{Quantum measurements}
In addition to quantum gates and quantum channels that manipulate quantum states, another special operation in quantum circuits is measurement. The aim of quantum measurements is to extract quantum information of the evolved state into the classical form. The quantum circuit diagram, which describes applying a unitary $U$ followed by the quantum measurement to a single-qubit state $\ket{\bm{a}}$, is shown in Figure \ref{fig:single-qubit-measure}. In particular, both types of measurements are depicted by the  `meter' symbol.

\begin{figure}[h!]
\centering
\includegraphics[width=0.48\textwidth]{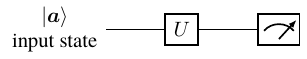}
\caption{\textbf{The quantum circuit diagram with measurement.}}
\label{fig:single-qubit-measure}  
\end{figure}  

The quantum measurements can be categorized into two types, i.e., \textit{projective measurements} and \textit{positive operator-valued measures} \citep{preskill1999lecture,nielsen2010quantum}. 

The projective measurement, which is also called the von Neumann measurement, is formally described by the Hermitian operator $A=\sum_i\lambda_i \ket{v_i}\bra{v_i}$, where $\{\lambda_i\}$ and $\{\ket{v_i}\}$ refer to the eigenvalues and eigenvectors of $A$, respectively. Supported by the Born rule \citep{nielsen2010quantum}, when the measurement operator $A\in\mathbb{C}^{2^N\times 2^N}$ is applied to an $N$-qubit state $\ket{\Phi}\in \mathbb{C}^{2^N}$, the probability of measuring any one of the eigenvalues in $\{\lambda_i\}$ is 
\begin{equation}
  \Pr(\lambda_i)= |\braket{v_i| \Phi}|^2~.
\end{equation}  
In the density operator representation, suppose that the state to be measured is $\rho\in\mathbb{C}^{2^N\times 2^N}$, the probability of measuring any one of the eigenvalues in $\{\lambda_i\}$ is 
\begin{equation}
  \Pr(\lambda_i)= \Tr(\rho \ket{v_i}\bra{v_i}).
\end{equation} 
Define $\Pi_i=\ket{v_i}\bra{v_i}$ as the $i$-th projective operator.  The complete set of projective operators $\{\Pi_i\}$ has the following properties 
\begin{equation}
    \text{1)} \Pi_i\Pi_j=\delta_{ij};~ \text{2)} \Pi^{\dagger}_i=\Pi;~ \text{3)} \Pi^2_i=\Pi;~ \text{4)} \sum_i \Pi_i = \mathbb{I}_{2^N}.
\end{equation}
A special set of projectors is defined as $\Pi_i=\ket{i}\bra{i}$ for $\forall i\in [2^N]$, which measures the probability corresponding to the basis state $\ket{i}$. For example, given the single-qubit state $\ket{\bm{\alpha}}$ in Eqn.~(\ref{eqn:sec-intro-1}), the probability to measure the computational basis state $\ket{i}$ is  
\begin{equation}\label{eqn:sec-intro-9}
  \Pr(i) = |\braket{v_i|\bm{\alpha}}|^2=|\alpha_i|^2~.
\end{equation}

The second type of quantum measurement is the \textit{positive operator-valued measures (POVM)}. A POVM is described by a collection of positive operators $0\preceq E_i$ satisfying $\sum_i E_i=\mathbb{I}$. Each positive operator $E_i$ is associated with an outcome of measurement. Specifically, applying the measurement $\{E_m\}$ to the state $\ket{\psi}$, the probability of outcome $i$ is given by
\begin{equation}
  \Pr(i) =  |\braket{\psi|E_i|\psi}|^2.
\end{equation}
In the density operator representation, suppose that the state to be measured is $\rho\in\mathbb{C}^{2^N\times 2^N}$,  the probability of outcome $i$ is given by 
\begin{equation}
  \Pr(\lambda_i)= \Tr(\rho E_i).
\end{equation}  
  We remark that the main difference between projective measurements and POVM elements is that the POVM elements do not have to be orthogonal. Due to this reason, the projective measurement is a special case of the generalized measurement (i.e., with setting $E_i=\Pi_i^{\dagger}\Pi_i$).

\begin{tcolorbox}[enhanced, 
    breakable,colback=gray!5!white,colframe=gray!75!black,title=Remark]
    
Here we address the accessible information through quantum measurements in the ideal and practical settings. For ease of discussion, suppose that the computation result corresponds to the probability amplitude $\bm{a}_1$ in the single-qubit state $\ket{\bm{a}}=\bm{a}_1\ket{0}+\bm{a}_2\ket{1}$ in Eqn.~(\ref{eqn:sec-intro-1}). To extract $\bm{a}_1$ from the quantum state into the classical form, we apply the projective operator $\Pi_1=\ket{0}\bra{0}$ to this state. Due to the law of quantum mechanics, after each measurement, the state is collapsed and the measured outcome $V_i$ can be viewed as a binary random variable with the Bernoulli distribution $\Ber(\bm{a}_i)$, i.e., $ \Pro(V_i=1)= \bm{a}_1$ and $\Pro(V_i=0)=1-\bm{a}_1$. Through applying the measurement $\Pi_i$ to $K$ copies of the state $\ket{\bm{a}}$, the obtained statistics, i.e., the sample mean, is denoted by $\bar{\bm{a}}_1=\sum_{i=1}^K V_i/K$. The large number theorem indicates $\bar{\bm{a}}_1=\bm{a}_1$ when $K\rightarrow \infty$. However,  only the finite number of measurements $K$ is allowed in practice and thus results in an estimation error.  
\end{tcolorbox}

\section{Quantum Read-in and Read-out protocols}\label{cha3:sec:in-out}

The terms \textit{quantum read-in} and \textit{read-out} refer to the processes of transferring information between classical systems and quantum systems. These are fundamental steps in the workflow of quantum machine learning shown in Figure~\ref{fig:chap1:comparison-clc-quantum}, responsible for loading data and extracting results. 

Quantum read-in and read-out pose significant bottlenecks in leveraging quantum computing to address classical computational tasks. As emphasized in \citep{aaronson2015read}, while quantum algorithms can offer exponential speed-ups in specific problem domains, these advantages can be negated if the processes of loading classical data into quantum systems (read-in) or extracting results from quantum systems (read-out) are inefficient. Specifically, the high-dimensional nature of quantum states and the constraints on measurement precision often lead to overheads that scale poorly with problem size. These challenges underscore the importance of optimizing quantum read-in and read-out protocols to realize the full potential of quantum computing. Below is a detailed introduction to quantum read-int and read-out protocols, including the basic concept and several typical algorithms.

\subsection{Quantum read-in protocols}\label{cha3:subsec:q-read-in}

Quantum read-in refers to the process of encoding classical information into quantum systems that can be manipulated by a quantum computer, which can be regarded as the \textit{classical-to-quantum mapping}. It acts as a bridge to utilize quantum algorithms to solve classical problems in quantum computing. Here, we will introduce several typical encoding methods, including basis encoding, amplitude encoding, angle encoding, and quantum random access memory. Some easy-to-use demonstrations are provided in Chapter~\ref{Chapter2:preliminary-code}.

\subsubsection{Basis encoding}\label{cha2:sec3:basis_encode}

Basis encoding is a basic method for processing classical data that can be represented in binary form. Given a classical binary vector \( \bm{x} = (\bm{x}_0, \ldots, \bm{x}_i, \ldots, \bm{x}_{N-1}) \in \{0, 1\}^N \), this encoding technique maps the vector directly into a quantum computational basis state as follows:
\begin{equation}
  \ket{\psi} = \ket{\bm{x}_0, \ldots, \bm{x}_{N-1}}.
\end{equation}
In this process, \( N \) qubits are required to represent a binary vector of length \( N \). To prepare the corresponding quantum state \( \ket{\psi} \), an \( X \) gate is applied to each qubit where the corresponding bit value is 1. The overall quantum state preparation can be expressed as:
\[
\ket{\psi} = \bigotimes_{i=0}^{N-1} X^{\bm{x}_i} \ket{0}^{\otimes N},
\]
where \( \ket{0}^{\otimes N} \) represents an initial state of all qubits set to \( |0\rangle \), and \( X^{\bm{x}_i} \) means applying the \( X \) gate to the $i$-th qubit only if \( \bm{x}_i = 1 \).

\begin{shadedbox}
\begin{example}
  (Basis encoding). Consider encoding the integer \( 6 \), which has the binary representation \( \bm{x} = (1, 1, 0) \). The corresponding quantum state is \( \ket{110} \). This state can be implemented by applying \( X \) gates to the first and second qubits, as shown in Figure~\ref{fig:basis-encode}.
\end{example}
\end{shadedbox}

\begin{figure}[h!]
  \centering
  \includegraphics[width=0.3\textwidth]{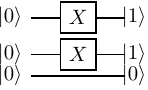}
  \caption{\textbf{Example of basis encoding for the integer $6$.}}
  \label{fig:basis-encode}
\end{figure}

\subsubsection{Amplitude encoding}\label{cha2:sec3:amplitude_encode}

Amplitude encoding is a technique that maps classical data into the amplitudes of a quantum state. Given a vector \( \bm{x} = (\bm{x}_0, \ldots, \bm{x}_i, \ldots, \bm{x}_{2^N-1}) \in \mathbb{C}^{2^N} \) containing complex values, we first apply \( L_2 \) normalization to obtain a normalized vector
\begin{equation}\label{eqn:amplitude_encode}
    \hat{\bm{x}} = \frac{\bm{x}}{\|\bm{x}\|_2},
\end{equation}
where \( \|\bm{x}\|_2 \) is the Euclidean norm. This ensures that the normalized vector \( \hat{\bm{x}} \) satisfies \( \sum_{i=0}^{2^N-1} |\hat{\bm{x}}_i|^2 = 1 \). The corresponding quantum state is then expressed as
\begin{equation}
  \ket{\psi} = \sum_{i=0}^{2^N-1} \hat{\bm{x}}_i \ket{i}
\end{equation}
with \( \ket{i} \) representing the \( N \)-qubit computational basis states.

\begin{shadedbox}
\begin{example}
(Amplitude encoding). Consider encoding a normalized vector \( \bm{x} = (\bm{x}_0, \bm{x}_1) \in \mathbb{C}^2 \) into the quantum state \( \ket{\psi} = \bm{x}_0 \ket{0} + \bm{x}_1 \ket{1} \). This can be achieved by applying a rotation gate \( U=R_Y(\theta) \) to the initial state \( \ket{0} \), where \( \theta = 2 \arccos(\bm{x}_0) \).
\end{example}
\end{shadedbox}

Amplitude encoding is highly efficient because it allows an exponentially large vector of length \( 2^N \) to be represented using only \( N \) qubits. However, preparing this quantum state requires constructing a unitary transformation \( U \) such that \( \ket{\psi} = U \ket{0}^{\otimes N} \). Efficiently finding such transformations can be challenging and is an active area of research (see Chapter~\ref{chap-preliminary-sec:Remark} for the discussions).

\subsubsection{Angle encoding}\label{cha2:sec3:angle_encode}

Basis encoding and amplitude encoding are fundamental techniques for mapping classical data to quantum states, but each comes with distinct resource costs. Basis encoding requires a number of qubits equal to the dimensionality of the binary representation of classical data and necessitates minimal gate operations for state preparation. In contrast, amplitude encoding is highly compact in terms of qubits, using only the logarithmic of the data dimensionality, but it involves a significant gate complexity.  

To address this limitation, an alternative is \textit{angle encoding}. The core idea of angle encoding is to embed classical data into a quantum state through rotation angles.

Given a real-valued vector \( \bm{x} = (\bm{x}_0, \ldots, \bm{x}_i, \ldots, \bm{x}_{N-1}) \in \mathbb{R}^{N} \), the encoded quantum state can be represented as:
\begin{equation}\label{eq:angle_encode}
  \ket{\psi} = \bigotimes_{i=0}^{N-1} R_{\sigma}(\bm{x}_i) \ket{0}^{\otimes N} = \bigotimes_{i=0}^{N-1} \exp\left(-i \frac{\bm{x}_i}{2}\sigma\right) \ket{0}^{\otimes N},
\end{equation}
where \( \sigma \in \{X, Y, Z\} \) denotes a Pauli operator, as defined in Figure~\ref{tab:Q-gates}. Since Pauli rotation gates are \( 2\pi \)-periodic, it is essential to scale each element \( \bm{x}_i \) into the range \([0, \pi)\) to ensure that different values are encoded into distinct quantum states. 

A key advantage of angle encoding is its ability to introduce nonlinearity. By mapping classical data into the parameters of quantum rotation gates, angle encoding leverages trigonometric functions to naturally capture non-linear relationships. This property is particularly important in quantum machine learning, where nonlinearity is essential for models to learn complex patterns in data, such as non-linearly separable decision boundaries.

\subsubsection{Quantum Random Access Memory (QRAM)}

Basis encoding, amplitude encoding, and angle encoding are generally designed to encode a single item of data at one time, which makes it challenging to process complicated classical datasets. The QRAM \citep{giovannetti2008quantum}, analogous to classical RAM, aims to simultaneously store, address, and access multiple quantum states.    

QRAM consists of two types of qubits: data qubits for storing classical data and address qubits for addressing. Given a classical dataset $\mathcal{D}=\{\bm{x}^{(j)}\}_{j=0}^{M-1}$ with $M$ training examples, assume we separately encode each data item into a quantum state $\ket{\bm{x}^{(j)}}_d$ using one of the encoding methods above. The QRAM can be constructed as follows: (1) Prepare an $N_a$-qubit address register where $N_a=\lceil \log_2(M) \rceil$; (2) Associate each data state $\ket{\bm{x}^{(j)}}_d$ with corresponding address state $\ket{j}_a$. The whole dataset is therefore encoded into a quantum state of the form
\begin{equation}
  \ket{\mathcal{D}}=\sum_{j=0}^{M-1}\frac{1}{\sqrt{M}}\ket{j}_a\ket{\bm{x}^{(j)}}_d.
\end{equation}

\begin{tcolorbox}[enhanced, 
    breakable,colback=gray!5!white,colframe=gray!75!black,title=Remark] The subscript $d$ in $\ket{\bm{x}^{(j)}}_d$ indicates that this quantum state resides in the data register, differentiating it from address qubits, which are denoted with the subscript $a$ (e.g., $\ket{j}_a$). This convention helps to distinguish between the roles of data and address qubits in QRAM operations.
\end{tcolorbox}

\begin{shadedbox}
\begin{example}
(QRAM Encoding). Consider a dataset \( \mathcal{D} = \{2, 3\} \). Using basis encoding, each sample is first converted into a two-qubit quantum state: \( \{\ket{10}_d, \ket{11}_d\} \). Each data state is then assigned an address state, \( \ket{0}_a \) for the first state \( \ket{10}_d \) and \( \ket{1}_a \) for the second state \( \ket{11}_d \). The resulting QRAM-encoded state is:
\begin{equation}
    \ket{\mathcal{D}} = \frac{1}{\sqrt{2}} \left( \ket{0}_a \ket{10}_d + \ket{1}_a \ket{11}_d \right).
\end{equation}
  
The corresponding quantum circuit for implementing this state is shown in Fig.~\ref{fig:qram-encode}.
\end{example}
\end{shadedbox}

\begin{figure}[h!]
  \centering
  \includegraphics[width=0.5\textwidth]{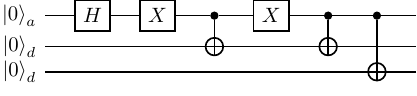}
\caption{\textbf{Example of QRAM encoding for the dataset $\mathcal{D}=\{2, 3\}$.}}
  \label{fig:qram-encode}
    \end{figure}
    
QRAM allows the dataset \( \mathcal{D} \) to be stored in a coherent quantum superposition, enabling simultaneous access to all data items through the entanglement of address and data qubits. While QRAM is theoretically powerful, its practical implementation remains a significant challenge due to the need for a large number of qubits and quantum operations (see Chapter~\ref{chap-preliminary-sec:Remark} for the discussion).

\subsection{Quantum read-out protocols}\label{chapter2:Sec-2.3.2-readout}

Quantum read-out refers to the process of translating the quantum state resulting from a quantum computation into classical data, enabling further processing, interpretation, or optimization in classical systems. This process can be considered the inverse operation of quantum read-in, representing a quantum-to-classical mapping. 

Based on the completeness of the information obtained during the read-out process, quantum read-out protocols can be broadly categorized into two types, i.e., \textit{full information and partial information read-out protocols}.  These protocols enable tailored read-out processes that match the requirements of different quantum applications, ranging from tomography to optimization and machine learning tasks.

\subsubsection{Full information read-out protocol}

The full information read-out protocol is used to completely reconstruct the quantum state, which is used to fully understand the quantum system's behavior. The most general approach to implementing this protocol is through quantum state tomography (QST) \citep{vogel1989determination}. 

QST involves performing quantum measurements, gathering measurement statistics, and using classical post-processing to reconstruct the quantum state. In what follows, we introduce two reconstruction techniques broadly used in QST, i.e., linear inversion \citep{qi2013quantum} and maximum likelihood estimation (MLE) \citep{hradil1997quantum}.

\medskip
\noindent\textit{QST with linear inversion.} Linear inversion is a straightforward method to reconstruct the quantum state from measurement data by directly solving linear systems of equations. Let $\rho$ be the explored quantum state and $\{E_i\}$ be a set of measurements. According to the Born rule, the probability of measurement outcome $i$ is given by
\begin{equation}
   \Pr(E_i|\rho) = \Tr(\rho E_i).
\end{equation}
In practice, $\Pr(E_i|\rho)$ is not directly accessible but is approximated by the frequency $p_i$ of measurement outcome $i$ over multiple measurements. By the law of large numbers, as the number of measurements increases, $p_i$ converges to the true probability $\Pr(E_i|\rho)$. Collecting measurements across all bases, we obtain a linear system
\begin{equation}
    \begin{bmatrix}
        \Tr(\rho E_0) \\
        \Tr(\rho E_1) \\
        \vdots
    \end{bmatrix} =
    \begin{bmatrix}
        \vec{E}_0^\dagger \cdot \vec{\rho} \\
        \vec{E}_1^\dagger \cdot \vec{\rho} \\
        \vdots
    \end{bmatrix} = A\vec{\rho} \approx \bm{p} = \begin{bmatrix}
        p_0 \\
        p_1 \\
        \vdots
    \end{bmatrix},
\end{equation}
where $\vec{E}$ and $\vec{\rho}$ refer to the vector representations of matrices $E_i$ and $\rho$, respectively. The vector representation of a matrix is obtained by stacking its columns into a single-column vector. For example, the vector representation of a $2\times 2$ identity matrix is $\vec{\mathbb{I}}_2=[1,0,0,1]^T$. The matrix $A$ is constructed such that each row corresponds to the vector representation of the measurement operator, i.e., $A=[\vec{E}_0^\dagger;\vec{E}_1^\dagger;...]$. The vector $\bm{p}$ contains the measured frequencies $p_i$.

Assuming the measurements are tomographically complete, i.e., $\{E_i\}$ forms a basis for the system's Hilbert space, we can reconstruct $\rho$ by solving the following linear systems of equations, i.e., 
\begin{equation}
    \vec{\rho} = (A^T A)^{-1} A^T \bm{p}.
\end{equation}

A common strategy is to use Pauli operators as measurement bases $\{E_i\}$. The density matrix $\rho$ of an $N$-qubit system can be expanded in terms of the Pauli basis as:
\begin{equation}
    \rho = \frac{1}{2^N} \sum_{i=0}^{4^N - 1} c_i P_i, \quad c_i \in \mathbb{R}, \quad P_i \in \{I, X, Y, Z\}^{\otimes N},
\end{equation}
where the coefficients $c_i$ represent projections of $\rho$ onto the Pauli basis, calculated as:
\begin{equation}
    c_i = \Tr(\rho P_i).
\end{equation}
To fully reconstruct $\rho$, the quantum state must theoretically be measured in all $4^N - 1$ Pauli bases to estimate each $c_i$.

\begin{tcolorbox}[enhanced, 
    breakable,colback=gray!5!white,colframe=gray!75!black,title=Remark] The Pauli basis consists of four Hermitian matrics $I$, $X$, $Y$ and $Z$ introduced in Figure~\ref{tab:Q-gates}. These operators form a complete basis for the space of $2\times 2$ complex matrices. For $N$-qubit systems, the tensor products of these single-qubit operators span the space of $2^N \times 2^N$ complex matrix. This makes the Pauli basis essential for representing quantum states, observables, and their transformations.
\end{tcolorbox}

A key limitation of linear inversion is that it does not guarantee a valid density matrix, as the estimated quantum state may lack properties such as positive semi-definiteness in Definition~\ref{def:psd}, especially with limited measurements. 

\medskip
\noindent\textit{Maximum Likelihood Estimation (MLE).} To ensure physical constraints on the quantum state during reconstruction, MLE is introduced. MLE reconstructs $\rho$ by maximizing the likelihood of observing the measurement outcomes, subject to the constraints that $\rho$ is Hermitian, positive semi-definite, and trace one. The likelihood function is given by
\begin{equation}
    L(\rho) = \prod_i \Tr(\rho E_i)^{p_i}.
\end{equation}
Reconstructing $\rho$ then reduces to solving the following optimization problem
\begin{equation}
    \argmax_{\rho'} L(\rho'), \quad \text{s.t.} \quad \rho' \succeq 0, \quad \rho' = \rho'^\dagger, \quad \Tr(\rho') = 1.
\end{equation}
Solving this typically requires iterative numerical optimization, which can be computationally intensive.

\begin{tcolorbox}[enhanced, 
    breakable,colback=gray!5!white,colframe=gray!75!black,title=Remark] A common challenge across all quantum state tomography (QST) methods, including linear inversion and MLE, is the \textit{exponential computational cost} with respect to the number of qubits. Specifically, the number of parameters to reconstruct grows exponentially with the system size, making QST methods feasible only for small-qubit systems in practice. This limitation underscores the need for scalable approaches to quantum state characterization in larger quantum systems.
\end{tcolorbox} 

\subsubsection{Partial information read-out protocol}

The partial information read-out protocol focuses on extracting specific, task-relevant information from a quantum state without reconstructing the entire density matrix. This protocol is efficient and can be applied to comprehend large-qubit systems. According to the type of collected information, current partial read-out protocols can be categorized into three classes, i.e., sampling, expectation value estimation, and shadow tomography.

\medskip
\noindent\textit{Sampling.} Sampling involves repeatedly measuring the quantum state in the computational basis to estimate the probability distribution over bit-strings. Given a state $\ket{\psi}$, the probability of observing a specific computational basis \( \ket{i} \) is given by
\begin{equation}
  \Pr(i) = \left| \braket{\psi|i}\right|^2.
\end{equation}
The empirical frequency of each outcome from repeated measurements provides an estimate of $\Pr(i)$. Sampling is particularly useful in the following applications
\begin{itemize}
    \item {Sampling over complicated distributions.} Quantum states can represent complex probability distributions that are difficult to sample classically. Quantum sampling allows efficient exploration of these distributions for specific applications, such as probabilistic modeling and Markov chain Monte Carlo \citep{layden2023quantum}.
    \item {Optimization problems.} Sampling can identify high-probability bitstrings in quantum algorithms, such as the Quantum Approximate Optimization Algorithm \citep{farhi2014quantum} and Grover search \citep{grover1996fast}, where these bitstrings often correspond to optimal or near-optimal solutions.
    \item {Verification.} Sampling facilitates the comparison of a quantum circuit's output with theoretical expectations or desired distributions, helping to verify the quantum systems \citep{boixo2018characterizing,bouland2019complexity,arute2019quantum}.
\end{itemize}

\medskip
\noindent\textit{Expectation value estimation.} For a wide class of quantum computation problems, such as quantum chemistry and quantum many-body physics, the computation outcome refers to the estimation of the expectation values of certain observables on the evolved quantum state \citep{kandala2017hardware,tilly2022variational}.

An observable $O\in \mathbb{C}^{2^N \times 2^N}$ mentioned here is a Hermitian operator that represents a measurable physical quantity. For an $N$-qubit system, $O$ can be expressed in terms of a Pauli basis expansion, i.e.,
\begin{equation}
    O = \sum_{i=1}^{4^N} \alpha_i P_i, \quad P_i \in \{\mathbb{I}_2, X, Y, Z\}^{\otimes N}, \quad \alpha_i \in \mathbb{R}.
\end{equation}
where $P_i$ denotes the $i$-th $N$-qubit Pauli string.

The expectation value of an observable $O$ with respect to an $N$-qubit state $\rho$ is 
\begin{equation}\label{eq:exp}
    \braket{O} = \Tr(\rho O).
\end{equation}

Substituting the Pauli expansion of $O$, the expectation value is expressed as the weighted sum of the expectation values of each Pauli basis term due to the linearity of the trace operation, i.e.,
\begin{equation}
    \braket{O} = \sum_{i=1}^{4^N} \alpha_i \Tr(\rho P_i) \equiv \sum_{i=1}^{4^N} \alpha_i \braket{P_i}.
\end{equation}

To estimate the expectation value of each individual Pauli term \( P_i \), the quantum state \( \rho \) must be measured on the basis of the eigenstates of \( P_i \). The measurement outcome is then associated with the corresponding eigenvalue of \( P_i \). Notably, the eigenstates and eigenvalues of \( P_i \) can be derived from the eigenstates and eigenvalues of its constituent single-qubit Pauli operators \( P_{ij} \):
\begin{itemize}
    \item Eigenvalues. The eigenvalues of \( P_i \) are the product of the eigenvalues of each single-qubit Pauli operator \( P_{ij} \), i.e., $P_i = \otimes_{j=1}^N P_{ij}$. For example, if the eigenvalues of \( P_{ij} \) are \( \pm 1 \), then the eigenvalues of \( P_i \) are products of these individual eigenvalues and remain in \( \{ \pm 1 \} \).
    \item Eigenstates. The eigenstates of \( P_i \) are the tensor products of the eigenstates of the single-qubit Pauli operators \( P_{ij} \). If \( \ket{\lambda_{ijk}} \) is one of the eigenstate of \( P_{ij} \), then the corresponding eigenstate of \( P_i \) is \( \bigotimes_{j=1}^N \ket{\lambda_{ijk}} \).
\end{itemize}
This structure allows \( P_i \) to be analyzed in terms of its simpler single-qubit components, significantly simplifying the process of determining the measurement basis for expectation value estimation. By repeating the measurements $M$ times and obtaining the corresponding measurement results $\{r_j\}_{j=1}^{M}$, the statistical value of $\braket{P_i}$ can be estimated by

\begin{equation}
    \braket{\hat{P}_i} = \frac{1}{M} \sum_{j=1}^{M} r_j.
\end{equation}
The expectation value of the observable $O$ is therefore statistically estimated by $\braket{\hat{O}}=\sum_{i=0}^{K-1}\alpha_i \braket{\hat{P}_i}$.

\begin{tcolorbox}[enhanced, 
    breakable,colback=gray!5!white,colframe=gray!75!black,title=Remark]
A key step in the process is to measure the quantum system in the basis of the eigenstates of $P_i$. If $P_i$ is diagonal in the computational basis (e.g., a tensor product of Pauli-Z operators), we can directly measure the state without additional operations. Otherwise (e.g., for Pauli-X or Pauli-Y operators), we need to apply a unitary transformation to rotate the quantum state into the desired basis. Specifically, when measuring in the Pauli-X basis (i.e., $\ket{+}$ and $\ket{-}$), a Hadamard gate $H$ is applied to the state $\rho$, i.e.,
\begin{equation}
    \rho' = H \rho H.
\end{equation}

When measuring in the Pauli-Y basis (i.e., $\frac{\ket{0} + i\ket{1}}{\sqrt{2}}$ and $\frac{\ket{0} - i\ket{1}}{\sqrt{2}}$), a phase gate $S = \sqrt{Z}$ followed by a Hadamard gate $H$ is applied, i.e.,
\begin{equation}
    \rho' = S^\dagger H \rho H S.
\end{equation}
Measuring the state $\rho'$ in the computational basis is equivalent to measuring the state $\rho$ in the corresponding Pauli basis.
\end{tcolorbox}

\medskip
\noindent\textit{Shadow tomography.} Performing full QST requires an exponential number of copies of the quantum state, making it impractical for systems beyond a small number of qubits. Instead of reconstructing the complete density matrix, shadow tomography \citet{aaronson2018shadow} focuses on efficiently obtaining specific properties of a quantum state, such as the expectation values of many observables.

\begin{definition}[Shadow tomography, \cite{aaronson2018shadow}]\label{def:shadow-tomo}
    Given an unknown $D$-dimensional quantum state $\rho$, as well as $M$ observables $O_1,...,O_M$, output real numbers $b_1,...,b_M$ such that $\left|b_i-\Tr(O_i\rho)\right|\leq\epsilon$ for all $i$, with success probability at least $1-\delta$. Do this via a measurement of $\rho^{\otimes k}$, where $k=k(D,M,\epsilon,\delta)$ is as small as possible.
\end{definition}

\citet{aaronson2018shadow} proved that the shadow tomography problem can be solved using a \textit{polylogarithmic} number of copies of states in terms of the dimension $D$ and number $M$ of observables. This result demonstrates that it is possible to estimate the expectation values of exponentially many observables for a quantum state of exponential dimension using only a polynomial number of measurements.

The central idea of shadow tomography is to create a compact measurement classical representation, or ``shadow'', of a quantum state that encodes sufficient information to estimate many properties of the state. Building on this concept, \citet{huang2020predicting} proposed a more practical and efficient approach, termed \textit{classical shadow}, which uses randomized measurements to construct this classical representation. The classical shadow approach consists of the following steps:

\begin{enumerate}
    \item {Randomized measurements}. Perform random unitary transformations on the quantum state and measure the transformed state in the computational basis. These random transformations can be drawn from specific ensembles, such as Clifford gates or local random rotations, which ensure that the measurement outcomes capture the essential properties of the quantum state.
    
    \item {Classical shadow construction.} Using the measurement results, construct a classical shadow of the quantum state. This compact representation encodes the quantum state in a way that allows for the efficient estimation of properties.

    \item {Property estimation.} Use the classical shadow to compute the desired properties of the quantum state, such as expectation values of specific observables, subsystem entropies, or fidelities with known states.
\end{enumerate}

Shadow tomography requires exponentially fewer measurements compared to full quantum state tomography, making it a practical solution for large-scale quantum systems. Moreover, the shadow of a quantum state serves as a versatile representation, enabling the efficient estimation of various properties such as expectation values, entanglement measures, and subsystem correlations.
 
\section{Quantum Linear Algebra}\label{chap2:preliminary-sec:linearAlgebra}

We next introduce quantum linear algebra, a potent toolbox for designing various \textsf{FTQC}-based algorithms introduced in Chapter~\ref{chapt1-sec-progress-FTQC}. For clarity, we start with the definition of block encoding in Chapter~\ref{chapt-3-sec-blockencoding}, which is about how to implement a matrix on the quantum computer. Based on this, we introduce some basic arithmetic rules for block encodings in Chapter~\ref{chapt-3-sect-arithmeticblockencoding}, like the multiplication, linear combination, and the Hadamard product.
Finally, in Chapter~\ref{chapt-3-sec-qsvt}, we introduce the quantum singular value transformation method, which enables one to implement functions onto singular values of block-encoded matrices.

\subsection{Block encoding\label{chapt-3-sec-blockencoding}}

For many computational problems, such as solving linear equations, we need to deal with a non-unitary matrix $A$.
However, remember that quantum gates as discussed in Chapter~\ref{cha3:sec:circuit} are unitaries. Therefore, if we want to solve these problems on quantum computers, it is essential to consider how to encode the matrix $A$ into a unitary.
This challenge can be addressed by the block encoding technique.

\begin{definition}[\label{blkencod}Block encoding, \cite{gilyenquantum2019}]
    Suppose that $A$ is an $N$-qubit operator, $\alpha, \varepsilon\geq 0$ and $a\in \mathbb{N}$. Then we say that the $(a+ N)$-qubit unitary $U$ is an $(\alpha,a,\varepsilon)$-block-encoding of $A$ if
    \begin{align}
        \|A-\alpha(\bra{0}^{\otimes a}&\otimes \mathbb{I}_{2^N})U(\ket{0}^{\otimes a}\otimes \mathbb{I}_{2^N})\|\leq \varepsilon.
    \end{align}
    Here, $\|\cdot\|$ represents the spectral norm, i.e., the largest singular value of the matrix.
\end{definition}

\begin{figure}
    \centering
\includegraphics[width=0.4\textwidth]{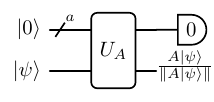}
    \caption{\textbf{Quantum circuit for block encoding.}}
    \label{fig:blockencoding}
\end{figure}

The circuit implementation of the block encoding is illustrated in Figure~\ref{fig:blockencoding}. The scaled matrix $A/\alpha$ interacts with the state $\ket{\psi}$ if the first qubit registers are measured as $\ket{0}$. By definition, we have $\alpha\geq \|A\|$ and any unitary $U$ is an $(1,0,0)$-block encoding of itself.

\begin{fact}\label{blockencoding.linearcombination}
    (Block encoding via the linear combination of unitaries (LCU) method, \cite{gilyenquantum2019}). Suppose that $A$ can be written in the form 
\begin{align}
    A=\sum_{k} \alpha_k U_k,
\end{align}
where $\{\alpha_k\}$ are real numbers and $U_k$ are some easily prepared unitaries such as Pauli strings.
Then, the LCU method allows us to have the access to two unitaries, i.e.,
    \begin{align}
        U_{\mathrm{SEL}}=&\sum_{k} |k\rangle\langle k|\otimes U_k,\\
        U_{\mathrm{PREP}}:&\ket{0}\rightarrow \frac{1}{\sqrt{\|\vec{\alpha}\|_1}}\sum_k \sqrt{\alpha_k}\ket{k},
    \end{align}
    where $\vec{\alpha}=(\alpha_1,\alpha_2,\dots)$.
    
After simple mathematical analysis, one can obtain $U=(U^\dagger_{\mathrm{PREP}}\otimes \mathbb{I}_{2^N})U_{\mathrm{SEL}}(U_{\mathrm{PREP}}\otimes \mathbb{I}_{2^N})$ is a $(\|\vec{\alpha}\|_1, m,0)$-block-encoding of $A$.
Here, $\mathbb{I}_{2^N}$ is the identity operator of $N$-qubit size and $\|\cdot\|_1$ denotes the $\ell_1$ norm of a given vector.
\end{fact}

Similar to the definition of block encoding, we can also define the state preparation encoding.
\begin{definition}[State preparation encoding\label{def.stateencoding} \cite{guo2024quantumlinear2024}]
We say a unitary $U_{\psi}$ is an $(\alpha,a,\epsilon)$-state-encoding of an $N$-qubit quantum state $\ket{\psi}$ if
\begin{align}
    \norm{\ket{\psi}-\alpha (\bra{0^{a}}\otimes I)U_{\psi}\ket{0^{a+N}}}_{\infty} \leq \epsilon,
\end{align} 
where $\|\cdot\|_\infty$ denotes the infinity norm of the given vector. 
\end{definition}

More straightforwardly, the $(\alpha,a,\epsilon)$-state-encoding $U_\psi$ prepares the state 
\begin{align}
    U_{\psi}\ket{0}\ket{0}= \frac{1}{\alpha}\ket{0}\ket{\psi'}+\sqrt{1-\alpha^2}\ket{1}\ket{\mathrm{bad}},\notag
\end{align}
where $\norm{\ket{\psi'}-\ket{\psi}}_{\infty}\leq \epsilon$ and $\ket{\mathrm{bad}}$ is an arbitrary quantum state.
One can further prepare the state $\ket{\psi'}$ by using $\mathcal{O}(\alpha)$ times of amplitude amplification \citep{brassard2002quantum}.
The state preparation encoding can be understood as a specific case of the block encoding, i.e., it is the block encoding of a $\mathbb{C}^{2^N \times 1}$ matrix.

\subsection{Basic arithmetic for block encodings\label{chapt-3-sect-arithmeticblockencoding}}
Now we introduce some arithmetic rules for block encoding unitaries. The following two facts describe the product and linear combination rules of block encoding unitaries, respectively.

\begin{fact}[Product of block encoding,  \cite{gilyenquantum2019}]\label{blockencoding.product}
    If $U$ is an $(\alpha,a,\delta)$-block encoding of an $N$-qubit operator $A$, and $V$ is a $(\beta,b,\varepsilon)$-block encoding of an $N$-qubit operator $B$, then $(\mathbb{I}_{2^b}\otimes U)(\mathbb{I}_{2^a}\otimes V)$ is an $(\alpha\beta,a+b,\alpha\varepsilon+\beta\delta)$-block-encoding of $AB$.
    Here, $\mathbb{I}_{2^a}$ is the identity operator of $a$-qubit size.
\end{fact}

\begin{fact}[Linear combination of block encoding,  \cite{gilyenquantum2019}]\label{LCU.blockencoding}
    Let $A=\sum_k x_k A_k$ be an $s$-qubit operator with $\beta\geq \|\vec{x}\|_1$ and $\varepsilon_1>0$, where $\vec{x}$ is the vector of coefficients.
    Suppose we have access to 
    \begin{align}
        P_L\ket{0}&=\sum_k c_k\ket{k},\\
        P_R\ket{0}&=\sum_k d_k\ket{k},\\
        W&=\sum_k |k\rangle\langle k|\otimes U_k+\left(\bigl(\mathcal{I}_s-\sum_k |k\rangle\langle k|\bigl)\otimes \mathcal{I}_a\otimes \mathcal{I}_b\right),
    \end{align}
    where $\sum_k |\beta c_k^*d_k-x_k|\leq \varepsilon_1$ and $U_k$ is an $(\alpha,a,\varepsilon_2)$-block-encoding of $A_k$.
    Then we can implement an $(\alpha\beta, a+b,\alpha\varepsilon_1+\beta\varepsilon_2)$-block-encoding of $A$ by using one time of $W, P_L$, and $P_R$.
\end{fact}
These results can be verified via direct computation.
Another arithmetic rule broadly employed in quantum machine learning is the Hadamard product, a.k.a, the element-wise product. The following lemma exhibits how to achieve this operation via the block encoding framework.

\begin{lemma}[Hadamard product of the block encoding unitaries,  \cite{guo2024quantumlinear2024}]\label{Hadamard.blockencoding}
With $N \in \mathbb{N}$, consider two matrices $A, B\in \mathbb{C}^{2^N\times 2^N}$, and assume that we have an $(\alpha,a,\delta)$-encoding $U_{A}$ of matrix $A$ and $(\beta,b,\epsilon)$-encoding $U_{B}$ of matrix $B$, then
we can construct an $(\alpha\beta,a+b+N,\alpha\epsilon+\beta\delta)$-encoding of matrix $A \circ B$ corresponding to the Hadamard product of $A$ and $B$.    
\end{lemma}

\begin{proof}[Proof sketch of Lemma~\ref{Hadamard.blockencoding}]
For simplicity, we only consider the perfect case, i.e., no errors. Refer to Ref.~\citep{guo2024quantumlinear2024} for the proof details under the more general cases. 

The intuition for achieving the Hadamard product is that  all the needed elements can be found in the tensor product, i.e.,
\begin{align}
    & (\bra{0^{a+b}}\otimes \mathbb{I}_{2^{2N}})(\mathbb{I}_{2^b}\otimes U_{A}\otimes \mathbb{I}_{2^N})(\mathbb{I}_{2^a}\otimes U_{B}\otimes \mathbb{I}_{2^N})(\ket{0^{a+b}}\otimes \mathbb{I}_{2^{2N}}) \\ \nonumber
    = & \frac{A\otimes B}{\alpha \beta}.
\end{align}

To this end, the question is reduced to finding proper permutation unitaries that can shift the required elements to the correct positions to achieve the Hadamard product.
Denote $P'=\sum_{i=0}^{d-1} |i\rangle\langle i|\otimes |0\rangle\langle i|$.
As proved  by \citet{zhao2021compiling}, the tensor product of $A$ and $B$ can be reformulated to the Hadamard product via $P$, i.e.,
$$P'(A\otimes B)P'^\dagger= (A \circ B)\otimes |0\rangle \langle 0|.$$
However, $P'$ is not a unitary.
Instead, we consider $P=\sum_{i,j=0}^{d-1} |i\rangle\langle i|\otimes |i\oplus j\rangle\langle j| $, which can be easily constructed by using $N$ CNOT gates, i.e., one CNOT gate between each pair of qubits consisting of one qubit from the first register and the corresponding qubit from the second register.
By direct computation, we have 
\begin{align}
    (\mathbb{I}_{2^N}\otimes \bra{0^N}) P(A\otimes B)P^\dagger (\mathbb{I}_{2^N}\otimes \ket{0^N} )=A\circ B.
\end{align}
Therefore, by direct computation, one can verify that $(
P\otimes \mathbb{I}_{2^{a+b}}) (\mathbb{I}_{2^b}\otimes U_{A}\otimes \mathbb{I}_{2^N})(\mathbb{I}_{2^a}\otimes U_{B}\otimes \mathbb{I}_{2^N})(P^\dagger\otimes \mathbb{I}_{2^{a+b}})$ is the desired block encoding. 
\end{proof}

\subsection{Quantum singular value transformation\label{chapt-3-sec-qsvt}}

Now we understand how to implement matrices on the quantum computer and some arithmetic methods among these matrices, here we further introduce how one can implement matrix functions on the quantum computer.
The method is called the quantum singular value transformation (QSVT), which is a powerful framework that can unify most known quantum algorithms.

For machine learning applications, we mostly deal with the real matrices. The computational cost of QSVT for the real matrix case is summarized in the following theorem.
For a matrix $A$, consider its singular value decomposition $A=\sum_{i}\sigma_i |\psi_i\rangle\langle \phi_i|$.
Given a function $P(x)$, we use $P^{(SV)}(A)$ to represent $P^{(SV)}(A)=\sum_{i} P(\sigma_i) |\psi_i\rangle\langle \phi_i|$.

\begin{fact}[Quantum singular value transformation - Real matrix case, \cite{gilyenquantum2019}]\label{QSVT}
        Suppose that $U_A$ is an $(\alpha, a, \varepsilon)$-block-encoding of a real matrix $A$.
        If $\delta \geq 0$ and $P:\mathbb{R}\rightarrow \mathbb{C}$ is a d-degree polynomial satisfying that
        \begin{align}
            \mathrm{for\ all}\ x \in[-1,1]:\left|P(x)\right| \leq \frac{1}{4},
        \end{align}
       then there is a quantum circuit $\tilde{U}$, which is an $(1, a+3, 4d \sqrt{\varepsilon / \alpha}+\delta)$-block-encoding of $P^{(SV)}(A / \alpha)$, and consists of d applications of $U_A$ and $U^{\dagger}_A$ gates.
        Further, the description of such a circuit can be computed classically in time $\mathcal{O}(\operatorname{poly}(d, \log (1 / \delta)))$.
\end{fact}

Notice that if the block-encoded matrix $A$ is Hermitian, the singular value transformation is equivalent to the eigenvalue transformation.
In this case, we can directly implement the matrix function via QSVT.

In the following, we introduce some applications of QSVT method.
There are several important applications, however, in this tutorial, we focus on introducing some applications that can be applied to machine-learning related tasks.
The first is matrix inversion, which has been widely used in traditional machine learning methods such as principal component analysis. Note that for a general matrix, we actually mean to implement the Moore-Penrose pseudoinverse, i.e., the inverse of all singular values.

\begin{lemma}[Matrix inversion, simplified \cite{gilyenquantum2019}]\label{chapt2:lemma:matrix-inversion}
    Let $U_A$ be a $(1,a,0)$-block encoding of matrix $A$.
    Further, for simplicity, assume the nonzero singular values of $A$ are lower bounded by $\delta>0$.
    Let $0\leq \epsilon\leq \delta \leq \frac{1}{2}$.
    One can construct a $(1/\delta,a+2,\epsilon)$-block encoding of $A^{-1}$ by using $\mathcal{\tilde{O}}(\frac{1}{\delta}\log(\frac{1}{\epsilon}))$ times of $U_A$ and $U_A^\dagger$.
\end{lemma}

\begin{proof}[Proof sketch of Lemma~\ref{chapt2:lemma:matrix-inversion}.]
    This can be achieved by finding a good polynomial approximation for the function $1/x$. One can not find such a polynomial on the whole interval $[-1,1]$, however, such a polynomial exists on the interval $[-1,-\delta]\cup [\delta,1]$ for some $\delta>0$.
\end{proof}

The second application of QSVT is the nonlinear amplitude transformation.
As mentioned in Chapter~\ref{cha3:subsec:q-read-in}, to deal with classical data, there are several ways to encode them into quantum. Here, we focus on the amplitude encoding case described in Chapter~\ref{cha2:sec3:amplitude_encode}, especially for the real amplitudes. The nonlinear transformation is achieved by combining the diagonal block encoding and QSVT.

\begin{fact}[Diagonal block encoding of amplitudes,  \cite{guononlinear2021, rattewnonlinear2023}\label{block encoding.amplitudes}]
    Given a state preparation unitary $U_{\psi}$ of an $N$-qubit state $\ket{\psi}=\sum_{j=1}^{2^N} \psi_j \ket{j}$, where $\{\psi_j\}$ are real, $\norm{\psi}_2=1$,
    one can construct an $(1,N+2,\epsilon)$-encoding of the diagonal matrix $A=\mathrm{diag}(\psi_1, \dots, \psi_{d})$ with $\mathcal{O}(N)$ circuit depth and $\mathcal{O}(1)$ times of controlled-$U$ and controlled-$U^\dagger$.
\end{fact}

As a straightforward generalization, one can replace the state preparation unitary with the general state preparation encoding, mentioned in Definition~\ref{def.stateencoding}.
By constructing the block encoding of amplitudes, one can implement many functions onto these amplitudes via QSVT. A direct application is performing the neural network on the quantum computer, as will be detailed in Chapter~\ref{chapt6:sec:quantum_transformer}.

\section{Code Demonstration}\label{Chapter2:preliminary-code}

This section provides code implementations for key techniques introduced earlier, including quantum read-in strategies and block encoding, to give readers the opportunity to practice and deepen their understanding. 

\subsection{Read-in implementations}

This subsection demonstrates toy examples of implementing data encoding methods in quantum computing, as discussed in earlier sections. Specifically, we cover basis encoding, amplitude encoding, and angle encoding from Chapter~\ref{cha2:sec3:angle_encode}. These examples aim to provide readers with hands-on experience in applying quantum data encoding techniques.

\subsubsection{Basis encoding}

PennyLane provides built-in support for basis encoding through its `BasisEmbedding' function. Below is the Python code demonstrating the basis encoding for the integer 6.
\begin{lstlisting}
import pennylane as qml

dev = qml.device("default.qubit", range(3))
@qml.qnode(dev)
def circuit(x):
  qml.BasisEmbedding(x, range(3))
  return qml.state()

# Call the function
circuit(6)
\end{lstlisting}

\subsubsection{Amplitude encoding}

PennyLane offers built-in support for amplitude encoding via the `AmplitudeEmbedding' function. Below is a Python example demonstrating amplitude encoding for a randomly generated complex vector.

\begin{lstlisting}[language=Python]
import pennylane as qml
import numpy as np

# Number of qubits
n_qubits = 8

# Define a quantum device with 8 qubits
dev = qml.device("default.qubit", wires=n_qubits)

@qml.qnode(dev)
def circuit(x):
    qml.AmplitudeEmbedding(features=x, wires=range(n_qubits), normalize=True, pad_with=0.)
    return qml.state()

# Generate a random complex vector of length 2^n_qubits
x_real = np.random.normal(loc=0, scale=1.0, size=2**n_qubits)
x_imag = np.random.normal(loc=0, scale=1.0, size=2**n_qubits)
x = x_real + 1j * x_imag

# Execute the circuit to encode the vector as a quantum state
circuit(x)
\end{lstlisting}

\subsubsection{Angle encoding}
 
PennyLane provides built-in support for angle encoding via the `AngleEmbedding' function. Below is a Python example demonstrating angle encoding for a randomly generated real vector.

\begin{lstlisting}[language=Python]
import pennylane as qml
import numpy as np

# Number of qubits
n_qubits = 8

# Define a quantum device with 8 qubits
dev = qml.device("default.qubit", wires=n_qubits)

@qml.qnode(dev)
def circuit(x):
    qml.AngleEmbedding(features=x, wires=range(n_qubits), rotation="X")
    return qml.state()

# Generate a random real vector of length n_qubits
x = np.random.uniform(0, np.pi, (n_qubits))

# Execute the circuit to encode the vector as a quantum state
circuit(x)
\end{lstlisting}

\subsection{Block encoding}
Here, we provide an example of how we may construct a block encoding.
We construct the block encoding via the linear combination, as \cref{blockencoding.linearcombination}.
We use PennyLane to keep consistency, yet there are many other platforms that are available as well.
Please note that it is time-consuming to do the Pauli decomposition (for an $N$-qubit matrix, it takes time $\mathcal{O}(N 4^N)$), so we suggest not trying a large matrix with this method.

\begin{lstlisting}[language=Python]
import numpy as np
import pennylane as qml
import matplotlib.pyplot as plt

a = 0.36
b = 0.64

# matrix to be decomposed
A = np.array(
    [[a,  0, 0,  b],
     [0, -a, b,  0],
     [0,  b, a,  0],
     [b,  0, 0, -a]]
)

# decompose the matrix into sum of Pauli strings
LCU = qml.pauli_decompose(A)
LCU_coeffs, LCU_ops = LCU.terms()

# normalized square roots of coefficients
alphas = (np.sqrt(LCU_coeffs) / np.linalg.norm(np.sqrt(LCU_coeffs)))

dev = qml.device("default.qubit", wires=3)

# unitaries
ops = LCU_ops
# relabeling wires: 0 --> 1, and 1 --> 2
unitaries = [qml.map_wires(op, {0: 1, 1: 2}) for op in ops]

@qml.qnode(dev)
def lcu_circuit():  # block_encode
    # PREP
    qml.StatePrep(alphas, wires=0)

    # SEL
    qml.Select(unitaries, control=0)

    # PREP_dagger
    qml.adjoint(qml.StatePrep(alphas, wires=0))
    return qml.state()

print(np.real(np.round(output_matrix,2)))
\end{lstlisting}

\section{Bibliographic Remarks}\label{chap-preliminary-sec:Remark}

We end this chapter by discussing the recent advancements in efficiently implementing fundamental components of quantum computing. For clarity, we begin with a brief discussion of advanced quantum read-in and read-out protocols, which are crucial for efficiently loading and extracting classical data in the pipeline of quantum machine learning. Next, we review the latest progress in quantum linear algebra.

\subsection{Advanced quantum read-in protocols}
Although conventional read-in protocols offer feasible solutions for encoding classical data into quantum computers, they typically face two key challenges that limit their broad applicability for solving practical learning problems. To address these limitations, initial efforts have been made to develop more advanced quantum read-in protocols.

\smallskip

\noindent\textit{Challenge I: high demand for quantum resources}. Encoding methods like amplitude encoding and basis encoding presented in Chapter~\ref{chapter2:Sec-2.3.2-readout} generally suffer from high quantum resource requirements. While amplitude encoding is highly compact in terms of qubit requirements, the trade-off is the requirement of an exponential number of quantum gates with the data size to prepare an exact amplitude-encoded state. In contrast, while basis encoding can be implemented with a small number of quantum gates, it requires a large number of qubits proportional to the input size. The high demand for either quantum gates or qubit counts makes these basic encoding strategies infeasible for practical use.

\noindent\textit{Challenge II: insufficient nonlinearity}. While quantum mechanics is inherently linear, most practical machine learning models require nonlinearity to capture complex data patterns effectively. Conventional encoding methods like angle encoding introduce some degree of nonlinearity; however, the representational power remains limited due to the linear nature of quantum operations and limited circuit depth.

\smallskip

For Challenge I, a practical alternative is the approximate amplitude encoding (AAE) \citep{nakaji2022approximate}. Instead of implementing exact amplitude encoding, AAE trains a parameterized quantum circuit with a constrained depth to approximate the desired quantum state with high fidelity. The training process optimizes the fidelity between the target state and the approximate state, ensuring that the representation error remains within a small bound. 

\smallskip
For Challenge II, techniques like \textit{data re-uploading} \citep{perez2020data} have been developed. \textit{Data re-uploading} involves feeding the same classical data into the quantum circuit multiple times, interspersed with trainable quantum operations. By alternating data encoding with trainable transformations, this approach allows the quantum model to capture non-linear relationships more effectively without requiring additional qubits. Additionally, neural quantum embedding \citep{hur2024neural} has been proposed, which leverages classical deep learning techniques to learn optimal quantum embeddings, effectively separating non-linearly separable classes of data.

\smallskip
To address both Challenges I \& II, hybrid encoding strategies have been introduced to leverage the respective advantages of each encoding method. For instance, basis-amplitude encoding combines basis encoding for discrete random variables with amplitude encoding for high-precision probabilities, effectively encoding both categorical and continuous features without requiring additional qubits \citep{schuld2018information}. Another widely used strategy involves classical preprocessing methods for high-dimensional data, such as principal component analysis (PCA) \citep{abdi2010principal}, to reduce input dimensionality before applying quantum encoding. This preprocessing step reduces the overall quantum resource requirements while preserving relevant information.

In addition to fixed encoding strategies, learning-based approaches have emerged to dynamically adjust data encoding for specific tasks. For example, \citet{lloyd2020quantum} achieves task-specific quantum embeddings by incorporating learnable parameters into the encoding layers, which are optimized to maximize class separability in Hilbert space. This technique is analogous to classical metric learning. Following this routine, a quantum few-shot embedding framework \citep{liu2022embedding} has been proposed to encode classical data into quantum states, which can be generalized to the downstream quantum machine learning tasks. These methods enable quantum circuits to adapt their encodings dynamically, improving efficiency and performance.

\subsection{Advanced quantum read-out protocols}
Conventional quantum read-out protocols often face significant challenges, including high computational overhead and resource inefficiencies. Below, we discuss the primary challenges and discuss solutions in two quantum read-out protocols: QST and observable estimation.

\smallskip

\noindent\textit{Challenge I: High computational overhead of QST.} QST aims to reconstruct the density matrix of a quantum state, but this becomes computationally infeasible as the system size increases. This is because the required number of measurements and the classical memory grows exponentially with the number of qubits. 

\noindent\textit{Challenge II: Resource inefficiency in observable estimation.} The required number of measurements for observable estimation grows linearly with the number of Pauli terms in the observable. For observables where the number of Pauli terms substantially increases with the system size, the measurement cost becomes prohibitive.

\smallskip

For Challenge I, the key idea is to focus on representing only a subspace of the quantum space, effectively capturing task-relevant properties while reducing the computational cost. 
For example, in many QML algorithms, such as the HHL algorithm for solving linear systems~\citep{harrow2009quantum} and quantum singular value decomposition~\citep{rebentrost2018quantum}, the solution state exists within the row or column space of the input matrix. When the input matrix is low-rank, state tomography can be obtained efficiently~\citep{zhang2021quantum} as the linear combination of a complete basis chosen from the input matrix.
Besides, an effective technique is matrix product state (MPS) tomography \citep{lanyon2017efficient,orus2019tensor}, which leverages the fact that many practical quantum states, such as those in Ising models or low-entanglement systems, can be efficiently represented with a reduced number of parameters. By focusing on states with limited entanglement, MPS tomography reconstructs the state using only a polynomial number of measurements with the qubit counts. 

Another promising approach is the use of neural networks to parameterize quantum states. Neural quantum states allow for the efficient representation and reconstruction of density matrices, particularly for complex or high-dimensional quantum systems. For instance, Restricted Boltzmann Machines \citep{fischer2012introduction} and Transformer \citep{vaswani2017attention} have been applied to approximate the probability of measurement outcome and density matrices \citep{torlai2018neuralnetwork, schmale2022efficient, wang2022predicting, zhao2023provable}. These approaches are particularly effective for systems that are difficult to capture using traditional methods.

\smallskip

For Challenge II, a measurement reduction technique can be applied by exploiting the commutativity of Pauli operators. When multiple Pauli terms commute, they can be measured simultaneously within the same measurement basis, significantly reducing the total measurement cost \citep{kandalaHardwareefficientVariationalQuantum2017a,verteletskyi2020measurement}. This approach has been widely adopted in hybrid quantum-classical algorithms, such as variational quantum Eigensolvers (VQE) \citep{cerezo2021variational}, where Hamiltonians are decomposed into sums of Pauli terms. Grouping commuting terms into clusters allows for efficient measurement strategies while preserving accuracy.

In addition to measurement grouping, adaptive measurement strategies further improve resource allocation during expectation value estimation. The key observation is that not all Pauli terms contribute equally to the total observable—terms with higher variance require more measurement shots for reliable estimation, while low-variance terms can be measured with fewer shots. Building on this insight, adaptive shot allocation techniques \citep{rubin2018application, arrasmith2020operator,qian2024shuffle} dynamically distribute measurement resources across Pauli terms based on their statistical properties and achieve more accurate estimations with a finite measurement budget.

\subsection{Advanced quantum linear algebra}

Quantum linear algebra, based on the block encoding and quantum singular value transformation framework, has proven its power for the design of quantum algorithms.
Compared to the traditional subroutines like quantum phase estimation and quantum arithmetic \citep{kitaev1995quantum, perez2017quantumarithmetic}, quantum linear algebra can exponentially improve the dependency on precision \citep{gilyenquantum2019}.
However, a major drawback is that it can only deal with the singular values of block-encoded matrices.

A natural consideration is to generalize the singular value transformation to the eigenvalue transformation.
One strong motivation from the application aspect for this is to solve the differential equations on the quantum computer \citep{liu2021efficient, childs2021highprecision, an2021quantumaccelerated, jin2022partialdifferential, shang2024design}.
This remains an active research field.
Quantum eigenvalue processing, proposed by \citep{low2024quantumeigen}, focuses on matrices with real spectra and Jordan forms, in which they prepare the Faber history state to achieve efficient eigenvalue transformation over the complex plane.
\citep{an2023linearcombination, an2024laplacetransform} shows that simulating a general class of non-unitary dynamics can be achieved by the linear
combination of Hamiltonian simulation (LCHS).

Another approach is to broaden the range of functions that can be implemented by quantum linear algebra. Quantum phase processing, proposed by \citep{wang2023phaseprocessing}, can directly apply arbitrary trigonometric transformations to eigenphases of a unitary operator. Similar results have been independently obtained by \citet{motlagh2024generalized}.
In addition, \citet{rossi2022multivariable} investigates how to implement multivariate functions.
For the application, a representative example is the multivariate state preparation achieved by  \citep{mori2024efficient}, enabling the amplitude encoding of classical multivariate data.

In Chapter~\ref{chap2:preliminary-sec:linearAlgebra}, we introduce the concept of diagonal block encoding, which can convert a state preparation unitary into a block encoding. As the efficient construction of block encodings is a prerequisite for achieving end-to-end quantum advantage, an important research direction is to investigate which types of matrices can be efficiently prepared. By leveraging state-of-the-art techniques in quantum state preparation \citep{zhang2022quantum,sun2023asymptotically} and the linear combination of unitaries \citep{childs2012hamiltonian}, it is possible to efficiently construct block encodings for certain classes of matrices \citep{guseynov2024efficient, guseynov2024explicitgate}. Additionally, explicit constructions have been explored for specific types of sparse matrices \citep{camps2023explicit}.

\chapter{Quantum Kernel Methods}
\label{Chapter3:kernel}

The fundamental goal of machine learning (ML) algorithms is to learn the underlying feature representations embedded in the training data, allowing data points to be effectively modeled using simple models like linear classifiers. \textbf{Kernel methods} are a powerful approach to achieving this by enabling non-linear patterns to be captured in a computationally efficient manner. In kernel methods, a kernel function is defined as the inner product between the high-dimensional feature representations of data points. These feature representations are generated by a hidden feature map that transforms the original data into a higher-dimensional space where complex patterns become easier to identify and model. The kernel function, therefore, serves as a measure of similarity between data points in this transformed space.

The effectiveness of kernel methods heavily depends on the hidden feature map's ability to capture relevant patterns in the data. The more effectively this feature map can reveal the underlying structure, the better the kernel method's performance in learning and generalizing from the data. However, classical kernel methods are inherently limited by the types of patterns they can recognize, as these are constrained by classical computational frameworks. Essentially, classical models excel at detecting patterns they are specifically designed to recognize but may struggle with patterns that deviate from this framework.

In contrast, quantum mechanics is known for generating complex, non-intuitive patterns that are often beyond the reach of classical algorithms. Quantum systems can produce statistical correlations that are computationally challenging—or even impossible—for classical computers to replicate. This suggests that employing quantum circuits as hidden feature maps could enable the detection of patterns that are difficult or impractical for classical models to capture. By leveraging quantum circuits, we can potentially access new regions of the feature space, leading to improved pattern recognition capabilities and, consequently, better learning performance.

These insights motivate the development of \textbf{quantum kernel methods}, where both the hidden feature map and the kernel function are implemented on a quantum computer. By harnessing the unique properties of quantum mechanics, such as superposition and entanglement, quantum kernel methods have the potential to surpass their classical counterparts in specific machine learning tasks, particularly those involving highly complex or subtle patterns. This could result in more powerful models with enhanced generalization capabilities.

In this chapter, we provide a step-by-step explanation of the transition from classical kernel machines to quantum kernel machines in Chapter~\ref{sec:classical_kernel} and Chapter~\ref{sec:quantum_kernel}. Moreover, we discuss the theoretical foundation of quantum kernel machines in Chapter~\ref{chapt3:sec:theo_foundation_QK} from the aspects of expressivity and generalization of quantum kernel machines. Finally, we demonstrate simple yet illustrative code implementations on MNIST dataset.

\section{Classical Kernel Machines}\label{sec:classical_kernel}

\subsection{Motivation of kernel methods}
\begin{figure}[h!]
    \centering
    \includegraphics[width=0.98\textwidth]{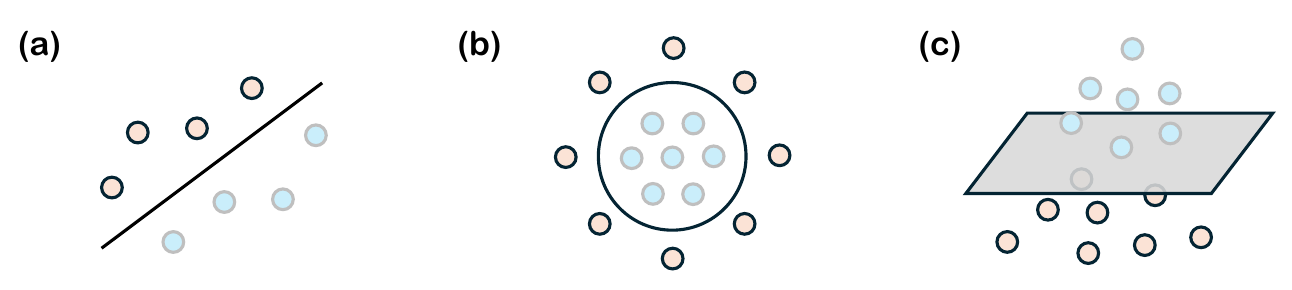}
    \caption{{\textbf{Various distributions of data points.} The left and middle panels show the cases where data points can and cannot be separated by a straight line. The right panel shows that the kernel function could map the linearly inseparable data points into the high dimensional linearly separable data points. } }
    \label{fig:hyperplane}  
    \end{figure}

Before delving into kernel machines, it is essential to first understand the motivation behind kernel methods. In many machine learning tasks, particularly in classification, the goal is to find a decision boundary that best separates different classes of data. When the data is linearly separable, this boundary can be represented as a straight line (in 2D), a plane (in 3D), or a hyperplane (in higher dimensions), as illustrated in Figure~\ref{fig:hyperplane}(a). Mathematically, given an input space $\mathcal{X}\subset \mathbb{R}^d$ with $d\ge 1$ and a target or output space $\mathcal{Y}=\{+1,-1\}$, we consider a training dataset $\mathcal{D}=\{(\bm{x}^{(i)},y^{(i)})\}_{i=1}^n \in (\mathcal{X} \times \mathcal{Y})^n$ where each data point $\bm{x}^{(i)}  \in \mathcal{X} $ is associated with a label $y^{(i)}  \in \mathcal{Y}$. For the dataset to be linearly separable, there must exist a vector $\bm{w} \in \mathbb{R}^{d}$ and a bias term $b\in \mathbb{R}$ such that
\begin{equation}
    \forall i\in [n], \quad y^{(i)}(\bm{w}^{\top} \bm{x}^{(i)}+b)\ge 0,
\end{equation}
where $\bm{w}^{\top}  \bm{x}^{(i)}$ represents the inner product of vectors $\bm{w}$ and $\bm{x}^{(i)}$.
This means that a hyperplane defined by $(\bm{w},b)$ can perfectly separate the two classes.

However, in real-world scenarios, data is often not linearly separable, as shown in Figure~\ref{fig:hyperplane}(b). The decision boundary required to separate classes may be curved or highly complex. Traditional linear models struggle with such non-linear data because they are inherently limited to creating only linear decision boundaries. This limitation highlights the need for more flexible approaches.

To address the challenge of non-linear data, one effective strategy is to transform the input data into a higher-dimensional space where the data may become linearly separable. This transformation is known as \textit{feature mapping}, denoted by
\begin{equation}\label{chapt3:eqn:feature-map}
  \phi:\bm{x} \to \phi(\bm{x})\in \mathbb{R}^D, 
\end{equation}
where the original input space $\mathcal{X}$ is mapped to a higher-dimensional feature space $\mathbb{R}^D$ with $D\ge d$. The idea is that, in this higher-dimensional space, complex patterns in the original data can be more easily identified using linear models.

However, explicitly computing the feature map $\phi(\bm{x})$ in Eqn.~(\ref{chapt3:eqn:feature-map}) can be computationally expensive, especially if the feature space is high-dimensional or even infinite-dimensional. Fortunately, many machine learning algorithms for tasks like classification or regression depend primarily on the inner product between data points, which will be explained in Chapter~\ref{chapt3:subsec:dual_rep}. In the feature space, this inner product is given by $\braket{{\phi}(\bm{x}^{(i)}),{\phi}(\bm{x}^{(j)})}$.

\begingroup
\allowdisplaybreaks
\begin{tcolorbox}[enhanced, 
  breakable,colback=gray!5!white,colframe=gray!75!black,title=Remark]
Throughout the whole tutorial, we interchangeably use $\bm{a}^{\top} \bm{b}$, $\bm{a}\cdot \bm{b}$,  $\langle \bm{a}, \bm{b} \rangle$,  and $\braket{\bm{a}|\bm{b}}$ to denote the inner production of two vectors $\bm{a}$ and $\bm{b}$.
\end{tcolorbox}
\endgroup

To circumvent the computational cost of explicitly calculating the feature map, we can use a \textit{kernel function}. A kernel function $k(\bm{x}^{(i)}, \bm{x}^{(j)})$ is defined as
\begin{equation}\label{eq:kernel_function}
  k(\bm{x}^{(i)}, \bm{x}^{(j)})=\braket{{\phi}(\bm{x}^{(i)}),{\phi}(\bm{x}^{(j)})}.
\end{equation} 
This allows us to compute the inner product in the higher-dimensional feature space indirectly, without ever having to compute ${\phi}(\bm{x})$ explicitly. This approach is commonly known as the \textit{kernel trick}. 

By using the kernel function directly within algorithms, we avoid the computational overhead of working in a high-dimensional space. The collection of kernel values for a dataset forms the kernel matrix (or Gram matrix), where each entry is given by $k(\bm{x}^{(i)}, \bm{x}^{(j)})$. This matrix is central to many kernel-based algorithms, as it captures the pairwise similarities between all training data points. 

To illustrate the kernel trick, let’s consider a simple example in a $d$-dimensional input space, where $\bm{x}=(\bm{x}_1,\cdots,\bm{x}_d)^{\top}$. Suppose we use the polynomial kernel of degree $2$, defined as
\begingroup
\allowdisplaybreaks
\begin{align}
  k(\bm{x},\bm{z})
  = & (\bm{x}^{\top}  \bm{z})^2 
  \nonumber \\
  = & (\bm{x}_1\bm{z}_1+ \cdots +\bm{x}_d\bm{z}_d)^2
  \nonumber \\
  = & \sum_{i=1}^d \sum_{j=1}^d \bm{x}_i\bm{z}_i \bm{x}_j\bm{z}_j
  \nonumber \\
  = & [\bm{x}_1^2, \cdots, \bm{x}_d^2, \sqrt{2}\bm{x}_1\bm{x}_2, \cdots, \sqrt{2}\bm{x}_d\bm{x}_{d-1}]
  \nonumber \\
  &  [\bm{z}_1^2, \cdots, \bm{z}_d^2, \sqrt{2}\bm{z}_1\bm{z}_2, \cdots, \sqrt{2}\bm{z}_d\bm{z}_{d-1}]^{\top}
  \nonumber \\
  = & {\phi}(\bm{x})^{\top}  {\phi}(\bm{z}).
\end{align}
\endgroup
Here, we see that the feature mapping, which comprises all second-order terms, takes the form as 
\begin{equation}
  {\phi}(\bm{x}) =  [\bm{x}_1^2, \cdots, \bm{x}_d^2, \sqrt{2}\bm{x}_1\bm{x}_2, \cdots, \sqrt{2}\bm{x}_{d}\bm{x}_{d-1}]^{\top}.
\end{equation}
Notably, directly computing the kernel function $(\bm{x}^{\top}  \bm{z})^2$ for a large $d$ is much more efficient than explicitly calculating the feature map ${\phi}(\bm{x})$ and then taking the inner product ${\phi}(\bm{x})^{\top} {\phi}(\bm{z})$. Specifically, using the kernel function only requires $\mathcal{O}(d)$ time, since it involves computing the dot product in the original input space $\mathbb{R}^d$. In contrast, if we were to explicitly compute the transformed feature vectors 
${\phi}(\bm{x})$ and their inner product, the time complexity could increase to $\mathcal{O}(D)$, where 
$D$ is the dimensionality of the feature space after mapping. For this example of a polynomial kernel with degree $2$, $D$ can grow to $\mathcal{O}(d^2)$. This demonstrates the computational efficiency of using the kernel trick.

\begingroup
\allowdisplaybreaks
\begin{tcolorbox}[enhanced, 
  breakable,colback=gray!5!white,colframe=gray!75!black,title=Remark]
  Throughout this manuscript, we use the notations $\mathcal{O}$ and $\Omega$ to represent the asymptotic upper and lower bounds, respectively, on the growth rate of a term, ignoring constant factors and lower-order terms.
\end{tcolorbox}
\endgroup

\subsection{Dual representation}\label{chapt3:subsec:dual_rep}
To understand why many machine learning algorithms rely primarily on the inner products between data points, we need to introduce the concept of the \textit{dual representation}. In essence, many linear parametric models used for regression or classification can be re-cast into an equivalent dual form, where the kernel function evaluated on the training data emerges naturally. 

Let’s start with a linear regression model with the training dataset $\{(\bm{x}^{(i)},y^{(i)})\}_{i=1}^n$, where the parameters are determined by minimizing a regularized sum-of-squares error function
\begin{equation}\label{chapter3:eq:Re_MSE}
    \mathcal{L}(\bm{w}) = \frac{1}{2} \sum_{i=1}^n \left(\bm{w}^{\top}  {\phi}(\bm{x}^{(i)})-y^{(i)} \right)^2 + \frac{\lambda}{2} \bm{w}^{\top}  \bm{w},
\end{equation}
where $\bm{w}^{\top}$ refers to the transpose of model parameters $\bm{w}$, ${\phi}(\bm{x}^{(i)})$ represents the feature mapping of the input $\bm{x}^{(i)}$, and  $\lambda \ge 0$  is the regularization factor that helps prevent overfitting.  

To find the optimal $\bm{w}$, we set the gradient of $\mathcal{L}(\bm{w})$ with respect to $\bm{w}$ to zero, i.e.,
\begin{equation}
    \frac{\partial \mathcal{L}(\bm{w})}{\partial \bm{w}}= \sum_{i=1}^n \left(\bm{w}^{\top}  {\phi}(\bm{x}^{(i)})-y^{(i)} \right) {\phi}(\bm{x}^{(i)}) + \lambda \bm{w} = 0,
\end{equation}
From this, we see that the solution for $\bm{w}$ can be expressed as a linear combination of the training data's feature vectors
\begin{equation}\label{chapter3:eqn:opt-para}
  \bm{w} = -\frac{1}{\lambda} \sum_{i=1}^n (\bm{w}^{\top}  {\phi}(\bm{x}^{(i)})-y^{(i)})   {\phi}(\bm{x}^{(i)}) = \sum_{i=1}^n \bm{a}^{(i)} {\phi}(\bm{x}^{(i)}) := \bm{\Phi}^{\top}   \bm{a},
\end{equation}
where $\bm{\Phi}=[{\phi}(\bm{x}^{(1)}), \cdots, {\phi}(\bm{x}^{(n)})]^{\top}$ is the design matrix, whose $i$-th row is given by ${\phi}(\bm{x}^{(i)})^{\top}$. Here, the coefficients $\bm{a}^{(i)}$ are functions of $\bm{w}$, defined as
\begin{equation}
  \bm{a}^{(i)} = - \frac{1}{\lambda} (\bm{w}^{\top} {\phi}(\bm{x}^{(i)})-y^{(i)}).
\end{equation}
Thus, instead of directly optimizing $\bm{w}$, we can reformulate the problem in terms of the parameter vector $\bm{a}$, giving rise to a dual representation. By substituting $\bm{w} = \bm{\Phi}^{\top} \bm{a}$ into the original objective function $\mathcal{L}(\bm{w})$ in Eqn.~(\ref{chapter3:eq:Re_MSE}), we obtain
\begin{equation}\label{eq:real_dual_rep}
    \mathcal{L}(\bm{a}) = \frac{1}{2} \bm{a}^{\top} \bm{\Phi} \bm{\Phi}^{\top} \bm{\Phi} \bm{\Phi}^{\top} \bm{a} - \bm{a}^{\top} \bm{\Phi} \bm{\Phi}^{\top} \bm{y} + \frac{1}{2}\bm{y}^{\top} \bm{y} + \frac{\lambda}{2} \bm{a}^{\top} \bm{\Phi} \bm{\Phi}^{\top} \bm{y},
\end{equation}
where $\bm{y}=(y^{(1)}, \cdots, y^{(n)})^{\top}$ denotes the vector representation of $n$ training labels. We define the kernel matrix ${K} = \bm{\Phi} \bm{\Phi}^{\top}$, where each element is given by
\begin{equation}
    {K}_{ij} = {\phi}(\bm{x}^{(i)})^{\top} {\phi}(\bm{x}^{(j)}) = k(\bm{x}^{(i)}, \bm{x}^{(j)}),
\end{equation}
using kernel function $k(\bm{x}, \bm{x}')$ defined by Eqn.~\eqref{eq:kernel_function}. The objective function in terms of $\bm{a}$ simplifies to
\begin{equation}
    \mathcal{L}(\bm{a}) = \frac{1}{2} \bm{a}^{\top} {K}^2 \bm{a} - \bm{a}^{\top} {K} \bm{y} + \frac{1}{2}y^{\top} \bm{y} + \frac{\lambda}{2} \bm{a}^{\top} {K} \bm{y},
\end{equation}
Setting the gradient of $\mathcal{L}(\bm{a})$ with respect to $\bm{a}$ to zero give us
\begin{equation}
    \bm{a} = ({K}+\lambda \mathbb{I}_n)^{-1} \bm{y},
\end{equation}
where $\mathbb{I}_n$ is the identity matrix of size $n\times n$.

Now, using this dual formulation, we can derive the prediction for a new input $\bm{x}$. Substituting $\bm{w}=\bm{\Phi}^{\top} \bm{a}$ in Eqn.~(\ref{chapter3:eqn:opt-para}), the prediction of $\bm{x}$ is given by
\begin{equation}\label{chapt3:eq:dual_rep_out}
    y(\bm{x}) = \bm{w}^{\top}   {\phi}(\bm{x}) = \langle \bm{\Phi}^{\top} \bm{a},  {\phi}(\bm{x}) \rangle = \bm{k}(\bm{x})^{\top}   ({K}+\lambda \mathbb{I}_n)^{-1} \bm{y},
\end{equation}
where $\bm{k}(\bm{x})\in \mathbb{R}^{n}$ is a vector with elements $\bm{k}_i(\bm{x}) = k(\bm{x}^{(i)}, \bm{x})={\phi}(\bm{x}^{(i)})^{\top}  {\phi}(\bm{x})$. This shows that the dual formulation allows us to express the solution entirely in terms of the kernel function $k(\bm{x},\bm{x}')$, rather than explicitly working with the feature map ${\phi}(\bm{x})$. This is particularly advantageous because it enables us to work in high-dimensional or even infinite-dimensional feature spaces implicitly.

In the dual formulation, we determine the parameter vector $\bm{a}$ by inverting an $n \times n$ matrix, compared to the original parameter space formulation which requires inverting a $d \times d$ matrix to determine $\bm{w}$.  Although this may not seem advantageous when $n>d$, the true benefit of the dual formulation lies in its ability to leverage the kernel trick. By expressing the solution in terms of the kernel function, we avoid the explicit computation of the feature vectors ${\phi}(\bm{x})$. This allows us to implicitly utilize feature spaces of very high, or even infinite, dimensionality, enabling the model to capture complex, non-linear relationships in the data without the associated computational cost.

\begingroup
\allowdisplaybreaks
\begin{tcolorbox}[enhanced, 
  breakable,colback=gray!5!white,colframe=gray!75!black,title=Remark]
  We standardize the notation used throughout this chapter to help readers follow the content more easily. The kernel function is represented by the lowercase letter $k$, or with subscripts $k_Q$ and $k_C$. The kernel matrix is denoted by the capital letter $K$, or with subscripts $K_Q$ and $K_C$. Additionally, we use the bold lowercase letter $\bm{k}(\bm{x})$ to represent the vector of kernel values, where each element is given by $\bm{k}_j(\bm{x}) = k(\bm{x}^{(j)}, \bm{x})$, corresponding to the training points $\bm{x}^{(j)} \in \{\bm{x}^{(1)}, \dots, \bm{x}^{(n)}\}$.
\end{tcolorbox}
\endgroup

\subsection{Kernel construction}
To utilize the kernel trick in machine learning algorithms, it is essential to construct valid kernel functions. One approach is to start with a feature mapping ${\phi}(\bm{x})$ and then derive the corresponding kernel.
For a one-dimensional input space, the kernel function is defined as
\begin{equation}
    k(\bm{x}, \bm{x}') = {\phi}(\bm{x})^{\top}{\phi}(\bm{x}') = \sum_{i=1}^D \langle {\phi}_i(\bm{x}),{\phi}_i(\bm{x}') \rangle,
\end{equation}
where ${\phi}_i(\bm{x})$ are the basis functions of the feature map.

Alternatively, kernels can be constructed directly without explicitly defining a feature map. In this case, we must ensure that the chosen function is a valid kernel, meaning it corresponds to an inner product in some (possibly infinite-dimensional) feature space. The validity of a kernel function is guaranteed by Mercer's condition.

\begin{fact}[Mercer's condition]
  \label{chap3:thm:mercer}
    Let $\mathcal{X} \subset \mathbb{R}^d$ be a compact set and let $k: \mathcal{X} \times \mathcal{X} \to \mathbb{R} $ be a continuous and symmetric function. Then, $k$ admits a uniformly convergent expansion of the form
    \begin{equation}\label{chap3:eq:kernel_unif_expan}
        k(\bm{x},\bm{x}') = \sum_{i=0}^{\infty} a_i \langle {\phi}_i(\bm{x}),{\phi}_i(\bm{x}') \rangle
    \end{equation}
    with $a_i > 0$ if and only if for any square-integrable function $c$, the following condition holds:
    \begin{equation}\label{chap3:eq:mercer}
        \int_{\mathcal{X}}\int_{\mathcal{X}} c(\bm{x}) c(\bm{x}') k(\bm{x}, \bm{x}') \mathrm{d} \bm{x} \mathrm{d} \bm{x}' \ge 0.
    \end{equation}
\end{fact}

Mercer's condition is crucial because it ensures that the optimization problem for algorithms like support vector machines (SVM) remains convex \citep{mohri2018foundations}, guaranteeing convergence to a global minimum. A condition equivalent to Mercer's condition (under the assumptions of the theorem) is that the kernel $k(\cdot, \cdot)$ be positive definite symmetric. This property is more general since it does not require any assumption about $\mathcal{X}$.
\begin{definition}[Positive definite symmetric kernels]
    A kernel $k: \mathcal{X} \times \mathcal{X} \to \mathbb{R} $  is said to be positive definite symmetric (PDS) if for any $\{\bm{x}^{(1)}, \cdots, \bm{x}^{(n)}\} \subset \mathcal{X}$, the matrix ${K} = [k(\bm{x}^{(i)}, \bm{x}^{(j)})]_{ij} \in \mathbb{R}^{n\times n} $ is symmetric positive semidefinite (SPSD).
\end{definition}
In other words, a kernel matrix $K$ associated with a PDS kernel function will always be SPSD, ensuring that the corresponding optimization problem remains well-behaved.

Below, we present several commonly used positive definite symmetric kernels.

\begin{shadedbox}
\begin{example}[Polynomial kernels]
    A polynomial kernel of degree  $m \in \mathbb{N}$ with a constant $c > 0$ is defined as
    \begin{equation}
        k(\bm{x},\bm{x}')=(\bm{x}\cdot \bm{x}' + c)^m,\quad \forall \bm{x},\bm{x}'\in \mathbb{R}^d.
    \end{equation}
    This kernel maps the input space to a higher-dimensional space of dimension $\binom{d+m}{m}$. For instance,  in a two-dimensional input space ($d = 2$) and with $m = 2$, the kernel expands as follows
    \begin{equation}\label{eq:poly_kernels}
        k(\bm{x},\bm{x}')=(\bm{x}_1\bm{x}_1' + \bm{x}_2\bm{x}_2' +c)^2 = \left[ \begin{array}{c} \bm{x}_1^2 \\ \bm{x}_2^2 \\ \sqrt{2}~\bm{x}_1\bm{x}_2 \\ \sqrt{2c}~\bm{x}_1 \\ \sqrt{2c}~\bm{x}_2 \\ c \end{array} \right] \cdot \left[ \begin{array}{c} \bm{x}_1^{'2} \\ \bm{x}_2^{'2} \\ \sqrt{2}~\bm{x}_1'\bm{x}_2' \\ \sqrt{2c}~\bm{x}_1' \\ \sqrt{2c}~\bm{x}_2' \\ c \end{array} \right]
    \end{equation}
In other words, this kernel corresponds to an inner product in a higher-dimensional space of dimension 6.
\end{example}
\end{shadedbox}

\smallskip

Thus, the features corresponding to a second-degree polynomial are the original features ($\bm{x}_1$ and $\bm{x}_2$), as well as products of these features, and the constant feature.
More generally, the features associated with a polynomial kernel of degree $m$ are all the monomials of degree at most $m$ based on the original features.

\begin{shadedbox}
\begin{example}[Gaussian kernels]
    The Gaussian kernel (or Radial Basis Function, RBF) is one of the most widely used kernels, defined as
    \begin{equation}
        \forall \bm{x},\bm{x}'\in \mathbb{R}^d, ~~~k(\bm{x},\bm{x}')=\exp\left(-\frac{\|\bm{x}-\bm{x}'\|}{2\sigma^2}\right),
    \end{equation}
    where $\sigma>0$ controls the width of the Gaussian function. 
\end{example}
\end{shadedbox}

\smallskip

Gaussian kernels are particularly effective in capturing complex non-linear patterns due to their infinite-dimensional feature space.

\begin{shadedbox}
\begin{example}[Sigmoid kernels]
  The sigmoid kernel over $\mathbb{R}^d$ is defined as:
    \begin{equation}
        k(\bm{x},\bm{x}')=\Tanh(a\cdot (\bm{x} \cdot \bm{x}')+b),\quad \forall \bm{x},\bm{x}'\in \mathbb{R}^d,
    \end{equation}
    where $a,b > 0$ are constants, and $\Tanh(c)=\frac{e^c-e^{-c}}{e^c+e^{-c}}$ is the hyperbolic tangent function, which squashes an arbitrary constant $c\in \mathbb{R}$ to a value between $-1$ and $1$.
\end{example}
\end{shadedbox}

\smallskip

This kernel is related to neural networks, as it resembles the activation function commonly used in multi-layer perceptrons as will be introduced in the next chapter. Using the sigmoid kernel with support vector machines results in a model similar to a simple neural network.
 
\begingroup
\allowdisplaybreaks
\begin{tcolorbox}[enhanced, 
  breakable,colback=gray!5!white,colframe=gray!75!black,title=Remark]
  Support Vector Machines (SVMs) are a well-known algorithm that heavily relies on kernel methods and are primarily used for classification tasks. The objective of an SVM is to identify the optimal hyperplane that separates data points from different classes with the \textit{maximum margin}. The margin is defined as the distance between the hyperplane and the closest data points from each class, known as support vectors. SVMs can be applied to both linear and non-linear classification problems. For non-linear cases, kernel functions are employed to map the data into higher-dimensional spaces, enabling the separation of data that is not linearly separable in the original space. For a detailed introduction to SVMs, please refer to \citet{mohri2018foundations}.
\end{tcolorbox}
\endgroup

\section{Quantum Kernel Machines}\label{sec:quantum_kernel}

\subsection{Motivations for quantum kernel machines}
To effectively introduce quantum kernel machines, it is essential to recognize the limitations of classical kernel machines. As discussed in Chapter~\ref{sec:classical_kernel}, classical kernel machines rely on manually tailored feature mappings, such as polynomials or radial basis functions.  However, these mappings may fail to capture the complex patterns behind the dataset. Quantum kernel machines emerge as a promising alternative, as they perform feature mapping using quantum circuits, enabling them to explore exponentially larger feature spaces that are otherwise infeasible for classical computation.

\begin{figure*}[h!]
  \centering 
  \includegraphics[width=0.96\textwidth]{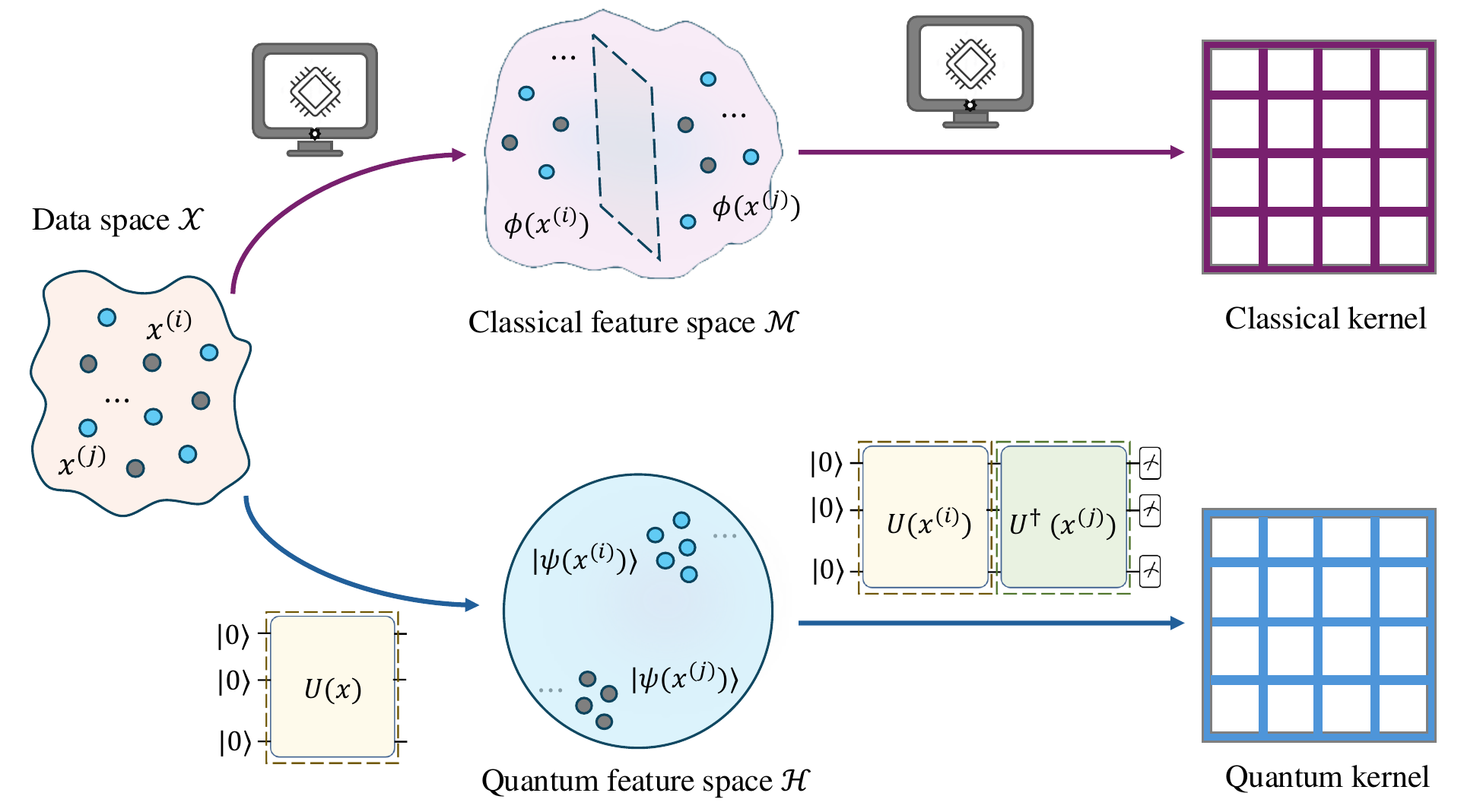}
  \caption{{\textbf{The paradigm of classical and quantum kernels.} Both of the classical and quantum kernels embed the data points from data space $\mathcal{X}$ into high-dimensional space, and then compute the kernel as the inner product of feature maps. The quantum kernel leverages quantum circuits to achieve this goal, as indicated by the blue color. }}
  \label{chap3:fig:QK_circuit}
\end{figure*}

Another crucial characteristic of quantum kernels is that they can be effectively implemented on near-term quantum devices, making them a practical tool for exploring the utility of near-term quantum technologies.

\subsection{Quantum feature maps and quantum kernel machines}
The key difference between quantum kernel machines and classical kernel machines lies in how the feature mapping is performed. In the quantum context, a feature map refers to the injective encoding of classical data  $\bm{x} \in \mathbb{R}^d$ into a quantum state $\ket{{\phi}(\bm{x})}=U(\bm{x})\ket{\phi}$ on an $N$-qubit quantum register, where $U(\bm{x})$ refers to the physical operation or quantum circuit that depends on the data $\bm{x}$. This feature map is implemented on a quantum computer and produces quantum states, which are referred to as quantum feature maps.

\begin{definition}
  [Quantum feature map]
  \label{cha3:def:quan_feat_map}
  Given an $N$-qubit quantum system initialized in state $\ket{\psi}$, let $\bm{x}\in \mathcal{X} \subset \mathbb{R}^d$ be classical data. The quantum feature map is defined as the mapping
  \begin{align}
    & \phi:\mathcal{X} \to \mathcal{F},
    \nonumber \\
    \phi(\bm{x}) = & \ket{{\phi}(\bm{x})}\bra{{\phi}(\bm{x})} = \rho(\bm{x}),
  \end{align}
  where $\mathcal{F}$ is the space of complex-valued $2^N \times 2^N$ matrices equipped with the Hilbert-Schmidt inner product $\braket{\rho, \sigma} = \Tr(\rho \sigma)$ for $\rho, \sigma \in \mathcal{F}$. In addition, the state $\ket{{\phi}(\bm{x})}$ can be implemented by applying a data-encoding quantum circuit $U(\bm{x})$ introduced in Chapter~\ref{cha3:subsec:q-read-in} on an initial state $\ket{\psi}$, leading to the expression of $\ket{{\phi}(\bm{x})}=U(\bm{x})\ket{\psi}$.
\end{definition}

Recall that one way of constructing kernels is adopting the inner product of the defined feature mappings. Using the Hilbert-Schmidt inner product from Definition~\ref{cha3:def:quan_feat_map}, the quantum kernel is defined as follows.

\begin{definition}
  [Quantum Kernel]
  \label{cha3:def:quan_kernel}
  Let $\phi$ be a quantum feature map over the domain $\mathcal{X}$. The quantum kernel $k_Q$ is the inner product between two quantum feature maps $\rho(\bm{x})$ and $ \rho(\bm{x}')$ for data points $\bm{x}, \bm{x}' \in \mathcal{X}$,
  \begin{align}
    & k_Q: \mathcal{X} \times \mathcal{X} \to \mathbb{R},
    \nonumber \\
    k_Q(\bm{x}, \bm{x}') = & \Tr(\rho(\bm{x}) \rho(\bm{x}')) = \left|\braket{{\phi}(\bm{x})|{\phi}(\bm{x}')}\right|^2.
  \end{align}
\end{definition}

To justify the term `kernel', we need to show that the quantum kernel is indeed a \textit{positive definite function}. 
A quantum kernel can be expressed as the product of a complex-valued kernel $\hat{k}_Q(\bm{x},\bm{x}')=\braket{{\phi}(\bm{x})|{\phi}(\bm{x}')} \in \mathbb{C}$ and its complex conjugate $\hat{k}_Q(\bm{x},\bm{x}')^*=\braket{{\phi}(\bm{x})|{\phi}(\bm{x}')}^*=\braket{{\phi}(\bm{x}')|{\phi}(\bm{x})}$. Since the product of two kernels is known to be a valid kernel, it suffices to show that $\hat{k}_Q(\bm{x},\bm{x}')$ is a valid complex-valued kernel and satisfies positive definiteness. For any $\bm{x}^{(i)} \in \mathcal{X}$, $i = 1,\cdots,n$, and any coefficients $c_i\in \mathbb{C}$, we have 
\begin{align}
  \sum_{i,j} c_i c_j^* \left(\hat{k}_Q(\bm{x}^{(i)},\bm{x}^{(j)}) \right) = & \sum_{i,j} c_i c_j^*\braket{{\phi}(\bm{x}^{(i)})|{\phi}(\bm{x}^{(j)})}
  \nonumber \\
  = & \left( \sum_i c_i \bra{{\phi}(\bm{x}^{(i)})} \right) \left( \sum_j c_j^* \ket{{\phi}(\bm{x}^{(j)})} \right)
  \nonumber \\
  = & \left\| \sum_i c_i^*  \ket{{\phi}(\bm{x}^{(i)})} \right\|^2 \ge 0.
\end{align}
This inequality confirms that $\hat{k}_Q(\bm{x},\bm{x}')$ satisfies Mercer's condition to be a valid kernel as illustrated in Eqn.~\eqref{chap3:eq:mercer}. Therefore, the quantum kernel $k_Q(\bm{x},\bm{x}')$ is also a valid kernel.

\begin{figure}[h!]
    \centering
    \includegraphics[width=0.9\textwidth]{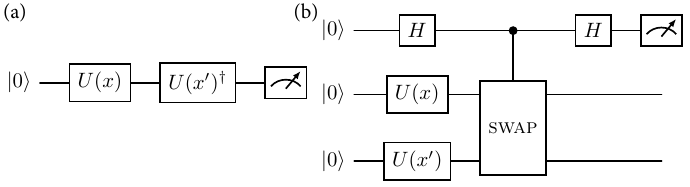}
    \caption{{\textbf{Two methods for computing the inner product of the kernel.} (a) Loschmidt echo test. (b) Swap test.  } }
    \label{fig:swap_test} 
    \end{figure}
  
The inner product between quantum states can be efficiently estimated on quantum computers using techniques such as Loschmidt echo test \citep{kusumoto2021experimental} and SWAP test \citep{blank2020quantum}. Both methods correspond to distinct quantum circuit architectures, as illustrated in Figure~\ref{fig:swap_test}.    

One key merit of quantum kernels is that their derivation does not require the explicit representation of the quantum feature maps. Instead, it relies only on the construction of the associated quantum circuits. This aligns with the essence of kernel methods: while feature mappings can be computationally complex, the kernel function itself must remain efficient to evaluate.

Below, we outline the core steps for constructing a quantum kernel.
\begingroup
\allowdisplaybreaks
\begin{tcolorbox}[colback=blue!5!white,colframe=blue!75!black,title=General construction rules of quantum kernels]
There are three steps to construct a quantum kernel:
\begin{itemize}
    \item[1.] {Quantum feature map construction}. Design a data-dependent quantum circuit $U(\bm{x})$ to encode classical input data $\bm{x}$ into the amplitudes or parameters of a quantum state $\ket{{\phi}(\bm{x})}=U(\bm{x})\ket{{\psi}}$ where the initial state $\ket{{\psi}}$ is typically $\ket{0}^{\otimes N}$. 
    \item[2.] {Kernel evaluation}. The quantum kernel is typically defined as the inner product of quantum states corresponding to two data points. Mathematically, this can be expressed as $k_Q(\bm{x},\bm{x}')= \left| \braket{{\phi}(\bm{x})|{\phi}(\bm{x}')}\right|^2$.
    \item[3.] {Post-processing}. After executing the quantum circuit for different input pairs, measure the output and calculate the kernel matrix. This matrix will then be used in machine learning models, such as SVM.
  \end{itemize}
\end{tcolorbox}
\endgroup

Below is a simple example of a quantum kernel with an angle encoding feature map introduced in Chapter~\ref{cha3:subsec:q-read-in}.
\begin{shadedbox}
\begin{example}[Single-qubit kernel]
  Consider an embedding that encodes a scalar input $x \in \mathbb{R}$ into the quantum state of a single qubit. The embedding is implemented by the Pauli-X rotation gate $\RX(x)=e^{-ix\sigma_x/2}$, where $\sigma_x$ is the Pauli-X operator. The quantum feature map is then given by $\phi:x \to \rho(x)=\ket{\phi(x)}\bra{\psi(x)}$ with
  \begin{align}
    \ket{\phi(x)}= & e^{-ix\sigma_x/2}\ket{0}
    \nonumber \\
    = & \left( \cos(x) \mathbb{I} - i\sin(x)\sigma_x  \right) \ket{0}
    \nonumber \\
    = &  \cos(x) \ket{0} - i\sin(x) \ket{1},
  \end{align}
  and hence the quantum kernel yields
  \begin{equation}
    k(x,x')= \left| \cos\left(\frac{x}{2}\right)\cos\left(\frac{x'}{2}\right) + \sin\left(\frac{x}{2}\right)\sin\left(\frac{x'}{2}\right) \right|^2 = \cos\left(\frac{x-x'}{2}\right)^2,
    \nonumber 
  \end{equation}
  which is a translation invariant squared cosine kernel. 
\end{example}
\end{shadedbox}

\subsection{Relation between quantum and  classical kernel machines}
An intuitive way to understand the connections and differences between classical and quantum kernel machines is by comparing their fundamental components.  As illustrated in Figure~\ref{chap3:fig:QK_circuit}, both types of kernel machines involve four fundamental components: input, feature mapping, kernel matrix, and the computation process. Table~\ref{cha3:tab:Q-C-Kernel} summarizes how these components are implemented in classical and quantum kernel machines.

\begin{table}[h!]
  \centering 
  \caption{Comparison between classical and quantum kernel machines.}
  \label{cha3:tab:Q-C-Kernel}
  \footnotesize
  \begin{tabular}{l|l|l}
    \toprule 
    \multicolumn{1}{c}{} & \multicolumn{1}{|c|}{Classical Kernel} & \multicolumn{1}{c}{Quantum Kernel} \\ \midrule
    Input   &  classical data $ \{\bm{x}^{(i)}\}_{i=1}^n \in \mathbb{R}^d $  &  classical data $ \{\bm{x}^{(i)}\}_{i=1}^n \in \mathbb{R}^d$ \\  
    Feature  &  Real vector ${\phi}(\bm{x}) \in \mathbb{R}^{D}$  &  Complex vector $\ket{{\phi}(\bm{x})} \in \mathbb{C}^{2^N}$                  \\  
    Kernel  & $n$-dimensional real matrix $K_C$                         &  $n$-dimensional real matrix $K_Q$ \\ 
    Computation & Digital logical circuits ${\phi}(\bm{x})$                &  Quantum circuits $U(\bm{x})\ket{{\psi}}$ \\
    \bottomrule
  \end{tabular}
\end{table}

The main distinctions between classical and quantum kernel machines lie in the computation processes for feature mapping and kernel matrix construction, as outlined below.
\begin{itemize}
  \item \textit{Classical versus quantum feature maps}. Quantum feature maps encode data into quantum states, resulting in exponentially large complex-valued vectors $\ket{{\phi}(\bm{x})} \in \mathbb{C}^{2^N}$, whereas classical feature maps operate in finite-dimensional real-valued spaces ${\phi}(\bm{x}) \in \mathbb{R}^{D}$. Although quantum feature maps can theoretically be simulated on classical computers by separating real and imaginary parts, this simulation becomes computationally infeasible as the number of qubits grows. Even feature maps generated by shallow quantum circuits are hard to simulate efficiently on classical hardware, demonstrating the inherent computational complexity of quantum feature maps.
  \item \textit{Classical versus quantum kernels}. The kernel function is determined by the feature mapping, but its computational properties differ significantly between classical and quantum methods. A key merit of kernel methods is that they allow the use of complex feature maps while maintaining efficient kernel evaluations.
  Quantum kernels leverage quantum circuits to compute the inner product of quantum states, enabling the recognition of intricate patterns that classical kernels fail to capture. If a quantum kernel is computationally hard to evaluate classically, it offers a significant advantage by exploiting quantum computing’s ability to process complex data representations efficiently.
\end{itemize}

\begingroup
\allowdisplaybreaks
\begin{tcolorbox}[enhanced, 
  breakable,colback=gray!5!white,colframe=gray!75!black,title=Remark]
  The efficiency discussed here refers to the computational time within the respective classical or quantum frameworks. Specifically,
  \begin{itemize}
    \item \textit{Classical Efficiency}: Determined by the depth of digital logical circuits used for feature mapping and kernel computation.
    \item \textit{Quantum Efficiency}: Determined by the depth of quantum circuits required to achieve the same tasks.
  \end{itemize}
  A computation process is considered efficient if it can be completed in polynomial time relative to the problem size in its corresponding framework (classical or quantum). 
\end{tcolorbox}
\endgroup

\subsection{Concrete examples of quantum kernels}

To better understand the concept of a quantum kernel, let’s examine the kernels associated with common information encoding strategies used in quantum machine learning. It is important to note that some kernels cannot be efficiently computed on classical computers \citep{liu2021rigorous}. While such results are significant, the question of which quantum kernels are practically useful for real-world problems remains an open challenge.

In the following examples, we will first review the various encoding strategies introduced in Chapter~\ref{cha3:subsec:q-read-in}, and then present the corresponding quantum kernels.

\begin{shadedbox}
  \begin{example}
    [Quantum kernel with basis encoding]
    Given a classical binary vector $\bm{x}=(\bm{x}_1, \cdots, \bm{x}_d) \in \{0,1\}^d$, the quantum feature mapping related to basis encoding refers to
    \begin{equation}
      \ket{{\phi}(\bm{x})}=\ket{\bm{x}_{1}, \cdots, \bm{x}_{d}},
    \end{equation} 
    and the induced quantum kernel yields
    \begin{equation}
      k(\bm{x},\bm{x}') = |\braket{{\phi}(\bm{x})|{\phi}(\bm{x}')}|^2 = \delta_{\bm{x}\bm{x}'},
    \end{equation}
    where $\delta_{\bm{x}\bm{x}'} = 1 $ if $\bm{x}=\bm{x}'$ and otherwise $0$.
  \end{example}
\end{shadedbox}

The basis encoding requires $N=d$ qubits. This kernel function is a very strict similarity measure on input space, and arguably not the best choice of data encoding for quantum machine learning tasks.

\begin{shadedbox}
  \begin{example}
    [Quantum kernel with amplitude encoding]
    Given a vector $\bm{x}=(\bm{x}_{1}, \cdots, \bm{x}_{d}) \in \mathbb{R}^{d}$, the quantum feature mapping related to amplitude encoding refers to
    \begin{equation}
      \ket{{\phi}(\bm{x})}=\sum_{i=1}^{d} \frac{\bm{x}_{i}}{\|\bm{x}\|_2} \ket{i},
    \end{equation} 
    where $\|\bm{x}\|_2$ is the Euclidean norm. The related quantum kernel is given by
    \begin{equation}
      k(\bm{x},\bm{x}') = |\braket{\bm{x}|\bm{x}'}|^2 = |\langle \bm{x},\bm{x}'\rangle|^2.
    \end{equation}
  \end{example}
\end{shadedbox}
The amplitude encoding requires $N=\lceil \log2(d) \rceil$ qubits. This encoding strategy leads to an identity feature map, which can be implemented by a non-trivial quantum circuit (for obvious reasons also known as “arbitrary state preparation”), which takes time $\mathcal{O}(d)$ in the worst case. Besides, this quantum kernel does not add much power to a linear model in the original feature space, and it is more of interest for theoretical investigations that want to eliminate the effect of the feature map.

\begin{shadedbox}
  \begin{example}
    [Quantum kernel with angle encoding]
    Given a vector $\bm{x}=(\bm{x}_{1}, \cdots, \bm{x}_{d}) \in \mathbb{R}^{d}$, the quantum feature mapping related to angle encoding refers to
    \begin{equation}\label{chap3:eq:angle_encode_feat}
      \ket{{\phi}(\bm{x})}= W_d e^{-i \bm{x}_{d} G_{d}} W_{d} \cdots W_2 e^{-i \bm{x}_{1} G_1}W_1 \ket{0}^{\otimes d},
    \end{equation} 
    where $W_0,\cdots,W_d$ are arbitrary unitary evolutions, and $G_i$ is $d_i \le d$-dimensional Hermitian operator called the generating Hamiltonian. 
        
    For a special case in which $W_i=\mathbb{I}$ and $G_i$ refers to the Pauli-X operators $\sigma_x$ acting on the $i$-th qubit, the quantum feature mapping refers to
    \begin{equation}
      \ket{{\phi}(\bm{x})}= \bigotimes_{i=1}^{d} \exp\left(-i\frac{\bm{x}_{i}}{2}\sigma_x \right) \ket{0}^{\otimes d},
    \end{equation}
    and the related quantum kernel is given by
    \begin{align}
      k(\bm{x},\bm{x}') = & \prod_{i=1}^d \left| \sin(\bm{x}_{i})\sin(\bm{x}_i^{\prime}) + \cos(\bm{x}^{(i)})\cos(\bm{x}_i^{\prime}) \right|^2 
      \nonumber \\
      = & \prod_{i=1}^d \left| \cos(\bm{x}_{i}-\bm{x}_{i}^{\prime}) \right|^2.
    \end{align}
  \end{example}
\end{shadedbox}

The angle encoding requires $d$-qubit, mapping the classical data to $2^d$-dimensional Hilbert space. 
One merit of angle encoding is introducing non-linearity, which is crucial for transforming low-dimensional, non-linearly separable data into higher-dimensional, linearly separable representations—a property essential for effective kernel-based machine learning. Additionally, angle encoding is well-suited for implementation on modern devices featuring limited qubits and circuit depth, making it practical for exploring the practical utility of near-term quantum computers.

The quantum kernels related to different data encoding strategies have a resemblance to kernels from the classical machine learning literature. This means that sometimes up to an absolute square value, they can be identified with standard kernels such as the polynomial or Gaussian kernel. For the special case of angle encoding, the resemblance to classical kernels is because the employed quantum circuit does not employ any entangled quantum gates such that it can be simulated classically. We now discuss the general form of the quantum kernels induced by quantum feature maps from angle encoding in Eqn.~\eqref{chap3:eq:angle_encode_feat}. We focus on the simplified case that each input $\bm{x}^{(i)}$ is only encoded once and that all the encoding Hamiltonians are the same, i.e., $G_1=\cdots=G_{d}=G$.

\begin{theorem}
  [Fourier representation of the quantum kernel]\label{chapt:kernel:thm-fourier} 
  Let $\mathcal{X}=\mathbb{R}^d$ and $U(\bm{x})$ be a quantum circuit that encodes the data inputs $\bm{x}=(\bm{x}_{1},\cdots,\bm{x}_{d})\in \mathcal{X}$ into a $d$-qubit quantum state $\ket{\phi(\bm{x})}$ via gates of the form $e^{-i\bm{x}_{i}G}$ for $i=1,\cdots, d$.  Without loss of generality, $G$ is assumed to be a $m \le 2^d$-dimensional diagonal operator with spectrum $\lambda_1,\cdots,\lambda_m$. Between such data-encoding gates, and before and after the entire encoding circuit, arbitrary unitary evolutions $W_{1},\cdots, W_{d+1}$ can be applied, so that
  \begin{equation}
    U(\bm{x}) = W_{d+1} e^{-i \bm{x}_{d} G_{d}} W_{d} \cdots W_2 e^{-i \bm{x}_{1} G_1}W_1.
  \end{equation}
  The quantum kernel $k_Q(\bm{x},\bm{x}')$ can be written as
  \begin{equation}
    k_Q(\bm{x},\bm{x}') = \sum_{\bm{s},\bm{t}\in \Omega} e^{i\bm{s}\bm{x}}e^{i\bm{t}\bm{x}'}c_{\bm{s}\bm{t}},
  \end{equation}
  where $\Omega \subset \mathbb{R}^d$, and $c_{\bm{s}\bm{t}} \in \mathbb{C}$. For every $\bm{s},\bm{t}\in \Omega$, we have $-\bm{s},-\bm{t}\in \Omega$ and $c_{\bm{s}\bm{t}}=c_{-\bm{s}-\bm{t}}^*$, which guarantees that the quantum kernel is real-valued.
\end{theorem}
\begin{proof}[Proof sketch of Theorem~\ref{chapt:kernel:thm-fourier}]
  The assumption that the generator $G$ is diagonal could be made without loss of generality because one can diagonalize Hermitian operators as $G=Ve^{-i\bm{x}_{i}\Sigma} V^{\dagger}$ with
  \begin{equation}
    e^{-i\bm{x}_{i}\Sigma} = 
    \left(\begin{matrix} 
      e^{-i\bm{x}_{i}\lambda_1} & 0 & \cdots & 0 \\
      0 & e^{-i\bm{x}_{i}\lambda_2} & \cdots & 0 \\
      \cdots & \cdots & ~ & ~ \\
      0 & \cdots & 0 & e^{-i\bm{x}_{i}\lambda_m} \\
    \end{matrix}
    \right),
  \end{equation}
  where $V^{\dagger}$ refers to the conjugate transpose of the matrix $V$, $\lambda_1,\cdots,\lambda_m$ are the eigenvalues of $G$. Formally, $V,V^{\dagger}$ can be absorbed into the arbitrary circuits $W_{i+1}$ and $W_i$ before and after the encoding gate.  In this regard, the quantum kernel can be written down as the inner product between the  feature state of the specific forms in Eqn.~\eqref{chap3:eq:angle_encode_feat}, i.e.,
  \begin{align}
    & k(\bm{x},\bm{x}')
    \nonumber \\
    = & \left|\braket{{\phi}(\bm{x}') | {\phi}(\bm{x})} \right|
    \nonumber \\
    = & \Big| \bra{\bm{0}}W_{1}^{\dagger} (e^{-i\bm{x}_{1}^{\prime}\Sigma})^{\dagger} \cdots (e^{-i\bm{x}_{d}^{\prime}\Sigma})^{\dagger}W_{d+1}^{\dagger} W_{d+1}
    e^{-i\bm{x}_{d}\Sigma} \cdots e^{-i\bm{x}_{1}\Sigma} W_{1} \ket{\bm{0}}\Big|^2
    \nonumber \\
    = & \left| \bra{\bm{0}}W_{1}^{\dagger} (e^{-i\bm{x}_{1}^{\prime}\Sigma})^{\dagger} \cdots (e^{-i\bm{x}_{d}^{\prime}\Sigma})^{\dagger} e^{-i\bm{x}_{d}\Sigma} \cdots e^{-i\bm{x}_{1}\Sigma} W_{1} \ket{\bm{0}}\right|^2
    \nonumber \\
    = & \Bigg|\sum_{j_1,\cdots,j_d=1}^m \sum_{k_1,\cdots,k_d=1}^m e^{-i(\lambda_{j_1}\bm{x}_{1}-\lambda_{k_1}\bm{x}_{1}^{\prime}+\cdots + \lambda_{j_d}\bm{x}_{d}-\lambda_{k_d}\bm{x}_{d}^{\prime})}  
    \nonumber \\
    & \times \left(W^{(1k_1)}_{1}  \cdots W^{(k_{d-1}k_d)}_{d}\right)^*  W^{(j_d j_{d-1})}_{d}  \cdots W^{(j_{1}1)}_{1} \Bigg|^2
    \nonumber \\
    = & \left|\sum_{\bm{j}} \sum_{\bm{k}} e^{-i(\Lambda_{\bm{j}}\bm{x}-\Lambda_{\bm{k}}\bm{x}')} (\omega_{\bm{k}})^* \omega_{\bm{j}} \right|^2
    \nonumber \\
    = & \sum_{\bm{j}} \sum_{\bm{k}} \sum_{\bm{h}} \sum_{\bm{l}}  e^{-i(\Lambda_{\bm{j}}-\Lambda_{\bm{l}})\bm{x}} e^{i(\Lambda_{\bm{k}}-\Lambda_{\bm{h}})\bm{x}'} (\omega_{\bm{k}} \omega_{\bm{h}} )^* \omega_{\bm{j}} \omega_{\bm{l}},
  \end{align}
  Here, the scalars $W^{(ab)}_{i}$ refer to the element $\bra{a} W_{i}\ket{b}$ of the unitary operator $W_{i}$, the bold multi-index $\bm{j}$ summarizes the set $(j_1,\cdots,j_d)$ and $\Lambda_j$ is a vector containing the eigenvalues selected by the multi-index (and similarly for $\bm{k},\bm{h},\bm{l}$ ).
  
  We can now summarize all terms where $\Lambda_{\bm{j}} - \Lambda_{\bm{l}} = \bm{s}$ and $\Lambda_{\bm{k}} - \Lambda_{\bm{h}} = \bm{t}$,  in other words where the differences of eigenvalues amount to the same vectors $\bm{s}, \bm{t}$. Then
  \begin{align}
    k(\bm{x},\bm{x}') & =  \sum_{\bm{s},\bm{t}\in \Omega} e^{-i\bm{s}\bm{x}} e^{i\bm{t}\bm{x}'} \sum_{\bm{j},\bm{l}:\Lambda_{\bm{j}} - \Lambda_{\bm{l}} = \bm{s}} \sum_{\bm{k},\bm{h}:\Lambda_{\bm{k}} - \Lambda_{\bm{h}} = \bm{t}} \omega_{\bm{j}} \omega_{\bm{l}} (\omega_{\bm{k}} \omega_{\bm{h}} )^*
    \nonumber \\
    & = \sum_{\bm{s},\bm{t}\in \Omega} e^{-i\bm{s}\bm{x}} e^{i\bm{t}\bm{x}'} c_{\bm{s}\bm{t}}.
  \end{align}
  The frequency set $\Omega$ contains all vectors $\{\Lambda_{\bm{j}} - \Lambda_{\bm{l}}\}$ with $\Lambda_{\bm{j}}=(\lambda_{j_1},\cdots,\lambda_{j_d})$ and $\lambda_{j_1},\cdots,\lambda_{j_d}\in [1,\cdots,m]$.
\end{proof}

We summarize the various strategies for the construction of quantum feature mappings and quantum kernels in Table~\ref{cha3:tab:Q-Kernel-sum}.

\begin{table}[h!]
  \centering 
  \caption{Overview of typical data encoding strategies and their quantum kernels. The input domain is assumed to be the $\bm{x}=(\bm{x}_{1}, \cdots, \bm{x}_{d})\in\mathcal{X} \subset \mathbb{R}^d$.}
  \label{cha3:tab:Q-Kernel-sum}
  \footnotesize
  \begin{tabular}{l|l|l|l}
    \toprule 
    \multicolumn{1}{c}{Encoding} & \multicolumn{1}{|c|}{Qubits} & \multicolumn{1}{|c|}{Dimension} & \multicolumn{1}{c}{Quantum Kernel $k(\bm{x},\bm{x}')$} \\ \midrule
    Basis encoding   &  $d$ & $ 2^d $  &  $\delta_{\bm{x},\bm{x}'}$ \\  
    Amplitude encoding  &  $\lceil \log_2(d) \rceil$ & $d$ &   $|\bm{x}^{\dagger}\bm{x}'|^2$                  \\  
    Angel encoding  & $d$   & $ 2^d $      &  $\prod_{k=1}^d|\cos(\bm{x}_{k}-\bm{x}_k^{\prime})|^2$ \\ 
    General angle encoding & $d$   & $ 2^d $  &  $\sum_{\bm{s},\bm{t}\in \Omega} e^{-i\bm{s}\bm{x}} e^{i\bm{t}\bm{x}'} c_{\bm{s}\bm{t}}$ \\
    \bottomrule
  \end{tabular}
\end{table}

\begingroup
\allowdisplaybreaks
\begin{tcolorbox}[enhanced, 
  breakable,colback=gray!5!white,colframe=gray!75!black,title=Remark]
After obtaining the quantum kernel matrix $K_Q$ for a given training dataset $\{(\bm{x}^{(i)},y^{(i)})\}_{i=1}^{n}$, we can use it to perform regression or classification tasks in a manner similar to the classical kernel methods. In particular, as discussed in Chapter~\ref{chapt3:subsec:dual_rep}, consider the linear regression model given by
\begin{equation}
  \mathcal{L}(\bm{w}) = \frac{1}{2} \sum_{i=1}^n \left(\bm{w}^{\top}\cdot {\phi}_Q(\bm{x}^{(i)})-y^{(i)} \right)^2 + \frac{\lambda}{2} \bm{w}^{\top} \cdot \bm{w},
\end{equation}
where ${\phi}_Q(\bm{x}^{(i)})$ denotes the quantum feature mapping related to the quantum kernel $k_{Q}$,   $\bm{w}$ denotes the model parameters, and  $\lambda \ge 0$  is the regularization factor.
Using the quantum kernel matrix, we can express the dual representation of the linear model given in Eqn.~\eqref{chapt3:eq:dual_rep_out} to predict the output for a new input as
\begin{equation}
    y(\bm{x}) = \bm{k}_Q(\bm{x})^{\top} \cdot ({K}_Q+\lambda \mathbb{I}_n)^{-1} \bm{y},
\end{equation}
where $\bm{y}=(y^{(1)},\cdots, y^{(n)})$ refers to the label vector, $\bm{k}_Q(\bm{x})$ is a vector with elements $\bm{k}_Q^{(i)}(\bm{x}) = k_Q(\bm{x}^{(i)}, \bm{x})$. 
\end{tcolorbox}
\endgroup

\section{Theoretical Foundations of Quantum Kernel Machines}\label{chapt3:sec:theo_foundation_QK}
In this section, we take a step further to explore the theoretical foundations of quantum kernels. Specifically, we focus on two crucial aspects: the \textit{expressivity} and \textit{generalization} properties of quantum kernel machines. As shown in Figure~\ref{fig:schem}, these two aspects are essential for understanding the potential advantages of quantum kernels over classical learning approaches and their inherent limitations. For ease of understanding, this section emphasizes the fundamental concepts necessary for evaluating the power and limitation of quantum kernels instead of exhaustively reviewing all theoretical results.

\begin{figure*} 
  \centering \includegraphics[width=0.9\textwidth]{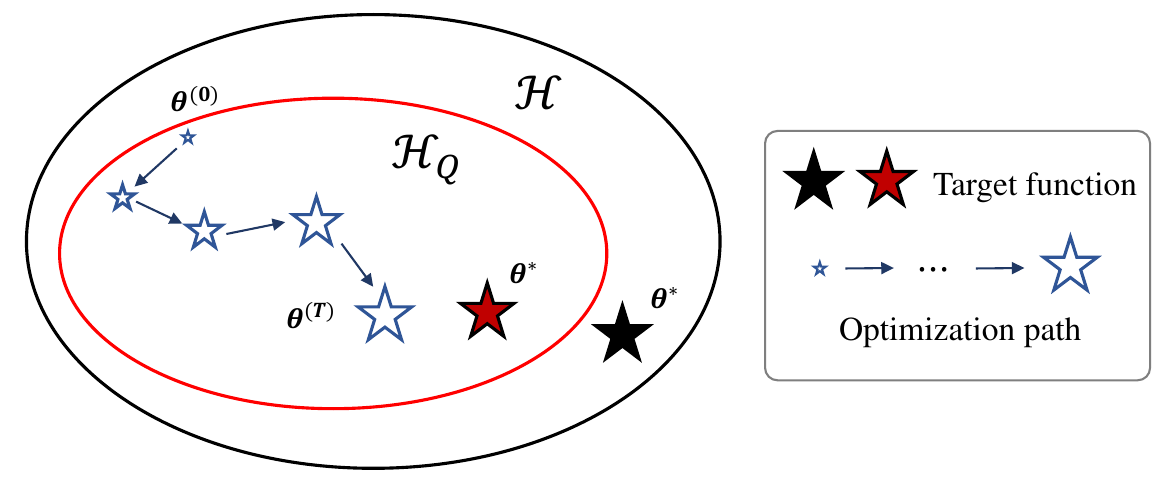}
  \caption{{\textbf{The expressivity and generalization ability of quantum kernels}. Expressivity concerns the size of the hypothesis space $\mathcal{H}_Q$ represented by quantum kernels, where $\mathcal{H}$ refers to the whole hypothesis space. Generalization ability considers the learned hypothesis that could predict the unseen data accurately features a small distance with the target concept.  }}
  \label{fig:schem}
\end{figure*}

The outline of this chapter is as follows. In Chapter~\ref{subsec:expressivity_QK}, we will discuss the expressivity of quantum kernels, which refers to the diversity of feature spaces that quantum kernels can represent. The achieved insights will help identify tasks that are particularly well-suited for quantum kernels. Then, in Chapter~\ref{subsec:gene_QK}, we will examine the potential advantage of quantum kernels in terms of generalization error compared to all classical kernel machines. This analysis highlights their ability to accurately predict labels or values for unseen data.

\subsection{Expressivity of quantum kernel machines}\label{subsec:expressivity_QK}

Quantum kernels, as discussed in Chapter~\ref{sec:quantum_kernel}, are constructed by explicitly defining quantum feature mappings. In this context, the expressivity of quantum kernel machines refers to the types of functions that quantum feature mappings can approximate and the kinds of correlations that quantum kernels can effectively model.

Following the conventions of \citet{gil2024expressivity}, we demonstrate that \textit{any kernel function can be approximated using finitely deep quantum circuits} by showing that the associated feature mapping can also be approximated using quantum circuits. This conclusion rests on two key theoretical foundations: Mercer’s feature space construction and the universality of quantum circuits. Together, these principles establish the theoretical feasibility of realizing any kernel function as a quantum kernel.

It is important to note that if exact mathematical equality were required, Mercer’s construction would demand an infinite-dimensional Hilbert space, which in turn would require quantum computers with infinitely many qubits—an impractical scenario. However, in practical applications, we are more interested in approximating functions to a certain level of precision rather than achieving exact evaluations. This perspective makes it feasible to implement the corresponding quantum feature mappings using a finite number of qubits. The following theorem confirms that any kernel function can be approximated as a quantum kernel to arbitrary precision with finite computational resources (We defer the proof details at the end of this subsection).

\begin{theorem}[Approximate universality of finite-dimensional
  quantum feature maps]
  Let $k:\mathcal{X} \times \mathcal{X} \to \mathbb{R}$ be a kernel function. Then, for any $\varepsilon \ge 0$ there exists $N \in \mathbb{N}$ and a quantum feature mapping $\rho_N$ onto the Hilbert space of quantum states of $N$ qubits such that
  \begin{equation}\label{chap3:eq:approx_univ}
    |k(\bm{x},\bm{x}')-2^N \Tr(\rho_N(\bm{x})\rho_N(\bm{x})') + 1| < \varepsilon
  \end{equation}
  for almost all $\bm{x},\bm{x}'\in \mathcal{X}$.
  \label{chap3:thm:universal_QFM}
\end{theorem}

Theorem~\ref{chap3:thm:universal_QFM}, instead of discussing the $\varepsilon$-approximation of quantum kernels in the form $|k(\bm{x},\bm{x}')- \Tr(\rho_N(\bm{x})\rho_N(\bm{x})')| < \varepsilon$, introduces additional multiplicative and additive factors, expressed as $|k(\bm{x},\bm{x}')-2^N \Tr(\rho_N(\bm{x})\rho_N(\bm{x}')) + 1| < \varepsilon$. These additional factors, explained below, do not impede the universality of the theorem. 

Moreover, the statement that Eqn.~\eqref{chap3:eq:approx_univ} holds for almost all  $\bm{x},\bm{x}'\in \mathcal{X}$ stems from measure theory. It signifies that the inequality is valid “except on sets of measure zero,” or equivalently “with probability 1.” In other words, while adversarial instances of  $\bm{x},\bm{x}'\in \mathcal{X}$ may exist for which the inequality does not hold, such instances are so sparse that the probability of encountering them when sampling from the relevant probability distribution is zero.

Last, Theorem~\ref{chap3:thm:universal_QFM} establishes that any kernel function can be approximated as a quantum kernel up to a multiplicative and an additive factor using a finite number of qubits. 

Before presenting the proof of this theorem, let us first introduce Algorithm~\ref{chap3:alg:C2QE}, which maps classical vectors to quantum states. These quantum states can then be used to evaluate Euclidean inner products as quantum kernels. Then, we demonstrate Lemma~\ref{chap3:lem:qk_express_correct} and Lemma~\ref{chap3:lem:EI_prod}, which separately formalize the correctness and runtime complexity of Algorithm~\ref{chap3:alg:C2QE}, as well as establish the relationship between the Euclidean inner product of encoded real vectors and the Hilbert-Schmidt inner product of the corresponding quantum states. 

\begin{algorithm}[H]
  \caption{Classical to quantum embedding (C2QE)}
  \textbf{Input:} a unit vector with  1-norm $\bm{r} \in \ell^d_1$. \\
  \textbf{Output:} Quantum state $\rho_{\bm{r}} \propto \mathbb{I} + \sum_{i=1}^d \bm{r}_i P_i$. \hfill $\triangleright$ See Lemma~\ref{chap3:lem:qk_express_correct}.
  \begin{algorithmic}[1]
    \State Set $N = \lceil \log_4 (d+1) \rceil$.
    \State Pad $\bm{r}$ with zeros until its length is $4^N - 1$.
    \State Draw $i \in \{1, \dots, 4^N - 1\}$ with probability $|\bm{r}_i|$.
    \State Prepare $\rho_i = \frac{1}{2^N} \left( \mathbb{I} + \text{sign}(\bm{r}_i) P_i \right)$.
    \State \textbf{return} $\rho_i$.
  \end{algorithmic}
  \label{chap3:alg:C2QE}
\end{algorithm}

The output of Algorithm~\ref{chap3:alg:C2QE}, $\frac{1}{2^N} ( \mathbb{I} \pm P)$, is a single
(pure) eigenstate of a Pauli operator $P$ with eigenvalue $\pm 1$. However, since Line $3$ involves sampling an index $i \in \{1, \cdots, 4^N - 1\}$,  
Algorithm~\ref{chap3:alg:C2QE} is inherently random, and the resulting quantum state is a classical mixture of pure states.

\begin{lemma}[Correctness and runtime of Algorithm~\ref{chap3:alg:C2QE}]
  \label{chap3:lem:qk_express_correct}
  Let $\bm{r} \in \mathcal{\ell}_1^d \subset \mathbb{R}^d$ be a unit vector with respect to the $1$-norm, i.e., $\|\bm{r}\|_1=1$. Take $N=\lceil \log_4(d+1) \rceil$ and pad $\bm{r}$ with zeros until its length is $4^N-1$. Let $(P_i)_{i=1}^{4^N-1}$ be the set of all Pauli matrices on $N$ qubits, excluding the identity. Then Algorithm~\ref{chap3:alg:C2QE} prepares the following state as a classical mixture
  \begin{align}
    \rho(\cdot): & \mathcal{\ell}_1^d \to \text{Herm}(2^N),
    \nonumber \\ 
    &\bm{r} \mapsto \rho_{\bm{r}}=   \frac{ \mathbb{I} + \sum_{i=1}^{4^N-1} \bm{r}_i P_i}{2^N}.
  \end{align}
  The total runtime complexity $t$ of Algorithm~\ref{chap3:alg:C2QE} fulfills $t\le \mathcal{O}(\text{poly}(d))$.  
\end{lemma}

\begin{proof}[Proof of Lemma~\ref{chap3:lem:qk_express_correct}] 
  The proof begins by expanding the state as follows
  \begin{equation}
    \frac{ \mathbb{I} + \sum_{i=1}^{4^N-1} \bm{r}_i P_i}{2^N} = \frac{1}{2^N} \left(\sum_{i=1}^{4^N-1} |\bm{r}_i| \mathbb{I} + \sum_{i=1}^{4^N-1}\bm{r}_iP_i \right),
  \end{equation}
where the first equality follows that $\|\bm{r}\|_1=1$ and $\bm{r}\in \mathbb{R}^{4^N-1}$.
  Rewriting the above equation using $\mbox{sign}(\bm{r}_i)$ yields
  \begin{align}
    \frac{ \mathbb{I} + \sum_{i=1}^{4^N-1} \bm{r}_i P_i}{2^N} = &\frac{1}{2^N} \left(\sum_{i=1}^{4^N-1} |\bm{r}_i| \mathbb{I} + \sum_{i=1}^{4^N-1} |\bm{r}_i| \mbox{sign}(\bm{r}_i)P_i \right)
    \nonumber \\
    = & \frac{1}{2^N} \sum_{i=1}^{4^N-1} |\bm{r}_i| \left(\mathbb{I} +  \mbox{sign}(\bm{r}_i)P_i \right) \succeq 0.
  \end{align}
  Here, it is used that $\sum_i |\bm{r}_i| = \|\bm{r}\|_1=1$ and $\mathbb{I} \pm P_i \ge 0$ for all Pauli operators $P_i$.  Notice that efficiently preparing $\mathbb{I} + P_i$ can be achieved by rotating each qubit's $\ket{0}$ basis state to the corresponding Pauli basis and flipping the necessary qubits individually. Since this state is a convex combination of quantum states, it can be efficiently prepared by mixing, when the number of terms is polynomial.
\end{proof}

\begin{lemma}[Euclidean inner products]
  \label{chap3:lem:EI_prod}
  Let $\bm{r},  \bm{r}' \in \mathbb{R}^d$ be unit vectors with respect to the $1$-norm, i.e., $\|\bm{r}\|_1=\|\bm{r}'\|_1=1$. 
  For the states $\rho_{\bm{r}},\rho_{\bm{r}}'$ produced in Algorithm~\ref{chap3:alg:C2QE}, the following identity holds
  \begin{align}
    \braket{\bm{r},\bm{r}'}=2^N \Tr(\rho_{\bm{r}}\rho_{\bm{r}}') -1.
  \end{align}
\end{lemma}

\begin{proof}[Proof of Lemma~\ref{chap3:lem:EI_prod}]
  The proof utilizes the following principles: (1) The trace is linear, and the trace of a tensor product equals the product of traces. (2) All Pauli words are traceless except for the identity, and each Pauli operator is its own inverse. Hence, the product of distinct Pauli operators is also traceless.
  
  Expanding the trace of $\rho_{\bm{r}} \rho_{\bm{r}'}$, we have
  \begin{align}
    \Tr(\rho_{\bm{r}} \rho_{\bm{r}'}) = & \Tr \left(\frac{1}{4^N} \left(\mathbb{I}+\sum_{j=1}^{4^N-1}\bm{r}_jP_j \right) \left(\mathbb{I}+\sum_{k=1}^{4^N-1}\bm{r}_k'P_k \right) \right),
  \end{align}
  which could be simplified as
  \begin{equation}
    \Tr(\rho_{\bm{r}} \rho_{\bm{r}'}) = \frac{1}{4^N} \Tr \left(\mathbb{I}+\sum_{j=1}^{4^N-1}\bm{r}_j\bm{r}_j'P_j^2 + \sum_{k\ne j}^{4^N-1}\bm{r}_j\bm{r}_k'P_jP_k \right).
  \end{equation}
  Using the properties of Pauli operators, the trace becomes
  \begin{equation}
    \Tr(\rho_{\bm{r}} \rho_{\bm{r}'}) = \frac{1}{4^N} \left( \Tr \left(\mathbb{I}\right)+\Tr \left(\sum_{j=1}^{4^N-1}\bm{r}_j\bm{r}_j^{\prime}\mathbb{I} \right) \right) =  \frac{1+\braket{\bm{r},\bm{r}'}}{2^N}.
  \end{equation}
  This completes the proof.
\end{proof}

Lemma~\ref{chap3:lem:EI_prod} clarifies the origin of the extra factors in Theorem~\ref{chap3:thm:universal_QFM}. In particular, the $2^N$ multiplicative factor is unproblematic, as $N\le \mathcal{O}(\log(d))$ and the methods are designed to scale polynomially with $d$. Moreover, the quantum state $\rho_{\bm{r}}$ is generally mixed but can be efficiently prepared. The mapping is injective but not surjective.

With these results in place, we now present the proof of Theorem~\ref{chap3:thm:universal_QFM}.

\begin{proof}[Proof of Theorem~\ref{chap3:thm:universal_QFM}]
  The proof follows from a corollary of Mercer’s theorem and the universality of quantum computing. First, by a direct corollary of the Mercer's Theorem (i.e., Fact~\ref{chap3:thm:mercer}) which states that an arbitrary kernel $k$ admits a uniformly convergent expansion of the form in Eqn.~\eqref{chap3:eq:kernel_unif_expan}, it is ensured that there exists a finite-dimensional feature map 
  $\Phi_m:\mathcal{X}\to \mathbb{R}^m$ such that
  \begin{equation}
    \left|k(\bm{x},\bm{x}')-\braket{\Phi_m(\bm{x}),\Phi_m(\bm{x}')} \right| < \varepsilon.
  \end{equation}
  Without loss of generality, it is assumed that $\|\Phi_m(\bm{x})\|=1$ for all $x\in\mathcal{X}$. The quantum state $\rho_{\Phi_m}$ can then be prepared, which requires $\lceil \log_4(m+1)\rceil$ qubits. By preparing two such states---one for $\Phi_m(\bm{x})$ and one for $\Phi_m(\bm{x}')$---their inner product can be computed as the Hilbert-Schmidt inner product of the quantum states, as shown in Lemma~\ref{chap3:lem:EI_prod}. This leads to
  \begin{equation}
    \braket{\Phi_m(\bm{x}),\Phi_m(\bm{x}')} = 2^N \Tr\left( \rho_{\Phi_m(\bm{x})} \rho_{\Phi_m(\bm{x}')} \right) - 1.
\end{equation}

For reference, it is noted that $\Tr\left( \rho_{\Phi_m(\bm{x})} \rho_{\Phi_m(\bm{x}')} \right)$ can be computed using the SWAP test to an additive precision determined by the number of measurement shots. This allows us to approximate the result efficiently to any desired polynomial additive precision. Consequently, it follows that
  \begin{equation}
    \left| k(\bm{x},\bm{x}') - 2^N \Tr\left( \rho_{\Phi_m(\bm{x})} \rho_{\Phi_m(\bm{x}')} \right) + 1 \right| < \varepsilon,
  \end{equation}
  for almost all $\bm{x},\bm{x}'\in \mathcal{X}$. This completes the proof.
\end{proof}

We remark that Theorem~\ref{chap3:thm:universal_QFM} does not aim to demonstrate any quantum advantage but rather establishes the ultimate expressivity of quantum kernels.  The theorem guarantees the existence of a quantum kernel using a finite number of qubits but does not address how quickly the number of required qubits grows with increasing computational complexity of the kernel function $k$ or with decreasing approximation error $\varepsilon > 0$.  The number of qubits $N$ will depend on certain properties of the kernel $k$ and the approximation error $\varepsilon$. For instance, if the required number of qubits scales exponentially with these parameters, Theorem~\ref{chap3:thm:universal_QFM} would have limited practical utility. Similarly, the time required to find such a quantum kernel approximation-independent of the memory and runtime requirements for preparing the feature vectors and computing their inner product—must also be considered.

\begingroup
\allowdisplaybreaks
\begin{tcolorbox}[enhanced, 
  breakable,colback=gray!5!white,colframe=gray!75!black,title=Remark]
Although Theorem~\ref{chap3:thm:universal_QFM} establishes that all kernel functions can be realized as quantum kernels, there may still exist kernel functions that cannot be \textit{realized efficiently} as quantum kernels. This observation requires us to identify quantum kernels that can be computed efficiently on quantum computers, i.e., in polynomial time.
\end{tcolorbox}
\endgroup

\subsection{Generalization of quantum kernel machines}\label{subsec:gene_QK}
Generalization, which quantifies the ability of learning models (both classical and quantum) to predict unseen data, is a critical metric for evaluating the quality of a learning model. Due to its importance, here we analyze the potential advantages of quantum kernels in terms of the generalization ability. 

For comprehensive, in this section, we first elucidate the generalization error bounds for general kernel machines, establishing a unified framework for a fair comparison between quantum kernels and classical kernels. Subsequently, we introduce a geometry metric to assess the potential quantum advantage of quantum kernels with respect to generalization error for a fixed amount of training data.

\subsubsection{Generalization error bound for kernel machines}
We begin by reviewing the optimal learning models based on the specified kernel machines which could be either classical or quantum, as discussed in Chapter~\ref{chapt3:subsec:dual_rep}. Suppose we have obtained $n$ training examples $\{(\bm{x}^{(i)},y^{(i)})\}_{i=1}^n$ with $\bm{x}^{(i)}\in \mathbb{R}^d$ and $y^{(i)}=f(\bm{x}^{(i)}) \in \mathbb{R}$, where $f$ is the target function. After training on this data, there exists a machine learning algorithm that outputs $h(\bm{x})=\bm{w}^{\dagger}  \phi(\bm{x})$, where $\phi(\bm{x}) \in \mathbb{C}^D$ refers to the hidden feature map corresponding the classical/quantum kernel function $k(\bm{x}^{(i)},\bm{x}^{(j)})= {K}_{ij}=\phi(\bm{x}^{(i)}) \cdot \phi(\bm{x}^{(j)})$. More precisely, considering the mean square error as the loss function for such a task, we have
\begin{equation}\label{eq:B2}
  \mathcal{L}(\bm{w}, \bm{x})=  \lambda \bm{w}^{\dagger} \bm{w} + \sum_{i=1}^n \left( \bm{w}^{\dagger}    \phi(\bm{x}^{(i)}) - y^{(i)}\right)^2,
\end{equation}
where $\lambda \ge 0$ is the regularization parameters for avoiding over-fitting. 

Then the optimal parameters for optimizing this loss function refer to
\begin{equation}\label{eqn:opt-omega}
  \bm{w}^* =  \arg\min_{\bm{w}\in \Theta}  \mathcal{L}(\bm{w}, \bm{x}).
\end{equation} 
As discussed in Chapter~\ref{chapt3:subsec:dual_rep}, the optimal solution $\bm{w}^*$ in Eqn.~\eqref{eqn:opt-omega} has the explicit form of  
\begin{equation}\label{eq:explicit_opt_omega}
  \bm{w}^* = \bm{\Phi}^{\dagger}  ({K}+\lambda \mathbb{I}_n)^{-1}  \bm{y} = \sum_{i=1}^n\sum_{j=1}^n \phi(\bm{x}^{(i)}) (({K}+\lambda \mathbb{I}_n)^{-1})_{ij} y^{(j)},
\end{equation}
where $\bm{y}=[y^{(1)}, ..., y^{(n)}]^{\top}$ refers to the vector of labels and $K\in \mathbb{R}^{n\times n}$ is the kernel matrix, and the second equality follows that $\bm{\Phi} = [\phi(\bm{x}^{(1)}),\cdots, \phi(\bm{x}^{(n)})]^{\dagger}$. Moreover, the norm of the optimal parameters has a simple form for the case of $\lambda \to 0$, i.e.,
\begin{equation}
  \|\bm{w}^*\|_2^2= \bm{y}^{\top} K^{-1} \bm{y}.
\end{equation}

We now expose the prediction error of these learning models, i.e.,
\begin{equation}\label{eq:learning_model}
  \epsilon_{\bm{w}^*}(\bm{x}) = \left| f(\bm{x}) - (\bm{w}^*)^{\dagger}  {\phi}(\bm{x}) \right|,
\end{equation} 
which is uniquely determined by the kernel matrix $K$ and the hyper-parameter $\lambda$ as shown in Eqn.~\eqref{eq:explicit_opt_omega}. In particular, we will focus on discussing the upper bound on the expected prediction error, which is the sum of training error and generalization error.

\begin{tcolorbox}[enhanced, 
  breakable,colback=blue!5!white,colframe=blue!75!black,title={Prediction, training, and generalization error}]
In the context of learning theory, the upper bound of the expected prediction error defined in Eqn.~\eqref{eq:learning_model} (a.k.a, expected risk) is achieved by separately analyzing the upper bounds of the training error (a.k.a, empirical risk) and the generalization error, i.e.,
\begin{equation}
    \mathbb{E}_{\bm{x} \sim \mathcal{D}} \epsilon_{\bm{w}^*}(\bm{x})  = \underbrace{\frac{1}{n} \sum_{i=1}^n \epsilon_{\bm{w}^*}(\bm{x}^{(i)})}_{\mbox{Training error}} + \underbrace{\mathbb{E}_{\bm{x} \sim \mathcal{D}} \epsilon_{\bm{w}^*}(\bm{x}) - \frac{1}{n} \sum_{i=1}^n \epsilon_{\bm{w}^*}(\bm{x}^{(i)})}_{\mbox{Generalization error}}.
  \end{equation}
This decomposition stems from the fact that data distribution $\mathcal{D}$ is inaccessible in most scenarios.  
\end{tcolorbox}

We now will separately give a rough derivation of the upper bound of training error and generalization error, which present the necessary steps for the derivation for a clear exposition and omit the specific details that could be found in \citet{huang2021power}.

\noindent\textit{$\bullet$ Training error}. Employing the convexity of function  and Jensen's inequality, the training error yields
\begin{align}
  \frac{1}{n} \sum_{i=1}^n \epsilon_{\bm{w}^*}(\bm{x}^{(i)}) \le \sqrt{\frac{1}{n} \sum_{i=1}^n \left( (\bm{w}^*)^{\dagger}  \phi(\bm{x}^{(i)}) - y^{(i)} \right)^2 }.
\end{align}
Moreover, combining with the expression for the optimal $\bm{w}^*$ given in Eqn.~\eqref{eq:explicit_opt_omega}, we can obtain the upper bound of training error in terms of the kernel matrix $K$ and hyper-parameter $\lambda$, i.e.,
\begin{equation}
  \frac{1}{n} \sum_{i=1}^n \epsilon_{\bm{w}^*}(\bm{x}^{(i)}) \le \sqrt{\frac{\lambda^2 \bm{y}^{\top} (K+\lambda \mathbb{I}_n)^{-2}\bm{y}}{n} }.
\end{equation}
We can see that when $\lambda=0$ and $K$ are invertible, the training error is zero. However, the hyper-parameter is usually set as $\lambda > 0$ in practice.

\noindent\textit{$\bullet$ Generalization error}. The derivation of generalization error is more complicated than training error, which involves a basic theorem in statistic and learning theory as presented below.
\begin{fact}[Theorem 3.3,  \citet{mohri2018foundations}]\label{fact:gene_rademacher}
  Let $\mathcal{G}$ be a family of function mappings from a set $\mathcal{Z}$ to $[0,1]$. Then for any $\delta>0$, with probability at least $1-\delta$ over identical and independent draw of $n$ samples from $\mathcal{Z}:\bm{z}^{(1)},\cdots, \bm{z}^{(n)}$, we have  for all $g \in \mathcal{G}$,
  \begin{equation}
    \mathbb{E}_{\bm{z}} g(\bm{z}) \le \frac{1}{n} \sum_{i=1}^n g(\bm{z}^{(i)}) + 2\mathbb{E}_{\sigma} \left[\sup_{g\in \mathcal{G}} \frac{1}{n} \sum_{i=1}^n \sigma_i g(\bm{z}^{(i)}) \right]+ 3\sqrt{\frac{\log(2/\delta)}{2n}},
  \end{equation}
  where $\sigma_1,\cdots,\sigma_n$ are in independent and uniform random variables over $\{1,-1\}$.
\end{fact}

For kernel functions defined in Eqn.~\eqref{eq:learning_model}, the set $\mathcal{Z}$ refers to the space of input vector with $\bm{z}^{(i)} = \bm{x}^{(i)}$ drawn from some input distribution. Each function $g$ would be equal to $\epsilon_{\bm{w}}/\alpha$ for some $\bm{w}$, where $\epsilon_{\bm{w}}$ is defined in Eqn.~\eqref{eq:learning_model} and $\alpha$ is a normalization factor such that the range of $\epsilon_{\bm{w}}/\alpha$ is $[0,1]$. Without loss of generality, we assume that $\alpha=1$. For any specific parameter $\bm{w}$, consider the special case of $\mathcal{G}$ with setting $\mathcal{G}_{\bm{w}}=\{\epsilon_{\bm{v}}~| ~\forall ~ \|\bm{v}\|\le \|\bm{w}\|\}$. Then we have the upper bound of generalization error for the optimal parameter,
\begin{align}\label{eq:gene_1}
  & \mathbb{E}_{\bm{x}} \epsilon_{\bm{w}^*}(\bm{x}) - \frac{1}{n} \sum_{i=1}^n \epsilon_{\bm{w}^*}(\bm{x}^{(i)}) 
  \nonumber \\
  \le & 2\mathbb{E}_{\sigma} \left[\sup_{\|\bm{v}\|\le \|\bm{w}^*\|} \frac{1}{n} \sum_{i=1}^n \sigma_i \epsilon_{\bm{v}}(\bm{x}^{(i)}) \right] + 3\sqrt{\frac{\log(2\|\bm{w}^*\|/\delta)}{2n}}.
\end{align}
Moreover, applying Talagrand's contraction lemma \citep{mohri2018foundations} to the first term on the right-hand side, we have
\begin{align}\label{eq:gene_1_term_1}
  \mathbb{E}_{\sigma} \left[\sup_{\|\bm{v}\|\le \|\bm{w}^*\|} \frac{1}{n} \sum_{i=1}^n \sigma_i \epsilon_{\bm{v}}(\bm{x}^{(i)}) \right] \le & \mathbb{E}_{\sigma} \left[\sup_{\|\bm{v}\|\le \|\bm{w}^*\|} \frac{1}{n} \sum_{i=1}^n \sigma_i (\bm{w}^*)^{\dagger} {\phi}(\bm{x}^{(i)}) \right]
  \nonumber \\
  \le & \sqrt{\frac{\|\bm{w}^*\|^2}{n}},
\end{align}
where the first inequality follows that $\epsilon_{\bm{v}}(\bm{x}^{(i)})$ is Lipschitz continuous with respect to $(\bm{w}^*)^{\dagger}\cdot{\phi}(\bm{x}^{(i)})$ with Lipschitz constant $1$, the second inequality follows direct algebra operation. For the detailed simplification processes, refer to Lemma~1 of \citet{huang2021power}.

In conjunction with Eqn.~\eqref{eq:gene_1}, Eqn.~\eqref{eq:gene_1_term_1}, and the expression of the optimal parameter $\bm{w}^*$ given in Eqn.~\eqref{eq:explicit_opt_omega}, we can reach the final upper bound of generalization error in terms of the kernel matrix, i.e.,
\begin{align}
  & \mathbb{E}_{\bm{x}} \epsilon_{\bm{w}^*}(\bm{x}) - \frac{1}{n} \sum_{i=1}^n \epsilon_{\bm{w}^*}(\bm{x}^{(i)}) 
  \nonumber \\
  \le & 5\cdot \frac{\bm{y}^{\top} (K+\lambda \mathbb{I}_n)^{-1} K  (K+\lambda \mathbb{I}_n)^{-1}\bm{y}}{n} + 3\sqrt{\frac{\log(2/\delta)}{2n}}.
\end{align}
For the case of $\lambda=0$, the first term in the generalization error bound has a simple form of $5\cdot \bm{y}^{\top}  K^{-1}  \bm{y}/n$.

By obtaining the upper bound of training and generalization error, we can directly get the prediction error of the learning model for a specific kernel matrix. We summarize these three errors below.

\begingroup
\allowdisplaybreaks
\begin{tcolorbox}[enhanced, 
  breakable,colback=gray!5!white,colframe=gray!75!black,title=Remark]
  The upper bound of the prediction error for kernel methods defined in Eqn.~\eqref{eq:learning_model} refers to
  \begin{align}\label{eq:pred_err_lambda}
    & \mathbb{E}_{\bm{x} \sim \mathcal{D}} \epsilon_{\bm{w}^*}(\bm{x})  \le \mathcal{O}\Bigg(\underbrace{\sqrt{\frac{\lambda^2 \bm{y}^{\top}  (K+\lambda \mathbb{I}_n)^{-2}\bm{y}}{n}} }_{\mbox{Training error}} + 
    \nonumber \\
    &\underbrace{\sqrt{\frac{\bm{y}^{\top} (K+\lambda \mathbb{I}_n)^{-1} K  (K+\lambda \mathbb{I}_n)^{-1}\bm{y}}{n}}+ \sqrt{\frac{\log(1/\delta)}{n}} }_{\mbox{Generalization error}} \Bigg),
  \end{align}
  where $K$ is a specific kernel related to the learning models, $\bm{y}=[y^{(1)},\cdots, y^{(n)}]$ refers to the label vector of $n$ training data. For the special case of $\lambda=0$, the training error is zero, and the prediction error reduce to the generalization error with a simple form
  \begin{align}\label{eq:pred_err_0}
    \mathbb{E}_{\bm{x} \sim \mathcal{D}} \epsilon_{\bm{w}^*}(\bm{x})  \le \mathcal{O}\Bigg(\sqrt{\frac{\bm{y}^{\top}  K^{-1}\bm{y}}{n}}+ \sqrt{\frac{\log(1/\delta)}{n}} \Bigg).
  \end{align}
\end{tcolorbox}
\endgroup

We remark that the derived upper bound of the prediction error applies to both classical and quantum kernels, as we have not imposed any restrictions on the kernel matrix $K$ during the derivation.

\subsubsection{Quantum kernels with prediction advantages}
Using the above theoretical results of generalization error for general kernel machines, we now elucidate how to access the potential quantum advantage of quantum kernels. For a clear understanding, we focus on the case of $\lambda=0$ in which the prediction error bound has a simple form of $\mathcal{O}(\sqrt{\bm{y}^{\top}  K^{-1}\bm{y}/n}+ \sqrt{\log(1/\delta)/n} )$ as shown in Eqn.~\eqref{eq:pred_err_0}. In particular, this bound has a key dependence on two quantities, namely (1) the size of training data $n$; (2) the kernel-dependent term $\bm{y}^{\top}  K^{-1}\bm{y}$, which we denote as 
\begin{equation}\label{eqn:chapt3-kernel-geomtric}
s_K(\bm{y})=\bm{y}^{\top}  K^{-1}\bm{y},
\end{equation}
in the following discussion for simplification.

The dependence on $n$ reflects the role of data to improve prediction performance. On the other hand, the quantity $s_K(\bm{y})$ is equal to the model complexity of the trained function $h(\bm{x}) = (\bm{w}^*)^{\dagger} \cdot \phi(\bm{x})$, where $s_K(\bm{y}) = \| \bm{w}^*\|^2 = (\bm{w}^*)^{\dagger} \cdot \bm{w}^*$ after training. A smaller value of $s_K(\bm{y})$ implies better generalization to new data $\bm{x}$ sampled from the distribution $\mathcal{D}$. Intuitively, $s_K(\bm{y})$ measures whether the closeness between $\bm{x}^{(i)}$ and $\bm{x}^{(j)}$ defined by the kernel function $k(\bm{x}^{(i)}, \bm{x}^{(j)})$ matches well with the closeness of the labels $y^{(i)}$ and $y^{(j)}$, recalling that a larger kernel value indicates two points are closer.

Based on the above discussion, we are now in the position to analyze the potential advantage of quantum kernel machines. Given a set of training data $\{(\bm{x}^{(i)},y^{(i)})\}_{i=1}^n$, let $\mathcal{Q}$ and $\mathcal{C}$ be the class of \textbf{Q}uantum and \textbf{C}lassical kernels respectively that can be efficiently evaluated on quantum and classical computers for any given $\bm{x}$. In order to formally evaluate the potential for quantum
prediction advantage generally, one must take the quantum kernel $K_Q\in \mathcal{Q}$ to satisfy the following two conditions:
\begin{itemize}
  \item $K_Q$ is hard to compute classically for any given $\bm{x}$.
  \item According to Eqn.~(\ref{eqn:chapt3-kernel-geomtric}), the quantity $s_{Q}(\bm{y})$ related to the quantum kernel $K_Q$ must be the minimal over all efficient classical models, namely $s_{Q}(\bm{y}) \le s_{C}(\bm{y}) $ for any $K_C \in \mathcal{C}$ with $s_{C}(\bm{y})$ being the $K_C$ related quantity.
\end{itemize}

\begin{figure*}%[ht]
  \centering \includegraphics[width=0.99\textwidth]{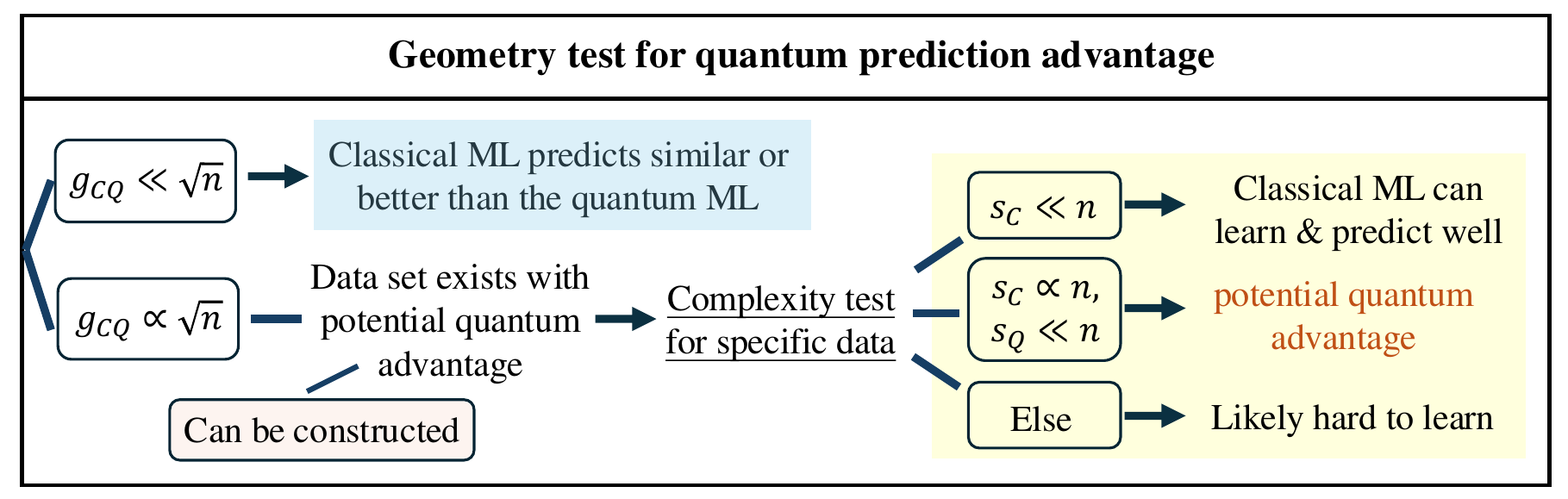}
  \caption{{\textbf{A flowchart for understanding the potential for quantum prediction
  advantage (Adapted from \citet{huang2021power})}.  }}
  \label{fig:qk_test_flowchart}
\end{figure*}

From the second condition, we can see that the potential advantage for the quantum kernel $K_Q$ to predict better than a classical kernel $K_C$ depends on the largest possible separation between $s_{Q}(\bm{y})$ and $s_{C}(\bm{y})$ for a dataset. \citet{huang2021power} define a geometry metric, namely \textbf{asymmetric geometric difference}, to characterize this separation for a fixed training dataset, which is given by
\begin{equation}
  g_{CQ} = g(K_C||K_Q)=\sqrt{\left\|\sqrt{K_Q}(K_C^{-1})\sqrt{K_Q}\right\|_{\infty}},
\end{equation}
where $\|\cdot\|_{\infty}$ is the spectral norm of the resulting matrix and we assume $\Tr(K_Q)=\Tr(K_C)=n$. The geometric difference $g(K_C||K_Q)$ can be computed on a classical computer by performing a singular value decomposition of the $n\times n$ matrices $K_C$ and $K_Q$ in time at most order $n^3$. 

Figure~\ref{fig:qk_test_flowchart} presents a detailed flowchart for evaluating the potential quantum prediction advantage using the defined geometric difference $g_{CQ}$ in a machine learning task. The input consists of $n$ data samples, along with both quantum and classical methods, each associated with its respective kernel. The tests are conducted as a function of $n$ to highlight the role of data size in determining the potential for a prediction advantage.

First, the geometric quantity $g_{CQ}$ is evaluated, which quantifies the potential for a separation between quantum and classical predictions, without yet considering the actual function to be learned. Specifically, a large value of $g_{CQ} \propto \sqrt{n}$ suggests the possibility of a quantum prediction advantage. If the test is passed, an adversarial dataset that saturates this limit can be constructed. In particular, there exists a dataset with $s_C = g_{CQ}^2 s_Q$, where the quantum model exhibits superior prediction performance, as will be described in the subsequent context.

Subsequently, to incorporate the provided data, a label-specific test can be performed using the model complexities $s_C$ and $s_Q$. For quantum kernels and classical learning models, when $s_Q\ll n$ and $s_C\propto n$, a prediction advantage for quantum models is possible, as supported by the generalization bound in Eqn.~\eqref{eq:pred_err_0}. In contrast, if $g_{CQ}$ is small such as $g_{CQ} \ll \sqrt{n}$, the classical learning model will likely have a similar or better model complexity  $s_C(\bm{y})$ compared to the quantum model. In this case, the classical model's prediction performance will be competitive or superior, and the classical model would likely be preferred.

\subsubsection{Construction of dataset with maximal quantum advantage}
We now elucidate how to explicitly construct such a  data set to enable the maximal separation between the model complexity of quantum kernels and classical kernels, as indicated by the geometry test in Figure~\ref{fig:qk_test_flowchart}. To separate between quantum and classical models related to kernel matrix $K_Q$ and $K_C$, we consider that the ratio between $s_C$ and $s_Q$ is as large as possible for a particular choice of targets $y^{(1)}, \cdots, y^{(n)}$. This could be achieved by solving the optimization problem
\begin{equation}
  \min_{\bm{y}\in \mathbb{R}^n} \frac{s_C}{s_Q} = \min_{\bm{y}\in \mathbb{R}^n} \frac{\bm{y}^{\top}  K_C^{-1}\bm{y}}{\bm{y}^{\top}  K_Q^{-1}\bm{y}},
\end{equation}
which has an exact solution given by a generalized eigenvalue problem. The solution is given by $\bm{y}=\sqrt{K_Q}\bm{v}$, where $\bm{v}$ is the eigenvector of $\sqrt{K_Q}K_C^{-1}\sqrt{K_Q}$ corresponding to the eigenvalue $g^2=\|\sqrt{K_Q}K_C^{-1}\sqrt{K_Q}\|_{\infty}$. This guarantees that $s_C=g^2s_Q$, and note that by definition of $g,s_C\le g^2 s_Q$. Hence this dataset fully utilized the geometric difference between the quantum and classical space. Finally, we can turn this dataset, which maps input $\bm{x}$ to a real value ${y}_{Q}$, into a classification task by replacing $y_Q$ with $+1$ if $y_{Q} > \mbox{median}(y^{(1)}, \cdots, y^{(n)})$ and $-1$ if $y_{Q} \le \mbox{median}(y^{(1)}, \cdots, y^{(n)})$. The constructed dataset will yield the largest separation between quantum and classical models from a learning theoretic sense, as the model complexity fully saturates the geometric difference. If there is no quantum advantage in this dataset, there will likely be none.

\section{Code Demonstration}

In this section, we explore the practical implementation of a quantum kernel. Before diving into concrete code examples, we discuss an efficient strategy for estimating the quantum kernel in practice.

As explained in Chapter.~\ref{sec:quantum_kernel}, one method for estimating the quantum kernel is the SWAP test, which is resource-intensive. Alternatively, we can encode the classical data vector $\bm{x}$ using a unitary operation $U(\bm{x})$ and apply the inverse embedding of $\bm{x}'$ using $U(\bm{x}')^\dagger$ on the same qubits. The quantum kernel $k_Q(\bm{x}, \bm{x}')$ is then estimated by measuring the expectation of the projector $O = (\ket{0} \bra{0})^{\otimes N}$ on the zero state.

The complete quantum circuit architecture for this process is illustrated in Figure~\ref{chap3:fig:QK_circuit}. Mathematically, the process is expressed as:
\begin{align}
  &\braket{0^{\otimes N} | U(\bm{x}') U(\bm{x})^\dagger O U(\bm{x}')^\dagger U(\bm{x}) | 0^{\otimes N}} \nonumber \\
  =& \braket{0^{\otimes N} | U(\bm{x}') U(\bm{x})^\dagger \ket{0}^{\otimes N} \bra{0}^{\otimes N} U(\bm{x}')^\dagger U(\bm{x}) | 0^{\otimes N}} \nonumber \\
  =& \left|\braket{0^{\otimes N} | U(\bm{x}')^\dagger U(\bm{x}) | 0^{\otimes N}}\right|^2 \nonumber \\
  =& \left|\braket{{\phi}(\bm{x}) | {\phi}(\bm{x}')}\right|^2 \nonumber \\
  =& k_Q(\bm{x}, \bm{x}').
\end{align}
This approach allows the quantum kernel estimation to use the same number of qubits required for the quantum feature mapping of the classical vector $\bm{x}$.

Next, we provide an example demonstrating the workflow of applying quantum kernels for classification tasks on the MNIST dataset, with step-by-step code implementation.

\subsection{Classification on MNIST dataset}

We train a Support Vector Machine (SVM) classifier associated with a quantum kernel on the MNIST dataset. The pipeline involves the following steps.
\begin{enumerate}
    \item [Step 1] Load and preprocess the dataset.
    \item [Step 2]  Define the quantum feature mapping.
    \item [Step 3] Construct the quantum kernel.
    \item [Step 4]  Train and evaluate the SVM classifier.
\end{enumerate}

We begin by importing the required libraries.

\begin{lstlisting}[language=Python]
import pennylane as qml
from sklearn.datasets import fetch_openml
from sklearn.decomposition import PCA
from sklearn.model_selection import train_test_split
from sklearn.preprocessing import StandardScaler
from sklearn.svm import SVC
from sklearn.metrics import accuracy_score
import numpy as np
\end{lstlisting}

\noindent \textbf{Step 1: Dataset preparation.} We focus on the digits 3 and 6 in the MNIST dataset, forming a binary classification problem. Principal Component Analysis (PCA) \citep{abdi2010principal} is applied to reduce the feature dimension of the images, minimizing the number of required qubits for encoding. The compressed features are normalized to align with the periodicity of the quantum feature mapping.

\begin{lstlisting}[language=Python]
def load_mnist(n_qubit):
    # Load MNIST dataset from OpenML
    mnist = fetch_openml('mnist_784', version=1)
    X, y = mnist.data, mnist.target
    
    # Filter out the digits 3 and 6
    mask = (y == '3') | (y == '6')
    X_filtered = X[mask]
    y_filtered = y[mask]
    
    # Convert labels to binary (0 for digit 3 and 1 for digit 6)
    y_filtered = np.where(y_filtered == '3', 0, 1)
    
    # Apply PCA to reduce feature dimension
    pca = PCA(n_components=n_qubit)
    X_reduced = pca.fit_transform(X_filtered)
    
    # Normalize the input features
    scaler = StandardScaler().fit(X_reduced)
    X_scaled = scaler.transform(X_reduced)
    
    # Split into training and testing sets
    X_train, X_test, y_train, y_test = train_test_split(X_scaled, y_filtered, test_size=0.2, random_state=42)
    return X_train, X_test, y_train, y_test
    
n_qubit = 8
X_train, X_test, y_train, y_test = load_mnist(n_qubit)
\end{lstlisting}

To better understand the structure of the dataset, we visualize the training data using t-distributed Stochastic Neighbor Embedding (t-SNE) \citep{van2008visualizing}. The following code generates the visualization:
\begin{lstlisting}[language=Python]
def visualize_dataset(X, labels):
    import matplotlib.pyplot as plt
    from sklearn.manifold import TSNE

    tsne = TSNE(n_components=2, random_state=42, perplexity=30)
    label2name = {
        0: '3',
        1: '6'
    }
    mnist_tsne = tsne.fit_transform(X)
    for label in np.unique(labels):
        indices = labels == label
        plt.scatter(mnist_tsne[indices, 0], mnist_tsne[indices, 1], cmap='coolwarm', s=20, label=f'Number {label2name[label]}')

    # Add labels and legend
    plt.title("t-SNE Visualization of Two Classes (3 and 6)")
    plt.xlabel("t-SNE Dimension 1")
    plt.ylabel("t-SNE Dimension 2")
    plt.legend()

    plt.tight_layout()
    plt.show()

visualize_dataset(X_train, y_train)
\end{lstlisting}

The resulting t-SNE visualization is shown in Figure~\ref{fig:mnist36}.

\begin{figure}[H]
\centering
\includegraphics[width=0.98\textwidth]{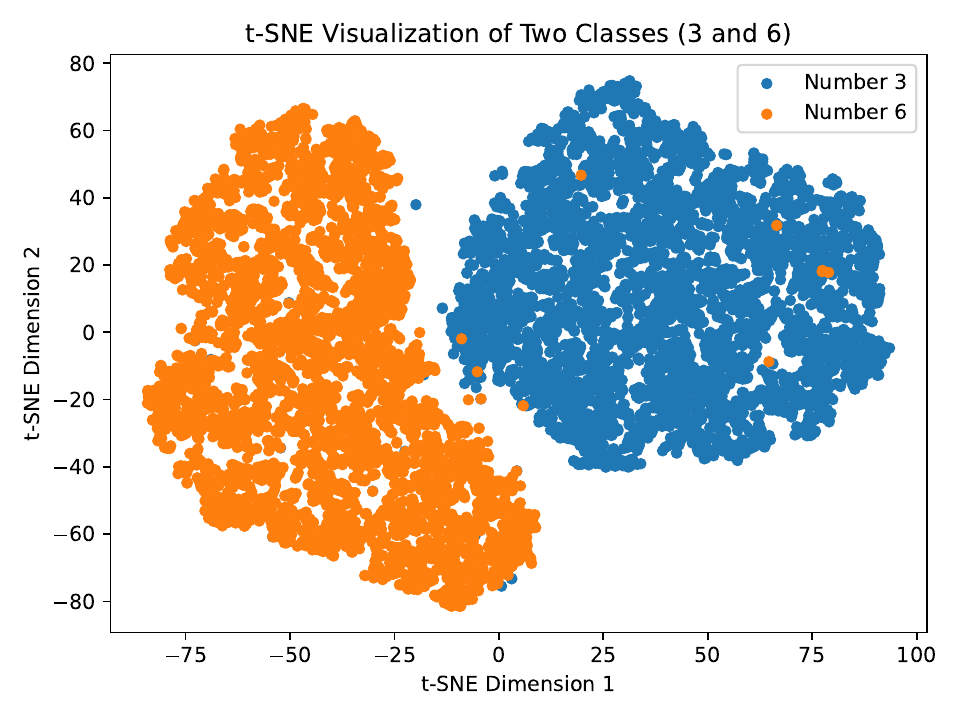}
\caption{{\textbf{T-SNE visualization of MNIST dataset of two classes `3' and `6'.} } }
\label{fig:mnist36} 
\end{figure}

\noindent \noindent \textbf{Steps 2\&3: Define quantum feature mapping and building quantum kernel.} We use angle embedding as the quantum feature mapping method. The quantum kernel is implemented as follows.

\begin{lstlisting}[language=Python]
dev = qml.device('default.qubit', wires=n_qubit)

@qml.qnode(dev)
def kernel(x1, x2, n_qubit):
    qml.AngleEmbedding(x1, wires=range(n_qubit))
    qml.adjoint(qml.AngleEmbedding)(x2, wires=range(n_qubit))
    return qml.expval(qml.Projector([0]*n_qubit, wires=range(n_qubit)))
\end{lstlisting}

Using the quantum kernel, we construct the kernel matrix by computing the kernel values for all pairs of samples:

\begin{lstlisting}[language=Python]
def kernel_mat(A, B):
    mat = []
    for a in A:
    row = []
    for b in B:
    row.append(kernel(a, b, n_qubit))
    mat.append(row)
    return np.array(mat)
\end{lstlisting}

Next, we visualize the quantum kernel matrix to gain insight into its structure.
\begin{lstlisting}[language=Python]
def visualize_kernel(X, y, n_sample):
    X_vis = []
    for label in np.unique(y):
        index = y == label
        X_vis.append(X[index][:n_sample])

    X_vis = np.concatenate(X_vis, axis=0)
    n_sample_per_class = len(X_vis) // 2

    sim_mat = kernel_mat(X_vis, X_vis)
    np.save('code/chapter_4_kernel/sim_mat.npy', sim_mat)

    import matplotlib.pyplot as plt
    plt.imshow(sim_mat, cmap='viridis', interpolation='nearest')

    # Add color bar to show the scale
    plt.colorbar(label='Similarity')

    plt.axhline(n_sample_per_class - 0.5, color='red', linewidth=1.5)  # Horizontal line
    plt.axvline(n_sample_per_class - 0.5, color='red', linewidth=1.5)  # Vertical line

    xticks = yticks = np.arange(0, len(X_vis))
    xtick_labels = [f"3-{i+1}" if i < n_sample_per_class else f"6-{i+1-n_sample_per_class}" for i in range(len(X_vis))]
    ytick_labels = xtick_labels

    plt.xticks(xticks, labels=xtick_labels, rotation=90, fontsize=8)
    plt.yticks(yticks, labels=ytick_labels, fontsize=8)

    # Title and axis labels
    plt.title("Quantum Kernel Matrix")
    plt.xlabel("Sample Index")
    plt.ylabel("Sample Index")


    plt.tight_layout()
    plt.show()

visualize_kernel(X_train, y_train, 10)
\end{lstlisting}

The resulting kernel matrix is shown in Figure~\ref{cha3:fig:kernel_mat}.

\begin{figure}[H]
\centering
\includegraphics[width=0.68\textwidth]{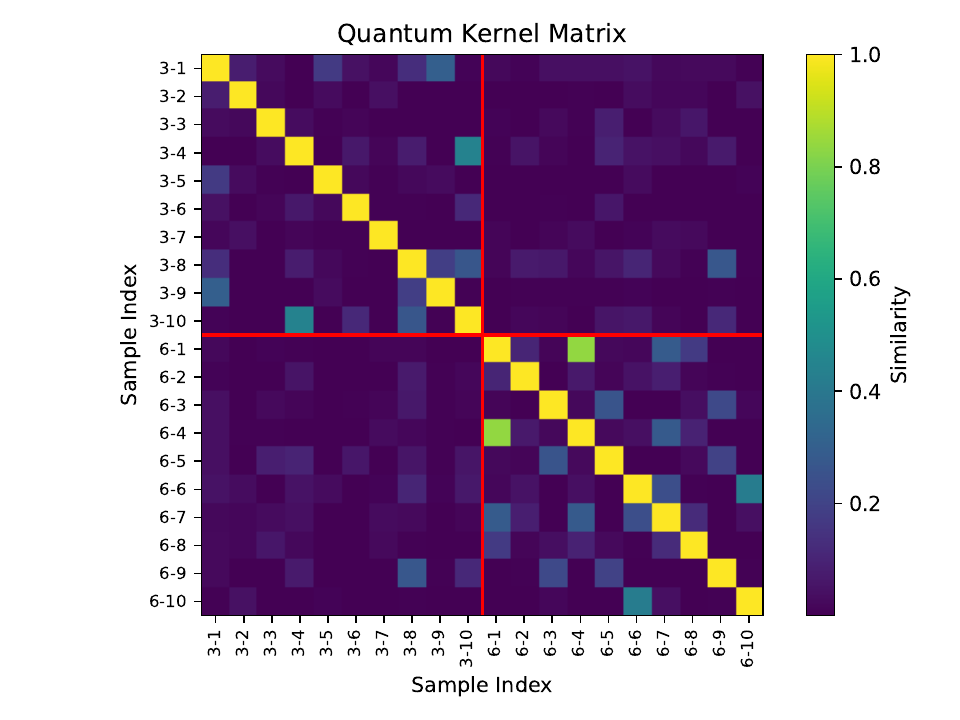}
\caption{{\textbf{Visualization of quantum kernel matrix on $20$ samples, equally drawn from two classes `3' and `6' in MNIST dataset.} } }
\label{cha3:fig:kernel_mat} 
\end{figure}

From the visualization, we observe a clear block structure:

\begin{itemize}
    \item Most of the elements in the top-left and bottom-right blocks, where samples belong to the same class, show higher similarity values.
    \item Most of the elements in the top-right and bottom-left blocks, where samples belong to different classes, exhibit lower similarity values.
\end{itemize}
This indicates that it may be possible to distinguish the two classes by setting a similarity threshold.

\noindent \textbf{Step 4: Training SVM.} To achieve higher classification accuracy, we build, train, and evaluate the SVM classifier with the quantum kernel matrix:

\begin{lstlisting}[language=Python]
svm = SVC(kernel=kernel_mat)
svm.fit(X_train, y_train)
pred = svm.predict(X_test)
print("Accuracy:", accuracy_score(y_test, pred))
\end{lstlisting}

To further analyze how the performance of the SVM with a quantum kernel depends on the size of the training dataset, we vary the number of training samples from $10$ to $100$ in increments of $10$. For each configuration, we record the corresponding classification accuracy on the test data.

\begin{lstlisting}[language=Python]
svm = SVC(kernel='precomputed')
n_sample_max = 100
X_train_sample = []
y_train_sample = []
for label in np.unique(y_train):
    index = y_train == label
    X_train_sample.append(X_train[index][:n_sample_max])
    y_train_sample.append(y_train[index][:n_sample_max])
X_train_sample = np.concatenate(X_train_sample, axis=0)
y_train_sample = np.concatenate(y_train_sample, axis=0)
kernel_mat_train = kernel_mat(X_train_sample, X_train_sample)
kernel_mat_test = kernel_mat(X_test, X_train_sample)

accuracy = []
n_samples = []
for n_sample in range(10, n_sample_max+10, 10):
    class1_indices = np.arange(n_sample)
    class2_indices = np.arange(n_sample_max, n_sample_max+n_sample)
    selected_indices = np.concatenate([class1_indices, class2_indices])

    svm.fit(kernel_mat_train[np.ix_(selected_indices, selected_indices)], np.concatenate([y_train_sample[:n_sample], y_train_sample[n_sample_max:n_sample_max+n_sample]]))
    pred = svm.predict(np.concatenate([kernel_mat_test[:, :n_sample], kernel_mat_test[:, n_sample_max:n_sample_max+n_sample]], axis=1))
    accuracy.append(accuracy_score(y_test, pred))
    n_samples.append(n_sample)

plt.plot(n_sample, accuracy, marker='o')
plt.title('Classification Accuracy vs. #Training Samples')
plt.xlabel('#Training Samples')
plt.xticks(n_sample, n_sample)
plt.ylabel('Accuracy')
plt.grid()
plt.tight_layout()
plt.show()
\end{lstlisting}

\begin{figure}[H]
    \centering
    \includegraphics[width=0.68\textwidth]{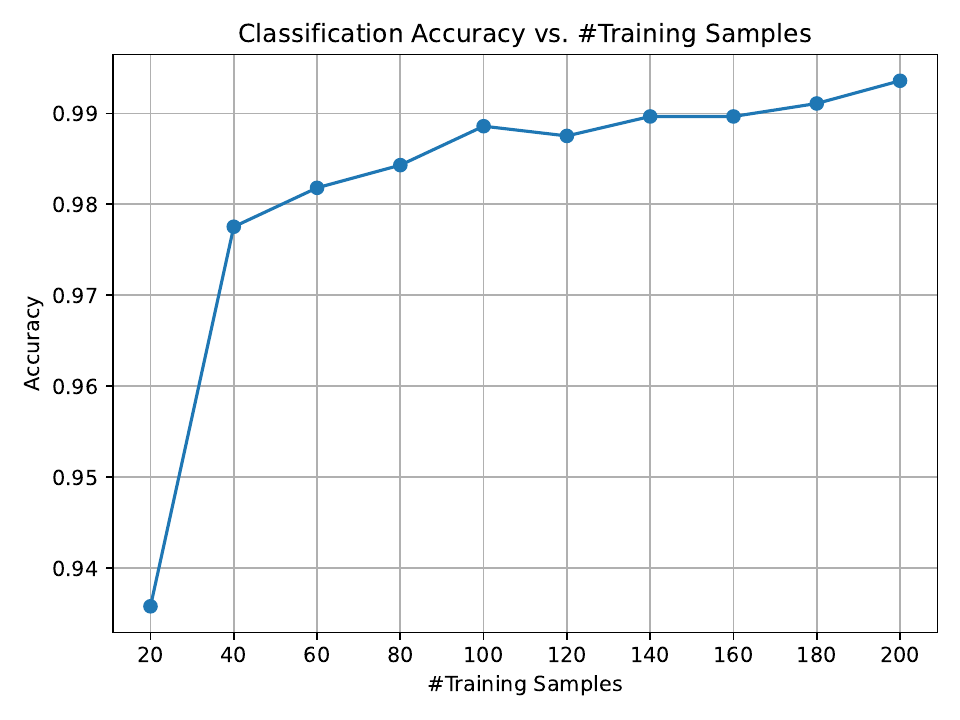}
    \caption{{\textbf{The classification accuracy on test data as a function of the number of training samples.} } }
    \label{fig:kernel_acc}  
\end{figure}

\section{Bibliographic Remarks}
The foundational concept of using quantum computers to evaluate kernel functions, namely the concept of quantum kernels, was first explored by \citet{schuld2017implementing}, who highlighted the fundamental differences between quantum kernels and quantum support vector machines. Building on this, \citet{havlivcek2019supervised} and \citet{schuld2019quantum} established a connection between quantum kernels and parameterized quantum circuits (PQCs), demonstrating their practical implementation. These works emphasized the parallels between quantum feature maps and the classical kernel trick. Since then, a large number of studies delved  into figuring out the potential of quantum kernels for solving practical real-world problems.

The recent advancements in quantum kernel machines can be roughly categorized into three key areas: \textit{kernel design}, \textit{theoretical findings}, and \textit{applications}. Specifically, the advances in kernel design focus on addressing challenges such as vanishing similarity and kernel concentration by exploring innovative frameworks.  Theoretical studies delve into the limitations and capabilities of quantum kernels, examining factors such as generalization error bounds, noise resilience, and their capacity to demonstrate quantum advantage. Finally, applications of quantum kernels showcase their potential across diverse domains. In the rest of this section, we separately review the existing developments within each of these three branches.

\subsection{Quantum kernel design}\label{chapt3:subsec:qk_design}

A crucial research line in this field focuses on constructing \textit{trainable quantum kernels} to maximize performance for specific datasets and problem domains.  In particular, traditional quantum kernels, with fixed data embedding schemes, are limited to specific feature representation spaces and often fail to capture the complex and diverse patterns inherent in real-world data. To address this limitation, \citet{lloyd2020quantum} explored the construction of task-specific quantum feature maps using measurement theory. Building on this idea, \citet{hubregtsen2022training} introduced a method to optimize quantum feature maps variationally within quantum kernels, linking this approach to the concept of data re-uploading techniques  \citep{perez2020data,schuld2021effect}.  Additionally, \citet{vedaie2020quantum} and \citet{lei2024neural} proposed leveraging multiple kernel learning and architecture search to construct quantum kernels, respectively. Last, \citet{glick2024covariant} proposed the covariant quantum kernels for efficiently solving the problems with group structure.

An orthogonal research direction is addressing the issue of exponential kernel concentration, also known as vanishing similarity, in quantum kernels \citep{thanasilp2022exponential}. Specifically, quantum kernels, which are defined as the inner product between quantum feature mappings, often suffer from the phenomenon of vanishing similarity. This issue was first highlighted by \citet{huang2021power}, who found that quantum feature mappings are typically ``far'' from one another in high-dimensional feature spaces, leading to vanishing similarity and, consequently, poor generalization performance.

To mitigate this issue, \citet{huang2021power} introduced projected quantum kernels, which store feature vectors in classical memory and evaluate a Gaussian kernel. This approach replaces the reliance on the inner product of quantum states, as seen in traditional quantum embedding kernels. Moreover, \citet{suzuki2022quantum} proposed the anti-symmetric logarithmic derivative quantum Fisher kernel, which avoids the exponential kernel concentration problem by encoding geometric information of the input data.
Beyond developing new types of quantum kernels, 
\citet{shaydulin2022importance} and \citet{canatar2022bandwidth} explored strategies to mitigate the exponential concentration issue for quantum embedding kernels by scaling input data with carefully chosen hyperparameters. This approach clusters the data-encoded quantum states closer together in feature space, reducing the risk of vanishing similarity at the cost of slightly lowering expressivity. 

\subsection{Theoretical studies of quantum kernels}
The theoretical studies of quantum kernels aim to rigorously understand their performance potential and limitations under realistic conditions, enabling the design of more effective, robust, and generalizable quantum kernel methods. Prior literature in this context focuses on exploring two aspects of quantum kernels, namely, expressivity and generalization ability.

\subsubsection{Expressivity of quantum kernels}
The expressivity of quantum kernels refers to their capacity to capture complex data relationships and represent intricate patterns in the feature space. A common approach to studying this is through the analysis of the reproducing kernel Hilbert space (RKHS), which provides insights into the underlying feature representations of quantum kernels.

\citet{schuld2021supervised} rigorously analyzed the RKHS of embedding-based quantum kernels and established the universality approximation theorem, demonstrating that quantum kernels can approximate a wide class of functions. Building on this, \citet{jerbi2023quantum} extended the analysis by investigating parameterized quantum embedding kernels, introducing a data-reuploading structure and proving a corresponding universality approximation theorem. These results underscore the expressive power of quantum kernels in representing complex data structures.

Despite these advances, the studies on expressive power and universality approximation often overlook the efficiency of constructing quantum kernels. Specifically, if the computational cost of constructing a universal quantum kernel is comparable to that of classical methods, the practical advantages of quantum kernels become questionable. 

To narrow this gap, \citet{gil2024expressivity} examined the expressive power of efficient quantum kernels that can be implemented on quantum computers within polynomial time. Their work provides a detailed analysis of the types of kernels that are achievable with a polynomial number of qubits and within polynomial time, offering insights into the feasibility and practical utility of quantum kernels in real-world scenarios.

However, alongside the exploration of expressive power, a significant challenge known as exponential kernel concentration has been identified. \citet{thanasilp2022exponential} identified four key factors contributing to this issue: high expressivity of data embeddings, global measurements, entanglement, and noise. To address this limitation, substantial research has focused on designing advanced quantum kernels to mitigate exponential kernel concentration, as discussed in Chapter~\ref{chapt3:subsec:qk_design}.

\subsubsection{Generalization of quantum kernels}
The generalization ability of a learning model—its capacity to perform well on unseen data—is a critical factor in evaluating its effectiveness. In this context, a considerable body of research has investigated the generalization ability of quantum kernels. \citet{huang2021power} established a data-dependent generalization error bound for quantum kernels and demonstrated that, for certain types of data (such as those generated by quantum circuits), quantum kernels can achieve a generalization advantage over classical learning models.

In addition, \citet{wang2021towards} explored the generalization performance of quantum kernels in the noisy scenario, where practical limitations such as finite measurements and quantum noise are taken into account. Their work rigorously showed that the generalization performance of quantum kernels could be significantly degraded in scenarios involving large training datasets, limited measurement repetitions, or high levels of system noise. To address these challenges, they proposed an effective method based on indefinite kernel learning to help preserve generalization performance under such constraints.

Beyond quantum data, \citet{liu2021rigorous} examined the generalization error of quantum kernels using artificial classical datasets, such as those based on the discrete logarithm problem. Their results demonstrated that quantum kernels could achieve accurate predictions in polynomial time for such problems, whereas classical learning models require exponential time, highlighting the potential computational advantages of quantum kernels.

Despite these promising results, \citet{kubler2021inductive} studied the generalization ability of quantum kernels from the perspective of inductive bias. They argued that quantum kernels, lacking inductive bias, often fail to outperform classical models in practical scenarios. This underscores the importance of carefully designing embeddings and aligning kernels to achieve meaningful and practical quantum advantages.

\subsubsection{Provable advantages of quantum kernels}
The potential for quantum kernels to demonstrate quantum advantage has been a central focus of research. For instance, \citet{huang2021power} provided evidence of generalization advantages for quantum kernels on quantum data. Similarly, \citet{liu2021rigorous} presented a rigorous framework showing that quantum kernels can efficiently solve problems like the discrete logarithm problem, which is believed to be intractable for classical computers under standard cryptographic assumptions. Moreover, \citet{sweke2021quantum} demonstrated quantum advantage in distribution learning tasks, offering some of the earliest theoretical evidence of quantum advantage in machine learning.

However, many of these tasks are artificial, designed specifically to showcase quantum advantages. This raises the question of how these theoretical benefits can be translated to real-world applications. In this regard, the next significant challenge is to demonstrate that quantum models can consistently outperform classical models in solving practical, real-world problems.

\subsection{Applications of quantum kernels}

Motivated by the potential of quantum kernels to recognize complex data patterns, numerous studies have explored their practical applications across diverse fields, including classification, drug discovery, anomaly detection, and financial modeling.

For instance, \citet{beaulieu2022quantum} investigate the use of quantum kernels for image classification, specifically in identifying real-world manufacturing defects. Similarly, \citet{rodriguez2024satellite} apply quantum kernels to satellite image classification, a task of particular importance in the earth observation industry. In the field of quantum physics, \citet{sancho2022quantum} and \citet{wu2023quantum} leverage quantum kernels to recognize phases of quantum matter, where quantum kernels outperform classical learning models in solving certain problems. In drug discovery, \citet{batra2021quantum} explore the potential of quantum kernels to accelerate and improve the identification of promising compounds.

Quantum kernels have also been explored in anomaly detection. \citet{liu2018quantum} demonstrate their superior performance over classical methods in detecting anomalies within quantum data.  Furthermore, \citet{grossi2022mixed} employ quantum kernel methods for fraud classification tasks, showing improvements when benchmarked against classical methods. \citet{miyabe2023quantum} expand their application to the financial domain by proposing a quantum multiple-kernel learning methodology. This approach broadens the scope of quantum kernels to include credit scoring and directional forecasting of asset price movements, highlighting their potential utility in financial services.

Despite the promise shown in these applications, the realization of quantum advantage in practical tasks remains an ongoing area of research, with current efforts directed toward identifying real-world problems where quantum kernels outperform classical alternatives.

\chapter{Quantum Neural Networks}\label{cha5:qnn}
 
Classical neural networks~\citep{lecun2015deep} are the foundation of modern artificial intelligence technologies and have achieved widespread success in fields such as computer vision~\citep{voulodimos2018deep} and natural language processing~\citep{otter2020survey}. However, despite these achievements, classical neural networks face significant challenges, including excessively large model sizes and the corresponding high computational costs~\citep{hoffmann2022training}, especially in terms of energy consumption~\citep{de2023growing}. These limitations result from their dependence on classical computational resources, which are becoming increasingly unsustainable as models grow in complexity. 

Quantum neural networks (QNNs)~\citep{jeswal2019recent} offer a promising solution by enhancing neural networks with the computational potential of quantum circuits~\citep{liu2024towards}. In QNNs, classical input data is encoded into quantum states, and quantum gates with trainable parameters process these states in ways that classical systems cannot easily replicate. This computational regime leverages quantum mechanics to explore new forms of pattern recognition and problem-solving that go beyond classical methods. Thus, QNNs have the potential to outperform classical neural networks in specific learning tasks~\citep{huang2022quantum}, where the advantages in processing and learning can be explored.  

Despite these promising features, there are challenges in realizing the full potential of QNNs, such as quantum noise~\citep{peters2021machine} and the requirement for scalable quantum hardware~\citep{acharya2024quantum}. Nonetheless, ongoing advancements in quantum hardware and algorithm design promise QNNs to address the inefficiencies of classical models, especially in areas such as quantum many-body physics~\citep{gardas2018quantum} and quantum chemistry~\citep{cao2019quantum}.

In this chapter, to provide a systematic overview, we begin by outlining the structure and function of classical neural networks in Chapter~\ref{chapt5:sec:classical_nn}, before transitioning to fault-tolerant and near-term quantum neural networks in Chapters~\ref{chapt5:sec:fault_quantum_perceptron} and \ref{chapt5:sec:qnn}, respectively. We also discuss the theoretical foundations of QNNs in Chapter~\ref{chapt5:sec:qnn_theory}, focusing on trainability, expressivity, and generalization capabilities. Finally, we provide illustrative code implementations of QNNs using the wine~\citep{Dua2019} and MNIST datasets~\citep{lecun1998gradient} in Chapter~\ref{chapt5:sec:qnn_code}.

\section{Classical Neural Networks}
\label{chapt5:sec:classical_nn}

Neural networks~\citep{mcculloch1943logical,gardner1998artificial,vaswani2017attention} are computer models inspired by the structure of the human brain, designed to process and analyze complex patterns in data. Originally developed from the concept of neurons connected by weighted pathways, neural networks have become one of the most powerful tools in artificial intelligence~\citep{lecun2015deep}. Each neuron processes its inputs by applying weights and non-linear activations, producing an output that feeds into the next layer of neurons. This structure enables neural networks to learn complex functions during training~\citep{hornik1993some}. For example, given a dataset of images and their labels, a neural network can learn to classify categories, such as distinguishing between cats and dogs, by adjusting its parameters during training. Guided by optimization algorithms such as gradient descent~\citep{amari1993backpropagation}, the learning process allows the network to gradually reduce the error between the predicted and actual outputs, allowing it to learn the best parameters for the given task. 

After nearly a century of exploration, neural networks have undergone remarkable advancements in both architectures and capabilities. The simplest model, the perceptron~\citep{mcculloch1943logical}, laid the foundation by showing how neural networks could learn to separate linearly classifiable categories. Building on this, deeper and more complex networks—such as multilayer perceptrons (MLPs)~\citep{gardner1998artificial} and  transformers~\citep{vaswani2017attention}—have enabled breakthroughs in tasks such as autonomous driving and content generation.

\subsection{Perceptron}
\label{chapt5:sec:classical_perceptron}

The perceptron model, first introduced by \citep{mcculloch1943logical}, is widely regarded as a foundational structure in artificial neural networks, inspiring architectures ranging from convolutional neural networks (CNNs)~\citep{lecun1989handwritten} and residual neural networks (ResNets)~\citep{he2016deep} to transformers \citep{vaswani2017attention}. Due to its fundamental role, we next introduce the mechanism of single-layer perceptrons. 

A single-layer perceptron comprises three fundamental components: \textit{input neurons, a weighted layer, and an output neuron} as illustrated in Figure~\ref{chap5_figure_perceptron}. Given a $d$-dimensional input vector $\bm{x} \in \mathbb{R}^d$, the input layer consists of $d$ neurons, each representing the feature $\bm{x}_{i}$ for $\forall i \in [d]$. This input is processed through a weighted summation, i.e.,
\begin{equation}\label{eq_ch5_perceptron_z}
z = \bm{w}^{\top} \bm{x},
\end{equation}
where $\bm{w}^{\top}$ is the transpose of the weight vector and $z$ is the output of the weighted layer. 
A non-linear activation function is then applied to produce the output neuron $\hat{y}$.
For the standard perceptron model, the sign function is typically used as the activation function:
\begin{equation}\label{eq_ch5_perceptron_yhat}
\hat{y} = f(z) = \left\{
\begin{aligned}
1, &\quad{} \text{if} \quad z \geq 0 , \\
-1, &\quad{} \text{if} \quad z < 0 .
\end{aligned}
\right. 
\end{equation}

\begin{figure}[t]
  \centering 
  \includegraphics[width=0.6\textwidth]{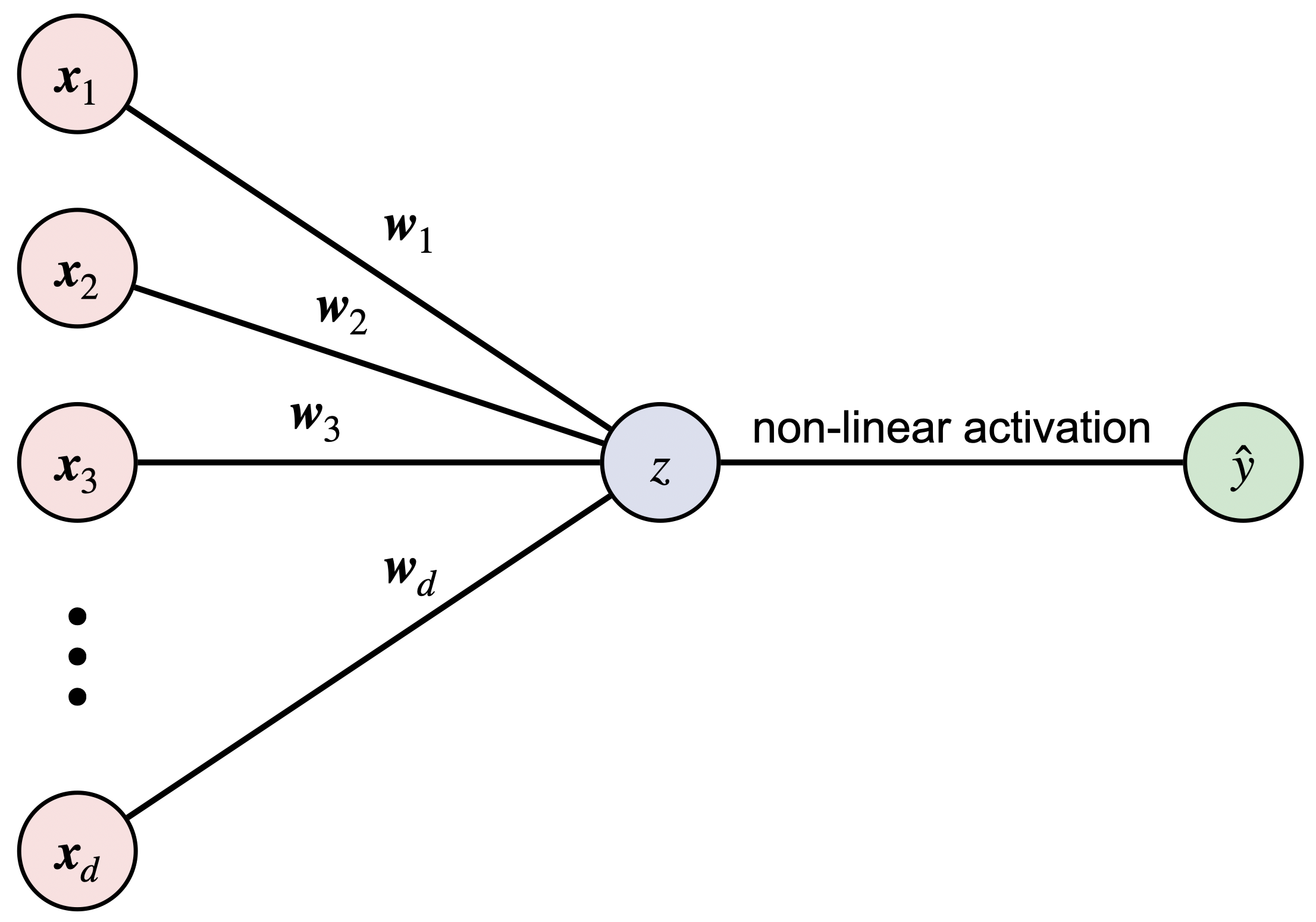}
  \caption{{\textbf{Illustration of the perceptron}. Inputs are processed with a weighted linear combination followed by a non-linear activation to produce the output. }}
  \label{chap5_figure_perceptron}
\end{figure}

The perceptron learns from input data by iteratively adjusting its trainable parameters $\bm{w}$. In particular, let $\mathcal{D}=\{(\bm{x}^{(a)}, y^{(a)})\}_{a=1}^n$ be the training dataset, where $\bm{x}^{(a)}$ represents the input features of the $a$-th example, and $y^{(a)} \in \{-1,1\}$ denotes the corresponding label. When the perceptron outputs a  prediction $\hat{y}^{(s)}$, the parameters are updated accordingly, i.e.,
\begin{align}
\bm{w} \leftarrow{}& \bm{w} + (y^{(s)} - \hat{y}^{(s)}) \bm{x}^{(s)}. \label{eq_ch5_perceptron_w_updation} 
\end{align}
The training process is repeated iteratively until the error reaches a predefined threshold. 

Perceptrons can perfectly classify linearly separable data with a finite number of mistakes, as stated in Theorem~\ref{theo_ch5_converge_perceptron}.

\begin{theorem}[Convergence of perceptrons~\citep{novikoff1962convergence}]\label{theo_ch5_converge_perceptron}
Suppose the training data consists of unit vectors separated by a margin of $\gamma$ with labels $y^{(i)} \in \{-1,1\}$. Then there exists a perceptron training algorithm that achieves zero error with at most $\mathcal{O}(\frac{1}{\gamma^2})$ mistakes. 
\end{theorem}

\begin{proof}[Proof of Theorem~\ref{theo_ch5_converge_perceptron}]

Consider the initial parameter of the perceptron $\bm{w}=\bm{0}$.
Since the training dataset is linearly separable by a margin of $\gamma$, there exists a unit vector $\bm{w}^*$ such that $y^{(i)} {\bm{w}^*}^{\top} \bm{x}^{(i)} \geq{} \gamma$ for all samples $i \in [n]$. Let $\bm{x}^{(s,t)}$ be the sample that is misclassified in the $t$-th step, which is then used for adjusting the parameter. Let $\bm{w}(t)$ be the parameter after the $t$-th step. Using Eqn.~(\ref{eq_ch5_perceptron_w_updation}), it can be shown that 
\begin{align}
{}& 
{\bm{w}^*}^{\top}   \bm{w}(t) - {\bm{w}^*}^{\top}  \bm{w}(t-1) \notag \\
={}& \left( y^{(s,t)}-\hat{y}^{(s,t)} \right) 
{\bm{w}^*}^{\top}  \bm{x}^{(s,t)} \notag \\
={}& 2 y^{(s,t)} {\bm{w}^*}^{\top}   \bm{x}^{(s,t)} \geq{} 2\gamma , \label{theo_ch5_converge_perceptron_1_3}
\end{align}
where Eqn.~(\ref{theo_ch5_converge_perceptron_1_3}) is derived by noticing the sample $(\bm{x}^{(s,t)}, y^{(s,t)})$ is misclassified with $\hat{y}^{(s,t)} \neq y^{(s,t)}$ and 
${y}^{(s,t)}, \hat{y}^{(s,t)} \in \{-1,1\}$. By considering the initialization $\bm{w}(0)=\bm{0}$, the norm of the parameter after the $t$-th step can be bounded by
\begin{align}
\left\| \bm{w}(t) \right\| \geq{}& \left| {\bm{w}^{*}}^{\top} \bm{w}(t) \right| \label{theo_ch5_converge_perceptron_2_1} \\
={}& \left| {\bm{w}^{*}}^{\top} \sum_{t'=1}^{t} \left( \bm{w}(t') - \bm{w}(t'-1) \right) \right| \label{theo_ch5_converge_perceptron_2_2} \\
\geq{}& 2\gamma t , \label{eq_ch5_perceptron_wb_norm_lowerbound} 
\end{align}
where Eqn.~(\ref{theo_ch5_converge_perceptron_2_1}) is follows from the condition $\|\bm{w}^*\|=1$. Eqn.~(\ref{eq_ch5_perceptron_wb_norm_lowerbound}) is derived by using the result in Eqn.~(\ref{theo_ch5_converge_perceptron_1_3}).
On the other hand, 
\begin{align}
{}& \left\| \bm{w}(t) \right\|^2  - \left\| \bm{w}(t-1) \right\|^2 \notag \\
={}& \left\| \bm{w}(t-1) + \left( y^{(s,t)}-\hat{y}^{(s,t)} \right) \bm{x}^{(s,t)} \right\|^2 - \left\| \bm{w}(t-1) \right\|^2 \label{theo_ch5_converge_perceptron_3_2} \\
={}& \left\| \bm{w}(t-1) + 2 y^{(s,t)} \bm{x}^{(s,t)} \right\|^2 - \left\| \bm{w}(t-1) \right\|^2 \label{theo_ch5_converge_perceptron_3_3} \\
={}& 4 \|\bm{x}^{(s,t)} \|^2 + 4 \bm{w}(t-1)^{\top}  y^{(s,t)} \bm{x}^{(s,t)} \notag \\
\leq{}& 4 + 4 \bm{w}(t-1)^{\top} y^{(s,t)} \bm{x}^{(s,t)} \label{theo_ch5_converge_perceptron_3_5} \\
\leq{}& 4 , \label{theo_ch5_converge_perceptron_3_6}
\end{align}
where Eqn.~(\ref{theo_ch5_converge_perceptron_3_2}) follows from the weight update rule in Eqn.~(\ref{eq_ch5_perceptron_w_updation}). Eqn.~(\ref{theo_ch5_converge_perceptron_3_3}) is derived by noticing that $y^{(s,t)} \neq \hat{y}^{(s,t)}$ and $y^{(s,t)}, \hat{y}^{(s,t)} \in \{-1,1\}$. Eqn.~(\ref{theo_ch5_converge_perceptron_3_5}) follows from the condition $\|\bm{x}^{(i)}\|=1$ for all samples. Eqn.~(\ref{theo_ch5_converge_perceptron_3_6}) is derived by noticing that the sample $(\bm{x}^{(s,t)}, y^{(s,t)})$ is misclassified by the perceptron with the parameter $\bm{w}(t-1)$, \ie 
$$y^{(s,t)} \bm{w}(t-1)^{\top} \bm{x}^{(s,t)} \leq 0. $$
Thus, after $t$ steps, the parameter is bounded by
\begin{align}
\|\bm{w}(t)\| \leq 2 \sqrt{t} . \label{eq_ch5_perceptron_wb_norm_upperbound}
\end{align}
Combining Eqn.~(\ref{eq_ch5_perceptron_wb_norm_lowerbound}) and Eqn.~(\ref{eq_ch5_perceptron_wb_norm_upperbound}), it can be shown that 
\begin{align}
t \leq{}& \frac{1}{\gamma^2} .
\end{align}

\end{proof}

Since the parameters are used in an inner product operation, as shown in Eqn.~\eqref{eq_ch5_perceptron_z}, the single-layer perceptron can be considered as a basic kernel method employing the identity feature mapping. Consequently, the single-layer perceptron can only classify linearly separable data and is inadequate for handling more complex tasks, such as the XOR problem~\citep{rosenblatt1958perceptron}. This limitation has driven the development of advanced neural networks, such as multilayer perceptrons (MLPs)~\citep{gardner1998artificial}, which can capture non-linear relationships by incorporating non-linear activation functions and multi-layer structures.

\subsection{Multilayer perceptron}
\label{chapt5:sec:classical_mlp}

\begin{figure}[t]
  \centering 
  \includegraphics[width=0.76\textwidth]{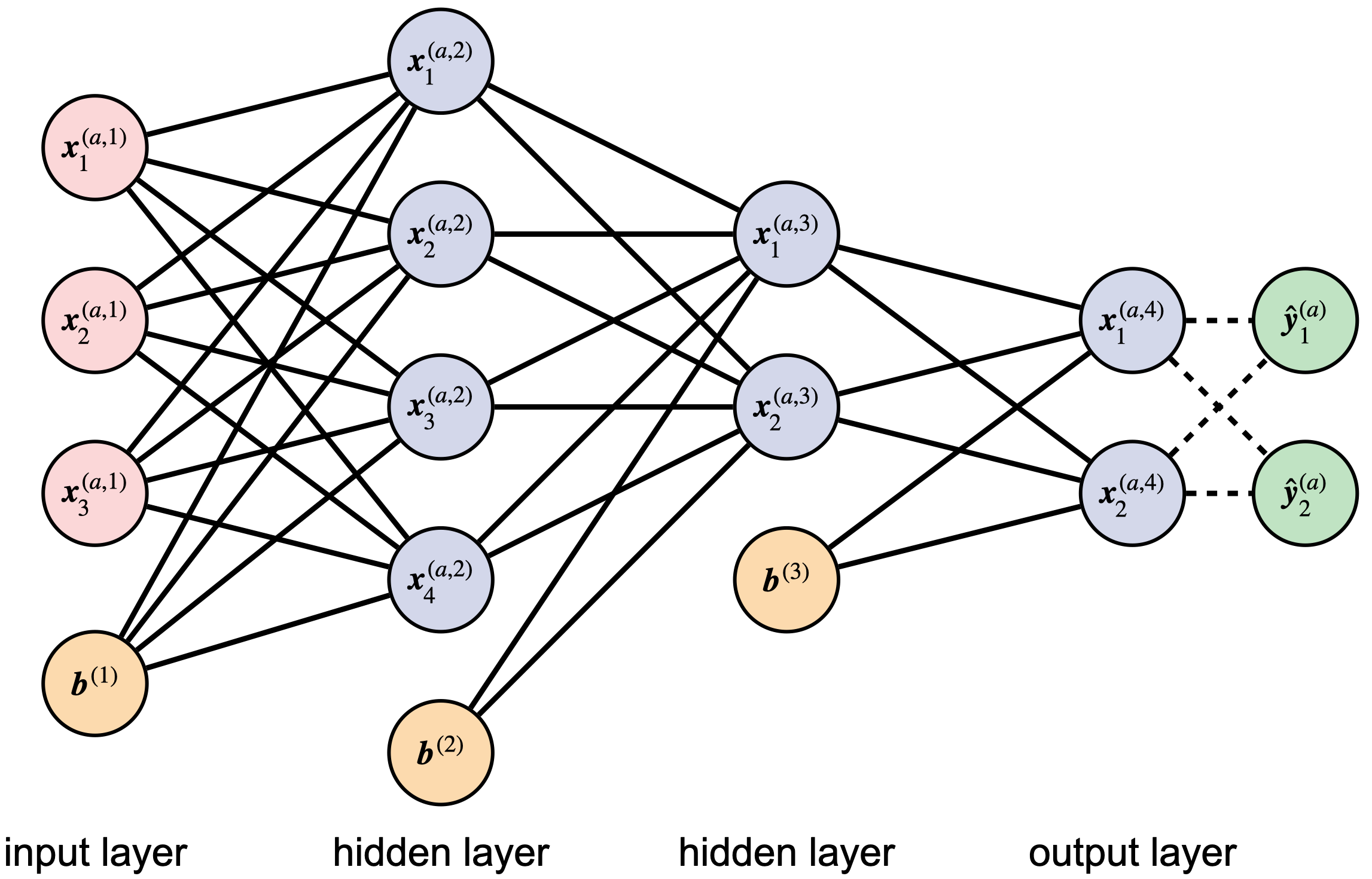}
  \caption{{\textbf{Illustration of a multilayer perceptron with two hidden layers}. Dashed lines denote softmax operations. }}
  \label{chap5_figure_mlp}
\end{figure}

The multilayer perceptron (MLP) is a fully connected neural network architecture consisting of three components: the \textit{input layer, hidden layers, and output layer,} as illustrated in Figure~\ref{chap5_figure_mlp}. Similar to the single-layer perceptron introduced in Chapter~\ref{chapt5:sec:classical_perceptron}, the neurons in the MLP are connected through weighted sums, followed by non-linear activation functions. 

The mathematical expression of MLP is as follows. Let $\bm{x}^{(a,1)}$ be the $a$-th input data and $\ell=1$ denote the input layer. Define $L$ as the number of total layers. 
The forward propagation at the $(\ell+1)$-th layer $\forall \ell \in \{1, 2, ..., L-2\}$ yields
\begin{align*}
\bm{z}^{(a, \ell+1)}  ={}& {W}^{(\ell)} \bm{x}^{(a, \ell)} + \bm{b}^{(\ell)}, \\
\bm{x}^{(a, \ell+1)} ={}& \sigma(\bm{z}^{(a, \ell+1)}),
\end{align*}
where $\sigma$ represents the non-linear activation function, and $W^{(\ell)}$ and $\bm{b}^{(\ell)}$ denotes trainable weight and the bias term, respectively. Similar to the notation $z$ in the perceptron in Chapter~\ref{chapt5:sec:classical_perceptron}, $\bm{z}^{(a,\ell+1)}$ denotes the output of the linear sum in the $\ell+1$-th layer, which is expressed in a more generalized vector form. Therefore, the parameter for the weighted linear sum is represented in matrix form as $W^{(\ell)}$. Various methodologies have been proposed for implementing non-linear activations, with some common approaches summarized in Table~\ref{chapter5-tab:nonlinear-activation}. 

After passing through $L-2$ hidden layers, the output of MLP given by the equation below serves as the prediction to approximate the target label $\bm{y}^{(a)}$, i.e.,  
\begin{equation*}
\hat{\bm{y}}^{(a)}={\rm softmax}(\bm{x}^{(a,L)}) := \frac{\left( \exp(\bm{x}^{(a,L)}_1) , \cdots, \exp(\bm{x}^{(a,L)}_{p})  \right)^{\top}}{\sum_{i=1}^{p} \exp(\bm{x}^{(a,L)}_i) },
\end{equation*} 
with $p$ here denotes the dimension of $\bm{x}^{(a,L)}$.

\begin{table}[]
\caption{Common non-linear activation functions.}
\label{chapter5-tab:nonlinear-activation}
\centering
\footnotesize
\begin{tabular}{c|c}
\toprule 
Name & Formulation \\ 
\midrule
sigmoid function & $\sigma(x)=1/(1+\exp(-x))$   \\ 
hyperbolic tangent function & $\sigma(x)=\tanh(x)$ \\
rectified linear unit (ReLU) function & $\sigma(x)=\max(0,x)$ \\
\bottomrule 
\end{tabular}
\end{table}

Next, we provide a toy example of binary classification to explain the MLP learning procedure. Let $\{(\bm{x}^{(a)}, \bm{y}^{(a)})\}_{a \in \mathcal{D}}$ be the training dataset $\mathcal{D}$, where $\bm{x}^{(a)}$ is the feature vector and $\bm{y}^{(a)} \in \{ (1,0)^{\top} , (0,1)^{\top} \}$ is the label for two categories. Consider the MLP with one hidden layer. The prediction can be expressed as follows: 
\begin{align*}
\hat{\bm{y}}^{(a)} ={}& {\rm softmax}(\bm{x}^{(a,3)})= {\rm softmax} \circ \sigma \left(\bm{z}^{(a,3)} \right) \\
={}& {\rm softmax} \circ \sigma \left({W}^{(2)} \bm{x}^{(a,2)} + \bm{b}^{(2)} \right) \\
={}& {\rm softmax} \circ \sigma \left({W}^{(2)} \sigma \left(\bm{z}^{(a,2)} \right) + \bm{b}^{(2)} \right) \\
={}& {\rm softmax} \circ \sigma \left({W}^{(2)} \sigma \left( W^{(1)} \bm{x}^{(a,1)} + \bm{b}^{(1)} \right) + \bm{b}^{(2)} \right) ,
\end{align*}
where $\circ$ denotes the function composition. We use $\sigma(x)=1/(1+\exp(-x))$ as the non-linear activation function.

\begin{figure}[t]
  \centering 
  \includegraphics[width=0.5\textwidth]{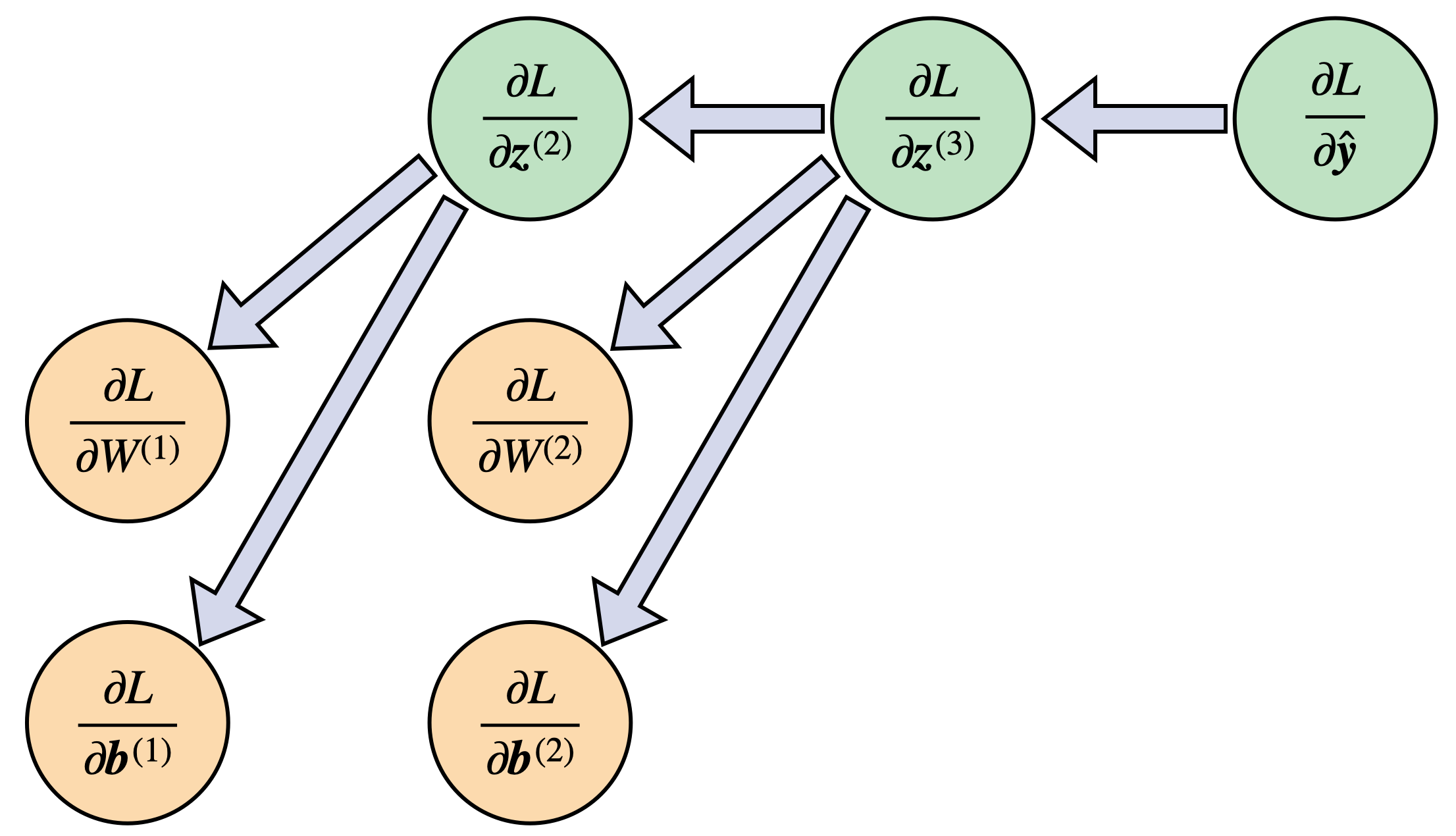}
  \caption{{\textbf{Illustration of backpropagation when calculating the gradient of an MLP with one hidden layer}. The index of sample $a$ is omitted for simplicity. }}
  \label{chap5_figure_back_propagation}
\end{figure}

MLP learns from the given dataset by minimizing the loss function with respect to the parameters $\bm{\theta}=({W}^{(1)}, {W}^{(2)}, \bm{b}^{(1)}, \bm{b}^{(2)})$, which is defined as the $\ell_2$ norm distance between the prediction and the label,
\begin{align}
\L(\bm{\theta} ) ={}& \frac{1}{|\mathcal{D}|} \sum_{a \in \mathcal{D}} \L^{(a)}(\bm{\theta} ) = \frac{1}{2|\mathcal{D}|} \sum_{a \in \mathcal{D}} \left\| \hat{\bm{y}}^{(a)}(\bm{\theta}) - \bm{y}^{(a)} \right\|^2 .
\end{align}
We use gradient descent with learning rate $\eta$ to optimize the parameters: 
\begin{equation*}
\bm{\theta}(t+1)=\bm{\theta}(t)-\eta \nabla_{\bm{\theta}} \L(\bm{\theta}(t)).
\end{equation*}
As illustrated in Figure~\ref{chap5_figure_back_propagation}, the gradient is computed using backpropagation~\citep{lecun1988theoretical} as follows. First, the gradient with respect to the output layer is given by
\begin{align*}
\frac{\partial \L^{(a)}}{\partial \hat{\bm{y}}^{(a)}} ={}&  \hat{\bm{y}}^{(a)} - \bm{y}^{(a)} , \\
\frac{\partial \L^{(a)}}{\partial {\bm{x}}^{(a,3)}} ={}& \frac{\partial \hat{\bm{y}}^{(a)}}{\partial {\bm{x}}^{(a,3)}}\frac{\partial \L^{(a)}}{\partial \hat{\bm{y}}^{(a)}}  ={} \left[ {\rm diag} \left( \hat{\bm{y}}^{(a)} \right) - \hat{\bm{y}}^{(a)} {\hat{\bm{y}}}^{(a)\top} \right] \frac{\partial \L^{(a)}}{\partial \hat{\bm{y}}^{(a)}} , \\
\frac{\partial \L^{(a)}}{\partial {\bm{z}}^{(a,3)}} ={}& \frac{\partial {\bm{x}}^{(a,3)}}{\partial {\bm{z}}^{(a,3)}} \frac{\partial \L^{(a)}}{\partial {\bm{x}}^{(a,3)}}  ={} {\rm diag} \left[ \left( \bm{1} - \bm{z}^{(a,3)} \right) \odot \bm{z}^{(a,3)} \right] \frac{\partial \L^{(a)}}{\partial {\bm{x}}^{(a,3)}} ,
\end{align*}
where $\bm{1}$ denotes the vector $(1,1,\cdots,1)^{\top}$, and $\odot$ denotes the element-wise multiplication (Hadamard product). For convenience, we omit the dimension of $\bm{1}$ here, which has the same dimension with $\bm{z}^{(a,3)}$. Next, the gradient with respect to the hidden layer can be obtained using the chain rule:
\begin{align*}
\frac{\partial \L^{(a)}}{\partial {\bm{x}}^{(a,2)}} ={}&  \frac{\partial {\bm{z}}^{(a,3)}}{\partial {\bm{x}}^{(a,2)}}\frac{\partial \L^{(a)}}{\partial {\bm{z}}^{(a,3)}} ={} W^{(2)} \frac{\partial \L^{(a)}}{\partial {\bm{z}}^{(a,3)}} , \\
\frac{\partial \L^{(a)}}{\partial {W}^{(2)}} ={}& \frac{\partial {\bm{z}}^{(a,3)}}{\partial {W}^{(2)}} \frac{\partial \L^{(a)}}{\partial {\bm{z}}^{(a,3)}}  ={} \frac{\partial \L^{(a)}}{\partial {\bm{z}}^{(a,3)}} {\bm{x}^{(a,2)\top}} , \\
\frac{\partial \L^{(a)}}{\partial \bm{b}^{(2)}} ={}& \frac{\partial {\bm{z}}^{(a,3)}}{\partial \bm{b}^{(2)}} \frac{\partial \L^{(a)}}{\partial {\bm{z}}^{(a,3)}}  ={} \frac{\partial \L^{(a)}}{\partial {\bm{z}}^{(a,3)}} , \\
\frac{\partial \L^{(a)}}{\partial \bm{z}^{(a,2)}} ={}&  \frac{\partial {\bm{x}}^{(a,2)}}{\partial \bm{z}^{(a,2)}} \frac{\partial \L^{(a)}}{\partial {\bm{x}}^{(a,2)}} ={} {\rm diag} \left[ \left( \bm{1} - \bm{z}^{(a,2)} \right) \odot \bm{z}^{(a,2)} \right] \frac{\partial \L^{(a)}}{\partial {\bm{x}}^{(a,2)}}  .
\end{align*}
The gradient with respect to the parameters for the input layer is derived similarly:
\begin{align*}
\frac{\partial \L^{(a)}}{\partial {W}^{(1)}} ={}& \frac{\partial {\bm{z}}^{(a,2)}}{\partial {W}^{(1)}} \frac{\partial \L^{(a)}}{\partial {\bm{z}}^{(a,2)}}  ={} \frac{\partial \L^{(a)}}{\partial {\bm{z}}^{(a,2)}} {\bm{x}^{(a,1)\top}} , \\
\frac{\partial \L^{(a)}}{\partial \bm{b}^{(1)}} ={}& \frac{\partial {\bm{z}}^{(a,2)}}{\partial \bm{b}^{(1)}} \frac{\partial \L^{(a)}}{\partial {\bm{z}}^{(a,2)}}  ={} \frac{\partial \L^{(a)}}{\partial {\bm{z}}^{(a,2)}} .
\end{align*}
After multiple training epochs, the loss function converges to a value below a predefined threshold, which leads to a small classification error.

Compared to single-layer perceptrons, MLP can model non-linear relationships by employing hidden layers and activation functions. This enables them to learn abstract representations by capturing the complex patterns inherent in the data. Mathematically, the power of MLPs is guaranteed by the universal approximation theorem, as stated in Theorem~\ref{ch5_fact_mlp_universal}, which asserts that a single hidden layer is sufficient to approximate any arbitrary continuous function.

\begin{fact}[Universal Approximation Theorem, informal version adapted from \citet{hornik1989multilayer}]\label{ch5_fact_mlp_universal}
Let $\mathcal{C}(\mathcal{X},\mathbb{R}^m)$ denote the set of continuous functions from a subset $\mathcal{X}$ of a Euclidean space $\mathbb{R}^n$ to a Euclidean space $\mathbb{R}^m$. Denote by $\sigma$ a function that is not polynomial. Then for every $n,m \in \mathbb{N}$, compact set $\mathcal{K} \subseteq \mathbb{R}^n$, $f \in \mathcal{C}(\mathcal{K},\mathbb{R}^m)$, and $\epsilon > 0$, there exist $k \in \mathbb{N}$, $A \in \mathbb{R}^{k \times n}$, $\bm{b} \in \mathbb{R}^k$, and $C \in \mathbb{R}^{m \times k}$ such that
\begin{equation*}
\mathop{\sup}_{\bm{x} \in \mathcal{X}} \left\| f(\bm{x}) - g(\bm{x}) \right\| < \epsilon ,
\end{equation*}
where $g(\bm{x})=C  \sigma(A \bm{x}+\bm{b})$.
\end{fact}

\begin{tcolorbox}[enhanced, breakable,colback=gray!5!white,colframe=gray!75!black,title=Remark]
MLPs involve a large number of parameters due to their fully connected multilayer architecture. This high parameter count enables MLPs to possess considerable representational power, allowing them to model complex data distributions. However, the excessive capacity to fit the training data often leads to overfitting~\citep{caruana2000overfitting}, where the MLP captures noise and irrelevant patterns instead of generalizable features. As a result, MLPs tend to perform poorly on unseen data, especially when the training set is limited or noisy. To mitigate this issue, advanced techniques such as dropout~\citep{srivastava2014dropout}, weight decay~\citep{krogh1991simple}, and attention mechanisms~\citep{vaswani2017attention} have been proposed to reduce overfitting in MLPs while maintaining sufficient expressivity. 
\end{tcolorbox}

\section{Fault-tolerant Quantum Perceptron}
\label{chapt5:sec:fault_quantum_perceptron}

The primary aim of advancing quantum machine learning is to harness the computational advantages of quantum mechanics to enhance performance across various learning tasks. As outlined in Chapter~\ref{chapt1:subsec:diff_quantum_adv}, these advantages manifest in several ways, including reduced runtime, lower query complexity, and improved sample efficiency compared to classical models. A notable example of this is the quantum perceptron model~\citep{kapoor2016quantum}. As a \textsf{FTQC}-based QML algorithm grounded in the Grover search, quantum perceptron offers a quadratic improvement in the query complexity during the training over its classical counterpart. For comprehensiveness, we first introduce the Grover search algorithm, followed by a detailed explanation of the quantum perceptron model.  

\subsection{Grover search}
\label{chapt5:sec:grover_search}

Grover search~\citep{grover1996fast} provides runtime speedups for unstructured search problems, which have broad applications in cryptography, quantum machine learning, and constraint satisfaction problems. Unlike classical search methods that require $\mathcal{O}(d)$ queries for a dataset with $d$ entries, Grover's algorithm can identify the target element with high probability using only $\mathcal{O}(\sqrt{d})$ queries to a quantum oracle. Consequently, quantum algorithms incorporating Grover search have the potential to achieve a quadratic speedup over classical approaches.

In general, a search task can be abstracted as a function $f({x})$ such that $f({x})=1$ if ${x}$ belongs to the solution set of the search problem, and $f({x})=0$ otherwise. We consider a dataset consisting of $d=2^N$ elements, where each element is represented by the quantum state $|{x}\>$ with $x=0,1,\cdots,d-1$. In this process, two key quantum oracles are introduced. The first oracle, $U_0 = 2(|0\>\<0|)^{\otimes N} -\mathbb{I}_{d}$, applies a phase shift of $e^{i\pi}=-1$ to all quantum states except $|0\>^{\otimes N}$, which remains unchanged. The second oracle, $U_f$, operates in a similar manner: it applies a phase shift of $-1$ to quantum states that belong to the solution set while leaving all other states unaffected. The procedure for Grover search is described in Algorithm~\ref{chapt5:sec:grover_search_alg}.

\begin{algorithm}
\caption{Grover search}\label{chapt5:sec:grover_search_alg}
\begin{algorithmic}[1]  
\Require Quantum oracles $U_f$ and $U_0$. The size of the dataset and the solution set, denoted by $d=2^N$ and $M$, respectively. 
\Ensure An index corresponds to one of the solution states with high probability. 
\State Initialize a register of $N$ qubits with the state of uniform superposition:
\begin{equation*}
|\phi_0\> = \frac{1}{\sqrt{d}} \sum_{x=0}^{d-1} |x\> = \bigotimes_{n=1}^{N} \frac{|0\>+|1\>}{\sqrt{2}} = \left( \bigotimes_{n=1}^{N} \Hada \right) |0\>.
\end{equation*}
\State Let $m= \lfloor \frac{\pi}{4} \sqrt{\frac{d}{M}}  -\frac{1}{2} \rfloor$. Apply the following operation:
\begin{equation*}
|\phi_m\> = \left[ {\Hada}^{\otimes N} U_0 {\Hada}^{\otimes N} U_f \right]^{m} |\phi_0\>.
\end{equation*}
\State Measure the state $|\phi_m\>$ to generate an index.
\end{algorithmic}
\end{algorithm}

\begin{theorem}[Time complexity of Grover search]\label{theo_ch5_time_com_grover}
Grover search finds a solution to the unstructured search problem with high probability in time $\mathcal{O}(\sqrt{\frac{d}{M}} (\log d + T_f))$, where $d$ is the size of the dataset, $M$ is the size of the solution set, and $T_f$ denotes the time complexity of implementing the oracle $U_f$.
\end{theorem}

\begin{tcolorbox}[enhanced, 
    breakable,colback=gray!5!white,colframe=gray!75!black,title=Remark]
The Grover search achieves a quadratic speed-up in the query complexity of the oracle $U_f$. It provides a quantum advantage in the runtime only if the time complexity of the oracle $U_f$, denoted as $T_f$, is less than $\sqrt{d}$. 
\end{tcolorbox}

\begin{proof}[Proof of Theorem~\ref{theo_ch5_time_com_grover}]  
Let the superposition of the solution state be 
\begin{align*}
|{\rm target} \> ={}& \frac{1}{\sqrt{M}} \sum_{x:f(x)=1} |x\>.  
\end{align*}
Similarly, the superposition of the states outside the solution set is given by
\begin{align*}
|{\rm other} \> ={}& \frac{1}{\sqrt{2^N-M}} \sum_{x|f(x)=0} |x\>. 
\end{align*}
Thus, the initial uniform superposition state can be expressed as
\begin{align*}
|\phi_0\> ={}& \sqrt{\frac{M}{2^N}} |{\rm target} \> + \sqrt{\frac{2^N-M}{2^N}} |{\rm other} \> := \alpha_0 |{\rm target} \> + \beta_0 |{\rm other} \> .
\end{align*}

In principle, the coefficient associated with the target state is expected to increase during the quantum state evolution, such that the solution could be obtained through quantum measurement with high probability. The dynamics of these coefficients can be described as follows:
\begin{align*}
\alpha_k ={}& \< {\rm target} | \phi_k \> \\
={}& \<{\rm target} | H^{\otimes N} U_0 H^{\otimes N} U_f | \phi_{k-1} \> \\
={}& \<{\rm target} | \left( 2 |\phi_0\> \< \phi_0 | - \mathbb{I} \right) U_f \left( \alpha_{k-1} |{\rm target} \> + \beta_{k-1} |{\rm other} \> \right) \\
={}& \<{\rm target} | \left( 2 |\phi_0\> \< \phi_0 | - \mathbb{I} \right) \left( - \alpha_{k-1} |{\rm target} \> + \beta_{k-1} |{\rm other} \> \right) \\
={}& \left(1-2\alpha_0^2 \right) \alpha_{k-1} + 2 \alpha_0 \beta_0 \beta_{k-1} , \\
\beta_{k} ={}& \< {\rm other} | \phi_k \> \\
={}& \<{\rm other} | \left( 2 |\phi_0\> \< \phi_0 | - \mathbb{I} \right) \left( - \alpha_{k-1} |{\rm target} \> + \beta_{k-1} |{\rm other} \> \right) \\
={}& \left( 2 \beta_0^2 - 1 \right) \beta_{k-1} - 2\alpha_0 \beta_0 \alpha_{k-1} .
\end{align*}
Let the angle $\theta = \arccos \sqrt{\frac{2^N-M}{2^N}}$, then by induction, it can be shown that 
\begin{equation*}
\alpha_k = \sin [(2k+1)\theta] , \quad \beta_k = \cos [(2k+1)\theta] .
\end{equation*}
To ensure that the coefficient $\alpha_m = \mathcal{O}(1)$, there is a condition $(2m+1)\theta \approx \pi/2$. Therefore, $m= \mathcal{O}(1/\theta)=\mathcal{O}(\sqrt{\frac{2^N}{M}})$ suffices to obtain the solution with high probability. 

\end{proof}

\subsection{Online quantum perceptron with quadratic speedups}
\label{chapt5:sec:online_quantum_perceptron_model}

As stated in Theorem~\ref{theo_ch5_converge_perceptron}, for a linearly separable dataset with a margin $\gamma$, a perceptron model can achieve perfect classification after making $\mathcal{O}(1/\gamma^2)$ mistakes during training. In classical approaches, identifying a sample that is misclassified by the current model may require up to $\mathcal{O}(d)$ queries, where $d$ denotes the size of the training dataset. In contrast, the quantum perceptron model~\citep{kapoor2016quantum} can identify misclassified samples more efficiently by employing the Grover search algorithm, achieving a quadratic speed-up in the query complexity. 

To begin, we introduce the setup of the input data.
We consider the classification of a dataset $\{z^{(i)}\}_{i=1}^{d} = \{(\bm{x}^{(i)}, y^{(i)})\}_{i=1}^{d}$, where the label $y^{(i)} \in \{-1,1\}$. For convenience, we assume that the number of samples is a power of $2$, i.e., 
$d=2^N$. Each data vector $\bm{x}^{(i)}$ is assumed to be represented by using $B$ bits. The information of each sample $z^{(i)}$ is stored in the quantum state $|z^{(i)}\>$ by using $B+1$ qubits. 

\begin{shadedbox}
\begin{example}
For the sample $(\bm{x}^{(i)},y^{(i)})=([0,0,1,0],1)$, the corresponding quantum state is $|z^{(i)}\>=|00101\>$, where the last qubit encodes the label (with ``$0$'' for the label ``$-1$''), and the remaining qubits represent the data vector. When $\bm{x}^{(i)}$ is a float vector, a similar bit sequence can be obtained by concatenating the binary representations of the elements in $\bm{x}^{(i)}$. 
\end{example}
\end{shadedbox}

Next, we introduce the oracle models. We assume the existence of a quantum oracle $U$ for encoding training data as the corresponding quantum state, i.e.,
\begin{align}
U |i\>|0\> ={}& |i\>   |z^{(i)}\> \ , \ U^\dag |i\>  |z^{(i)}\> = |i\>|0\> . \label{ch5_eq_quantum_perceptron_oracle_U}
\end{align}
Due to the linearity of unitary,
\begin{align}
U \sum_{i=0}^{d-1} \frac{1}{\sqrt{d}}|i\>|0\> ={}& \sum_{i=0}^{d-1} \frac{1}{\sqrt{d}}|i\> |z^{(i)}\> .
\end{align}
In addition to the input oracle $U$ described in Eqn.~(\ref{ch5_eq_quantum_perceptron_oracle_U}), the quantum perceptron model employs another oracle to distinguish between correctly classified and misclassified quantum states. Specifically, the oracle ${F}_{\bm{w}}'$ satisfies 
\begin{align}
{F}_{\bm{w}}' |z^{(i)}\> ={}& (-1)^{f(\bm{w},z^{(i)})} |z^{(i)}\> ,
\end{align}
where $f:(\bm{w},z^{(i)}) \rightarrow \{0,1\}$. 
The function outputs $1$ if the current perceptron model with weight $\bm{w}$ misclassifies the training sample $z^{(j)}$; otherwise, it outputs $0$. Furthermore, we define 
\begin{align}
F_{\bm{w}} ={}& U^\dag (\mathbb{I} \otimes {F}_{\bm{w}}') U , 
\end{align}
which is used as the oracle $U_f$ in the Grover search. The online quantum perceptron procedure is given in Algorithm~\ref{chapt5:sec:online_quantum_perceptron_alg}. The query complexity of the online quantum perceptron is provided in Theorem~\ref{theo_ch5_online_quantum_perceptron}.

\begin{algorithm}
\caption{Online quantum perceptron}\label{chapt5:sec:online_quantum_perceptron_alg}
\begin{algorithmic}[1] 
\Require Linearly separable dataset $\{z^{(i)}\}_{i=1}^{d} = \{(\bm{x}^{(i)}, y^{(i)})\}_{i=1}^{d}$, where $d=2^N$. Margin threshold $\gamma$. Constants $\epsilon \in (0,1)$ and $c \in (1,2)$.
\Ensure Weight $\bm{w}$ for a perceptron that correctly classifies the dataset with a margin $\gamma$ with probability at least $1-\epsilon$.
\State Initialize the weight $\bm{w}=\bm{0}$.
\For{$h=1,\cdots,\lceil \frac{1}{\gamma^2} \rceil $} \label{chapt5:sec:online_quantum_perceptron_alg_2}
\For{$k=1,\cdots,\lceil \log_{3/4} \gamma^2 \epsilon \rceil $} \label{chapt5:sec:online_quantum_perceptron_alg_3}
\For{$j=1,\cdots,\lceil \log_{c} \frac{1}{\sin (2\sin^{-1} (1/\sqrt{d}))} \rceil $} \label{chapt5:sec:online_quantum_perceptron_alg_4}
\State Draw $m$ uniformly from $\{0,\cdots, \lceil c^j \rceil -1 \}$. \label{chapt5:sec:online_quantum_perceptron_alg_5}
\State Prepare the quantum state
\begin{equation*}
|\phi_0\> = \frac{1}{\sqrt{d}} \sum_{i=0}^{d-1} |i\> .
\end{equation*} \label{chapt5:sec:online_quantum_perceptron_alg_6}
\State Generate the state 
\begin{equation*}
|\phi_1\> = \left\{ \left[ (2|\phi_0\>\<\phi_0|-\mathbb{I}_{d}) \otimes \mathbb{I}_d \right] F_{\bm{w}} \right\}^m |\phi_0\> |0\>^{\otimes N} .
\end{equation*} \label{chapt5:sec:online_quantum_perceptron_alg_7}
\State Measure the first register of the state $|\phi_1\>$ to obtain an outcome $q$. \label{chapt5:sec:online_quantum_perceptron_alg_8} 
\If{$f(\bm{w},z^{(q)})=1$}   
\State Update $\bm{w} \leftarrow \bm{w}+y^{(q)} \bm{x}^{(q)}$. \label{chapt5:sec:online_quantum_perceptron_alg_10}
\EndIf  
\EndFor \label{chapt5:sec:online_quantum_perceptron_alg_12}
\EndFor \label{chapt5:sec:online_quantum_perceptron_alg_13}
\EndFor \label{chapt5:sec:online_quantum_perceptron_alg_14}
\State Output $\bm{w}$. \label{chapt5:sec:online_quantum_perceptron_alg_15}
\end{algorithmic}
\end{algorithm}

\begin{theorem}[Online quantum perceptron~\citep{kapoor2016quantum}]\label{theo_ch5_online_quantum_perceptron}
Consider a training dataset that consists of unit vectors $\{\bm{x}^{(1)}, \cdots, \bm{x}^{(d)}\}$ and labels $\{{y}^{(1)}, \cdots, {y}^{(d)}\}$ with a margin $\gamma$, Denote by $n_{\rm quant}$ the number of queries to $F_{\bm{w}}$ needed to learn the weight $\bm{w}$, such that the training dataset is perfectly classified with probability at least $1-\epsilon$, then
\begin{equation*}
n_{\rm quant} \in \mathcal{O} \left( \frac{\sqrt{d}}{\gamma^2} \log \frac{1}{\gamma^2 \epsilon} \right).
\end{equation*}
For the classical case where the training vectors are uniformly sampled from the training dataset, the number of queries to $f_{\bm{w}}$ is bounded by
\begin{equation*}
\Omega({d}) \ni n_{\rm class} \in \mathcal{O} \left( \frac{{d}}{\gamma^2} \log \frac{1}{\gamma^2 \epsilon} \right).
\end{equation*}
\end{theorem}

\begin{proof}[Proof of Theorem~\ref{theo_ch5_online_quantum_perceptron}]

The main idea of the quantum perceptron model in Algorithm~\ref{chapt5:sec:online_quantum_perceptron_alg} is to replace the procedure of finding the misclassified sample in classical perceptrons with the Grover search. Due to convergence result for perceptrons in Theorem~\ref{theo_ch5_converge_perceptron}, $h=1,\cdots,\lceil \frac{1}{\gamma^2} \rceil$ iterations of Steps (\ref{chapt5:sec:online_quantum_perceptron_alg_3}-\ref{chapt5:sec:online_quantum_perceptron_alg_13}) suffice to update the weight $\bm{w}$ towards the case of perfect classification. Therefore, Theorem~\ref{theo_ch5_online_quantum_perceptron} is the direct consequence of the following lemmas and Theorem~\ref{theo_ch5_converge_perceptron}. The query complexity of classical perceptrons has the lower bound $\Omega(d)$, since the model needs to go through the entire dataset in the worst case.

\end{proof}

\begin{lemma}\label{lemma_ch5_query_perceptron_classical}
Given only uniform sampling access to the training dataset, there exists a classical perceptron that either finds a misclassified sample to update the weight $\bm{w}$ or concludes that no such example exists with probability $1-\epsilon \gamma^2$, using $\mathcal{O}(d\log(1/\epsilon\gamma^2))$ queries to $f_{\bm{w}}$.
\end{lemma}

\begin{lemma}\label{lemma_ch5_query_perceptron_quantum}
The procedure of Steps \ref{chapt5:sec:online_quantum_perceptron_alg_3}-\ref{chapt5:sec:online_quantum_perceptron_alg_13} in Algorithm~\ref{chapt5:sec:online_quantum_perceptron_alg} either finds a misclassified sample to update the weight $\bm{w}$ or concludes that no such example exists with probability $1-\epsilon \gamma^2$, using $\mathcal{O}(\sqrt{d}\log(1/\epsilon\gamma^2))$ queries to $F_{\bm{w}}$.
\end{lemma}

\begin{proof}[Proof of Lemma~\ref{lemma_ch5_query_perceptron_classical}]
First, let $m_c=d\lceil \log (1/\epsilon\gamma^2)\rceil$ be the number of samples drawn from the dataset uniformly in each iteration of training. Suppose these samples are classified correctly, then the probability that the entire dataset is classified correctly is
\begin{align*}
\Pr(\text{Correct classification}) \geq {}& 1- \left( 1 - \frac{1}{d} \right)^{m_c} \geq 1 - \exp \left( - \frac{m_c}{d} \right) \geq 1 - \epsilon \gamma^2 .
\end{align*}
 
\end{proof}

\begin{proof}[Proof of Lemma~\ref{lemma_ch5_query_perceptron_quantum}]

For convenience, denote $\theta_a := \arccos \sqrt{\frac{d-d_{0}}{d}}$, where $d_{0}$ the number of misclassified samples in the dataset according to the current model. Let $d_1:=\lceil \log_{c} \frac{1}{\sin (2\sin^{-1} (1/\sqrt{d}))} \rceil$. 
Here, an exponential expansion strategy is used in Steps \ref{chapt5:sec:online_quantum_perceptron_alg_4}-\ref{chapt5:sec:online_quantum_perceptron_alg_12} to handle the scenario of unknown $d_0$. Namely, quantum operations in the Grover search are repeated for $m$ times, where $m$ is drawn from an exponentially expanded set ${0,\cdots, \lceil c^j \rceil -1}$ uniformly for a predefined $c \in (1,2)$ and $j = 1,\cdots,d_1$. It can be shown that this strategy can find a misclassified sample before the convergence of Algorithm~\ref{chapt5:sec:online_quantum_perceptron_alg} with an average probability at least $1/4$:
\begin{align*}
 \Pr \left( f(\bm{w}, z^{(q)})=1 \right) ={}& \sum_{j=1}^{d_1} \frac{1}{\lceil c^j \rceil }  \sum_{m=0}^{\lceil c^j \rceil -1 } \sin^2 ((2m+1) \theta_a) \\
\geq{}& \frac{1}{\lceil c^{d_1} \rceil }  \sum_{m=0}^{\lceil c^{d_1} \rceil -1 } \sin^2 ((2m+1) \theta_a) \\
={}&  \frac{1}{2} \left[ 1 - \frac{\sin (4 \lceil c^{d_1} \rceil \theta_a)}{2 \lceil c^{d_1} \rceil \sin (2 \theta_a)} \right] \\
\geq{}& \frac{1}{4} .
\end{align*}

The procedure of Steps \ref{chapt5:sec:online_quantum_perceptron_alg_4}-\ref{chapt5:sec:online_quantum_perceptron_alg_12} is repeated for $k=1,\cdots,\lceil \log_{3/4} \gamma^2 \epsilon \rceil $ iterations to accumulate the success probability. The probability of finding a misclassified sample in Steps \ref{chapt5:sec:online_quantum_perceptron_alg_3}-\ref{chapt5:sec:online_quantum_perceptron_alg_13} before the convergence of Algorithm~\ref{chapt5:sec:online_quantum_perceptron_alg} is at least
\begin{align}
1 - \left( 1 - \frac{1}{4} \right)^{\lceil \log_{3/4} \epsilon \gamma^2 \rceil}  \geq{}& 1 - \epsilon \gamma^2 . 
\end{align}

Finally, the query complexity $Q$ of Steps \ref{chapt5:sec:online_quantum_perceptron_alg_3}-\ref{chapt5:sec:online_quantum_perceptron_alg_13} in Algorithm~\ref{chapt5:sec:online_quantum_perceptron_alg} can be upper bounded as follows:
\begin{align*}
Q \leq{}& \sum_{k=1}^{\lceil \log_{3/4} \gamma^2 \epsilon \rceil} \sum_{j=1}^{d_1} c^j  \\
\leq{}& \left( 1 + \log_{3/4} \gamma^2 \epsilon \right) \frac{c}{1-c} \left[ 1 - c^{d_1} \right] \\
\leq{}& \left( 1 + \log_{3/4} \gamma^2 \epsilon \right) \frac{c^2}{c-1} \left[ \frac{1}{\sin (2\sin^{-1} (1/\sqrt{d}))} -1 \right] \\
={}& \mathcal{O}(\sqrt{d} \log \frac{1}{\epsilon \gamma^2}) .
\end{align*}

\end{proof}

\section{Near-term Quantum Neural Networks}
\label{chapt5:sec:qnn}

Following recent experimental breakthroughs in superconducting quantum hardware architectures~\citep{arute2019quantum,acharya2024quantum,abughanem2024ibm,gao2024establishing}, researchers have devoted considerable effort to developing and implementing quantum machine learning algorithms optimized for current and near-term quantum devices~\citep{wang2024comprehensive}. Compared to fault-tolerant quantum computers, these devices face three primary limitations: quantum noise, limited coherence time, and circuit connectivity constraints. Regarding quantum noise, state-of-the-art devices have single-qubit gate error rates of $10^{-4} \sim 10^{-3}$ and two-qubit gate error rates of approximately $10^{-3} \sim 10^{-2}$~\citep{abughanem2024ibm,gao2024establishing}. The coherence time is around $10^2 \mu s$~\citep{acharya2024quantum,abughanem2024ibm,gao2024establishing}, primarily limited by decoherence in noisy quantum channels. Regarding circuit connectivity, most superconducting quantum processors employ architectures that exhibit two-dimensional connectivity patterns and their variants~\citep{acharya2024quantum,abughanem2024ibm,gao2024establishing}. Gate operations between non-adjacent qubits must be executed through intermediate relay operations, leading to additional error accumulation. To address these inherent limitations, the quantum neural network (QNN) framework has been proposed. Specifically, these QNNs are designed to perform meaningful computations on near-term quantum devices.

\subsection{General framework}
\label{chapt5:sec:qnn_imple}

\begin{figure}
\centering
\includegraphics[width=0.8\textwidth]{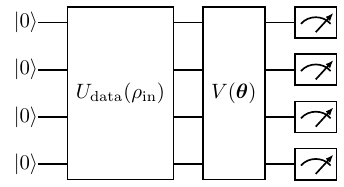}
\caption{\textbf{Illustration of a QNN}. The input state $\rho_{\textrm{in}}$ is prepared using the operation $U_{\rm data}$, followed by a variational quantum circuit (VQC) $V(\bm{\theta})$ and the measurement operation. }
\label{chap5_figure_vqc}
\end{figure}

In this section, we introduce the basic architecture of QNNs. As illustrated in Figure~\ref{chap5_figure_vqc}, a basic QNN consists of three components: the input, the model circuit, and the measurement.

\noindent\textbf{Input}. The QNN uses quantum states $\rho_{\rm in}$ as input data. As shown in Table~\ref{ch5_tab_classical_quantum_data}, QNNs can process both classical and quantum data. Specifically, the input states $\rho_{\rm in}$ may be introduced from physical processes such as quantum Hamiltonian evolutions or be constructed to encode classical vectors using encoding protocols introduced in Chapter~\ref{cha3:subsec:q-read-in}, such as angle encoding and amplitude encoding.

\noindent\textbf{Model circuit}. QNNs employ variational quantum circuits (VQCs), a.k.a, ansatzes, to extract and learn features from input data. A typical VQC, denoted as $V(\bm{\theta})$, adopts a layered architecture that consists of both parameterized and fixed quantum gates, with the former being trainable. For problem-agnostic implementations, an effective parameterization strategy is to use the parameters $\bm{\theta}$ as the phases of single-qubit rotation gates $\RX, \RY, \RZ$, while quantum entanglement is introduced through fixed two-qubit gates, such as $CX$ and $CZ$. Standard circuit architectures include the hardware-efficient circuit (HEC)~\citep{kandala2017hardware}, shown in Figure~\ref{chap5_figure_hec}, and the quantum convolutional neural network (QCNN)~\citep{cong2019quantum}, shown in Figure~\ref{chap5_figure_qcnn}. For problem-specific applications, such as finding the ground states of molecular Hamiltonians, specialized circuits like the unitary coupled cluster \textit{ansatz}~\citep{peruzzo2014variational} are employed.

\begin{figure}
\centering
\includegraphics[width=0.8\textwidth]{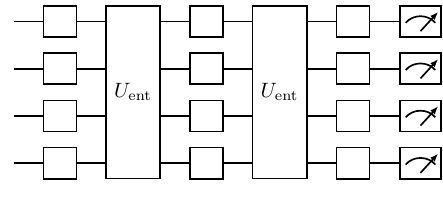}
\caption{\textbf{Illustration of a hardware-efficient circuit with two entanglement layers}. }
\label{chap5_figure_hec}
\end{figure}

\begin{figure}
\centering
\includegraphics[width=0.8\textwidth]{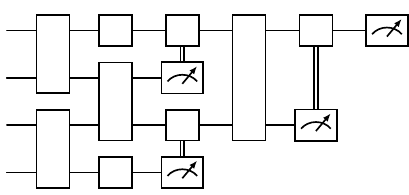}
\caption{\textbf{Illustration of a quantum convolutional neural network}. }
\label{chap5_figure_qcnn}
\end{figure}

\begin{table}
\centering  
\caption{Examples of classical and quantum data employed in QNNs, where $H$ denotes the system Hamiltonian, $k_B$ is the Boltzmann constant, and $|\phi_0\>$ is a predefined initial state.}
\label{ch5_tab_classical_quantum_data}
\footnotesize
\begin{tabular}{c|c|c}
\toprule 
 & Example & Input state formulation \\ \midrule
\multirow{2}{*}{Classical data}  & Angle encoding & $\bigotimes_{n=1}^{N} \left[ \RY(\bm{x}_n) |0\> \right]$ \\   
  & Amplitude encoding & $ \sum_{i=0}^{d-1} \bm{x}_i/\|\bm{x}\|_2 |i\>$ \\  
  \hline
\multirow{2}{*}{Quantum data}  & Gibbs state & $ \frac{\exp(-H/k_B T)}{{\rm Tr} \left[\exp(-H/k_B T) \right]}$ \\  
  & Hamiltonian evolution & $ \exp(-iHt)|\phi_0\>$ \\  
\bottomrule
\end{tabular}
\end{table}

\begin{shadedbox}
\begin{example}
Hardware-efficient circuits incorporate several widely adopted \textit{ansatzes}. Single-qubit rotations $\{\RX, \RY, \RZ\}$ are used to construct parameterized single-qubit unitaries. The entangled unitary layer can be implemented using two-qubit gates such as
\begin{align*}
U_{\rm ent} ={}& \bigotimes_{n=1}^{\lceil \frac{N}{2} \rceil} U(2n-1+k\%2,2n+k\%2)   \quad {\text{for the $k$-th layer}}, \\
U_{\rm ent} ={}& \bigotimes_{n=1}^{\lceil \frac{N}{2} \rceil} U(2n-1,2n) \bigotimes_{n=1}^{\lceil \frac{N-1}{2} \rceil} U(2n,2n+1) ,
\end{align*}
where $U \in \{ CX, CZ\}$.
\end{example}
\end{shadedbox}

\noindent\textbf{Measurement}. After implementing the model circuit, the quantum state is measured using specific observables, denoted as $O$, to extract classical information. 
The choice of observables depends on the experimental objectives. In the case of a variational quantum eigensolver, where the goal is to find the ground state and energy of a given Hamiltonian, the observable is chosen to be the target Hamiltonian itself. In quantum machine learning applications involving classical data, the measurement outcomes are used to approximate label information, which typically lacks direct physical significance. As a result, the observable can, in principle, be any Hermitian operator. However, for practical experimental considerations, a linear combination of Pauli-Z operators is commonly used as the observable:
\begin{equation}\label{ch5_qnn_eq_observable_cnzn}
O = \sum_{j=1}^{N} \bm{c}_j \mathbb{I}^{\otimes (j-1)} \otimes Z_{j} \otimes \mathbb{I}^{\otimes (N-j)} ,
\end{equation}
where $\bm{c} \in \mathbb{R}^{N}$ is a weight vector.
The measurement outcome of QNN can be expressed as a function of $\bm{\theta}$, i.e.,
\begin{align}
f(\bm{\theta};\rho_{\rm in},V,O) ={}& \Tr \left[ O V(\bm{\theta}) \rho_{\rm in}  V(\bm{\theta})^\dag \right] . \label{ch5_qnn_eq_ftheta}
\end{align}

\noindent\textbf{Training of QNNs}. As a QML framework, the optimization of QNNs amounts to updating parameters $\bm{\theta}$ using gradient-based methods. Due to the linearity of quantum mechanics and the unitary evolution constraint, for certain cases, the acquisition of gradients can be elegantly performed using the parameter-shift rule.

\begin{theorem}[Parameter-shift rule~\citep{crooks2019gradients}]\label{theo_ch5_parameter_shift}
Suppose the gate $G_j (\bm{\theta}_j)$ in a VQC $V(\bm{\theta})$ has a unitary Hamiltonian $H_j$, then the corresponding gradient could be obtained as
\begin{equation*}
\frac{\partial f}{\partial \bm{\theta}_j} (\bm{\theta}) = \frac{1}{2} \left[ f \left(\bm{\theta} + \frac{\pi}{2} \bm{e}^{(j)} \right) - f \left(\bm{\theta} - \frac{\pi}{2} \bm{e}^{(j)} \right) \right] ,
\end{equation*}
where the function $f$ follows the Eqn.~(\ref{ch5_qnn_eq_ftheta}) and the one-hot vector $\bm{e}^{(j)}$ has the same dimension with $\bm{\theta}$ with the $j$-th element being $1$.
\end{theorem}

\begin{proof}[Proof of Theorem~\ref{theo_ch5_parameter_shift}]

For convenience, we denote the detailed structure of VQC as 
\begin{equation*}
V(\bm{\theta}) = \prod_{i=L}^{1} G_i(\bm{\theta}_i) W_i,  
\end{equation*}
where $L$ is the number of parameters in VQC, $G_i$ is the parameterized gate, and $W_i$ is the fixed gate. By assumption, the gate takes the form as
\begin{equation*}
G_j(\bm{\theta}_j) =\exp(-i H_j \bm{\theta}_j /2),  
\end{equation*}
where the Hamiltonian $H_j$ is a unitary. For convenience, unnecessary parameterized and fixed gates can be merged into the state $\rho_{\rm in}$ and the observable $O$, i.e., 
\begin{align*}
\rho_{\rm in}' ={}&   W_j \left( \prod_{i=j-1}^{1} G_i(\bm{\theta}_i) W_i \right) \rho_{\rm in} \left( \prod_{i=1}^{j-1} W_i^\dag G_i(\bm{\theta}_i)^\dag \right) W_j^\dag , \\
O' ={}& \left( \prod_{i=j}^{L} W_i^\dag G_i(\bm{\theta}_i)^\dag \right) O \left( \prod_{i=L}^{j} G_i(\bm{\theta}_i) W_i \right) .
\end{align*}
It can be shown that
\begin{align}
f(\bm{\theta})={}& \Tr \left[ O V(\bm{\theta}) \rho_{\rm in}  V(\bm{\theta})^\dag \right] \notag \\
={}& \Tr \left[ O' G_j(\bm{\theta}_j) \rho_{\rm in}'  G_j(\bm{\theta}_j)^\dag \right] \notag \\
={}& \Tr \left[ O' \exp(-i H_j \bm{\theta}_j /2) \rho_{\rm in}'  \exp(i H_j \bm{\theta}_j /2) \right] \notag \\
={}& \cos^2 \frac{\bm{\theta}_j}{2} \Tr \left[ O' \rho_{\rm in}' \right] + \frac{i}{2} \sin  \bm{\theta}_j \left[ [H_j , O']  \rho_{\rm in}' \right] + \sin^2 \frac{\bm{\theta}_j}{2}  \Tr \left[ H_j O' H_j \rho_{\rm in}' \right], \label{ch5_qnn_eq_ftheta_prime}
\end{align}
where $[A,B]:=AB-BA$ denotes the commutator.

After some calculations from Eqn.~(\ref{ch5_qnn_eq_ftheta_prime}), it can be shown that
\begin{align*}
f \left(\bm{\theta} + \frac{\pi}{2} \bm{e}^{(j)} \right) ={}& \frac{1-\sin \bm{\theta}_j}{2} \Tr \left[ O' \rho_{\rm in}' \right] + \frac{i}{2} \cos  \bm{\theta}_j \left[ [H_j , O']  \rho_{\rm in}' \right] \\
+{}& \frac{1+\sin \bm{\theta}_j}{2} \Tr \left[ H_j O' H_j \rho_{\rm in}' \right] \\
f \left(\bm{\theta} - \frac{\pi}{2} \bm{e}^{(j)} \right) ={}& \frac{1+\sin \bm{\theta}_j}{2} \Tr \left[ O' \rho_{\rm in}' \right] - \frac{i}{2} \cos  \bm{\theta}_j \left[ [H_j , O']  \rho_{\rm in}' \right] \\
+{}& \frac{1-\sin \bm{\theta}_j}{2} \Tr \left[ H_j O' H_j \rho_{\rm in}' \right] \\
\frac{\partial f}{\partial \theta_j} (\bm{\theta}) ={}& - \frac{1}{2} \sin \bm{\theta}_j \Tr \left[ O' \rho_{\rm in}' \right] + \frac{i}{2} \cos \bm{\theta}_j \left[ [H_j , O']  \rho_{\rm in}' \right] \\
+{}& \frac{1}{2} \sin \bm{\theta}_j \Tr \left[ H_j O' H_j \rho_{\rm in}' \right].
\end{align*}

Comparing the above equations, Theorem~\ref{theo_ch5_parameter_shift} is proved.

\end{proof}

\subsection{Discriminative learning with QNNs}
\label{chapt5:sec:qnn_dis}

In this section, we present an example in which a QNN is employed for discriminative learning. Specifically, we focus on binary classification, where the label $y^{(i)} = \pm 1$ corresponds to the input state $\rho^{(i)}$. In the case of classical data, the state $\rho^{(i)}=|\psi(\bm{x}^{(i)})\>\<\psi(\bm{x}^{(i)})|$ can be generated from the classical vector $\bm{x}^{(i)}$ using a read-in approach 
\begin{equation}
|\psi(\bm{x}^{(i)})\> = U_\phi (\bm{x}^{(i)}) |0\> ,  
\end{equation}
where a simple feature map can be constructed via angle encoding, as introduced in Chapter~\ref{cha3:subsec:q-read-in},
\begin{equation}
U_\phi (\bm{x}^{(i)}) = \bigotimes_{n=1}^{N} \RY (\bm{x}_n^{(i)}) = \bigotimes_{n=1}^{N} \exp(-iY \bm{x}_n^{(i)} /2).
\end{equation}
Denote by $O$ and $V(\bm{\theta})$ the quantum observable and the VQC, respectively. The prediction function of the QNN is given by
\begin{equation}
\hat{y}^{(i)}(\bm{\theta}) = \Tr [ O V(\bm{\theta}) \rho^{(i)} V(\bm{\theta})^\dag ] .  
\end{equation}

In the binary classification task, the QNN learns by training the parameter $\bm{\theta}$ to minimize the distance between the label $y^{(i)}$ and the prediction $\hat{y}^{(i)}(\bm{\theta})$. Specifically, the mean square error (MSE) is used as the loss function:
\begin{equation}\label{theo_ch5_qnn_discriminative_lossL}
\bm{\theta}^* = {\rm argmin} \mathcal{L}(\bm{\theta}) , \ \text{where } \mathcal{L}(\bm{\theta}) = \sum_{i=1}^{n} \ell (\bm{\theta}, \bm{x}^{(i)}, y^{(i)}) = \frac{1}{2}   \sum_{i=1}^{n} \left( \hat{y}^{(i)} (\bm{\theta}) - y^{(i)} \right)^2 .
\end{equation}
The gradient of the loss in Eqn.~(\ref{theo_ch5_qnn_discriminative_lossL}) can be calculated via the chain rule, i.e.,
\begin{equation}
\nabla_{\bm{\theta}} \mathcal{L}(\bm{\theta}) = \sum_{i=1}^{n} \left( \hat{y}^{(i)} (\bm{\theta}) - y^{(i)} \right) \nabla_{\bm{\theta}} \hat{y}^{(i)} (\bm{\theta}) , 
\end{equation}
where the gradient of the prediction $\hat{y}^{(i)}$ can be obtained by using the parameter-shift rule in Theorem~\ref{theo_ch5_parameter_shift}. Consequently, a variety of gradient-based optimization algorithms, such as stochastic gradient descent~\citep{amari1993backpropagation}, Adagrad~\citep{duchi2011adaptive}, and Adam~\citep{kingma2014adam}, can be employed to train QNNs.

\begin{tcolorbox}[enhanced, 
breakable,colback=gray!5!white,colframe=gray!75!black,title=Remark]
The QNN binary classification framework can be naturally extended to multi-label classification using the \textbf{one-vs-all} strategy. Specifically, we train $k$ QNN binary classifiers for $k$ classes, with each classifier distinguishing a specific class from the others.
\end{tcolorbox}

\begin{tcolorbox}[enhanced, breakable,colback=gray!5!white,colframe=gray!75!black,title=Remark]
The QNN classification framework presented in this section can be extended to quantum regression learning by incorporating continuous labels.
\end{tcolorbox}

\subsection{Generative learning with QNNs}
\label{chapt5:sec:qnn_gen}

In this section, we introduce a quantum generative model implemented by QNNs, namely quantum generative adversarial network (QGAN)~\citep{lloyd2018quantum}. Similar to its classical counterparts, QGAN learns to generate samples by employing a discriminator and a generator, which are engaged in a two-player minimax game. Specifically, both the discriminator $D$ and the generator $G$ can be implemented using QNNs. By leveraging the expressive power of QNNs, QGAN has the potential to exhibit quantum advantages in certain tasks~\citep{bravyi2018quantum,zhu2022generative}.

To illustrate the training and sampling processes of QGAN, we present two examples based on the quantum patch and batch GANs proposed by \citet{huang2021experimental}. Let $N$ denote the number of qubits and $M$ the number of training samples. The patch and batch strategies are designed for the cases where $N < \lceil \log M \rceil$ and $N > \lceil \log M \rceil$, respectively. Specifically, the patch strategy enables the generation of high-dimensional images with limited quantum resources, while the batch strategy facilitates parallel training when sufficient quantum resources are available.

\begin{figure}[t]
  \centering 
  \includegraphics[width=.8\textwidth]{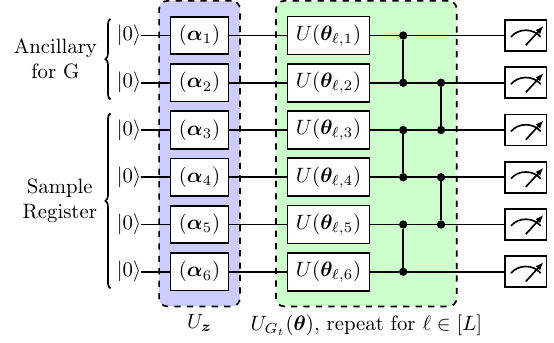}
  \caption{{\textbf{The quantum generator used in the quantum patch GAN}, where each $U(\bm{\theta}_{\ell,n}) \in \mathcal{U}(2)$ is a trainable single-qubit unitary. }}
  \label{chap5_figure_batchqgan}
\end{figure}

\subsubsection{Quantum patch GAN}

We begin by introducing the quantum patch GAN, which consists of a quantum generator, as illustrated in Figure~\ref{chap5_figure_batchqgan}, a classical discriminator, and a classical optimizer. Both the learning and sampling processes of an image are performed in patches, involving $T$ sub-generators. For the $t$-th sub-generator, the model takes a latent state $\bm{z}$ as input and generates a sample $G_t(\bm{z})$.
Specifically, the latent state is prepared from the initial state $|0\>^{\otimes N}$ using a single-qubit rotation layer, where the parameters $\{\bm{\alpha}_n\}_{n=1}^{N}$ are sampled from the uniform distribution over $[0,2\pi)$. The latent state is then processed through an $N$-qubit hardware-efficient circuit $U_{G_t} (\bm{\theta})$, which leads to the state
\begin{equation}
|\psi_t (\bm{z}) \> = U_{G_t} (\bm{\theta}) |\bm{z}\> .
\end{equation}

To perform non-linear operations, partial measurements are conducted, and a subsystem $\mathcal{A}$ (ancillary qubits) is traced out from the state $|\psi_t(\bm{z})\>$. 
The resulting mixed state is
\begin{equation}
\rho_t(\bm{z}) = \frac{\Tr_{\mathcal{A}} \left[ \Pi \otimes \mathbb{I} |\psi_t (\bm{z}) \> \< \psi_t (\bm{z}) | \right] }{\Tr \left[ \Pi \otimes \mathbb{I} |\psi_t (\bm{z}) \> \< \psi_t (\bm{z}) | \right]} ,
\end{equation}
where $\Pi$ is the projective operator acting on the subsystem $\mathcal{A}$.
Subsequently, the mixed state $\rho_t(\bm{z})$ is measured in the computational basis to obtain the sample $G_t(\bm{z})$. Specifically, let $\Pr(J=j):=\Tr[|j\>\<j|\rho_t(\bm{z})]$, where the probabilities of the outcomes can be estimated by the measurement. The sample $G_t(\bm{z})$ is then defined as
\begin{align}
G_t(\bm{z}) ={}& [ \Pr(J=0), \cdots, \Pr(J=j), \cdots, \Pr(J=2^{N-N_{\mathcal{A}}}-1) ],
\end{align}
where $N_{\mathcal{A}}$ is the number of qubits in $\mathcal{A}$.
Finally, the complete image is reconstructed by aggregating these samples from all sub-generators as follows:
\begin{equation}
G(\bm{z}) = [G_1(\bm{z}), \cdots, G_T (\bm{z})] .
\end{equation}

In principle, the discriminator $D$ in a quantum patch GAN can be any classical neural network that takes the training data $\bm{x}$ or the generated sample $G(\bm{z})$ as input, with the output
\begin{equation}
D(\bm{x}), \ D(G(\bm{z})) \in [0,1] . 
\end{equation}
Let $\bm{\gamma}$ and $\bm{\theta}$ denote the parameters of the discriminator $D$ and the generator $G$, respectively. The optimization problem for the quantum patch GAN can be formulated as:
\begin{align}
\min_{\bm{\theta}} \max_{\bm{\gamma}} \mathcal{L} ( D_{\bm{\gamma}} (G_{\bm{\theta}}(\bm{z})) , D_{\bm{\gamma}}(\bm{x})):= \mathop{\mathbb{E}}_{\bm{x}} \left[ \log D_{\bm{\gamma}}(\bm{x}) \right] + \mathop{\mathbb{E}}_{\bm{z}} \left[ \log (1-D_{\bm{\gamma}}(G_{\bm{\theta}}(\bm{z}))) \right] . \label{ch5_eq_gan_loss}
\end{align}
Similar to quantum discriminative learning, the quantum patch GAN can be trained using gradient-based optimization algorithms.

\subsubsection{Quantum batch GAN}

\begin{figure}[t]
  \centering 
  \includegraphics[width=0.8\textwidth]{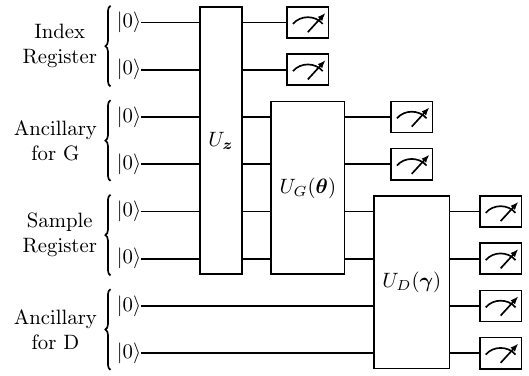}
  \caption{{\textbf{The main structure of the
quantum batch GAN. } The figure illustrates the process of generation and training using fake data. The oracle $U_{\bm{z}}$ for encoding latent vectors, the quantum generator $U_G(\bm{\theta})$, and the quantum discriminator $U_D(\bm{\gamma})$ are applied sequentially to the initial state $|0\>^{\otimes N}$. Both $U_G$ and $U_D$ share the same hardware-efficient structure as shown in Figure~\ref{chap5_figure_batchqgan}. In the case of real data, the operations $U_{\bm{z}}$ and $U_G(\bm{\theta})$ are replaced by the oracle $U_{\bm{x}}$.}}
  \label{chap5_figure_batchqgan2}
\end{figure}

As illustrated in Figure~\ref{chap5_figure_batchqgan2}, the quantum batch GAN differs from the quantum patch GAN by employing a quantum discriminator. In a quantum batch GAN, all qubits are divided into two registers: the index register, consisting of $N_I$ qubits, and the feature register, consisting of $N_F$ qubits. The qubits in the feature register are further partitioned into three parts: $N_D$ qubits for generating quantum samples, $N_{A_G}$ qubits for implementing non-linear operations in the generator $G_{\bm{\theta}}$, and $N_{A_D}$ qubits for implementing non-linear operations in the discriminator $D_{\bm{\gamma}}$. For a batch with size $|B_k|=2^{N_I}$, two oracles are used to encode the information of latent vectors and training samples:
\begin{align}
|0\>_{I} \otimes |0\>_{F} \xrightarrow{U_{\bm{z}}}{}& \frac{1}{2^{N_I}} \sum_i |i\>_{I} \otimes |\bm{z}^{(i)}\>_{F} , \\  
|0\>_{I} \otimes |0\>_{F} \xrightarrow{U_{\bm{x}}}{}& \frac{1}{2^{N_I}} \sum_i |i\>_{I} \otimes |\bm{x}^{(i)}\>_{F} .
\end{align}

\begin{tcolorbox}[enhanced, 
breakable,colback=gray!5!white,colframe=gray!75!black,title=Remark]
For data with $M$ features, state preparation for amplitude encoding in $U_{\bm{x}}$ requires $\tilde{\mathcal{O}}(2^{N_I}M)$ multi-controlled quantum gates, which is infeasible for current quantum devices. This challenge can be addressed by employing pre-trained shallow circuit approximations of the given oracle~\citep{benedetti2019generative}.
\end{tcolorbox}

After the encoding stage, a PQC $U_G(\bm{\theta})$ and the corresponding partial measurement are employed as the quantum generator. Thus, the generated state corresponding to $|B_k|$ fake samples is obtained as follows:
\begin{align*}
{}& \frac{1}{2^{N_I}} \sum_i |i\>_{I} \otimes |\bm{z}^{(i)}\>_{F} \\
\xrightarrow{U_G(\bm{\theta})}{}& \frac{1}{2^{N_I}} \sum_i |i\>_{I} \otimes \left( U_G (\bm{\theta}) \otimes \mathbb{I}_{2^{N_{A_D}}} |\bm{z}^{(i)}\>_{F} \right) := |\psi (\bm{z}) \> \\
\xrightarrow{\Pi_{A_G}}{}& \frac{ \mathbb{I}_{2^{N_I}} \otimes \Pi_{A_G} \otimes \mathbb{I}_{2^{N_D+N_{A_D}}} |\psi(\bm{z})\> }{ \Tr \left[ \mathbb{I}_{2^{N_I}} \otimes \Pi_{A_G} \otimes \mathbb{I}_{2^{N_D+N_{A_D}}} |\psi(\bm{z})\> \< \psi(\bm{z})| \right] } := | G_{\bm{\theta}}(\bm{z})\> ,
\end{align*}
where the partial measurement $\Pi_{A_G}=(|0\>\<0|)^{\otimes N_{A_G}}$ serves as the non-linear operation. In the sampling stage, the reconstructed image is generated similarly to the quantum patch GAN. Specifically, the $i$-th image $G_{\bm{\theta}}(\bm{z}^{(i)})$ in the batch is 
\begin{align}
G_{\bm{\theta}}(\bm{z}^{(i)}) ={}& \left[ \Pr(J=0|I=i), \cdots, \Pr(J=2^{N_D}-1|I=i) \right],\label{chap5_eq_Gzi}
\end{align}
where 
\begin{align}
\Pr(J=j|I=i) ={}& \Tr \left[ |i\>_I |j\>_F \<i|_I \<j|_F |G(\bm{z})\>\<G(\bm{z})| \right] .
\end{align}

Finally, we introduce the training stage. A quantum discriminator is applied to either the fake generated state $| G_{\bm{\theta}}(\bm{z})\>$ or the real data state $|\bm{x}\>$. Similar to the quantum generator, the quantum discriminator $D_{\bm{\gamma}}$ consists of a PQC $U_D(\bm{\gamma})$, followed by the corresponding partial measurement. In the case of the real state, the state evolution proceeds as follows:
\begin{align*}
{}& \frac{1}{2^{N_I}} \sum_i |i\>_{I} \otimes |\bm{x}^{(i)}\>_{F} \\
\xrightarrow{U_D(\bm{\gamma})}{}& \frac{1}{2^{N_I}} \sum_i |i\>_{I} \otimes \left( \mathbb{I}_{2^{N_{A_G}}} \otimes U_D (\bm{\gamma}) |\bm{x}^{(i)}\>_{F} \right) := |\psi (\bm{x}) \> \\
\xrightarrow{\Pi_{A_D}}{}& \frac{ \mathbb{I}_{2^{N-N_{A_D}}} \otimes \Pi_{A_G} |\psi(\bm{x})\> }{ \Tr \left[ \mathbb{I}_{2^{N-N_{A_D}}} \otimes \Pi_{A_G} |\psi(\bm{x})\> \< \psi(\bm{x})| \right] } := | D_{\bm{\gamma}}(\bm{x})\> ,
\end{align*}
where the partial measurement is $\Pi_{A_G}=(|0\>\<0|)^{\otimes N_{A_D}}$. The classical description $D_{\bm{\gamma}}(\bm{x})$ is generated similarly to Eqn.~(\ref{chap5_eq_Gzi}). The generated state $G_{\bm{\theta}}|\bm{z}\>$ undergoes the same procedure to obtain the description $D_{\bm{\gamma}}(G_{\bm{\theta}}(\bm{z}))$. These classical vectors are then used in the loss function in Eqn.~(\ref{ch5_eq_gan_loss}) to train parameters $\bm{\theta}$ and $\bm{\gamma}$.

\section{Theoretical Foundations of Quantum Neural Networks}
\label{chapt5:sec:qnn_theory}

The primary goal of QNNs is to make accurate predictions on unseen data. Achieving this goal depends on three key factors: expressivity, generalization ability, and trainability, as illustrated in Figure~\ref{chap5:fig:qnn_learnability}. A thorough analysis of these factors is crucial for understanding the potential advantages and limitations of QNNs compared to classical counterparts. Instead of providing an exhaustive review of all theoretical results, this section focuses on emphasizing key conceptual insights of QNNs.

\begin{figure}[h!]
  \centering
  \includegraphics[width=0.96\textwidth]{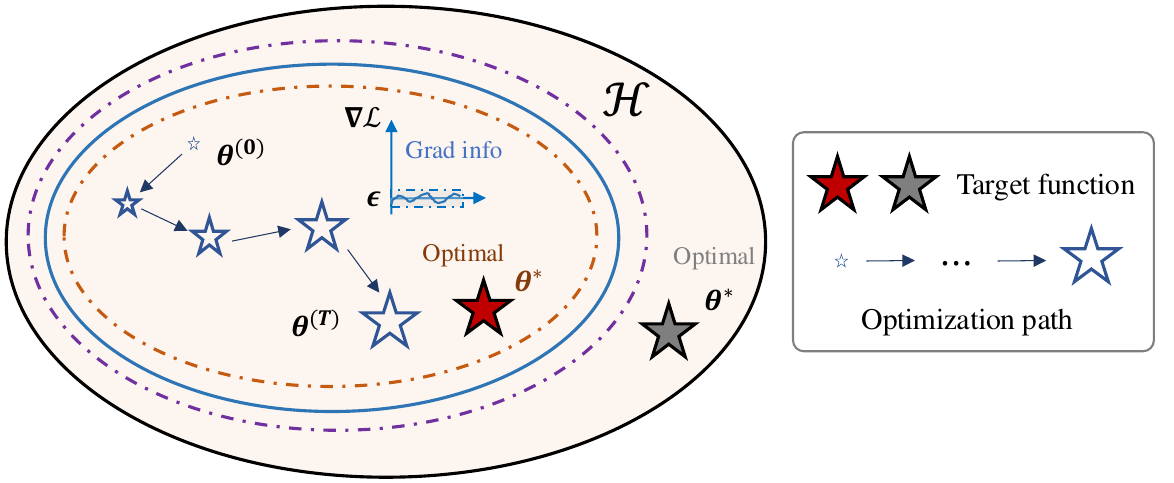}
    \caption{{\textbf{Overview of the expressivity, generalization ability, and trainability of QNNs.} The expressivity of the employed QNNs determines its hypothesis space $\mathcal{H}$ (solid blue ellipse). When $\mathcal{H}$ has a moderate size and encompasses the target concept (solid red star), QNNs can achieve good performance. Conversely, if $\mathcal{H}$ fails to cover the target concept (solid gray star) due to limited expressivity, the performance of QNNs deteriorates. During QNN optimization, a significant challenge arises from the vanishing gradient problem, commonly referred to as the barren plateau. This issue prohibits a good estimation near the target parameters $\bm{\theta}^*$. }}
  \label{chap5:fig:qnn_learnability}  
\end{figure}

As explained in Chapter~\ref{chapt3:sec:theo_foundation_QK}, expressivity refers to a model's ability to represent a wide range of functions, determining the smallest achievable training error. In Chapter~\ref{chapt5:sec:expr_gene}, we will characterize the expressivity of QNNs using the \textit{covering number}, an advanced tool from statistical learning theory. This analysis will reveal the relationship between the expressivity of QNNs and their structural factors, such as the size of the quantum system and the number of exploited quantum gates. Understanding this connection helps clarify how QNNs' expressivity scales with their architecture.

Generalization ability evaluates the discrepancy between a model's performance on the training data and on unseen test data. In Chapter~\ref{chapt5:sec:expr_gene}, we will further explore the relationship between the generalization ability and expressivity of QNNs by deriving a generalization error bound in terms of the covering number. This bound provides insights into how the expressivity of QNNs--specifically their structural factors--may impact their ability to generalize and offers a framework to assess their potential advantages over classical ML models.

While expressivity and generalization ability are crucial for both quantum kernels and QNNs, trainability emerges as an additional consideration for QNNs due to the introduction of trainable parameters in quantum circuits. This leads to fundamentally different optimization challenges, where many existing results from classical ML models no longer apply. Specifically, trainability refers to a model's ability to efficiently converge to a good solution during training, directly influencing the computational cost of training. In Chapter~\ref{chapt5:sec:trainability}, we will introduce a well-known challenge in training QNNs, referred to as the barren plateau problem, where gradients vanish exponentially as the system size increases, making optimization intractable. Additionally, we will discuss various strategies to address this issue, offering practical insights into enhancing the trainability of QNNs.

\subsection{Expressivity and generalization of quantum neural networks}
\label{chapt5:sec:expr_gene}
The expressivity and generalization 
are deeply interconnected within the framework of statistical learning theory for understanding the prediction ability of any learning model. To better understand these two terms in the context of quantum neural networks, we first review the framework of empirical risk minimization (ERM), which is a popular framework for analyzing these abilities in statistical learning theory. 

Consider the training dataset $\mathcal{D} = \{(\bm{x}^{(i)},{y}^{(i)})\}_{i=1}^{n} \in \mathcal{X} \times \mathcal{Y} $ sampled independently from an unknown distribution $\mathcal{P}$, a learning algorithm $\mathcal{A}$ aims to use the dataset $\mathcal{D}$ to infer a hypothesis $h_{\bm{\theta}^*}: \mathcal{X}\to \mathcal{Y}$ from the hypothesis space $\mathcal{H}$ that could accurately predict all labels of $\bm{x}\in \mathcal{X}$ following the distribution $\mathcal{P}$. This amounts to identifying an optimal hypothesis in $\mathcal{H}$
minimizing the expected risk
\begin{equation}\label{chapt4:eq:expct_loss}
  R(h) = \mathbb{E}_{(\bm{x},{y})\sim \mathcal{P}} \ell(h_{\bm{\theta}^*}(\bm{x}), {y}),
\end{equation}
where $\ell(\cdot,\cdot)$ refers to the per-sample loss predefined by the learner. Unfortunately, the inaccessible distribution $\mathcal{P}$ forbids us to assess the expected risk directly. In practice, $\mathcal{A}$ alternatively learns an empirical hypothesis $h_{\hat{\bm{\theta}}} \in \mathcal{H}$, as the global minimizer of the (regularized) loss function
\begin{equation}\label{chapt5:eq:loss_qnn}
  \mathcal{L}(\bm{\theta}, \mathcal{D}) = \frac{1}{n} \sum_{i=1}^{n} \ell (h_{\bm{\theta}}(\bm{x}^{(i)}), {y}^{(i)}) + \mathcal{R}(\bm{\theta}),
\end{equation}
where $\mathcal{R}(\bm{\theta})$ refers to an optional regularizer, as will be detailed in the following. Moreover, the first term on the right-hand side refers to the empirical risk
\begin{equation}
  R_{\ERM}(h_{\hat{\bm{\theta}}})= \frac{1}{n} \sum_{i=1}^{n} \ell (h_{\hat{\bm{\theta}}}(\bm{x}^{(i)}), {y}^{(i)}),
\end{equation}
which is also known as the training error.  To address the intractability of $R(h_{\hat{\bm{\theta}}})$, one can decompose it into two measurable terms,
\begin{equation}\label{chapt4:eq:expct_risk}
  R(h_{\hat{\bm{\theta}}}) = R_{\ERM}(h_{\hat{\bm{\theta}}}) + R_{\Gene}(h_{\hat{\bm{\theta}}}),
\end{equation}
where $R_{\Gene}(h_{\hat{\bm{\theta}}}) = R(h_{\hat{\bm{\theta}}}) - R_{\ERM}(h_{\hat{\bm{\theta}}})$ refers to the generalization error. In this regard, achieving a small prediction error requires the learning model to achieve both a small training error and a small generalization error.
 
\subsubsection{An overview}
Before moving to analyze the training error (ERM) and generalization error of QNNs rigorously, we first delve into better understanding the meaning of expressivity and generalization ability of the learning models with the ERM framework and try to give an intuition about the necessitates and benefits of exploring such theoretical aspects of QNNs as a special learning model. 

In particular, the expressivity could be directly understood as the size of the hypothesis space $\mathcal{H}=\{h_{\bm{\theta}}: \bm{\theta} \in \Theta\}$ related to the learning model. Intuitively, the achievable smallest empirical risk is determined by the expressivity of learning models. Specifically, a learning model with low expressivity may not fit the training data with complex patterns, e.g., the hypothesis space of linear model $\mathcal{H}=\{h_{\bm{\theta}}=\bm{\theta} \cdot \bm{x}\}$ cannot fit the nonlinear data $\{\bm{x}^{(i)},(\bm{x}^{(i)})^2\}$ perfectly. 

In general, the cardinality of the hypothesis space is infinity, as the parameters $\bm{\theta}$ are continuous. This makes it hard to compare the expressivity of different learning models. An alternative measure is model complexity, which measures the richness of the hypothesis space through the structural factors of the specific learning models, such as the number of parameters, depth, or architectural design. Remarkably, model complexity is measurable and bounded. In this tutorial, we will employ the covering number to measure the model complexity of QNNs later.

The generalization capability of learning models is directly measured by the generalization error $R_{\Gene}$ in Eqn.~(\ref{chapt4:eq:expct_risk}). A good generalization ability means that the learning model predicts well on unseen data as well as on the training data. In this regard, a small generalization error with a small training error implicates a small prediction error, as a small generalization error guarantees that the prediction performance is as well as the training performance. 

\begin{figure}[h!]
  \centering
  \includegraphics[width=0.66\textwidth]{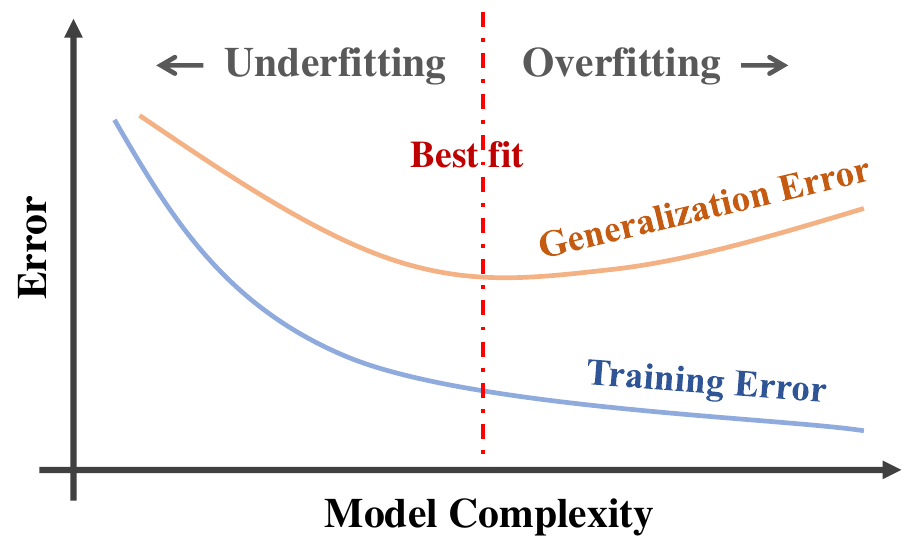}
  \caption{{\textbf{Influence of model complexity on generalization error.}} }
  \label{chap5:fig:u-curve} 
\end{figure}

In statistical learning theory, it has been well-established that a bias-variance trade-off governs the interplay between model complexity and generalization performance for any learning model, highlighting the delicate balance required for a model to generalize well to unseen data. The relationship is often depicted by a U-shaped curve, as shown in Figure~\ref{chap5:fig:u-curve}. This curve suggests that there exists an optimal level of model complexity for improving the generalization ability of any learning model. When under the point related to optimal expressivity, increasing model complexity improves performance on training data and enhances generalization. However, beyond a certain point, higher complexity leads to overfitting, resulting in poor generalization on test data. For QNNs, identifying this optimal level of complexity is crucial for achieving the best balance between training performance and generalization. 

\subsubsection{Expressivity of QNNs}
In this chapter, we analyze the generalization error of QNNs through a specific measure of model complexity: the covering number. By leveraging this measure, we aim to understand better and characterize the generalization performance of QNNs.

To elucidate the specific definition of the covering number, we first review the general structures of QNNs. Define $\rho\in \mathbb{C}^{2^N \times 2^N}$ as the $N$-qubit input quantum states, $O\in  \mathbb{C}^{2^N \times 2^N}$ as the quantum observable, $U(\bm{\theta})=\prod_{l=1}^{N_g} u_{l}(\bm{\theta}) \in \mathcal{U}(2^N)$ as the applied ansatz, where $\bm{\theta}\in \Theta$ are the trainable parameters living in the parameter space $\Theta$, $u_{l}(\bm{\theta}) \in \mathcal{U}(2^k)$ refers to the $l$-th quantum gate operated with at most $k$-qubits with $k\le N$, and $\mathcal{U}(2^N)$ stands for the unitary group in dimension $2^N$. In general, $U(\bm{\theta})$ is formed by $N_{gt}$ trainable gates and $N_g-N_{gt}$ fixed gates, e.g., $\Theta \subset [0, 2\pi)^{N_{gt}}$. Under the above definitions, the explicit form of the output of QNN under the ideal scenarios is 
\begin{equation}
  h(\bm{\theta},O,\rho):= \Tr\left( U(\bm{\theta})^{\dagger} O U(\bm{\theta}) \rho \right).
\end{equation}
Given the training data set $\mathcal{D}=\{(\rho^{(i)}, {y}^{(i)})\}_{i=1}^{n}$ and loss function $\mathcal{L}(\bm{\theta},\mathcal{D})$ defined in Eqn.~\eqref{chapt5:eq:loss_qnn}, QNN is optimized to find a good approximation $h^*(\bm{\theta},O,\rho)= \arg \min_{h(\bm{\theta},O,\rho) \in \mathcal{H}} \mathcal{L}(\bm{\theta},\mathcal{D})$ that can well approximate the target concept, where $\mathcal{H}$ refers to the hypothesis space of QNNs with
\begin{equation}\label{chap5:eqn:hypo-class}
  \mathcal{H} = \left\{ \Tr\left( U(\bm{\theta})^{\dagger} O U(\bm{\theta}) \rho \right)\Big| \bm{\theta}\in \Theta \right\}.
\end{equation}

An intuition about how the hypothesis space $\mathcal{H}$ affects the performance of QNNs is depicted in Figure~\ref{chap5:fig:qnn_learnability}.  When $\mathcal{H}$ has a modest size and covers the target concepts, the
estimated hypothesis could well approximate the target concept. By contrast, when the complexity of $\mathcal{H}$ is too low, there exists a large gap between the estimated hypothesis and the target concept. An effective measure to evaluate the complexity of $\mathcal{H}$ is covering number, an advanced tool broadly used in statistical learning theory, to bound the complexity of $\mathcal{H}$ and measure the expressivity of QNNs.
\begin{figure}[h!]
  \centering
  \includegraphics[width=0.48\textwidth]{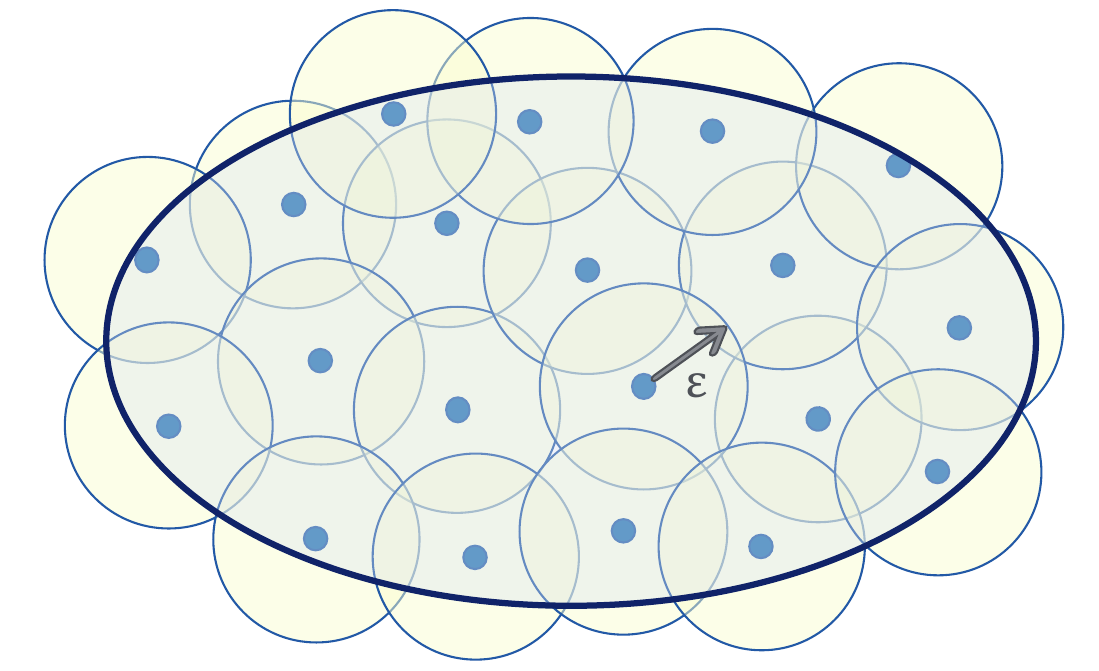}
  \caption{{\textbf{The geometric intuition of covering number.} Covering number concerns the minimum number of spherical balls with radius $\epsilon$ that occupy the whole space.} }
  \label{chap5:fig:cvr_num} 
\end{figure}

\begin{definition}
  [Covering number]
  \label{chap5:def:cvr_num}
  The covering number $\mathcal{N}(\mathcal{U}, \epsilon, \|\cdot\|)$ denotes the least cardinality of any subset $V \subset U$ that covers $U$ at scale $\epsilon$ with a norm $\|\cdot\|$, i.e., $\sup_{A \in \mathcal{U}}\min_{B\in \mathcal{V}}\|A-B\| \le \epsilon$. Here we use this notion to measure the expressivity of QNNs.
\end{definition}

The geometric interpretation of the covering number is depicted in Figure~\ref{chap5:fig:cvr_num}, which refers to the minimum number of spherical balls with radius $\epsilon$ that are required to completely cover a given space with possible overlaps. This notion has been employed to study other crucial topics in quantum physics such as Hamiltonian simulation and entangled states. Note that $\epsilon$ is a predefined hyper-parameter, i.e., a small constant with $\epsilon \in (0,1)$, and is independent of any factor. This convention has been broadly adopted in the regime of machine learning to evaluate the model capacity of various learning models.

Following the convention of \citet{du2022efficient}, we now give a step-by-step analysis of the model complexity of the hypothesis space $\mathcal{H}$ of QNNs defined in Eqn.~\eqref{chap5:eqn:hypo-class}. We will show that the covering number of QNNs is controlled by their structural factors, including the number of parameterized gates $N_{gt}$, the number of qubits $k$ the gates acting on, and the type of the quantum observable $O$. In the end, we first look at a simpler hypothesis space consisting of the 
operator group
\begin{equation}\label{chapt5:eq:H_circ}
  \mathcal{H}_{\Circ}:= \left\{ U(\bm{\theta})^{\dagger}OU(\bm{\theta}) \big| \bm{\theta} \in \Theta \right\},
\end{equation}
where we remove the factor of input states $\rho$ compared to the hypothesis space $\mathcal{H}$ related to QNNs. Actually, the covering number of $\mathcal{H}$ under the metric $d$ could be connected to the covering number of $\mathcal{H}_{\Circ}$ under the related metric $d_{\Circ}$ through employing their Lipschitz properties, which is encapsulated in the following Fact.

\begin{fact}\label{chapt5:lem:lip_cvr}
  Let $(\mathcal{H}_1,d_1)$ and $(\mathcal{H}_2,d_2)$ be two metric spaces satisfying $f: \mathcal{H}_1 \to \mathcal{H}_2$ be bi-Lipschitz such that
  \begin{equation}\label{chapt5:eq:bi_lip}
    c_l d_1(\bm{x},\bm{z}) \le d_2(f(\bm{x}),f(\bm{z})) \le c_r d_1(\bm{x},\bm{z}), ~\forall \bm{x},\bm{z} \in \mathcal{H}_1 .
  \end{equation}
  Then their covering number obey
  \begin{equation}\label{chapt5:eq:bound_lip_cvr}
    \mathcal{N}(\mathcal{H}_1,2\epsilon/c_l,d_1) \le  \mathcal{N}(\mathcal{H}_2,\epsilon,d_2) \le   \mathcal{N}(\mathcal{H}_1,\epsilon/c_r,d_1),
  \end{equation}
  where the left inequality requires $\epsilon\le c_lc_u/2$ with $c_u$ being the upper bound of the distance between any two points in $\mathcal{H}_1$, namely, $d_1(\bm{x},\bm{z})\le c_u$ for $\bm{x},\bm{z} \in \mathcal{H}_1 $.
\end{fact}

Fact~\ref{chapt5:lem:lip_cvr} indicates that we can derive the covering number of the metric space $(\mathcal{H},d)$ by analyzing the covering number of the metric space $(\mathcal{H}_{\Circ},d_{\Circ})$ and the Lipschitz constants of the mapping between $\mathcal{H}$ and $\mathcal{H}_{\Circ}$. Intuitively, a quantum circuit consisting of a large number of multi-qubit parameterized gates leads to a complicated QNN with a large model complexity. These intuitions are formalized into Theorem~\ref{chap5:thm:cvr_bound}. Specifically, the result of the covering number of the metric space $(\mathcal{H}_{\Circ}, d_{\Circ})$ is encapsulated in Lemma~\ref{chapt5:lem:cvr_circuit}.

\begin{lemma}\label{chapt5:lem:cvr_circuit}
  Suppose that the employed $N$-qubit quantum circuit containing in total $N_g$ gates with $N_g >N$, each gate $u_i(\bm{\theta})$ acting on most $k$ qubits, and $N_{gt}\le N_g$ gates in $U(\bm{\theta})$ are trainable. The $\epsilon$-covering number for the operator group $\mathcal{H}_{\Circ}$ in Eqn.~\eqref{chapt5:eq:H_circ} with respect to the operator-norm distance obeys
  \begin{equation}
    \mathcal{N}(\mathcal{H}_{\Circ},\epsilon,\|\cdot\|) \le \left( \frac{7N_{gt}\|O\|}{\epsilon} \right)^{2^{2k}N_{gt}},
  \end{equation}
  where $\|O\|$ denotes the operator norm of $O$.
\end{lemma}

\begin{proof}[Proof of Lemma~\ref{chapt5:lem:cvr_circuit}]
  To measure the covering number the operator group of $\mathcal{H}_{\Circ}= \{ U(\bm{\theta})^{\dagger}OU(\bm{\theta}) \big| \bm{\theta} \in \Theta \}$, one could first consider a fixed $\epsilon$-covering $\mathcal{S}$ for the set $\mathcal{N}(U(2^k), \epsilon, \|\cdot\|)$ of all possible gates and define the set 
  \begin{equation}
    \tilde{\mathcal{S}}:=\left\{\prod_{i\in\{N_{gt}\}}{u}_i(\bm{\theta}_i)\prod_{j\in\{N_g-N_{gt}\}}{u}_j\Big| {u}_i(\bm{\theta}_i)\in \mathcal{S}\right\},
  \end{equation}
  where ${u}_i(\bm{\theta}_i)$ and ${u}_j$ specify the trainable and fixed quantum gates in the employed quantum circuit, respectively. Note that for any circuit ${U}(\bm{\theta}) = \prod_{i=1}^{N_g}{u}_i(\bm{\theta}_i)$, one can always find a ${U}_{\epsilon}(\bm{\theta})\in\tilde{\mathcal{S}}$ where each ${u}_i(\bm{\theta}_i)$ of trainable gates is replaced with the nearest element in \textit{the covering set $\mathcal{S}$}, and the discrepancy $\|{U}(\bm{\theta})^{\dagger}O{U}(\bm{\theta})-{U}_{\epsilon}(\bm{\theta})^{\dagger}O{U}_{\epsilon}(\bm{\theta})\|$ satisfies
  \begin{align}\label{eqn:append-proof-lemma2-1}
    & \|{U}(\bm{\theta})^{\dagger}O{U}(\bm{\theta})-{U}_{\epsilon}(\bm{\theta})^{\dagger}O{U}_{\epsilon}(\bm{\theta})\| \nonumber\\
    \leq &  \|{U} -{U}_{\epsilon}\|   \|O\| \nonumber\\
    \leq &    N_{gt} \|O\|  \epsilon,
  \end{align} 
   where the first inequality uses the triangle inequality, and the second inequality follows from $ \|{U} -{U}_{\epsilon}\| \leq N_{gt}\epsilon$.

   Therefore, by Definition~\ref{chap5:def:cvr_num},  $\tilde{\mathcal{S}}$ forms an $N_{gt} \|O\| \epsilon$-covering set for $\mathcal{H}_{\Circ}$. An upper bound for the group $\mathcal{S}$, as established by \citet[Lemma 1]{barthel2018fundamental}, gives $|\mathcal{S}|\leq\left(\frac{7}{\epsilon} \right)^{2^{2k}}$. Since there are $|\mathcal{S}|^{N_{gt}}$ combinations for the gates in $\tilde{\mathcal{S}}$, it follows that $|\tilde{\mathcal{S}}|\leq\left(\frac{7}{\epsilon} \right)^{2^{2k}N_{gt}}$  and the covering number for $\mathcal{H}_{\Circ}$ satisfies 
   \begin{equation}
      \mathcal{N}(\mathcal{H}_{\Circ},  N_{gt}  \|O\|  \epsilon, \|\cdot\|) \leq  \left(\frac{7 }{\epsilon} \right)^{2^{2k}N_{gt}}.   
   \end{equation}
   An equivalent representation of the above inequality is
   \begin{equation}
     \mathcal{N}(\mathcal{H}_{\Circ},   \epsilon, \|\cdot\|) \leq  \left(\frac{7 N_{gt}  \|O\| }{\epsilon} \right)^{2^{2k}N_{gt}}.  
   \end{equation}
   
\end{proof}

With the established covering number of operator group $\mathcal{H}_{\Circ}$, one could directly analyze the covering number of the hypothesis space $\mathcal{H}$ related to QNNs, which is encapsulated in the following theorem.

\begin{theorem}\label{chap5:thm:cvr_bound}
  For $0<\epsilon < 1/10$, the covering number of   the hypothesis space $\mathcal{H}$ in Eqn.~\eqref{chap5:eqn:hypo-class} yields
  \begin{equation}
    \mathcal{N}(\mathcal{H}, \epsilon, |\cdot|) \le \left( \frac{7N_{gt}\|O\|}{\epsilon}\right)^{2^{2k}N_{gt}},
  \end{equation} 
  where $\|O\| $ denotes the operator norm of $O$.
\end{theorem}
\begin{proof}[Proof of Theorem~\ref{chap5:thm:cvr_bound}]
  The intuition of the proof is as follows. Recall the definition of the hypothesis space $\mathcal{H}$ in Eqn.~\eqref{chap5:eqn:hypo-class} and Lemma \ref{chapt5:lem:lip_cvr}. When $\mathcal{H}_1$ refers to the hypothesis space $\mathcal{H}$ and $\mathcal{H}_2$ refers to the unitary group $\mathcal{U}(2^N)$, the upper bound of the covering number of $\mathcal{H}$, i.e., $\mathcal{N}(\mathcal{H}_1, d_1, \epsilon )$, can be derived by first quantifying $c_r$  Eqn.~\eqref{chapt5:eq:bi_lip}, and then interacting with $\mathcal{N}(\mathcal{H}_{\Circ}, \epsilon, \|\cdot\|)$ in Lemma \ref{chapt5:lem:cvr_circuit}. Based on the above observations, the following addresses the upper bound of the covering number $\mathcal{N}(\mathcal{H},  \epsilon, |\cdot|)$.
  
  The Lipschitz constant $c_r$ in Eqn.~\eqref{chapt5:eq:bi_lip} is derived as a prerequisite for establishing the upper bound of $\mathcal{N}(\mathcal{H},  \epsilon, |\cdot |)$. Define ${U}\in \mathcal{U}(2^N)$ as the employed quantum circuit composed of $N_g$ gates, i.e., ${U}=\prod_{i=1}^{N_g} {u}_l$.  Let ${U}_{\epsilon}$ be the quantum circuit where each of the $N_g$ gates is replaced by the nearest element in the covering set. The relation between the distance $d_2(\Tr({U}_{\epsilon}^{\dagger}O{U}_{\epsilon}\rho), \Tr({U}^{\dagger}O{U} \rho))$ and the distance $d_1({U}_{\epsilon}, {U})$ yields  
  \begin{align}
    & d_2(\Tr({U}_{\epsilon}^{\dagger}O{U}_{\epsilon}\rho), \Tr({U}^{\dagger}O{U} \rho)) \nonumber\\
  =   & |\Tr({U}_{\epsilon}^{\dagger}O{U}_{\epsilon}\rho) - \Tr({U}^{\dagger}O{U} \rho) | \nonumber\\
    = & \left|\Tr\left(({U}_{\epsilon}^{\dagger}O{U}_{\epsilon}  -  {U}^{\dagger}O{U} ) \rho \right) \right| \nonumber\\
    \leq & \left\|{U}_{\epsilon}^{\dagger}O{U}_{\epsilon}  -  {U}^{\dagger}O{U} \right\|\Tr(\rho) \nonumber\\
    = & d_1({U}_{\epsilon}^{\dagger}O{U}_{\epsilon}, {U}^{\dagger}O{U}), 
  \end{align}
  where the first equality comes from the explicit form of the hypothesis, the first inequality uses the Cauchy-Schwartz inequality, and the last inequality employs $\Tr(\rho)=1$ and  
  \begin{equation}\label{eqn:dist-cov}
    \left\|{U}_{\epsilon}^{\dagger}O{U}_{\epsilon}  -  {U}^{\dagger}O{U} \right\| = d_1({U}_{\epsilon}^{\dagger}O{U}_{\epsilon}, {U}^{\dagger}O{U}).
  \end{equation}
  
  The above equation indicates $c_r=1$. Combining the above result with   Lemma \ref{chapt5:lem:lip_cvr} (i.e., Eqn.~\eqref{chapt5:eq:bi_lip}) and Lemma \ref{chapt5:lem:cvr_circuit}, we obtain
  \begin{equation}
    \mathcal{N}(\mathcal{H},  \epsilon,  |\cdot |)  \leq   \mathcal{N}(\mathcal{H}_{\Circ}, \epsilon, \|\cdot\|) \leq  \left(\frac{7 N_{gt}  \|O\| }{\epsilon} \right)^{2^{2k}N_{gt}}.
  \end{equation} 
  
  This relation ensures 
  \begin{equation}
    \mathcal{N}(\mathcal{H},  \epsilon, |\cdot|) \leq \left(\frac{7 N_{gt}  \|O\| }{\epsilon} \right)^{2^{2k}N_{gt}}.
  \end{equation}

\end{proof}

Theorem~\ref{chap5:thm:cvr_bound} indicates that the most decisive factor, which controls the complexity of $\mathcal{H}$, is the employed quantum gates in $U(\bm{\theta})$. This claim is ensured by the fact that the term $2^{2^kN_{gt}}$ exponentially scales the complexity $\mathcal{N}(\mathcal{H}, \epsilon, |\cdot|)$. Meanwhile, the qubits count $N$ and the operator norm $\|O\|$ polynomially scale the complexity of $\mathcal{N}(\mathcal{H}, \epsilon, |\cdot|)$. These observations suggest a succinct and direct way to compare the expressivity of QNNs with different quantum circuits. Moreover, the dependence of the expressivity of QNNs on the type of quantum
gates (denoted by the term $k$) demonstrated that the expressivity of QNNs depends on the structure information of ansatz such as the location of different quantum gates and the types of the employed quantum gates. The expressivity measured by the covering number could provide practical guidance for designing the circuit structure of QNNs.

\subsubsection{Generalization error of QNNs}

As the relation between generalization error and covering number is well-established in statistical learning theory, we can directly obtain the generalization error bound with the above bounds of covering number following the same conventions.

\begin{theorem}\label{chapt4:them:gene-QNN}
  Assume that the loss function $\ell$ defined in Eqn.~\eqref{chapt4:eq:expct_loss} is $L_1$-Lipschitz and upper bounded by a constant $C$, the QNN-based learning algorithm outputs a hypothesis $h_{\hat{\bm{\theta}}}$ from the training dataset $\mathcal{S}$ of size $n$. Following the notations of risk $R_{\Gene}(h_{\hat{\bm{\theta}}})=R(h_{\hat{\bm{\theta}}})-R_{\ERM}(h_{\hat{\bm{\theta}}})$ defined in Eqn.~\eqref{chapt4:eq:expct_risk}, for $0<\epsilon<1/10$, with probability at least $1-\delta$ with $\delta \in (0,1)$, we have
  \begin{equation}\label{chapt4:eq:gene_qnn}
    R_{\Gene}(h_{\hat{\bm{\theta}}}) \le \mathcal{O}\left(\frac{8L+c+24L \sqrt{N_{gt}}\cdot 2^k}{\sqrt{n}} \right).
  \end{equation}
\end{theorem}
\begin{proof}
    [Proof sketch of Theorem~\ref{chapt4:them:gene-QNN}] 
    Recall that the bound of generalization error in terms of Rademacher complexity has been established by \citet{kakade2008complexity} as follows
    \begin{equation}\label{chapt4:eq:gene_radem_qnn}
        R_{\Gene}(h_{\hat{\bm{\theta}}}) \le 2L_1 \mathcal{R}(\mathcal{H}_{QNN}) + 3C \sqrt{\frac{\ln(2/\delta)}{2n}},
    \end{equation}
    where $\mathcal{R}(\mathcal{H}_{QNN})$ represents the empirical Rademacher complexity of the hypothesis space of QNNs.
    Furthermore,  the relationship between Rademacher complexity and covering number can be derived using the Dudley entropy integral bound \citet{dudley1967sizes}, which is given by 
    \begin{equation}\label{chapt4:eq:radem_covering}
        \mathcal{R}(\mathcal{H}) \le \inf_{\alpha>0} \left( 4\alpha + \frac{12}{\sqrt{n}} \int_{\alpha}^1 \sqrt{\ln \mathcal{N}(\mathcal{H}_{|\mathcal{S}}, \epsilon,\|\cdot\|_2) }\mathrm{d} \epsilon \right),
    \end{equation}
    where $\mathcal{H}_{|\mathcal{S}}$ denotes the set of vectors formed by the hypothesis with $n$ examples in the dataset $\mathcal{S}$. In this regard, the generalization error bound in Eqn.~\eqref{chapt4:eq:gene_qnn} could be obtained by combining the Eqn.~\eqref{chapt4:eq:gene_radem_qnn} and Eqn.~\eqref{chapt4:eq:radem_covering} with direct but tedious calculations, which is omitted here. For details of the calculations, please refer to the proof of Theorem~2 in \citet{du2022efficient}.
\end{proof}

The assumption used in this analysis is quite mild, as the loss functions in QNNs are generally Lipschitz continuous and can be bounded above by a constant $C$.
This property has been broadly employed to understand
the capability of QNNs. The results obtained have three key implications. First, the generalization bound exhibits an exponential dependence on the term $k$ and a sublinear dependence on the number of trainable quantum gates $N_{gt}$. This observation reflects the quantum version of Occam's razor \citep{haussler1987occam}, where the parsimony of the output hypothesis implies greater predictive power. 
Second, increasing the number of training examples $n$ improves the generalization bound. This suggests that incorporating more training data is essential for optimizing complex quantum circuits. Lastly, the sublinear dependence on $N_{gt}$ may limit the ability to accurately assess the generalization performance of overparameterized QNNs \citep{larocca2023theory}. Together, these implications provide valuable insights for designing more powerful QNNs.

\subsection{Trainability of quantum neural networks}
\label{chapt5:sec:trainability}

The parameters in QNNs are often trained using gradient-based optimizers. As such, the magnitude of the gradient plays a crucial role in the trainability of QNNs. Specifically, large gradients are desirable, as they allow the loss function to decrease rapidly and consistently. However, this favorable property does not hold across a wide range of problem settings. In contrast, training QNNs usually encounters the barren plateau (BP) problem~\citep{mcclean2018barren}, i.e., \textit{the variance of the gradient, on average, decreases exponentially as the number of qubits increases}. In this section, we first introduce an example demonstrating how quantum circuits that form unitary $2$-designs~\citep{dankert2009exact} lead to BP, and then discuss several techniques to avoid or mitigate this issue.

We begin by introducing some basic notations. For convenience, let $L$ denote the number of parameters in the QNN $V(\bm{\theta})$. Consider the loss function defined as the measurement outcome of an $N$-qubit quantum state $\rho$ after applying the QNN operation, i.e.,
\begin{equation}\label{ch5_eq_bp_loss} f(\bm{\theta}) = \Tr \left[ O V(\bm{\theta}) \rho V(\bm{\theta})^\dag \right], \end{equation}
Then, the mathematical formulation of the BP phenomenon is given by
\begin{equation} \mathbb{E}_{\mathcal{P}} \left[ \frac{\partial f}{\partial \theta_k} \right] = 0, \quad {\rm Var}_{\mathcal{P}} \left[ \frac{\partial f}{\partial \theta_k} \right] = \exp(-\alpha N) \cdot \beta, 
\end{equation}
where $\mathcal{P}$ represents the probability distribution of the quantum circuit, and $\alpha,\beta>0$ are constants. In the case where the circuit $V(\bm{\theta})$ has a random structure with a polynomial number of single-qubit rotations and CNOT or CZ gates in $N$, a uniform distribution over the parameter space can approximate a 2-design for the unitary $V(\bm{\theta})$~\citep{harrow2009random,haferkamp2022random}. Moreover, a unitary sampled from an exact 2-design exhibits the following statistical properties.

\begin{fact}[ \citet{cerezo2021cost}]\label{vqacr_qntk_tdesign_tr_wawb}
Let $\{W_y\}_{y \in Y} \subset \mathcal{U}(d)$ form a unitary $2$-design, and let $A,B,C,D: \mathcal{H}_w \rightarrow \mathcal{H}_w$ be arbitrary linear operator. Then
\begin{align}
{}& \frac{1}{|Y|} \sum_{y \in Y} \Tr [W_y A W_y^\dag B] =\frac{\Tr[A]\Tr[B]}{d}, \label{vqacr_qntk_tdesign_tr_wawb_eq} \\
{}& \frac{1}{|Y|} \sum_{y \in Y} \Tr [W_y A W_y^\dag B ] \Tr [ W_y C W_y^\dag D] \notag \\
={}& \frac{1}{d^2-1} \left( \Tr[AC] \Tr[BD] + \Tr[A]\Tr[B]\Tr[C]\Tr[D] \right) \notag \\
-{}& \frac{1}{d(d^2-1)} \left( \Tr[A]\Tr[C] \Tr[BD] + \Tr[AC]\Tr[B] \Tr[D] \right) , \label{vqacr_qntk_tdesign_trwawbtrwcwd_eq} \\
{}& \frac{1}{|Y|} \sum_{y \in Y} \Tr [W_y A W_y^\dag B W_y C W_y^\dag D] \notag \\
={}& \frac{1}{d^2-1} \left( \Tr[A]\Tr[C] \Tr[BD] + \Tr[AC]\Tr[B] \Tr[D] \right) \notag \\
-{}& \frac{1}{d(d^2-1)} \left( \Tr[AC] \Tr[BD] + \Tr[A]\Tr[B]\Tr[C]\Tr[D] \right) . \label{vqacr_qntk_tdesign_tr_wawbwcwd_eq}
\end{align}
\end{fact}

Fact~\ref{vqacr_qntk_tdesign_tr_wawb} can be derived from Facts~\ref{app_haar_design_def_ave_1_design} and \ref{app_haar_design_def_ave_2_design} in Appendix~\ref{App:Haar_design}, which provides a more detailed discussion of unitary designs, potentially of independent interest. By applying Fact~\ref{vqacr_qntk_tdesign_tr_wawb}, it can be shown that QNNs with quantum circuits forming $2$-designs exhibit barren plateau loss landscapes.

\begin{theorem}[Adapted from \citet{mcclean2018barren}]\label{ch5_train_bp_original}
Consider the loss function given in Eqn.~(\ref{ch5_eq_bp_loss}), where the QNN $V(\bm{\theta})=\prod_{j=1}^{L} V_j (\bm{\theta}_j) W_j$ with fixed gate $W_j$ and variational gate $V_j(\bm{\theta}_j)=\exp(-i\bm{\theta}_j H_j /2)$. Suppose all hermitian matrices $\{H_j\}$ are traceless. For a integer $k \in [1,L]$, denote $U_- = \prod_{j=1}^{k-1} V_j (\bm{\theta}_j) W_j$ and $U_+ = \prod_{j=k+1}^{L} V_j (\bm{\theta}_j) W_j$. Then, if both $U_-$ and $U_+$ form $2$-designs, there is
\begin{align}
\mathbb{E} \left[ \frac{\partial f}{\partial \bm{\theta}_k} \right] = 0 , \quad {\rm Var} \left[ \frac{\partial f}{\partial \bm{\theta}_k} \right] \approx{}& \frac{1}{2^{3N+1}} \Tr \left[O^2 \right] \Tr \left[\rho^2 \right] \Tr \left[ H_j^2 \right]. \label{ch5_train_bp_original_eq}
\end{align}
\end{theorem}

\begin{proof}[Proof of Theorem~\ref{ch5_train_bp_original}]

By using notations $U_-$ and $U_+$, the function $f(\bm{\theta})$ in Eqn.~(\ref{ch5_eq_bp_loss}) can be expressed as:
\begin{align*}
f ={}& \Tr \left[ O V \rho V^\dag \right] \\
={}& \Tr \left[ O U_- V_k W_k U_+ \rho U_+^\dag W_k^\dag V_k^\dag U_-^\dag \right] \\
={}& \Tr \left[ U_-^\dag O U_- V_k W_k U_+ \rho U_+^\dag W_k^\dag V_k^\dag \right] \\
={}& \Tr \left[ O' \exp(-i \bm{\theta}_k H_k /2) \rho' \exp(i \bm{\theta}_k H_k /2) \right],
\end{align*}
where $O':=U_-^\dag O U_-$ and $\rho':=W_k U_+ \rho U_+^\dag W_k^\dag $. Thus, the gradient could be calculated as
\begin{align*}
\frac{\partial f}{\partial \bm{\theta}_k} ={}& \frac{i}{2} \Tr \left[ O' \left[ V_k \rho' V_k^\dag , H_k \right] \right] .
\end{align*}
The expectation of the gradient is zero since
\begin{align}
\mathop{\mathbb{E}}_{U_+, U_-} \frac{\partial f}{\partial \bm{\theta}_k} ={}& \mathbb{E} \frac{i}{2} \Tr \left[ O' \left[ V_k \rho' V_k^\dag , H_k \right] \right] \notag \\
={}& \mathop{\mathbb{E}}_{U_+, U_-} \frac{i}{2} \Tr \left[ U_-^\dag O U_- \left[ V_k \rho' V_k^\dag , H_k \right] \right] \notag \\
={}& \mathop{\mathbb{E}}_{U_+} \frac{i}{2^{N+1}} \Tr \left[ O \right] \Tr \left[ \left[ V_k \rho' V_k^\dag , H_k \right] \right] \label{ch5_train_bp_original_2_3} \\
={}& 0, \label{ch5_train_bp_original_2_4}
\end{align}
where Eqn.~(\ref{ch5_train_bp_original_2_3}) follows from Eqn.~(\ref{vqacr_qntk_tdesign_tr_wawb_eq}), and Eqn.~(\ref{ch5_train_bp_original_2_4}) is derived by noticing $\Tr [[A,B]]=\Tr[AB-BA]=0$. Therefore, the variance of the gradient equals to the expectation of its square, i.e., 
\begin{align}
\mathop{{\rm Var}}_{U_+, U_-} \left[ \frac{\partial f}{\partial \bm{\theta}_k} \right] ={}& \mathop{\mathbb{E}}_{U_+, U_-} \left[ \frac{\partial f}{\partial \bm{\theta}_k} \right]^2 \notag \\
={}& -\frac{1}{4} \mathbb{E}_{U_+, U_-} \Tr \left[ U_-^\dag O U_- \left[ V_k \rho' V_k^\dag , H_k \right] \right]^2 \notag \\
={}& -\frac{1}{4\times (2^{2N}-1)} \mathop{\mathbb{E}}_{U_+} \Tr \left[ O^2 \right] \Tr \left[ \left[ V_k \rho' V_k^\dag , H_k \right]^2 \right] \notag \\
+{}& \frac{1}{2^{N+2} \left( 2^{2N} -1 \right) } \mathop{\mathbb{E}}_{U_+} \Tr \left[ O \right]^2 \Tr \left[ \left[ V_k \rho' V_k^\dag , H_k \right]^2 \right] \label{ch5_train_bp_original_3_3} \\
={}& - \frac{ 1 }{4 \times \left( 2^{2N} -1 \right) } \Tr \left[ O^2 \right] \mathop{\mathbb{E}}_{U_+} \Tr \left[ \left[ V_k \rho' V_k^\dag , H_k \right]^2 \right] \label{ch5_train_bp_original_3_4},
\end{align}
where Eqn.~(\ref{ch5_train_bp_original_3_3}) follows from Eqn.~(\ref{vqacr_qntk_tdesign_trwawbtrwcwd_eq}), and Eqn.~(\ref{ch5_train_bp_original_3_4}) follows from ${\rm Tr}[O]=0$.
Further, it can be shown that
\begin{align}
{}& \mathop{\mathbb{E}}_{U_+} \Tr \left[ \left[ V_k \rho' V_k^\dag , H_k \right]^2 \right] \notag \\
={}& 2 \mathop{\mathbb{E}}_{U_+} \Tr \left[ \left( V_k \rho' V_k^\dag H_k \right)^2 \right] - 2 \mathop{\mathbb{E}}_{U_+} \Tr \left[ \left( V_k \rho' V_k^\dag \right)^2 \left( H_k \right)^2 \right] \notag \\
={}& 2 \mathop{\mathbb{E}}_{U_+} \Tr \left[ \left( V_k W_k U_+ \rho U_+^\dag W_k^\dag V_k^\dag H_k \right)^2 \right] \notag \\
-{}& 2 \mathop{\mathbb{E}}_{U_+} \Tr \left[ \left( V_k W_k U_+ \rho U_+^\dag W_k^\dag V_k^\dag \right)^2 \left( H_k \right)^2 \right] \notag \\
={}& \frac{2}{2^{2N}-1} \left\{ \Tr [\rho]^2 \Tr [H_k^2] + \Tr \left[\rho^2 \right] \Tr[H_k]^2 \right\} \notag \\
-{}& \frac{2}{2^N (2^{2N}-1)} \left\{ \Tr \left[ \rho^2 \right] \Tr \left[ H_k^2 \right] + \Tr[\rho]^2 \Tr[H_k]^2 \right\} \notag \\
-{}& \frac{2}{2^N} \Tr \left[ H_k^2 \right] \Tr \left[ \rho^2 \right] \label{ch5_train_bp_original_4_4} \\
\approx{}& - \frac{2^{N+1}}{2^{2N}-1} \Tr \left[ H_k^2 \right] \Tr \left[ \rho^2 \right] \label{ch5_train_bp_original_4_5},
\end{align}
where Eqn.~(\ref{ch5_train_bp_original_4_4}) follows from Eqn.~(\ref{vqacr_qntk_tdesign_tr_wawb_eq}) and Eqn.~(\ref{vqacr_qntk_tdesign_tr_wawbwcwd_eq}). Eqn.~(\ref{ch5_train_bp_original_4_5}) is derived by ignoring minor terms and using ${\rm Tr} [H_k]=0$. Combining Eqn.~(\ref{ch5_train_bp_original_3_4}) and Eqn.~(\ref{ch5_train_bp_original_4_5}), it can be shown that
\begin{align*}
\mathop{{\rm Var}}_{U_+, U_-} \left[ \frac{\partial f}{\partial \bm{\theta}_k} \right] \approx{}& \frac{2^N}{2 \times \left( 2^{2N} -1 \right)^2 } \Tr \left[ O^2 \right] \Tr \left[ H_k^2 \right] \Tr \left[ \rho^2 \right] \\
\approx{}& \frac{1}{2^{3N+1}} \Tr \left[ O^2 \right] \Tr \left[ \rho^2 \right] \Tr \left[ H_k^2 \right] .
\end{align*}
Thus, Theorem~\ref{ch5_train_bp_original} is proved. 

\end{proof}

\begin{tcolorbox}[enhanced, 
    breakable,colback=gray!5!white,colframe=gray!75!black,title=Remark]
The influence of barren plateau can be categorized into three folds.
\begin{enumerate} 
\item \textbf{Training efficiency}. Exponentially small gradients imply that the training of QNNs with gradient-based optimizers may require exponential numbers of iterations to converge. 
\item \textbf{Optimization effectiveness}. The calculation of the gradient via the parameter-shift rule would introduce unbiased statistical noise due to finite shot numbers. A small gradient could be vulnerable to these measurement noises, which may induce additional barriers in optimizations.
\item \textbf{Quantum advantage}. QNNs are expected to exhibit quantum advantages by employing intermediate to large numbers of qubits. However, the barren plateau phenomenon may induce exponential training steps, which can offset potential quantum advantages.
\end{enumerate}
\end{tcolorbox}

Since the barren plateau could seriously affect the trainability of scaled QNNs and raise concerns about the utility of QNNs for achieving quantum advantages, researchers have been focused on developing techniques to address this problem. Existing efforts including specific architecture design~\citep{cerezo2021cost,pesah2021absence,zhang2021toward}, parameter initialization schemes~\citep{grant2019initialization,zhang2022escaping}, and advanced training protocols~\citep{skolik2021layerwise,haug2021optimal}. Here we briefly introduce two related results with theoretical guarantees.

\begin{fact}[Shallow hardware-efficient circuits are BP-free, informal version adapted from \citet{cerezo2021cost}]\label{ch5_fact_bpfree_shallow}
Suppose the observable has a local form in the Pauli basis decomposition. For QNNs employing $N$-qubit shallow hardware-efficient circuits with logarithmic depths, the variance of the gradient has the lower bound
\begin{equation}
{\rm Var} \left[ \frac{\partial f}{\partial \bm{\theta}_k} \right] \geq \Omega \left( \frac{1}{{\rm poly} (N)} \right). 
\end{equation}
\end{fact}

\begin{fact}[Gaussian initializations help to escape the BP region, informal version adapted from \citet{zhang2022escaping}]\label{ch5_fact_bpfree_gaussian}
Suppose the observable is the tensor product of Pauli matrices $\sigma_{\bm{i}}=\sigma_{i_1}\otimes \cdots \otimes \sigma_{i_N}$, where the number of non-identity matrices in $\{\sigma_{i_1}, \cdots, \sigma_{i_N}\}$ is $S$. For QNNs employing $N$-qubit shallow hardware-efficient circuits with the depth $L$, the gradient norm has the lower bound
\begin{equation}\label{tqnn_cost_gaussian_main_eq}
\mathop{\mathbb{E}}\limits_{\bm{\theta}} \|\nabla_{\bm{\theta}} f\|^2 \geq \frac{L}{S^{S} (L+2)^{S+1}} {\rm Tr} \left[ \sigma_{\bm{j}} \rho_{\rm in} \right]^2,
\end{equation}
where $S$ is the number of non-zero elements in $\bm{i}$, and the index $\bm{j}=(j_1,j_2,\cdots,j_N)$ such that $j_m = 0, \forall i_m = 0$ and $j_m = 3, \forall i_m \neq 0$. The expectation is taken with the Gaussian distribution $\mathcal{N}\left(0, \frac{1}{4S(L+2)}\right)$ for the parameters $\bm{\theta}$.
\end{fact}

\section{Code Demonstration}
\label{chapt5:sec:qnn_code}

This section provides hands-on demonstrations of QNNs for both discriminative and generative tasks, where we will illustrate practical implementations of quantum classifiers and quantum patch GAN. Each subsection corresponds to a specific application, offering a step-by-step explanation and code walkthrough.

\subsection{Quantum classifier}

We now demonstrate how to utilize QNN to solve discriminative tasks, specifically a binary classification problem based on the Wine dataset. Below, we outline the major steps in the pipeline:

\begin{enumerate}
    \item [Step 1] Load and preprocess the dataset.
    \item [Step 2] Implement a quantum read-in protocol to encode classical data into quantum states.
    \item [Step 3] Construct a parameterized quantum circuit model to process the input quantum states.
    \item [Step 4] Train and test the QNN to evaluate its performance.
\end{enumerate}

We first import all the necessary libraries:

\begin{lstlisting}[language=Python]
import sklearn
import sklearn.datasets
import pennylane as qml
from pennylane import numpy as np
from pennylane.optimize import AdamOptimizer
import matplotlib.pyplot as plt    
\end{lstlisting}

\noindent\textbf{Step 1: Dataset preparation.} We prepare the Wine dataset for the classification task. For simplicity, we focus on the first two classes of the Wine dataset. The dataset consists of 13 attributes per sample, each with a distinct range. We apply normalization to rescale these attributes to the interval $[0, \pi]$. Furthermore, the labels are remapped from $\{0,1\}$ to $\{-1,1\}$ to align with the output range of the quantum circuit model. The dataset is split into training and test sets to fairly evaluate the classifier.

\begin{lstlisting}[language=Python]
def load_wine(split_ratio = 0.5):
    feat, label = sklearn.datasets.load_wine(return_X_y=True)

    # normalization
    feat = np.pi * (feat - np.min(feat, axis=0, keepdims=True)) / np.ptp(feat, axis=0, keepdims=True)

    index_c0 = label == 0
    index_c1 = label == 1

    label = label * 2 - 1

    n_c0 = sum(index_c0)
    n_c1 = sum(index_c1)

    X_train = np.concatenate((feat[index_c0][:int(split_ratio*n_c0)], feat[index_c1][:int(split_ratio*n_c1)]), axis=0)
    y_train = np.concatenate((label[index_c0][:int(split_ratio*n_c0)], label[index_c1][:int(split_ratio*n_c1)]), axis=0)
    X_test = np.concatenate((feat[index_c0][int(split_ratio*n_c0):], feat[index_c1][int(split_ratio*n_c1):]), axis=0)
    y_test = np.concatenate((label[index_c0][int(split_ratio*n_c0):], label[index_c1][int(split_ratio*n_c1):]), axis=0)

    return X_train, y_train, X_test, y_test
X_train, y_train, X_test, y_test = load_wine()   
\end{lstlisting}

To better understand the dataset, we apply t-SNE to visualize its distribution. As shown in Figure~\ref{fig:wine01}, each data point is projected into a 2D space for visualization, with distinct colors representing different classes.

\begin{lstlisting}[language=Python]
def visualize_dataset(X, y):
    from sklearn.manifold import TSNE

    tsne = TSNE(n_components=2, random_state=42, perplexity=30)

    wine_tsne = tsne.fit_transform(X)
    for label in np.unique(y):
        indices = y == label
        plt.scatter(wine_tsne[indices, 0], wine_tsne[indices, 1], edgecolor='black', cmap='coolwarm', s=20, label=f'Class {label}')

    # Add labels and legend
    plt.title("t-SNE Visualization of Wine dataset (two classes)")
    plt.xlabel("t-SNE Dimension 1")
    plt.ylabel("t-SNE Dimension 2")
    plt.legend()

    plt.tight_layout()
    plt.show()
visualize_dataset(X_train, y_train)
\end{lstlisting}

\begin{figure}[H]
\centering
\includegraphics[width=0.8\textwidth]{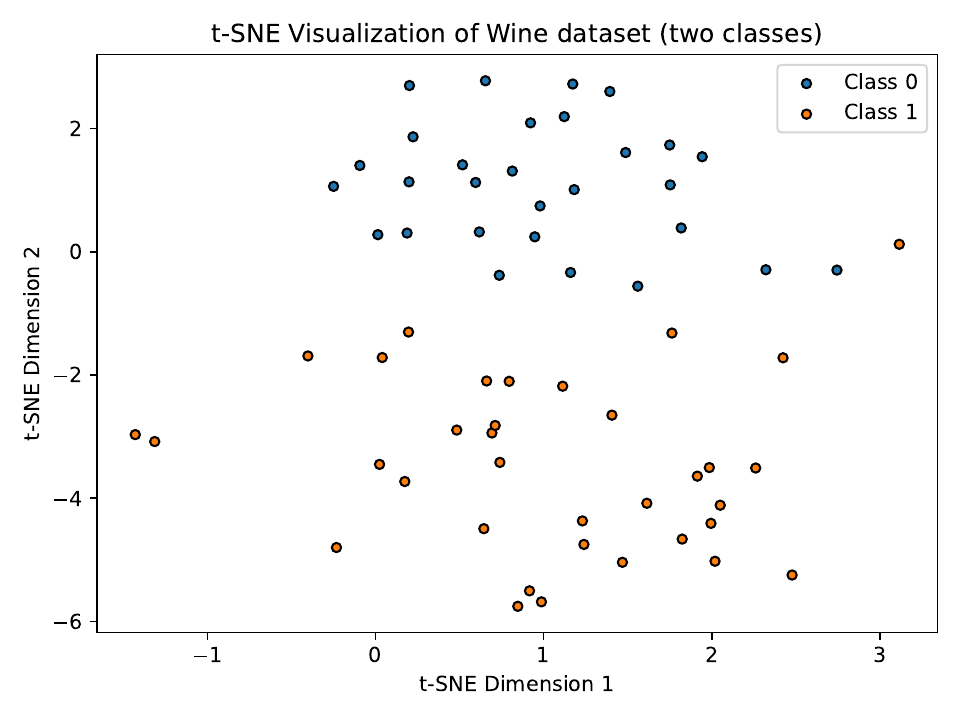}
\caption{{\textbf{T-SNE visualization of Wine dataset of the first two classes.} } }
\label{fig:wine01}  
\end{figure}

\noindent\textbf{Step 2: Data encoding.} To encode the $13$ attributes of the Wine dataset into a quantum system, we use angle encoding introduced in Section~\ref{cha2:sec3:angle_encode}, followed by a layer of CNOT gates acting on neighboring qubits to introduce entanglement.

\begin{lstlisting}[language=Python]
def data_encoding(x):
    n_qubit = len(x)
    qml.AngleEmbedding(features =x , wires = range(n_qubit) , rotation ="X")
    for i in range(n_qubit):
        if i+1 < n_qubit:
            qml.CNOT(wires=[i, i+1])
\end{lstlisting}

\noindent \textbf{Step 3: Building quantum classifier.} With the data encoding in place, we construct a quantum binary classifier. The circuit model is composed of multiple layers, where each layer includes parameterized single-qubit rotation gates with trainable angles, followed by a block of non-parametric CNOT gates to introduce entanglement among qubits. To read-out the category information of each input sample from the prepared quantum state, we compute the expectation value of the Pauli-Z operator on the first qubit.

\begin{lstlisting}[language=Python]
def classifier(param, x=None):
    data_encoding(x)
    n_layer, n_qubit = param.shape[0], param.shape[1]
    for i in range(n_layer):
        for j in range(n_qubit):
            qml.Rot(param[i, j, 0], param[i, j, 1], param[i, j, 2], wires=j)
        for j in range(n_qubit):
            if j+1 < n_qubit:
                qml.CNOT(wires=[j, j+1])
    return qml.expval(qml.PauliZ(0))

n_qubit = X_train.shape[1]
dev = qml.device('default.qubit', wires=n_qubit)
circuit = qml.QNode(classifier, dev)
\end{lstlisting}

We visualize the whole quantum circuit of 2 layers by drawing the diagram, as shown in Figure~\ref{fig:wine_classifier}.

\begin{lstlisting}[language=Python]
fig, ax = qml.draw_mpl(circuit)(np.pi * np.random.randn(2, n_qubit, 3), X_train[0])
fig.show()
\end{lstlisting}

\begin{figure}[H]
\centering
\includegraphics[width=0.8\textwidth]{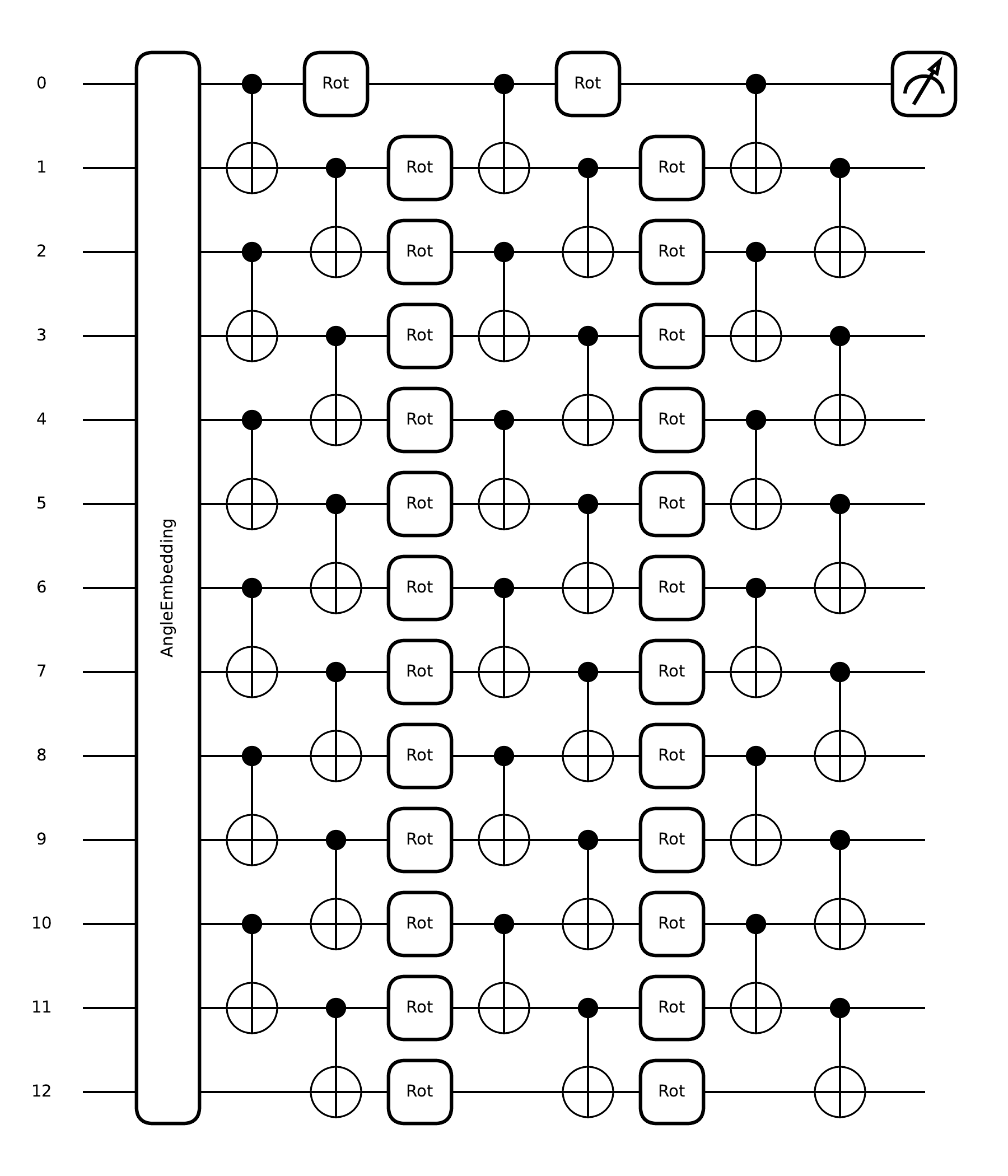}
\caption{{\textbf{The circuit diagram of the quantum classifier.} } }
\label{fig:wine_classifier} 
\end{figure}

\noindent\textbf{Step 4: Training and evaluation of quantum classifier.} With the data and circuit model ready, we now move to the optimization of the quantum classifier. The mean squared error (MSE) is used as the loss function. The goal is to minimize the difference between the predicted and actual labels over the training set.

\begin{lstlisting}[language=Python]
def mse_loss(predict, label):
    return np.mean((predict - label)**2)

def cost(param, circuit, X, y):
    exp = []
    for i in range(len(X)):
        pred = circuit(param, x=X[i])
        exp.append(pred)
    return mse_loss(np.array(exp), y)
\end{lstlisting}

To evaluate the performance of the quantum classifier, we use classification accuracy as the metric. Specifically, if the sign of the read-out result matches the corresponding label, the prediction is considered correct; otherwise, it is deemed incorrect. The accuracy is then calculated as the proportion of correctly classified samples out of the total.

\begin{lstlisting}[language=Python]
def accuracy(predicts, labels):
    assert len(predicts) == len(labels)
    return np.sum((np.sign(predicts)*labels+1)/2)/len(predicts)
\end{lstlisting}

The Adam optimizer is utilized to minimize the loss function. To ensure efficient computation, the dataset is divided into smaller batches for each training iteration. At the end of each epoch, both the training and test losses, along with the classification accuracy, are recorded to track the model's performance.

\begin{lstlisting}[language=Python]
lr = 0.01
opt = AdamOptimizer(lr)
batch_size = 4

n_epoch = 50
cost_train, cost_test, acc_train, acc_test = [], [], [], []
for i in range(n_epoch):
    index = np.random.permutation(len(X_train))
    feat_train, label_train = X_train[index], y_train[index]
    for j in range(0, len(X_train), batch_size):
        feat_train_batch = feat_train[j*batch_size:(j+1)*batch_size]
        label_train_batch = label_train[j*batch_size:(j+1)*batch_size]
        param = opt.step(lambda v: cost(v, circuit, feat_train_batch, label_train_batch), param)
    
    # compute cost
    cost_train.append(cost(param, circuit, X_train, y_train))
    cost_test.append(cost(param, circuit, X_test, y_test))

    # compute accuracy
    pred_train = []
    for j in range(len(X_train)):
        pred_train.append(circuit(param, x=X_train[j]))
    acc_train.append(accuracy(np.array(pred_train), y_train))

    pred_test = []
    for j in range(len(X_test)):
        pred_test.append(circuit(param, x=X_test[j]))
    acc_test.append(accuracy(np.array(pred_test), y_test))
\end{lstlisting}

After training the QNN, the training and test cost, as well as the accuracy over epochs, can be visualized using the following code.

\begin{lstlisting}[language=Python]
plt.figure(figsize=(12, 6))

plt.subplot(1, 2, 1)
epochs = np.arange(n_epoch) + 1
plt.plot(epochs, cost_train, label='Training Cost', marker='o')
plt.plot(epochs, cost_test, label='Test Cost', marker='o')
plt.title('Cost Over Epochs')
plt.xlabel('Epochs')
plt.ylabel('Cost')
plt.legend()
plt.grid()

# Plot training and test accuracy
plt.subplot(1, 2, 2)
plt.plot(epochs, acc_train, label='Training Accuracy', marker='o')
plt.plot(epochs, acc_test, label='Test Accuracy', marker='o')
plt.title('Accuracy Over Epochs')
plt.xlabel('Epochs')
plt.ylabel('Accuracy')
plt.legend()
plt.grid()

plt.tight_layout()
plt.show()
\end{lstlisting}

\begin{figure}[H]
\centering
\includegraphics[width=0.98\textwidth]{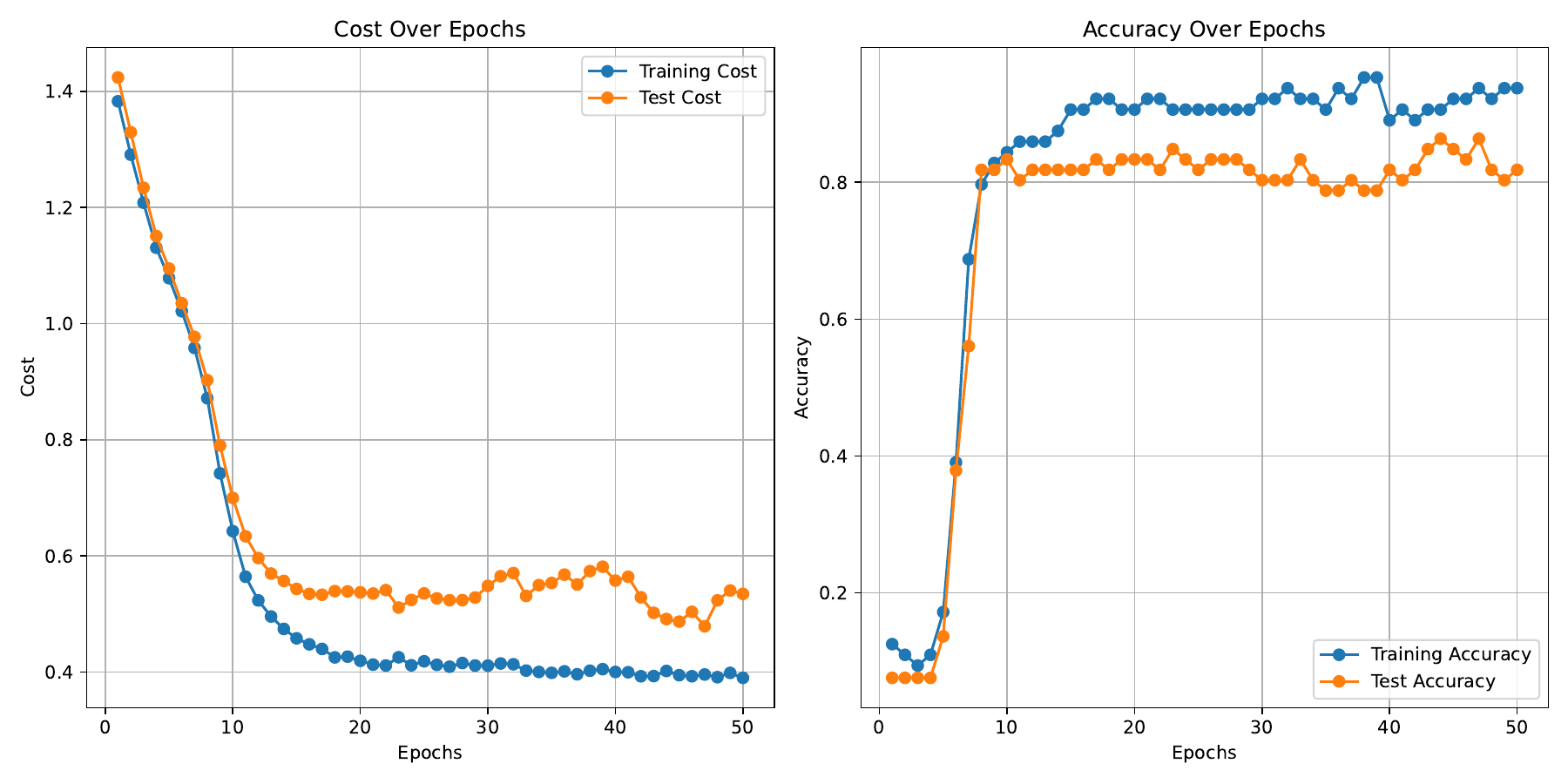}
\caption{{\textbf{The training curve of the quantum classifier.} } }
\label{fig:qnn_training_curve}
\end{figure}

As demonstrated in Figure~\ref{fig:qnn_training_curve}, this QNN achieves a test accuracy exceeding $0.8$. The performance of the QNN could potentially be further enhanced by employing more advanced read-in protocols, as discussed in Chapter~\ref{cha3:subsec:q-read-in}, which could enable more efficient and expressive representations of the input data. Additionally, optimizing the circuit design, such as adjusting the arrangement of layers or introducing more complex parameterized gates to increase the model's capacity, as highlighted in Chapter \ref{chapt5:sec:review}, could further improve the QNN's ability to capture intricate patterns in the dataset.

\subsection{Quantum patch GAN}

We next demonstrate how to implement a quantum patch GAN introduced in Chapter~\ref{chapt5:sec:qnn_gen} for the generation of hand-written digits of five. The whole pipeline includes the following steps:

\begin{enumerate}
    \item [Step 1] Load and pre-process the dataset.
    \item [Step 2] Build the classical discriminator.
    \item [Step 3] Build the quantum generator.
    \item [Step 4] Train the quantum patch GAN.
    \item [Step 5] Visualize the generated images.
\end{enumerate}

We begin by importing the required libraries:

\begin{lstlisting}[language=Python]
import torch
import torch.nn as nn
import torch.optim as optim
from torch.utils.data import Dataset, DataLoader

import numpy as np
import matplotlib.pyplot as plt
import pennylane as qml
import math
\end{lstlisting}

\noindent \textbf{Step 1: Dataset preparation.} We will use the Optical Recognition of Handwritten Digits dataset (\texttt{optdigits}), where each data point represents an $8\times 8$ grayscale image. The following code defines a custom dataset class to load and process the data.

\begin{lstlisting}[language=Python]
class OptdigitsData(Dataset):
    def __init__(self, data_path, label):
        """
        Dataset class for Optical Recognition of Handwritten Digits.
        """
        super().__init__()

        self.data = []
        with open(data_path, 'r') as f:
            for line in f.readlines():
                if int(line.strip().split(',')[-1]) == label:
                    # Normalize image pixel values from [0,16) to [0, 1)
                    image = [int(pixel)/16 for pixel in line.strip().split(',')[:-1]]
                    image = np.array(image, dtype=np.float32).reshape(8, 8)
                    self.data.append(image)
        self.label = label

    def __len__(self):
        return len(self.data)
    
    def __getitem__(self, index):
        return torch.from_numpy(self.data[index]), self.label
\end{lstlisting}

After defining the dataset, we can visualize a few examples to better understand the structure of the dataset.

\begin{lstlisting}[language=Python]
def visualize_dataset(data_path):
    """
    Visualizes the dataset by displaying examples for each digit label.
    """
    plt.figure(figsize=(10, 5))
    for i in range(10):
        plt.subplot(1, 10, i + 1)
        data = OptdigitsData(data_path, label=i)
        plt.imshow(data[0][0], cmap='gray')
        plt.title(f"Label: {i}")
        plt.axis('off')
    plt.tight_layout()
    plt.show()

visualize_dataset('code/chapter5_qnn/data/optdigits.tra')
\end{lstlisting}

\begin{figure}[H]
    \centering
    \includegraphics[width=0.98\textwidth]{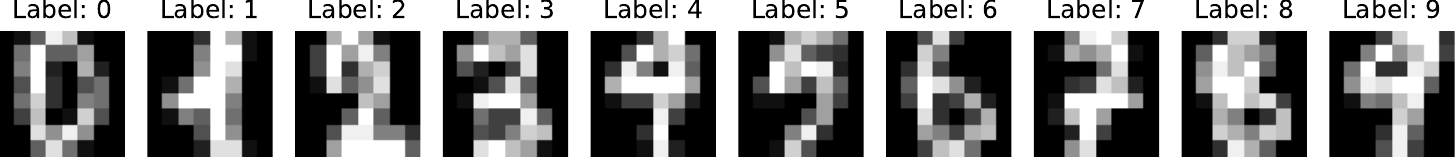}
    \caption{{\textbf{Samples in the dataset \texttt{optdigits}.} } }
    \label{fig:optdigit_sample}
\end{figure}

\noindent \textbf{Step 2: Building the classical discriminator.} The discriminator is a classical neural network responsible for distinguishing real images from fake ones. It consists of fully connected layers with ReLU activations. The final output is passed through a Sigmoid function, which scales the output to the range $(0,1)$, serving as a probabilistic indicator of whether the input image is real or fake.

\begin{lstlisting}[language=Python]
class ClassicalDiscriminator(nn.Module):
    """
    A classical discriminator for classifying real and fake images.
    """
    def __init__(self, input_shape):
        super().__init__()
        self.model = nn.Sequential(
            nn.Flatten(),
            nn.Linear(int(np.prod(input_shape)), 256),
            nn.ReLU(),
            nn.Dropout(),
            nn.Linear(256, 128),
            nn.ReLU(),
            nn.Dropout(),
            nn.Linear(128, 1),
            nn.Sigmoid()
        )

    def forward(self, img):
        return self.model(img)
\end{lstlisting}

\noindent \textbf{Step 3: Defining the quantum patch generator.}  The generator in the quantum patch GAN consists of parameterized quantum circuits (PQC). These circuits are responsible for generating patches of the target image. Specifically, the PQC follows the layout in Figure~\ref{chap5_figure_batchqgan}, which applies layers of single-qubit rotation gates and entangling gates to the latent state.

\begin{lstlisting}[language=Python]
def PQC(params):
    n_layer, n_qubit = params.shape[0], params.shape[1]
    for i in range(n_layer):
        for j in range(n_qubit):
            qml.Rot(params[i, j, 0], params[i, j, 1], params[i, j, 2], wires=j)
        # Control Z gates
        for j in range(n_qubit - 1):
            qml.CZ(wires=[j, j + 1])
\end{lstlisting}

Then, we implement the quantum generator for each patch of an image, i.e., sub-generators. The sub-generator transforms the latent variable $\bm{z}$ into a latent quantum state $\ket{\bm{z}}$, applies a PQC, performs partial measurements on the ancillary system $\mathcal{A}$, and outputs the probabilities of each computational basis state in the remaining system, which correspond to the generated pixel values.

\begin{lstlisting}[language=Python]
def QuantumGenerator(params, z=None, n_qubit_a=1):
    n_qubit = params.shape[1]

    # angle encoding of latent state z
    for i in range(n_qubit):
        qml.RY(z[i], wires=i)

    PQC(params)

    # partial measurement on the ancillary qubits
    qml.measure(wires=n_qubit-1)
    return qml.probs(wires=range(n_qubit-n_qubit_a))
\end{lstlisting}

Using the sub-generators, the quantum patch generator combines the output patches from multiple sub-generators to construct the complete image.

\begin{lstlisting}[language=Python]
class PatchQuantumGenerator(nn.Module):
    """
    Combines patches generated by quantum circuits into full images.
    """
    def __init__(self, qnode_generator, n_generator, n_qubit, n_qubit_a, n_layer):
        super().__init__()

        self.params_generator = nn.ParameterList([
            nn.Parameter(torch.rand((n_layer, n_qubit, 3)), requires_grad=True) for _ in range(n_generator)
        ])
        self.qnode_generator = qnode_generator
        self.n_qubit_a = n_qubit_a

    def forward(self, zs):
        images = []
        for z in zs:
            patches = []
            for params in self.params_generator:
                patch = self.qnode_generator(params, z=z, n_qubit_a=self.n_qubit_a).float()

                # post-processing: min-max scaling
                patch = (patch - patch.min()) / (patch.max() - patch.min() + 1e-8)

                patches.append(patch.unsqueeze(0))
            patches = torch.cat(patches, dim=0)
            images.append(patches.flatten().unsqueeze(0))
        return torch.cat(images, dim=0)
\end{lstlisting}

\noindent\textbf{Step 4: Training the quantum patch GAN.} With the dataset and models ready, we initialize the quantum generator, classical discriminator, and their optimizers.

\begin{lstlisting}[language=Python]
# Hyperparameters
torch.manual_seed(0)
image_width = 8
image_height = 8
n_generator = 4
n_qubit_d = int(np.log2((image_width * image_height) // n_generator))
n_qubit_a = 1
n_qubit = n_qubit_d + n_qubit_a
n_layer = 6

# Quantum device
dev = qml.device("lightning.qubit", wires=n_qubit)
qnode_generator = qml.QNode(QuantumGenerator, dev)

# Initialize generator and discriminator
discriminator = ClassicalDiscriminator([image_height, image_width])
discriminator.train()
generator = PatchQuantumGenerator(qnode_generator, n_generator, n_qubit, n_qubit_a, n_layer)
generator.train()

# Optimizers
lr_generator = 0.3
lr_discriminator = 1e-2
opt_discriminator = optim.SGD(discriminator.parameters(), lr=lr_discriminator)
opt_generator = optim.SGD(generator.parameters(), lr=lr_generator)

# Construct dataset and dataloader
batch_size = 4
dataset = OptdigitsData('code/chapter5_qnn/data/optdigits.tra', label=5)
dataloader = DataLoader(dataset, batch_size=batch_size, shuffle=True, drop_last=True)

# Loss function
loss_fn = nn.BCELoss()
labels_real = torch.ones(batch_size, dtype=torch.float)
labels_fake = torch.zeros(batch_size, dtype=torch.float)

# Testing setup
n_test = 10
z_test = torch.rand(n_test, n_qubit) * math.pi
\end{lstlisting}

The GAN training process involves alternating updates for the discriminator and the generator.
The discriminator is trained to distinguish between real and fake images, while the generator learns to create images that can successfully deceive the discriminator. 
To track the generator's progress during training, the generated images are saved at the end of each epoch.

\begin{lstlisting}[language=Python]
n_epoch = 10
record = {}
for i in range(n_epoch):
    for data, _ in dataloader:

        zs = torch.rand(batch_size, n_qubit) * math.pi
        image_fake = generator(zs)

        # Training the discriminator
        discriminator.zero_grad()
        pred_fake = discriminator(image_fake.detach())
        pred_real = discriminator(data)

        loss_discriminator = loss_fn(pred_fake.squeeze(), labels_fake) + loss_fn(pred_real.squeeze(), labels_real)
        loss_discriminator.backward()
        opt_discriminator.step()

        # Training the generator
        generator.zero_grad()
        pred_fake = discriminator(image_fake)
        loss_generator = loss_fn(pred_fake.squeeze(), labels_real)
        loss_generator.backward()
        opt_generator.step()

    print(f'The {i}-th epoch: discriminator loss={loss_discriminator: 0.3f}, generator loss={loss_generator: 0.3f}')

    # test
    generator.eval()
    image_generated = generator(z_test).view(n_test, image_height, image_width).detach()

    record[str(i)] = {
        'loss_discriminator': loss_discriminator.item(),
        'loss_generator': loss_generator.item(),
        'image_generated': image_generated.numpy().tolist()
    }
    generator.train()
\end{lstlisting}

\noindent \textbf{Step 5: Visualizing the generated images.}  After training, we visualize the images generated by the quantum generator to evaluate its performance. These visualizations allow us to see how well the model has learned to produce realistic image patches.

\begin{lstlisting}[language=Python]
n_epochs_to_visualize = len(record) // 2
n_images_per_epoch = 10

fig, axes = plt.subplots(n_epochs_to_visualize, n_images_per_epoch, figsize=(n_images_per_epoch, n_epochs_to_visualize))

# Iterate through the recorded epochs and visualize generated images
for epoch_idx, (epoch, data) in enumerate(record.items()):
    if epoch_idx % 2 == 1:
        continue
    images = np.array(data['image_generated'])
    
    for img_idx in range(n_images_per_epoch):
        ax = axes[epoch_idx // 2, img_idx]
        ax.imshow(images[img_idx], cmap='gray')
        ax.axis('off')
        
        # Add epoch information to the title of each row
        if img_idx == 0:
            ax.set_title(f"Epoch {epoch}", fontsize=10)

plt.tight_layout()
plt.show()
\end{lstlisting} 

\begin{figure}[H]
    \centering
    \includegraphics[width=0.98\textwidth]{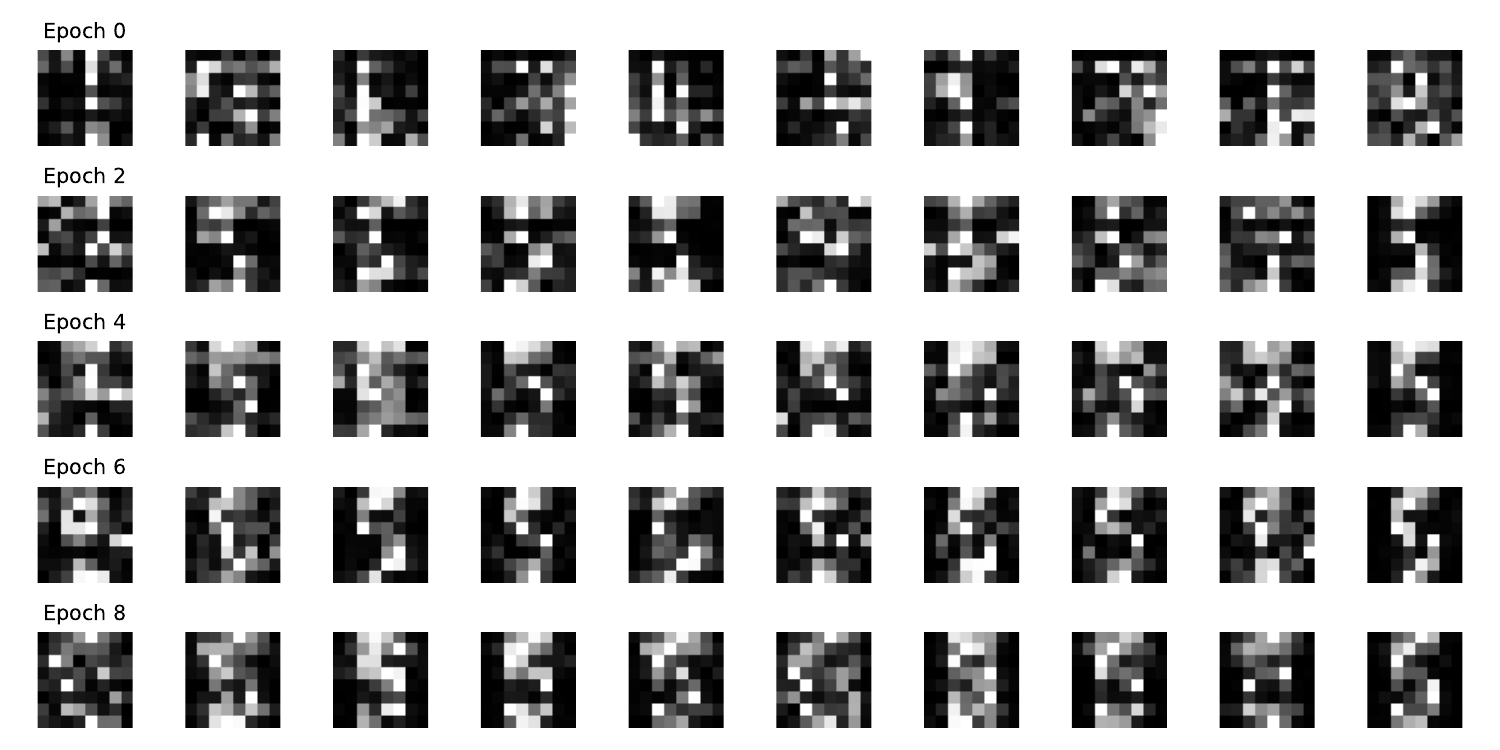}
    \caption{{\textbf{The images generated by quantum patch GAN during training.} } }
    \label{fig:qgan_training_curve}
\end{figure}

\section{Bibliographic Remarks}\label{chapt5:sec:review}

Quantum neural networks (QNNs) have emerged as a prominent paradigm in quantum machine learning, demonstrating potential in both discriminative and generative learning tasks. For discriminative tasks, QNNs utilize high-dimensional Hilbert spaces to efficiently capture and represent complex intrinsic relationships. In generative tasks, QNNs leverage variational quantum circuits to generate complex probability distributions that may exceed the capabilities of classical models. While these approaches share common learning strategies, they each introduce unique challenges and opportunities in terms of model design, theoretical exploration, and practical implementation. In the remainder of this section, we briefly review recent advances in QNNs. Interested readers can refer to \citet{ablayev2019quantum,li2022recent,massoli2022leap,tian2023recent,wang2024comprehensive} for a more comprehensive review.

\subsection{Discriminative learning with QNN}

QNNs for discriminative tasks have emerged as one of the most active research areas in quantum machine learning, demonstrating potential advantages in feature representation and processing efficiency. The quantum learning approach leverages the high-dimensional Hilbert space and quantum parallelism to potentially handle complex classification boundaries more effectively than classical neural networks. Research has shown particular promise in handling datasets with inherent quantum properties and problems where quantum entanglement can be meaningfully exploited~\citep{huang2022quantum}. 

\subsubsection{Model designs}
In the realm of quantum discriminative models, researchers have developed various quantum neural architectures. In general, variational quantum classifiers~\citep{havlivcek2019supervised,mitarai2018quantum} could employ parameterized quantum circuits for classification tasks. Subsequently, quantum convolutional neural networks~\citep{cong2019quantum} are designed for processing structured data. Hybrid quantum-classical architectures~\citep{arthur2022hybrid} are proposed to combine quantum layers with classical neural networks. Other notable works include the development of quantum versions of popular classical architectures like recurrent neural networks~\citep{bausch2020recurrent} and attention mechanisms~\citep{shi2024qsan}. Finally, \citet{perez2020data,fan2022compact} have explored quantum re-uploading strategies for encoding classical data, achieving QML models with more expressive feature maps. 

In addition to manually designed architectures, various lightening strategies have been explored to enhance the efficiency of quantum neural networks. For example, quantum architecture search methods have been developed by \citep{du2022quantum,zhang2022differentiable,linghu2024quantum} to automatically discover optimal quantum circuit designs with reduced gate complexity. \citet{sim2021adaptive,wang2022symmetric} introduced quantum pruning techniques that systematically identify and remove redundant quantum gates while preserving the performance. In the realm of knowledge distillation, researchers have demonstrated how to transfer knowledge from the teacher model given as quantum~\citep{alam2023knowledge} or classical~\citep{li2024hybrid} neural networks to more compact quantum circuit architectures that are more robust against quantum noises. These optimization approaches have collectively contributed to improving the practical performance of QNNs on real quantum devices, particularly in the NISQ era.

\subsubsection{Theoretical foundations}
To gain a deeper understanding of the potential advantages and limitations of QNNs, a crucial research topic is analyzing their learnability. More concisely, the learnability is determined by the interplay of three key aspects: expressivity, trainability, and generalization, as preliminarily introduced in Chapter~\ref{chapt5:sec:qnn_theory} with essential theoretical results. Beyond these foundational insights, an extensive body of research has conducted more comprehensive and detailed investigations into these three aspects, which will be reviewed individually in the following.

\smallskip
\noindent\textbf{Expressivity}. The expressivity of QNNs refers to their ability to represent complex functions or quantum states efficiently. Universal approximation theorems (UAT) incorporating data re-uploading strategies have been established by~\citet{perez2020data} firstly with subsequent works ~\citep{schuld2021effect,yu2022power} in various problem settings. Beyond the UAT, \citet{sim2019expressibility}, \citet{nakaji2021expressibility}, and \citet{holmes2022connecting} analyze the expressivity of QNNs by investigating how well the parameterized quantum circuits used in QNNs can approximate the Haar distribution, a critical measure of expressive capacity in quantum systems. Moreover, \citet{yu2024non} analyze the non-asymptotic error bounds of variational quantum circuits for approximating multivariate polynomials and smooth functions.

\smallskip
\noindent\textbf{Trainability}. The trainability of QNNs corresponds to two aspects during the optimization of QNNs: the \textit{gradient magnitude} and the \textit{convergence rate}. 

For the first research line, \citet{mcclean2018barren} first found the phenomenon of vanishing gradients, dubbed as the barren plateau, where the gradient magnitude scales exponentially small with the quantum system size. Since then, a series of studies explored the cause of barren plateau, including global measurements~\citep{cerezo2021cost}, highly entanglement~\citep{ortiz2021entanglement}, and quantum system noise \citep{wang2021noise}. 

Efforts to address the barren plateau problem, a major challenge in training deep quantum circuits, have yielded strategies such as proper circuit and parameter initialization techniques~\citep{grant2019initialization,zhang2022escaping}, cost function design~\citep{cerezo2021cost}, and proper circuits~\citep{pesah2021absence}. Quantum-specific regularization techniques have also been developed to mitigate these effects~\citep{larocca2022diagnosing}.   

Another research line on the trainability of QNNs focuses on the convergence of QNNs, which is not introduced in this Chapter. In particular, \citet{kiani2020learning} and \citet{wiersema2020exploring} experimentally found that overparameterized QNNs embrace a benign landscape and permit fast convergence towards near optimal local minima. Afterward, initial attempts have been made to theoretically explain the
superiority of over-parameterized QNNs. Specifically, \citet{larocca2023theory} and \citet{anschuetz2021critical} utilized tools from dynamical Lie algebra and random matrix theory, respectively, to quantify the critical points in the optimization landscape of overparameterized QNNs. Moreover, \citet{you2022convergence} extended the classical convergence results of \citet{xu2018convergence} to the quantum domain, proving that overparameterized QNNs achieve an exponential convergence rate. Additionally, \citet{liu2023analytic} and \citet{wang2023symmetric} introduced the concept of the quantum neural tangent kernel (QNTK) to further demonstrate the exponential convergence rate of overparameterized QNNs. Besides the overparameterization theory, \citep{du2021learnability,du2022distributed,qi2023theoretical,qian2024shuffle} investigated the required conditions to ensure the convergence of QNNs towards local minima.

\smallskip
\noindent\textbf{Generalization}. Research has also focused on understanding the sample complexity and generalization error bounds of quantum machine learning algorithms by using different statistical learning tools. In particular, \citet{abbas2021power}  compared the generalization power of QNNs and classical learning models based on an information geometry metric. \citet{caro2022generalization} and \citet{du2022efficient} established generalization error bounds using covering numbers, revealing the impact of circuit structural factors—such as the number of gates and types of observables—on generalization ability. Similarly, \citet{bu2022statistical} analyzed the generalization ability of QNNs from the perspective of quantum resource theory, emphasizing the role of quantum resources such as entanglement and magic in influencing generalization.

Furthermore, frameworks for demonstrating quantum advantage in specific learning scenarios have been proposed~\citep{huang2021information,huang2022quantum}, providing insights into the conditions under which quantum models outperform their classical counterparts. \citet{zhang2024curse} investigate the training-dependent generalization abilities of QNNs, while \citet{du2023problem} study problem-dependent generalization, highlighting key factors that enable QNNs to achieve strong generalization performance.

Beyond analyses focused on specific datasets and problems, the generalization ability of QNNs has also been examined through the lens of the No-Free-Lunch theorem. \citet{poland2020no} explore the average performance of QNNs across all possible datasets and problems, providing a broader perspective on their generalization potential. Extending this work, \citet{sharma2022reformulation} and \citet{wang2024transition} adapt the No-Free-Lunch theorem to scenarios involving entangled data, demonstrating the potential benefits of entanglement in certain settings. Additionally, \citet{wang2024separable} establish a No-Free-Lunch framework for various learning protocols, considering different quantum resources used in these protocols.

\noindent\textbf{Potential advantages}. A critical challenge in quantum learning theory is identifying tasks where QNNs demonstrate provable computational advantages over classical approaches when solving classical problems. While recent studies have explored unconditional quantum advantage, most focus on synthetic tasks. For instance, polynomial advantages over commonly used classical models have been demonstrated through quantum contextuality in sequential learning \citep{anschuetz2023interpretable,anschuetz2024arbitrary}. Additionally, quantum entanglement has been shown to reduce communication requirements in non-local machine-learning tasks \citep{zhao2024entanglement}. Despite these advancements, a significant issue remains: most QNNs, without careful design, can be efficiently simulated or approximated by classical surrogate models. The review \citep{cerezo2023does} presented strong evidence that commonly used models with provable absence of barren plateaus are also classically simulable. Concrete algorithms in this research line include Pauli path simulators \citep{bermejo2024quantum,angrisani2024classically,lerch2024efficient}, tensor-network simulators \citep{shin2024dequantizing}, and learning protocols \citep{landman2022classically,schreiber2023classical,du2024efficient}.

\subsubsection{Applications}

Practical applications of quantum neural networks for discriminative learning have spanned multiple domains. In computer vision, researchers have demonstrated quantum approaches to image classification~\citep{henderson2020quanvolutional} and pattern recognition~\citep{alrikabi2022face}. In quantum chemistry, QNNs have been applied to proton affinity predictions
~\citep{jin2024integrating} and catalyst developments~\citep{roh2024hybrid}. Financial applications include market trend classification~\citep{li2023quantum} and fraud detection~\citep{innan2024financial}. Medical applications encompass drug discovery~\citep{batra2021quantum} and disease diagnosis~\citep{enad2023review}.

\subsection{Generative learning with QNNs}

QNNs for generative tasks represent a promising avenue in quantum machine learning, offering new methodologies for generating complex data distributions. By leveraging variational quantum circuits, these models aim to generate potentially more complex probability distributions compared to classical counterparts, particularly in domains where high-dimensional data or quantum properties are prominent. We refer the readers to \citet{tian2023recent} for a survey of QNNs in generative learning.

\subsubsection{Model designs} Researchers have developed various QNN architectures tailored for generative tasks. Quantum circuit Born machines (QCBMs)~\citep{benedetti2019generative} are one of the pioneering models, utilizing parameterized quantum circuits to generate discrete probability distributions. Quantum generative adversarial networks~\citep{lloyd2018quantum} extend the adversarial framework to the quantum domain, where quantum generators and discriminators compete to learn complex distributions. Quantum Boltzmann machines~\citep{amin2018quantum} are another notable model, employing quantum devices to prepare Boltzmann distributions for estimating target discrete distributions. Additionally, quantum autoencoders~\citep{romero2017quantum} have been proposed for tasks like quantum state compression and reconstruction, offering potential advantages in quantum information processing. Recently, quantum diffusion models~\citep{zhang2024generative,kolle2024quantum} have been proposed for generating quantum state ensembles or classical images.  
These models showcase the versatility of QNNs in addressing diverse generative tasks.

\subsubsection{Theoretical foundations} The theoretical understanding of generative quantum neural networks has advanced in several directions. Similar to QNNs for discriminative tasks, quantum generative models like QCBMs face the barren plateau issue with additional mechanisms from the Kullbach-Leibler (KL) divergence loss function~\citep{rudolph2024trainability}. In parallel, QCBMs are more efficient in the data-limited regime than the other classical generative models~\citep{hibat2024framework}. Besides, \citet{gao2018quantum} proved the existence of quantum generative model that is more capable of representing probability distributions compared with classical generative models, which has exponential speedup in learning and inference. Similarly, \citet{gao2022enhancing} proved the separation in expressive power between a class of widely used generative models, known as Bayesian networks, and its minimal quantum-inspired extension. \citet{du2022power} analyzed the generalization bounds of QCBMs and QGANs under the maximum mean discrepancy loss.

\subsubsection{Applications} Generative QNNs have shown potential in various practical applications. In finance, they have been used to model complex financial data distributions and generate synthetic financial datasets, demonstrating better performance than classical models in certain scenarios~\citep{alcazar2020classical,zhu2022generative}. In the domain of quantum physics, quantum generative models have been applied for quantum state tomography and quantum simulation, aiding in the understanding of quantum systems~\citep{benedetti2019adversarial}. In image generation, QGANs have been employed to produce high-quality images, showcasing their capability in handling complex visual data~\citep{huang2021experimental}. Furthermore, quantum generative models have been explored for drug discovery, where they can potentially accelerate the process by efficiently exploring large chemical spaces~\citep{li2021quantum}. These applications highlight the broad potential of QNNs in generative tasks across different domains.

\chapter{Quantum Transformer}\label{Chapter5:Transformer}

Transformers, introduced by \citet{vaswani2017attention}, have become one of the most important and widely adopted deep learning architectures in modern AI. Transformers were first developed to improve previous architectures for natural language processing on the ability to handle long-range dependencies and capture intricate relationships in data. Unlike previous sequential models, such as recurrent neural networks, which process information in a step-by-step manner, transformers use a mechanism called \textit{self-attention} to capture correlations among all elements in a sequence simultaneously. This parallel processing capability significantly reduces training time and improves learning performance.

Despite its many advantages, the transformer architecture has several drawbacks, particularly the required computational resources. As discussed in previous chapters, quantum computing provides unique advantages over classical computing in certain applications by leveraging quantum phenomena such as superposition, entanglement, and interference. These capabilities have inspired researchers to explore whether integrating quantum computing with Transformers could lead to superior performance compared to their classical counterparts in specific tasks.

To figure out this question, in this chapter, we start by introducing the mechanism of Transformers in Chapter~\ref{chapt6:sec:classical_transformer}. We then illustrate how to construct a quantum Transformer on a fault-tolerant quantum computer in Chapter~\ref{chapt6:sec:quantum_transformer}. We also analyze the runtime of quantum transformers combined with numerical observations in Chapter~\ref{chapt6:sec:runtime_qtransformer},  demonstrating a quadratic speedup over the classical counterpart.

\section{Classical Transformer}\label{chapt6:sec:classical_transformer}

The transformer architecture is designed to predict the next \textit{token} (formally present in Chapter~\ref{chapter5:sec:token}) in a sequence by leveraging sophisticated neural network components. Its modular design—including residual connections, layer normalization, and feed-forward networks (FFNs) as introduced in Chapter~\ref{chapt5:sec:classical_nn}—makes it highly scalable and customizable. This versatility has enabled its successful application in large-scale foundation models across diverse domains, including natural language processing, computer vision, reinforcement learning, robotics, and beyond. 

The full architecture of Transformer is illustrated in Figure~\ref{fig:transformer_arch}. Note that while the original paper by \citet{vaswani2017attention} introduced both encoder and decoder components, contemporary large language models primarily adopt \textit{decoder-only architectures}, which have demonstrated superior practical performance. Therefore, in the remainder of this section, we focus on detailing the implementation of each building block and discussing the optimization of decoder-only Transformer architectures. To deepen the understanding, a toy example of a classical Transformer with the code demonstration is provided in Chapter~\ref{chapt5:section:code-implementation}. 

\begin{figure}[h!]
\centering
\includegraphics[width=0.9\textwidth]{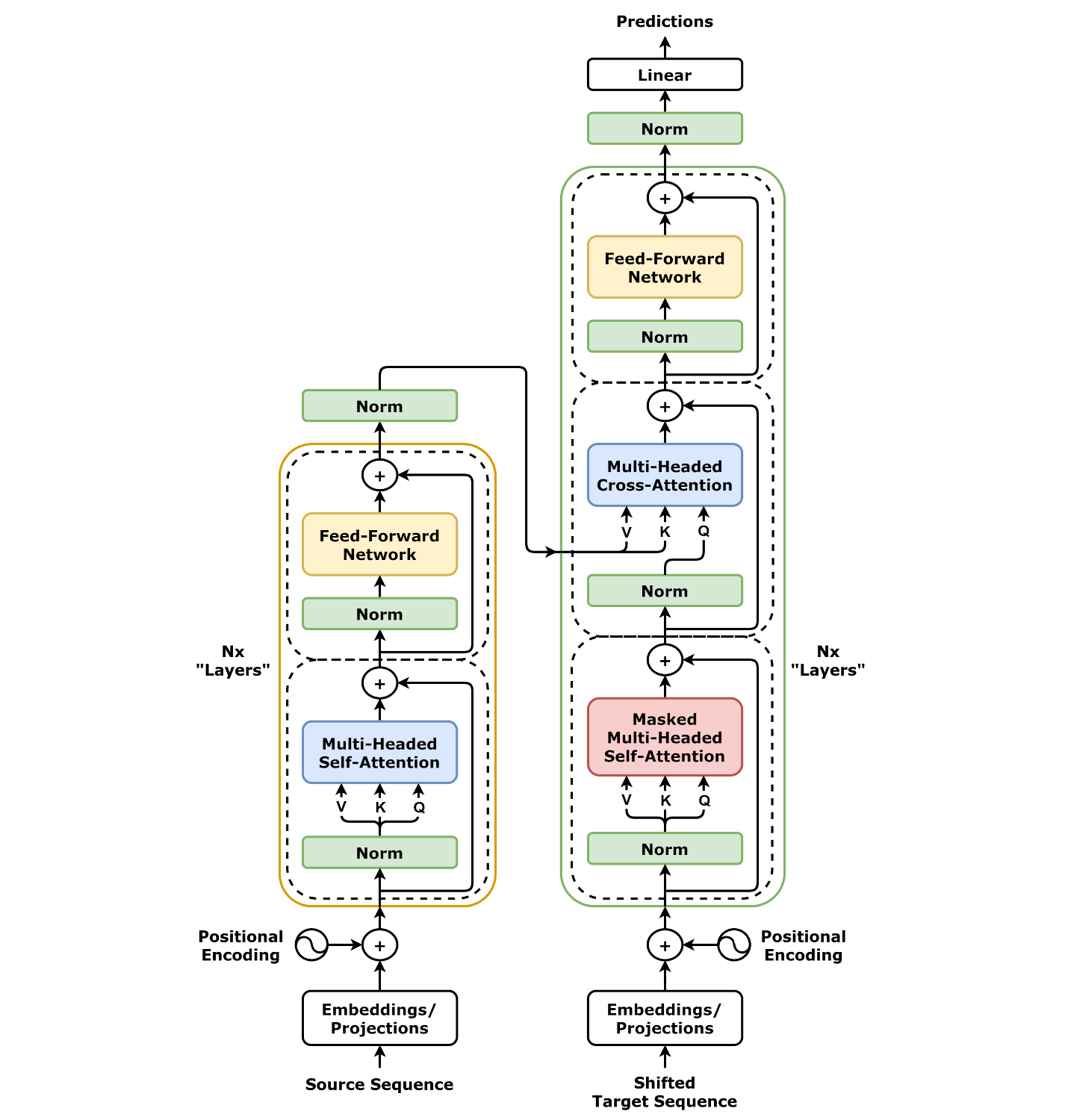}
\caption{{\textbf{A standard Transformer architecture, showing an encoder on the left, and a decoder on the right.} Image by \href{https://github.com/dvgodoy/dl-visuals}{Daniel Voigt Godoy} under \href{https://creativecommons.org/licenses/by/4.0/}{CC BY}. The encoder processes the source sequence through multiple layers of multi-headed self-attention and feed-forward networks, augmented with residual connections, layer normalization, and positional encodings. The decoder incorporates masked multi-headed self-attention to process the target sequence and multi-headed cross-attention to integrate information from the encoder. The output is passed through a feed-forward network and a final linear layer to generate predictions.}}
\label{fig:transformer_arch}
\end{figure}

\subsection{Tokenization and embedding}\label{chapter5:sec:token}
To handle sequential data, such as natural language, Transformers employ \textit{tokenization} to convert it into discrete units. This preprocessing step makes the data compatible with computational models and optimizes it for parallel processing, particularly on GPUs. More concisely, Tokenization breaks a sentence into smaller pieces called \emph{tokens}, which could be words, subwords, or even characters, depending on the tokenization strategy. For example, the sentence ``Transformers are amazing!'' might become tokens like (``Transform'', ``ers'', ``are'', ``amazing'', ``!'') if subwords are used. Modern tokenization methods~\citep{sennrich2016neural, kudo2018sentencepiece, mielke2021words} enable sophisticated mapping of complex inputs into token spaces.

For Transformer, tokens are mapped to high-dimensional real vector representations via \textit{embedding} \citep{vaswani2017attention}, as highlighted by the solid box ``Embeddings/Projections'' in Figure~\ref{fig:transformer_arch}. Let $d_{\mathrm{token}}$ denote the dictionary's token count and $d_{\mathrm{model}}$ represent the embedding vector dimension. We define the set containing all token embedding vectors in the dictionary as 
\[\mathcal{W}\coloneqq \{\mathcal{W}_j \in \mathbb{R}^{d_{\mathrm{model}}} : \mathcal{W}_j \text{ is the embedding of token } j\in [d_{\mathrm{token}}] \}.\]
An $\ell$-length sentence is represented as a sequence of vectors $\{S_j\}_{j=1}^\ell$, where $S_j\in \mathcal W$.
Mathematically, this sequence can be interpreted as a real matrix $S\in \mathbb{R}^{\ell\times d_{\mathrm{model}}}$ whose $j$-th row $S_j$ representing the $j$-th token.

\subsection{Self-attention}

\emph{Self-attention} is a core building block of the transformer architecture, which captures intrinsic correlations among tokens. By allowing each token in a sequence to attend to every other token, Transformer generates attention matrices via the inner-product operations, encoding complex inter-token relationships into a transformative vector representation, colloquially termed ``scaled dot-product attention''. The generated attention matrices highlight how relevant each part of the input is to every other part. This allows Transformers to handle contextual dependencies across a variety of data structures. 

The self-attention mechanism, as highlighted by the blue or red box in Figure~\ref{fig:transformer_arch}, involves three parameterized weight matrices: $W_q, W_k\in \mathbb{R}^{d_{\mathrm{model}}\times d_k}$ and $W_v\in \mathbb{R}^{d_{\mathrm{model}}\times d_v}$.  
\begingroup
\allowdisplaybreaks
\begin{tcolorbox}[enhanced,breakable,colback=gray!5!white,colframe=gray!75!black,title=Remark]
Following conventions~\citep{vaswani2017attention}, we use the notation $d$ to specify $d_{\mathrm{model}}$, $d_k$, and $d_v$ in the rest of this chapter, i.e., $d:=d_{\mathrm{model}}=d_k=d_v$. 
This is a widely used setting in practice.
\end{tcolorbox}
\endgroup

Given a sequence $S\in \mathbb{R}^{\ell\times d}$, we define the three new matrices after interacting it with three parameterized weight matrices $W_q, W_k, W_v$, i.e.,  
\begin{itemize}
    \item Query matrix: $Q\coloneqq S W_q$;
    \item Key matrix: $K\coloneqq SW_k$;
    \item Value matrix: $V\coloneqq SW_v$.
\end{itemize}

The attention block computes the matrix $\Gsoft \in \mathbb R^{\ell \times d}$ via 
\begin{align}  
    \mathrm{Attention}(Q,K,V)=\mathrm{softmax}(Q K^{\top}/\alpha_0)V\eqqcolon \Gsoft,\label{matrix.self-attention}
\end{align}
where $\alpha_0>0$ is a scaling factor, and $\mathrm{softmax}(\cdot)$ is a row-wise nonlinear transformation such that $\mathrm{softmax}(\bm{z})_j\coloneqq e^{\bm{z}_j}/(\sum_{k\in[\ell]} e^{\bm{z}_k})$ for $\bm{z}\in\mathbb{R}^\ell$ and $j\in [\ell]$. The row-wise softmax application ensures the controlled attention distribution. The scaling factor $\alpha_0 = \sqrt{d}$ empirically prevents excessive value amplification, particularly when input matrix rows have zero mean and unit standard deviation.

For decoder-only architectures, \textit{masked} self-attention is employed, strategically hiding tokens subsequent to the current query token, i.e.,
\begin{align}
M_{jk}=\begin{cases}
0 & \quad k\leq j, \\
-\infty & \quad k>j.
\end{cases}\label{eq.mask}
\end{align}
Conceptually, the mask $M$ is applied to the scaled dot product $QK^{\top} / \alpha_0$ in Eq.~(\ref{matrix.self-attention}) before the softmax operation. Specifically, the matrix in the softmax operation is modified as $ QK^{\top} / \alpha_0 + M$.

\smallskip
Another crucial technique in Transformer is the multi-head attention, which further extends the self-attention mechanism by computing and concatenating multiple attention matrices, enabling parallel representation learning across different subspaces. In practice, the embedding dimensions are often much larger (e.g., $d=512$ or $d=768$), with multiple attention heads working simultaneously, each capturing distinct relationships between words in the sequence. This mechanism plays a crucial role in modern AI, as it allows words to dynamically interact with one another within the context of the sequence.

While multi-head attention is a pivotal advancement in Transformer architectures, we will not delve into its details here, as \textit{the quantum Transformer implementations presented below primarily support single-head attention}. Nonetheless, these techniques remain essential in classical AI and are promising directions for future developments in quantum Transformers.

\subsection{Residual connection}

Residual connections (the arrows bypassing the main components, such as the attention and feed-forward layers in Figure~\ref{fig:transformer_arch}), paired with layer normalization (green box in the figure), provide crucial architectural flexibility and robustness. By enabling direct information flow between layers, they mitigate challenges in training deep neural networks \citep{he2015deep, ba2016layer}.

For the $j$-th token in an $\ell$-length sentence, the residual connection generates $\Gsoft_j + S_j \in \mathbb{R}^d$ for $\forall j \in [\ell]$, which is subsequently normalized to standardize the vector representation. Let $$\Bar{s}_j :=\frac{1}{d}(\sum_{k=1}^d(\Gsoft_{jk}+S_{jk}),\dots, \sum_{k=1}^d(\Gsoft_{jk}+S_{jk}))\in \mathbb{R}^{d},$$ where $\varsigma := \sqrt{\frac{1}{d}\sum_{k=1}^d ((\Gsoft_{j}+S_{j}-\Bar{s}_j)_k)^2}$. 
The complete residual connection with the \textit{layer normalization} $\mathrm{LN}(\cdot, \cdot)$ can be expressed as
\begin{align}
    \mathrm{LN}_{\gamma,\beta}(\Gsoft_j,S_j)=\gamma\frac{\Gsoft_j+S_j-\Bar{s}_j}{\varsigma}+\beta, \label{eq.Resnet}
\end{align}
where $\gamma$ and $\beta$ denote the scale and bias parameters, respectively.

\subsection{Feed-forward network}\label{chapter5:sec:FFN}

Recall the definitions of fully-connected neural networks (FFN) in Chapter~\ref{chapt5:sec:classical_nn}. Transformers employ a two-layer fully connected transformation (yellow box in Figure~\ref{fig:transformer_arch}) to proceed with the output of residual connection, i.e.,
\begin{align}
    \mathrm{FFN}(\mathrm{LN}(z_j,S_j))=\sigma(\mathrm{LN}(\Gsoft_j,S_j)M_1+b_1)M_2+b_2,
    \label{eq.neuralnetwork}
\end{align}
where $M_1\in \mathbb{R}^{d\times d'}, M_2\in \mathbb{R}^{d' \times d}$ are linear transformation matrices, and $b_1,b_2$ are vectors.
In most practical cases, $d'=4d$. Here $\sigma(x)$ is an activation function, such as $\tanh(x)$ and $\relu(x)=\max(0,x)$. Another activation function that has been widely used in Transformers is the Gaussian Error Linear Units function (GELU), i.e.,
\begin{align*}
\mathrm{GELU}(x) = x \cdot \frac{1}{2}\left(1 + \mathrm{erf}\left(\frac{x}{\sqrt{2}}\right)\right).
\end{align*}

\begingroup
\allowdisplaybreaks
\begin{tcolorbox}[enhanced, breakable,colback=blue!5!white,colframe=blue!75!black,title={Single-head and single-block Transformer}]
By integrating all ingredients introduced in Chapters~\ref{chapter5:sec:token}, we reach out the explicit form of single-head and single-block Transformer, i.e.,
\begin{align}\label{chapter5:eqn:Transformer-formula}
    \mathrm{Transformer}(S, j) := \mathrm{LN}(\mathrm{FFN}(\mathrm{LN}(\mathrm{Attention}(S,j)))).
\end{align}
For the final output, i.e., to predict the next token, one can implement the softmax function to make the vector into a distribution $\mathrm{Pr}(\cdot|S_1,\dots, S_{j-1})$, where the dimension is the size of the token dictionary $d_{\mathrm{token}}$, and sample from this distribution.
\end{tcolorbox}
\endgroup
 
In modern architectures, multiple computational blocks are applied iteratively. Similar to multi-head attention, we will not explore this in detail here, as the quantum Transformer implementations introduced below primarily support single-head attention.

\subsection{Optimization and inference}

Upon the architecture of Transformer, its \emph{optimization} involves training the model to achieve high performance on a given task. When applied to language processing tasks, a feasible loss function of the Transformer is the cross entropy between the predicted probability distribution and the correct distribution. Mathematically, given a $\ell$-length sequence $\{S_1,\dots,S_{\ell}\}$ as input, the loss function can be written as
\begin{align}
    \mathcal{L}=-\frac{1}{\ell}\sum_{j=1}^{\ell}\Pr(S_j|S_1,\dots, S_{j-1}),
\end{align}
where $\Pr(S_j|S_1,\dots, S_{j-1})$ is the predicted probability of the correct $S_j$ coming out of the softmax layer, based on the previous tokens. 

 This process of training Transformers can be achieved by using Adam optimizer~\citep{kingma2015adam}. Learning rate schedules, such as warm-up followed by decay, can provide stable and effective convergence.

After optimization, we could use the trained Transformer for \emph{inference}, which refers to making predictions on new data. Given a new initial sequence $S'=\{S'_1,\dots,S'_{j-1}\}$, we feed it to the trained Transformer and obtain the distribution $\mathrm{Pr}(S'_j|S'_1,\dots, S'_{j-1})$. Then we can use a decoding strategy (e.g., greedy decoding or sampling) to select the token based on the highest probability or a sampling approach. 

Efficient inference is crucial for deploying models in real-world applications. Speed optimization techniques, such as quantization, reduce the precision of weights and activations to accelerate inference with minimal accuracy loss~\citep{jacob2018quantization}. Pruning removes redundant weights or attention heads to reduce the model size and computational cost~\citep{han2015learning}. Batching and parallelism are also critical, with batched inference allowing the processing of multiple inputs simultaneously and GPU or TPU acceleration enabling parallel computations. Efficient attention at inference, such as caching key and value tensors, reduces redundant computations in autoregressive tasks like text generation~\citep{shoeybi2019megatron}.

The inference cost is up to $10$ times the training cost as large language models (LLMs) are trained once and applied millions of times~\citep{ mcdonald2022great,desislavov2023trends}.
For this reason, in the next section, we explore how to harness quantum computing to address this issue, which is crucial from both scientific and societal perspectives.

\section{Fault-tolerant Quantum Transformer}\label{chapt6:sec:quantum_transformer}
In this section, we move on to show an end-to-end transformer architecture implementable on a quantum device, which includes all the key building blocks introduced in Chapter~\ref{chapt6:sec:classical_transformer}, and a discussion of the potential runtime speedups of this quantum model. In particular, here we focus on the  \textit{inference process} in which a classical Transformer has already been trained and is queried to predict the single next token.

Recall that in Chapter~\ref{chapt6:sec:classical_transformer}, we suppose that the three parameterized matrices in the self-attention mechanism have the same size, i.e., $W_q, W_k, W_v \in \mathbb{R}^{d \times d}$. Besides, the input sequence $S$ and the matrix returned by the attention block $G^{\text{soft}}$ has the size $\ell \times d$. Here, we further suppose the length of the sentence and the dimension of hidden features exponentially scale with $2$, i.e., $\ell=2^N$ and $\log d\in \mathbb{N^+}$. This setting aligns with the scaling of quantum computing, making it easier to understand. For other cases, padding techniques can be applied to expand the matrix and vector dimensions to conform to this requirement.

Since the runtime speedups depend heavily on the capabilities of the available input oracles, it is essential to specify the input oracles used in the quantum Transformer before detailing the algorithms. For the classical Transformers, the memory access to the inputs such as the sentence and the query, key, and value matrices is assumed. In the quantum version, we assume the access of several matrices via \textit{block encoding techniques} introduced in Chapter~\ref{chap2:preliminary-sec:linearAlgebra}.
\begin{assumption}[Input oracles]\label{chapt5:assmpt-input}
Following the explicit form of the single-head and single-layer Transformer in Eqn.~(\ref{chapter5:eqn:Transformer-formula}), there are five parameterized weight matrices, i.e., $W_q$, $W_k$, $W_v \in \mathbb{R}^{d\times d}$ in the attention block, as well as $M_1\in \mathbb R^{d' \times d}$ and $M_2 \in \mathbb R^{d \times d'}$ in FFN.
Note that here, $M_1$ and $M_2$ are actually the transpose of parameterized matrices in the classical transformer.
Quantum Transformer assumes access to these five parameterized weight matrices, as well as the input sequence $S\in \mathbb{R}^{\ell \times d}$ via block encoding.

Mathematically, given any $A\in \{W_q, W_k, W_v, M_1, M_2, S\}$ corresponding to an $N$-qubit operator, $\alpha, \varepsilon\geq 0$ and $a\in \mathbb{N}$, there exists a $(a+ N)$-qubit unitary $U_A$ referring to the  $(\alpha,a,\varepsilon)$-block-encoding of $A$ with
    \begin{align}
        \|A-\alpha(\bra{0}^{\otimes a}&\otimes \mathbb{I}_{2^N})U_A(\ket{0}^{\otimes a}\otimes \mathbb{I}_{2^N})\|\leq \varepsilon,
    \end{align}
    where $\|\cdot\|$ represents the spectral norm.
\end{assumption}
Under this assumption, we know that the quantum Transformer can access $U_S$, $U_{W_q}$, $U_{W_k}$, and $U_{W_v}$ corresponding to the $\Sinput$-encoding of $S$ and $\Winput$-encodings of $W_q,W_k$ and $W_v$, respectively. Moreover, the quantum Transformer can access $(\alpha_{m},a_{m})$-encodings $U_{M_1}$ and $U_{M_2}$ corresponding two weight matrices $M_1\in \mathbb R^{d' \times d}$ and $M_2 \in \mathbb R^{d \times d'}$ in FFN.  
 
\begin{tcolorbox}[enhanced, 
breakable,colback=gray!5!white,colframe=gray!75!black,title=Remark]
For simplicity and clarity, in the following, we consider the \textit{perfect} block encoding of input matrices \textit{without errors}, i.e., $\varepsilon=0$. As such, we will not explicitly write the error term of the block encoding, and use $(\alpha,a)$ instead of $(\alpha, a,0)$. The output of the quantum transformer is a quantum state corresponding to the probability distribution of the next token.
The complete single-layer structure is described in Figure~\ref{fig:qtransformer_arch}.
\end{tcolorbox}

\begin{figure}[h!]
\centering
\includegraphics[width=\textwidth]{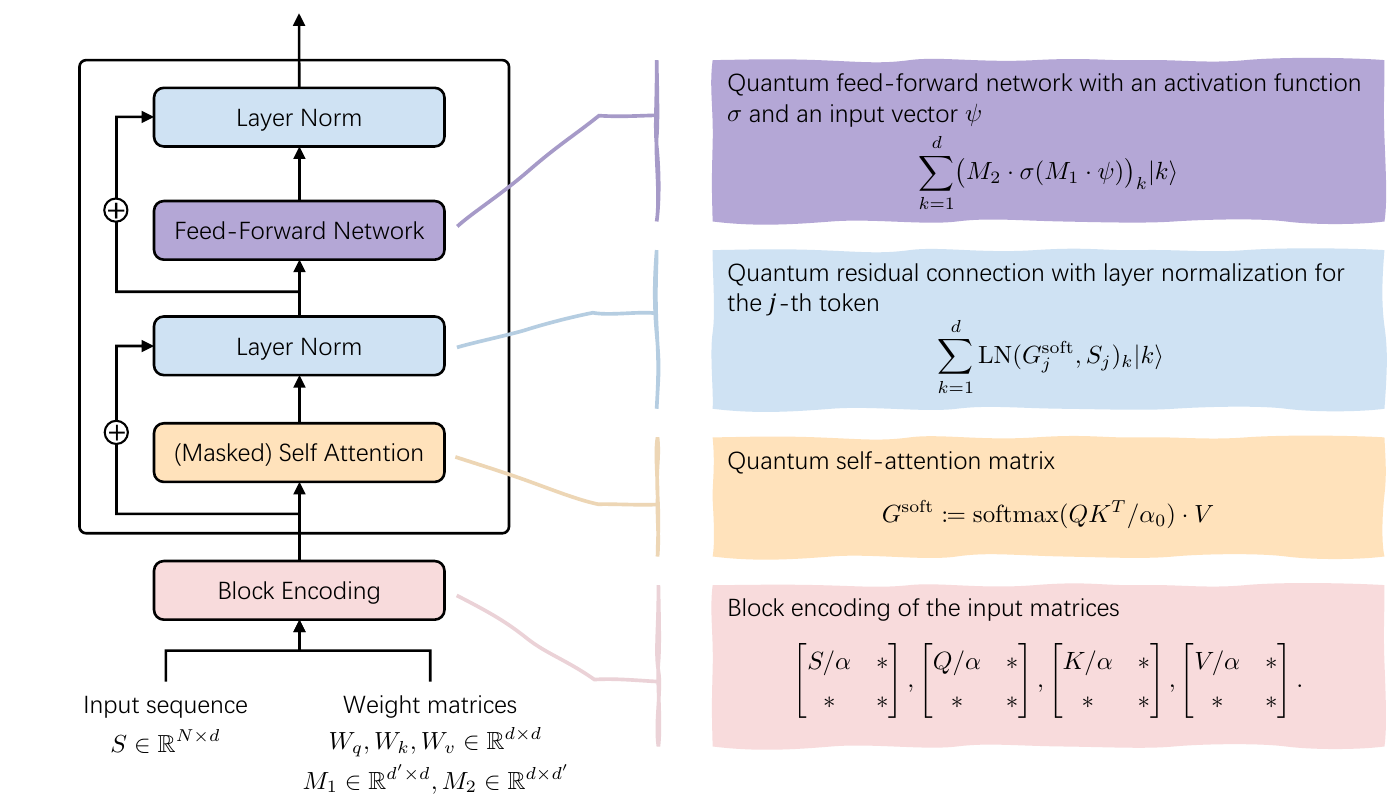}
\caption{{\textbf{Overview of the single-layer decoder-only quantum transformer.} A quantum transformer consists of a self-attention and a feed-forward network sub-layer, incorporating residual connections with layer normalization. The inputs of the quantum transformer are block encodings of matrices for the input sequence and pre-trained weights, from which the relevant matrices for the transformer are constructed (query $Q$, key $K$, and value $V$). Each of the components accepts the block encoding from the prior component as the input and prepares a block encoding of the target matrix using quantum linear algebra as the output.}}
\label{fig:qtransformer_arch}
\end{figure}

Under the above assumptions about access to the read-in protocols, the following theorem indicates how to implement a single-head and single-block transformer architecture in Eqn.~(\ref{chapter5:eqn:Transformer-formula}) on the quantum computer.

\begin{theorem}[Quantum Transformer, informal\label{thmTransformer.informal_main}]
For a single-head and single-block Transformer depicted in  Figure~\ref{fig:qtransformer_arch}, suppose its embedding dimension is $d$ and its input sequence $S$ has the length $\ell=2^N$.  Under Assumption~\ref{chapt5:assmpt-input} about the input oracles, for the index $j\in [\ell]$, one can construct an $\epsilon$-accurate quantum circuit for the quantum state proportional to
\begin{align}\label{chapter5:eqn:Q-trans-form}
\sum_{k=1}^d \mathrm{Transformer}(S,j)_k \ket{k},
\end{align}
by using ${\mathcal{\tilde{O}}}(d N^2 \alpha_s\alpha_w \log^2(1/\epsilon))$ times of the input block encodings.  
\end{theorem}

We show this theorem by explicitly designing the quantum circuit for each computation block of the transformer architecture in a coherent way, i.e., without intermediate measurement. In addition, a subsequent transformation of the amplitude-encoded state, followed by measurement in the computational basis, yields the index of the next predicted token based on the probabilities modeled by the Transformer architecture.

\medskip 
\noindent\textbf{Roadmap}. In the remainder of this section, we detail the implementation of quantum Transformers, proceeding from the bottom to the top as illustrated in Figure~\ref{fig:qtransformer_arch}.  Specifically, we first demonstrate how to quantize the attention block $\mathrm{Attention}(S,j)$ in Chapter~\ref{chapter5:sec:self-attention}. Next, we present the quantization of residual connections and layer normalization operations (i.e., the operation $\mathrm{LN}(\mathrm{Attention}(S,j))$ in Eqn.~(\ref{chapter5:eqn:Transformer-formula})) in Chapter~\ref{chapter5:subsec:Q-resi-ln}. Last, we exhibit the quantization of the fully connected neural network to complete the computation $\mathrm{LN}(\mathrm{FFN}(\mathrm{LN}(\mathrm{Attention}(S,j))))$ in Chapter~\ref{chapter5:subsec:QFFN}. This end-to-end approach ensures that the generated quantum state corresponds to the one described in Eqn.~(\ref{chapter5:eqn:Q-trans-form}).

\subsection{Quantum self-attention}\label{chapter5:sec:self-attention}
We now describe how to achieve the quantum self-attention block, aiming to complete the computation $$ \mathrm{Attention}(Q,K,V)=\mathrm{softmax}(Q K^{\top}/\alpha_0)V\eqqcolon \Gsoft$$ in Eqn.~(\ref{matrix.self-attention}) on quantum computers. More specifically, under Assumption~\ref{chapt5:assmpt-input}, the quantum self-attention block outputs a block encoding of a matrix $G$ whose $j$-th row is the same as the output of the classical attention block, as described in the following theorem.
\begin{theorem}[Quantum self-attention\label{attention.attention}]
    Let $\alpha_0=\alpha_s^2\alpha_w^2$.
    For the index $j\in [\ell]$, one can construct a block encoding of a matrix $G$ such that $G_{j\star}=G^{\mathrm{soft}}_{j}\coloneqq(\mathrm{softmax}\left({QK^{\top}}/{\alpha_{0}}\right)V)_{j\star}$.
\end{theorem}

\begin{tcolorbox}[enhanced, 
breakable,colback=gray!5!white,colframe=gray!75!black,title=Remark]
For quantum self-attention, we make a slight change by setting the scaling factor $\alpha_0=\alpha_s^2\alpha_w^2$ for the following reasons.
The first is that the usual setting $\alpha_0=1/\sqrt{d}$ is chosen somehow in a heuristic sense, and there are already some classical works considering different scaling coefficients which may even achieve better performance \citep{yang2022tensor, ma2024era}.
The second, which is more important, is that the quantum input assumption using the block encoding format naturally contains the normalization factor $\alpha$ which plays a similar role to the scaling factor.
Therefore, for the quantum case in the context of our work, it suffices to use $\alpha$ directly.  
\end{tcolorbox}

The implementation of the quantum self-attention block can be decomposed into three steps:
\begin{enumerate}
    \item Construction of the block encoding of the matrix $QK^{\top}$ and the matrix $V$ given access to $U_S$, $U_{W_q}$, $U_{W_k}$, and $U_{W_v}$;
    \item Implementation of the quantum algorithm to compute the softmax function $\text{softmax}(QK^{\top}/\alpha_0)$ given access to $U_{KQ}$;
    \item Multiply with $V$ via the product of block encodings.
\end{enumerate}
In what follows, we iteratively detail the implementation of each step. 

\medskip
\noindent\textbf{Step I}.  The construction of the block encoding of matrix $QK^{\top}$ and $V$ builds upon the employment of Fact~\ref{blockencoding.product}. That is, given access to $(\alpha,a)$-encoding $U_{A}$ of matrix $A$ and $(\beta,b)$-encoding $U_{B}$ of matrix $B$,  an $(\alpha\beta,a+b)$-encoding can be constructed for the matrix $AB$. This result leads to the efficient construction of the block encodings of $QK^{\top}$ and $V$, i.e.,
\begin{itemize}
    \item [-] For the matrix $V=W_vS$, it is straightforward to set $A=W_v$ and $B=S$ to construct the $\Vinput$-encoding $U_V$, where $\alpha_v=\alpha_s\alpha_w$ and $a_v = a_s + a_w$.
    \item [-] For the matrix $QK^{\top}$, we first use Fact~\ref{blockencoding.product} to construct the $\Qinput$-encoding $U_Q$ and $\Kinput$-encoding $U_K$ with $Q=W_qS$ and $K=W_kS$, respectively. Then, we use Fact~\ref{blockencoding.product} again to construct the $\QKinput$-encoding $U_{QK^{\top}}$ of $QK^{\top}$, where $\alpha_{0} = \alpha_s^2\alpha_w^2$ and $a_{0}=2a_s+2a_w$.  Note that for a real matrix $M$ and its block encoding unitary $U_M$, $U_M^\dagger$ is the block encoding of $M^{\top}$. 
\end{itemize}

\medskip
\noindent\textbf{Step II}. Once the unitary $U_{QK^{\top}}$ is prepared,  we move to implement the quantum algorithm corresponding to the softmax function, i.e., $\mathrm{softmax}(Q K^{\top}/\alpha_0)$. Note that the softmax function relies on the exponential function, which is generally resource-intensive to implement on quantum computers. To circumvent this bottleneck, the quantum Transformer uses polynomial functions to approximate the softmax function, as supported by the following fact.
\begin{fact}\label{approximation.exp}
    For $x\in [-1,1]$, the function $f(x)\coloneqq  e^{x}$ can be approximated with error bound $\epsilon$ with an $\mathcal{O}(\log(1/\epsilon))$-degree polynomial function.
\end{fact}

The insight provided by Fact~\ref{approximation.exp} is to use polynomial functions to approximate the softmax function, i.e., we first approximate $\exp\circ(Q K^{\top}/\alpha_{0})$ using polynomial functions, then multiply with different coefficients (normalization) for each row.

\begin{tcolorbox}[enhanced, 
breakable,colback=gray!5!white,colframe=gray!75!black,title=Remark]
The notation $\exp\circ(A)$ indicates that the exponential operation is applied elementwise to each entry of the matrix $A$, rather than representing a matrix exponential.

Moreover, the element-wise functions mean that functions are implemented on each matrix element.
\end{tcolorbox}

In this context, the challenge of implementing a quantized softmax function reduces to the implementation of a quantized polynomial function.  
The key technique for achieving this lies in applying polynomial functions to each element of block-encoded matrices, as detailed in the following lemma.

\begin{lemma}[Element-wise polynomial function of block-encoded matrix\label{elementfunction.blockencoding}]
Let $N,k \in \mathbb{N}$.
Given access to an $(\alpha, a)$-encoding of a matrix $A\in \mathbb{C}^{2^N\times 2^N}$ and an $r$-degree polynomial function $f_{r}(x)=\sum_{j=1}^{r} c_j x^j$, $c_j \in \mathbb C$ for $j \in [r]$, one can construct a $(C,b)$-encoding of $f_{r}\circ(A/\alpha)$ by using $\mathcal{O}(r^2)$ times the input unitary, where $C\coloneqq \sum_{j=1}^{r} |c_j| $, $b\coloneqq r a+(r-1)N+\lceil\log (r+1)\rceil$.

Moreover, for a polynomial function $g_r(x)=\sum_{j=0}^{r} c_j x^j$ with constant term $c_0$, one can construct a $(C',b)$-encoding of $g_{r}\circ (A/\alpha)$, where $C'=r c_0+C$.
\end{lemma}

\begin{proof}[Proof of Lemma~\ref{elementfunction.blockencoding}]
To achieve this implementation, we construct two state-preparation unitaries $P_L$ and $P_R$, which act on $\lceil \log(r+1)\rceil$ qubits such that
\begin{align}
P_L:\ \ket{0^{\lceil \log(r+1)\rceil}}&\to\frac{1}{\sqrt{C}}\sum_{j=1}^{r} \sqrt{|c_j|} \ket{j},\\
P_R:\ \ket{0^{\lceil \log(r+1)\rceil}}& \to\frac{1}{\sqrt{C}}\sum_{j=1}^{r} \sqrt{|c_j|}e^{i\theta_j} \ket{j},
\end{align}
where $C=\sum_{j=1}^{r} |c_j|$ and $|c_j| e^{i\theta_j}=c_j$. These two unitaries encode the polynomial coefficients $\{c_j\}$ into the quantum circuit, which is needed for block encoding via the linear combination of unitaries indicated in \cref{blockencoding.linearcombination}.
Note that the construction of 
$P_L$ and $P_R$ is efficient for small $r$, as the corresponding circuit is $\mathcal{O}(r)$-depth with only elementary quantum gates \citep{sun2023asymptotically, zhang2022quantum}.

For $j\in [r]$, let $U_{A^j}$ be the $(1,ja+(j-1)N)$-encoding of $$(A/\alpha)^{\circ j}:=\underbrace{(A/\alpha)\circ (A/\alpha)\circ \cdots\circ (A/\alpha)}_{j-1\ \text{times of Hadamard product}},$$ which is constructed by iteratively applying \cref{Hadamard.blockencoding}.
For simplicity, these block encodings $U_{A^j}$ can also be considered as $(1,r a+ (r-1)N)$-encoding of $(A/\alpha)^{\circ j}$.
Then, we construct a $(r a+r N+\lceil\log (r+1)\rceil)$-qubit unitary 
$W=\sum_{j=1}^{r}|j\rangle \langle j|\otimes U_{A^j} +(\mathbb{I}_{2^{\lceil \log (r+1) \rceil}}-\sum_{j=1}^{r}|j\rangle \langle j|)\otimes \mathbb{I}_{2^{(r a+r N)}}$.
Taking the linear combination of block encodings via \cref{blockencoding.linearcombination}, we can implement a $(C,r a+(r-1)N+\lceil\log (r+1)\rceil)$-encoding of $f_{r}\circ(A/\alpha)$.

To implement element-wise functions including constant terms, we also need access to the block encoding of a matrix whose elements are all $1$.
Notice that this matrix can be written as the linear combination of the identity matrix and the reflection operator, i.e.,
  \begin{align}
    \sum_{k,k' \in [2^N]}|k\rangle\langle k'| &= \frac{2^N}{2}\left(\mathbb{I}_{2^N}-(\mathbb{I}_{2^N}- \frac{2}{2^N} \sum_{k, k'\in [2^N]}\ket {k}\bra{k'})\right)\\
    &=\frac{2^N}{2}\biggl(\mathbb{I}_{2^N}-H^{\otimes N}\left(\mathbb{I}_{2^N}-2 \ket { 0^N}\bra{ 0^N} \right)H^{\otimes N}\biggl),\label{eq.sum_all_one}
  \end{align}
where $H$ is the Hadamard gate.
Define $U_{\mathrm{ref}}= |0\rangle\langle 0|\otimes \mathbb{I}_{2^N}+|1\rangle\langle 1| \otimes (H^{\otimes N} (\mathbb{I}_{2^N}-2|0^N\rangle\langle 0^N|)H^{\otimes N})$.
By direct computation, one can show that $U_0=(XH\otimes \mathbb{I}_{2^N})U_{\mathrm{ref}}(H\otimes \mathbb{I}_{2^N})$ is an $(2^N,1)$-encoding of $\sum_{k,k'}|k\rangle\langle k'|$.
One can achieve the element-wise function by following the same steps as above and taking linear combinations among $U_0,\dots, U_{A^r}$. One point to notice is that we can only construct $(2^N,1)$-encoding of the matrix whose elements are all $1$ since the spectral norm of this matrix is $2^N$. 
Therefore, we encode $2^N c_0$ into the state instead of $c_0$ to amplify the constant term. 
\end{proof}

Supported by the above lemma, we can complete Step II (i.e., the quantum softmax for self-attention), as shown in the following theorem. 

\begin{theorem}[Quantum softmax for self-attention, informal]\label{attention.softmax}
    Given an $(\alpha,a)$-encoding $U_A$ of a matrix $A\in \mathbb{R}^{\ell\times \ell}$, a positive integer $d$, and an index $j\in [\ell]$,
    one can prepare a state-encoding of
    \[
    \ket{A_j}\coloneqq \sum_{k=1}^\ell \sqrt{\mathrm{softmax}\left(A/\alpha\right)_{jk}} \ket{k}=
    \frac{1}{\sqrt{Z_j}}\sum_{k=1}^\ell \exp\circ \Big(\frac{A}{2\alpha} \Big)_{jk}\ket{k},\] 
    where $Z_j=\sum_{k=1}^\ell \exp\circ(A/\alpha)_{jk}$.
 \end{theorem}

\begin{proof}[Proof sketch of Theorem~\ref{attention.softmax}]
    We first construct the block encoding of $\exp\circ(\frac{A}{2\alpha})$.
    Note that Taylor expansion of $\exp(x)$ contains a constant term $1$.
    This can be achieved with \cref{elementfunction.blockencoding} and \cref{approximation.exp}.
    Here, since we are only focusing on the $j$-th row, instead of taking linear combination with the matrix whose elements are all $1$, we take sum with the matrix whose $j$-th row elements are all $1$ and else are $0$.
    This enables us to have a better dependency on $\ell$, i.e., from $\ell$ to $\sqrt{\ell}$.
    For index $j\in [\ell]$, let $U_j:\ket{0}\rightarrow \ket{j}$.
    One can achieve this by changing Eqn.~(\ref{eq.sum_all_one}) to the following,
    \begin{align}
        \sum_{k}|j\rangle\langle k|=\frac{\sqrt{\ell}}{2}(U_j H^{\otimes N}-U_j\left(\mathbb{I}_{2^N} -2 \ket { 0^N}\bra{0^N} \right)H^{\otimes N}).
    \end{align}
    Following the same steps in \cref{elementfunction.blockencoding}, one can achieve the construction.
    There are two error terms in this step. Denote $U_{f\circ(A)}$ as the constructed block encoding unitary.
    By \cref{elementfunction.blockencoding} and some additional calculation, one can show that $U_{f\circ(A)}$ is a block-encoding of $\exp\circ(\frac{A}{2\alpha})$. Note that $\exp\circ(\frac{A}{2\alpha})_{jk}=\exp\circ(\frac{A}{2\alpha})^{\top}_{kj}$.    
With unitary $U_{f\circ(A)}^\dagger (I\otimes U_j)$ and amplitude amplification, one can prepare a state-encoding of the target state
\begin{align}
    \ket{A_j}\coloneqq\frac{1}{\sqrt{Z_j}}\sum_{k=1}^\ell \exp\circ \Big(\frac{A}{2\alpha} \Big)_{jk}\ket{k},
\end{align}
where $Z_j=\sum_{k=1}^{\ell} \exp\circ(A/\alpha)_{jk}$ is the normalization factor of softmax function for the $j$-th row.
\end{proof}

\medskip
\noindent\textbf{Step III}.  Finally, we implement the matrix multiplication with $V$. This can be easily achieved by using Fact~\ref{blockencoding.product}, with $U_{f(QK^{\top})}^\dagger$ and $U_V$. Consequently, we obtain an encoded quantum state  analogous to
    \begin{align}
        \sum_{k}^\ell (\mathrm{softmax}(QK^{\top}/\alpha_0)V)_{jk}\ket{k}.
    \end{align}

Combining the results of Steps I, II, and III, we are now ready to present the proof of Theorem~\ref{attention.attention}.

\begin{proof}[Proof of Theorem~\ref{attention.attention}]
    In the first step, we construct the block encoding of matrix $QK^{\top}$ and $V$.
    Note that for a real matrix $M$ and its block encoding unitary $U_M$, $U_M^\dagger$ is the block encoding of $M^{\top}$.
    By Fact~\ref{blockencoding.product}, one can construct an $\QKinput$-encoding $U_{QK^{\top}}$ of $QK^{\top}$, where $\alpha_{0}\coloneqq \alpha_s^2\alpha_w^2$ and  $a_{0}=2a_s+2a_w$. One can also construct an $\Vinput$-encoding $U_V$ of $V$, where $\alpha_v=\alpha_s\alpha_w$ and $a_v = a_s + a_w$.

    By \cref{attention.softmax}, using $U_{QK^{\top}}$ one can prepare a state-encoding of the state 
    \[\sum_{k=1}^\ell \sqrt{\mathrm{softmax}(QK^{\top}/\alpha_0)_{jk}}\ket{k},\]
    where $Z_j=\sum_{k=1}^\ell \exp\circ(QK^{\top}/\alpha_{0})_{jk}$.
    Remember that state encoding is also a block encoding.
    By \cref{Hadamard.blockencoding}, one can construct a block encoding of a matrix whose $j$-th column is \[[\mathrm{softmax}(QK^{\top}/\alpha_0)_{j1},\dots, \mathrm{softmax}(QK^{\top}/\alpha_0)_{j\ell}]\] ignoring other columns.
    Let this block-encoding unitary be $U_{f(QK^{\top})}$.

    Last, by exploiting Fact~\ref{blockencoding.product} again, with $U_{f(QK^{\top})}^\dagger$ and $U_V$, we obtain an encoded quantum state  analogous to
    \[ 
        \sum_{k}^\ell (\mathrm{softmax}(QK^{\top}/\alpha_0)V)_{jk}\ket{k}.\]
   
\end{proof}

\subsubsection{Extension to implement quantum masked self-attention}

Now we consider how to implement the \textit{masked} self-attention, which is essential for the decoder-only structure.
This can be achieved by slightly changing some steps as introduced in previous theorems.

\begin{corollary}
[Quantum masked self-attention]\label{cor:masked_attention.prehalf}
    For the index $j\in [\ell]$, one can construct a block encoding of a matrix $G^{\mathrm{mask}}$ such that $G^{\mathrm{mask}}_{j\star}=(\mathrm{softmax}(\frac{QK^{\top}}{\alpha_{0}}+M)V)_{j\star}$, where $M$ is the masked matrix as Eqn.~(\ref{eq.mask}), $Z_j=\sum_{k=1}^\ell \exp\circ(\frac{QK^{\top}}{\alpha_{0}}+M)_{jk}$.
\end{corollary}
\begin{proof}[Proof of Corollary~\ref{cor:masked_attention.prehalf}]
    To achieve masked self-attention, we slightly change the steps mentioned in \cref{attention.softmax}.
   
    First, to approximate the exponential function, we move beyond using a matrix where all elements in the $j$-th row are set to $1$ while other rows remain $0$. Instead, we refine this approach by considering only the first  $2^{\lceil \log (j+1) \rceil}$  elements in the  $j$-th row to be $1$.  Note that this matrix can be achieved similarly to the original one.
    The encoding factor of this matrix is $2^{\lceil \log (j+1)\rceil/2}$.
    
    Second, after approximating the function,
    for index $j\in [\ell]$, we multiply the block encoding with a projector $\sum_{k,k\leq j}\ket{k}\bra{k}$ to mask the elements.
    Though the projector $\sum_{k\in \mathcal{S}} \ket{k}\bra{k}$ for $\mathcal S \subseteq [\ell]$ is not unitary in general, one can construct a block encoding of the projector by noticing that it can be written by the linear combination of two unitaries:
    \begin{align}
        \sum_{k\in \mathcal S} \ket{k}\bra{k}=\frac{1}{2}\mathbb{I}+\frac{1}{2}\Big(2\sum_{k\in \mathcal S} \ket{k}\bra{k}-\mathbb{I} \Big).
    \end{align}
    Define $U_{\rm proj}\coloneqq\ket{0}\bra{0}\otimes \mathbb{I}+|1\rangle\langle 1|\otimes (2\sum_{k\in \mathcal S} |k\rangle\langle k|- \mathbb{I})$.
    One can easily verify that $(H\otimes \mathbb{I})U_{\rm proj}(H\otimes \mathbb{I})$ is a $(1,1,0)$-encoding of $\sum_{k\in \mathcal S} |k\rangle\langle k|$, where $H$ is the Hadamard gate.
    The following steps follow the same with \cref{attention.softmax} and \cref{attention.attention}.
\end{proof}

One may further achieve the multi-head self-attention case by using the linear combination of unitaries.
 
\subsection{Quantum residual connection and layer normalization}\label{chapter5:subsec:Q-resi-ln}

In this subsection, we discuss how to implement the residual connection with layer normalization on a quantum computer. 
We continue based on the result in \cref{attention.attention} to implement the Layer Norm block shown in Figure~\ref{fig:qtransformer_arch}.

\begin{theorem}[Quantum residual connection with layer normalization]\label{theorem.residual}
Given access to the block encoding of the matrix $G$ in Theorem~\ref{attention.attention}, One can construct a quantum-encoded state 
    \[
    \sum_{k=1}^d \mathrm{LN}(G^{\mathrm{soft}}_{j},S_{j})_k\ket{k}=\frac{1}{\varsigma}\sum_{k=1}^d (G^{\mathrm{soft}}_{jk}+S_{jk}-\Bar{s}_j)\ket{k},\]
    where $\Bar{s}_j\coloneqq \frac{1}{d}\sum_{k=1}^d (G^{\mathrm{soft}}_{jk}+S_{jk})$ and $\varsigma\coloneqq \sqrt{\sum_{k=1}^d (G^{\mathrm{soft}}_{jk}+S_{jk}-\Bar{s}_j)^2}$.
\end{theorem}

\begin{proof}[Proof of sketch of Theorem~\ref{theorem.residual}]
    As shown in \cref{attention.attention}, we can construct a block-encoding of a matrix $G$ whose $j$-th row is the same row as that of $G^{\mathrm{soft}}$.
    By Assumption~\ref{chapt5:assmpt-input}, we are given $U_s$ which is an $\Sinput$-encoding of $S$.
    By Lemma~\ref{LCU.blockencoding} with the state preparation pair $(P,P)$ such that 
    \begin{equation}
        P\ket{0} = \frac{1}{\sqrt{\alpha_g+\alpha_s}}(\sqrt{\alpha_g}\ket{0}+\sqrt{\alpha_s}\ket{1}),
    \end{equation}
    one can construct a quantum circuit $\Ures$ which is an $(\alpha_g+\alpha_s, a_g+1)$-encoding of an $\ell \times d$ matrix whose $j$-th row is the same as that of $G^{\mathrm{soft}}+S$. 
    
    Now we consider how to create a block encoding of a diagonal matrix $\Bar{s}_j\cdot \mathbb{I}$, where $\Bar{s}_j\coloneqq \frac{1}{d}\sum_{k=1}^d (G^{\mathrm{soft}}_{jk}+S_{jk})$. Let us define a unitary $H^{\log d} \coloneqq H^{\otimes \log d}$.
    Note that $H^{\log d}$ is a $(1,0,0)$-encoding of itself, and the first column of $H^{\log d}$ is $\frac{1}{\sqrt{d}}(1,\dots,1)^{\top}$.
    By \cref{blockencoding.product}, one can multiply $\Gsoft+S$ with $H^{\log d}$ to construct a block encoding of an $\ell \times d$ matrix, whose $(j,1)$-element is $\sqrt{d}\Bar{s}_i$.
    One can further move this element to $(1,1)$ by switching the first row with the $j$-th row.
    By tensor product with the identity $\mathbb{I}$ of $\log d$ qubits, one can construct a block encoding of $\sqrt{d}\Bar{s}_i\cdot \mathbb{I}$.

    With $U_j:\ket{0}\rightarrow \ket{j}$, one can prepare the state 
    \begin{align}
        \Ures^\dagger (\mathbb{I}\otimes U_j)\ket{0}\ket{0}=\frac{1}{\alpha_g+\alpha_s}\ket{0}\sum_{k=1}^d \psi_k\ket{k}+\sqrt{1-\frac{\sum_k \psi_k^2}{(\alpha_g+\alpha_s)^2}}\ket{1}\ket{\mathrm{bad}}.
    \end{align}
    By the diagonal block encoding of amplitudes mentioned as \cref{block encoding.amplitudes}, this can be converted to a block encoding of the diagonal matrix $\diag(G_{j1}+S_{j1},\dots,G_{jd}+S_{jd})$.
    
    By taking the linear combination as \cref{LCU.blockencoding} with state preparation pair $(P_1,P_2)$, where
    \begin{equation}
        P_1\ket{0}=\frac{1}{\sqrt{1+1/\sqrt{d}}}(\ket{0}+\frac{1}{\sqrt{d}}\ket{1})
    \end{equation}
    and 
    \begin{equation}
        P_2\ket{0}=\frac{1}{\sqrt{1+1/\sqrt{d}}}(\ket{0}-\frac{1}{\sqrt{d}}\ket{1}),
    \end{equation}
    one can construct a block encoding of $\mathrm{diag}(G_{j1}+S_{j1}-\Bar{s}_j,\dots,G_{jd}+S_{jd}-\Bar{s}_j)$, and we call it $\Uln$.
    Then the unitary $\Uln(\mathbb{I}\otimes H^{\log d})$ is an state-encoding of the state 
    \[\frac{1}{\varsigma}\sum_{k=1}^d (G^{\mathrm{soft}}_{jk}+S_{jk}-\Bar{s}_j)\ket{k},\]
    where $\varsigma\coloneqq \sqrt{\sum_{k=1}^d (G^{\mathrm{soft}}_{jk}+S_{jk}-\Bar{s}_j)^2}$.
\end{proof}

\subsection{Quantum feedforward neural network}\label{chapter5:subsec:QFFN}

We turn our attention to the third main building block of the transformer architecture, the feed-forward neural network. This block often is a relatively shallow neural network with linear transformations and ReLU activation functions~\citep{vaswani2017attention}. More recently, activation functions such as the GELU have become popular, being continuously differentiable. 
We highlight that they are ideal for quantum Transformers, since the QSVT framework requires functions that are well approximated by polynomial functions.
Functions like $\relu(x)=\max (0,x)$ can not be efficiently approximated.
The GELU is constructed from the error function, which is efficiently approximated as follows.

\begin{fact}[Polynomial approximation of error function \citep{low2017quantum}\label{polyappro.error}]
    Let $\epsilon>0$.
    For every $k>0$, the error function $\mathrm{erf}(kx)\coloneqq \frac{2}{\sqrt{\pi}}\int_{0}^{kx} e^{-t^2} \,dt$ can be approximated with error up to $\epsilon$ by a polynomial function with degree $\mathcal{O}(k\log(\frac{1}{\epsilon}))$.
\end{fact}
This lemma implies the following efficient approximation of the GELU function with polynomials.
\begin{corollary}[Polynomial approximation of GELU function\label{thm.gelu}]
Let $\epsilon>0$ and $\lambda \in \mathcal{O}(1)$.
For every $k>0$ and $x\in [-\lambda, \lambda]$, the $\mathrm{GELU}$ function $\mathrm{GELU}(kx)\coloneqq kx\cdot \frac{1}{2}(1+\mathrm{erf}(\frac{kx}{\sqrt{2}}))$ can be approximated with error up to $\epsilon$ by a polynomial function with degree $\mathcal{O}(k\log(\frac{k\lambda}{\epsilon}))$.
\end{corollary}
\begin{proof}[Proof of Cororllary~\ref{thm.gelu}]
    It suffices to approximate the error function with precision $\frac{\epsilon}{k\lambda}$ by \cref{polyappro.error}.
\end{proof}

In the following theorem, we consider how to implement the two-layer feedforward network on quantum computers. As mentioned, the GELU function is widely used in transformer-based models and we explicitly consider it as the activation function in the theorem.
Cases for other activation functions like sigmoid follow the same analysis.
An example is the $\mathrm{tanh}(x)$ function, which can be well approximated by a polynomial for $x\in [-\pi/2,\pi/2]$ \citep{guononlinear2021}.

\begin{theorem}[Two-layer feedforward network with GELU function, informal\label{theorem.ffn}]
Assume we have access to $(\alpha,a)$-state encoding of an $N$-qubit state $\ket{\psi}=\sum_{k=1}^{2^N} \psi_k \ket{k}$, where $\{\psi_k\}$ are real and $\norm{\psi}_2=1$.
Further, assume access to $(\alpha_{m},a_{m})$-encodings $U_{M_1}$ and $U_{M_2}$ of weight matrices $M_1\in \mathbb R^{d' \times d}$ and $M_2 \in \mathbb R^{d \times d'}$.  
Let the activation function be $\mathrm{GELU}(x)\coloneqq x\cdot \frac{1}{2}(1+\mathrm{erf}(\frac{x}{\sqrt{2}}))$.
One can prepare a state-encoding of the state
\begin{align}
\ket \phi = \frac{1}{C} \sum_{k=1}^{d}  \Big(M_2\cdot\mathrm{GELU}(M_1\cdot\psi) \Big)_k\ket{k},
\end{align}
where $C$ is the normalization factor. 
\end{theorem}

\begin{proof}[Proof of sketch of Theorem~\ref{theorem.ffn}]
We have 
\begin{align}
(\mathbb{I}_{2^a}\otimes U_{M_1})(\mathbb{I}_{2^{a_{m}}} \otimes U_{\psi}) \ket{0^{a+a_{m}+N}}=\frac{1}{\alpha \alpha_{m}}\ket{0^{a+a_{m}}}M_1\ket{\psi}+\ket{\widetilde{\perp}},
\end{align}
where $\ket{\widetilde{\perp}}$ is an unnormalized orthogonal state.
For the case $d' \geq \ell$, this can be achieved by padding ancilla qubits to the initial state.
By \cref{block encoding.amplitudes}, one can construct a block encoding of the diagonal matrix $\diag((M_1\psi)_1,\dots, (M_1\psi)_{d'})$.
Note that the $\mathrm{GELU}$ function does not have a constant term, and is suitable to use the importance-weighted amplitude transformation as in \citet{rattewnonlinear2023}.
Instead of directly implementing the GELU function, we first implement the function $f(x)=\frac{1}{2}(1+\mathrm{erf}(\frac{x}{\sqrt{2}}))$.
Note that the value of $|\mathrm{erf}(x)|$ is upper bounded by $1$.
By \cref{block encoding.amplitudes} with function $\frac{1}{4}(1+\mathrm{erf}(\alpha \alpha_{m} \frac{x}{\sqrt{2}}))$, one can construct a block encoding of matrix $\diag(f(M_1\psi)_1,\dots, f(M_1\psi)_{d'})$.

Let the previously constructed block-encoding unitary be $U_{f(x)}$.
We have
\begin{align}
U_{f(x)}(\mathbb{I}\otimes U_{M_1})(\mathbb{I}\otimes U_{\psi})\ket{0}\ket{0}=\frac{1}{2\alpha\alpha_{m}}\ket{0}\sum_k \mathrm{GELU}(M_1\psi)_k\ket{k}+\ket{\widetilde{\perp'}},
\end{align}
where $\ket{\widetilde{\perp'}}$ is an unnormalized orthogonal state. Finally, by implementing the block-encoding unitary $U_{M_2}$, we have 
\begin{align}
&(\mathbb{I}\otimes U_{M_2})(\mathbb{I}\otimes U_{f(x)})(\mathbb{I}\otimes U_{M_1})(\mathbb{I}\otimes U_{\psi})\ket{0}\ket{0}\notag\\
=&\frac{C}{2\alpha\alpha^2_{m}}\ket{0}\sum_{k=1}^{d}  \Big(M_2\cdot\mathrm{GELU}(M_1\cdot\psi) \Big)_k\ket{k}+\ket{\widetilde{\perp''}},
\end{align}
where $C$ is the normalization factor,
and $\ket{\widetilde{\perp}''}$ is an unnormalized orthogonal state.
\end{proof}

\begingroup
\allowdisplaybreaks
\begin{tcolorbox}[enhanced,breakable,colback=gray!5!white,colframe=gray!75!black,title=Remark]
The quantum feedforward network discussed in this subsection is a quantum implementation of the classical feedforward network under the input assumption of block encoding, which is essentially different from the quantum analog of neural networks introduced in \cref{cha5:qnn}.
\end{tcolorbox}
\endgroup

\section{Runtime Analysis with Quadratic Speedups}\label{chapt6:sec:runtime_qtransformer}

In this section, we provide a combined analytical and numerical analysis to explore the potential of a quantum speedup in time complexity.

\subsection{Overview}

Having the quantum implementation of self-attention, residual connection, layer normalization, and feed-forward networks, we are able to construct the quantum transformer by combining these building blocks as in \cref{thmTransformer.informal_main}. 

We obtain this final complexity on the basis of the following considerations: the single-head and single-block transformer architecture includes one self-attention, one feed-forward network, and two residual connections with layer normalization, as shown in Figure~\ref{fig:qtransformer_arch}.

Starting from the input assumption as Assumption~\ref{chapt5:assmpt-input}, for the index $j\in [\ell]$, we first construct the block encoding of self-attention matrix, as described in Section~\ref{chapter5:sec:self-attention}.
This output can be directly the input of the quantum residual connection and layer normalization, as Section~\ref{chapter5:subsec:Q-resi-ln}, which output is a state encoding. Remind the definition of state encoding as \cref{def.stateencoding}.
The state encoding can directly be used as the input of the feed-forward network, as \cref{chapter5:subsec:QFFN}.
Finally, we put the output of the feed-forward network, which is a state encoding, into the residual connection block. This is possible by noticing that state encoding is a specific kind of block encoding.
Multiplying the query complexity of these computational blocks, one can achieve final result.
The detailed analysis of runtime is referred to \citet{guo2024quantumlinear2024}

As \cref{thmTransformer.informal_main} shows,
the quantum transformer uses ${\mathcal{\widetilde{O}}}(d N^2 \alpha_s\alpha_w)$ times the input block encodings, where $\alpha_s$ and $\alpha_w$ are encoding factors.
By analyzing naive matrix multiplication, the runtime of classical single-head and single-block Transformer during the inference stage is $\cO(\ell d + d^2)$, where $d$ is the embedding dimension and $\ell=2^N$ is the input sequence length. 
From the comparison, we can see that $\alpha_s$ and $\alpha_w$ are the dominant factors that affect the potential quantum speedup.
We will explore the properties of these two factors via numerical studies.

\subsection{Numerical evidence}

The encoding factors $\alpha_s$ and $\alpha_w$ appear in the block encodings of matrices $S$ and $W_q,W_k,W_v$.
Recall that the encoding factor $\alpha$ is lower bounded by the spectral norm of a block-encoded matrix $A$, i.e., $\alpha \geq \norm{A}$. Given access to quantum Random Access Memory (QRAM) and a quantum data structure \citep{lloyd2014quantum, kerenidisQuantumRecommendationSystems2016}, there are well-known implementations that enable the construction of a block encoding for an arbitrary matrix $A$ where the encoding factor is upper bounded by the Frobenius norm $\norm{A}_F$.
Based on these considerations, we numerically study
these two norms of the input matrices of several open-source large language models\footnote{Parameters are obtained from the \href{https://huggingface.co/}{Hugging Face} website, which is an open-source platform for machine learning models.} to provide upper and lower bound of $\alpha_s$ and $\alpha_w$.

We first investigate the input sequence matrix $S$, which introduces the dependency on $\ell$.
We consider input data in real-world applications sampled from the widely-used Massive Multitask Language Understanding (MMLU) dataset~\citep{hendry2021measuring} covering 57 subjects across science, technology, engineering, mathematics, the humanities, the social sciences, and more. The scaling of the spectral norm and Frobenius norm of $S$ on the MMLU dataset is demonstrated in Figure~\ref{fig:MMLU_dataset_fro_spe_main}. We can find that the matrix norms of the input matrix of all LLMs scale at most as $\mathcal{O}(\sqrt{\ell})$.

As additional interest, this analysis of the matrix norm provides new insights for the classical tokenization and embedding design. We also observe that comparatively more advanced LLM models like {\it Llama2-7b} and {\it Mistral-7b} have large variances of matrix norms. This phenomenon is arguably the consequence of the training in those models; the embeddings that frequently appear in the real-world dataset are actively updated at the pre-training stage, and therefore, are more broadly distributed.
 
\begin{figure}[h]
\centering
\includegraphics[width=0.9\textwidth]{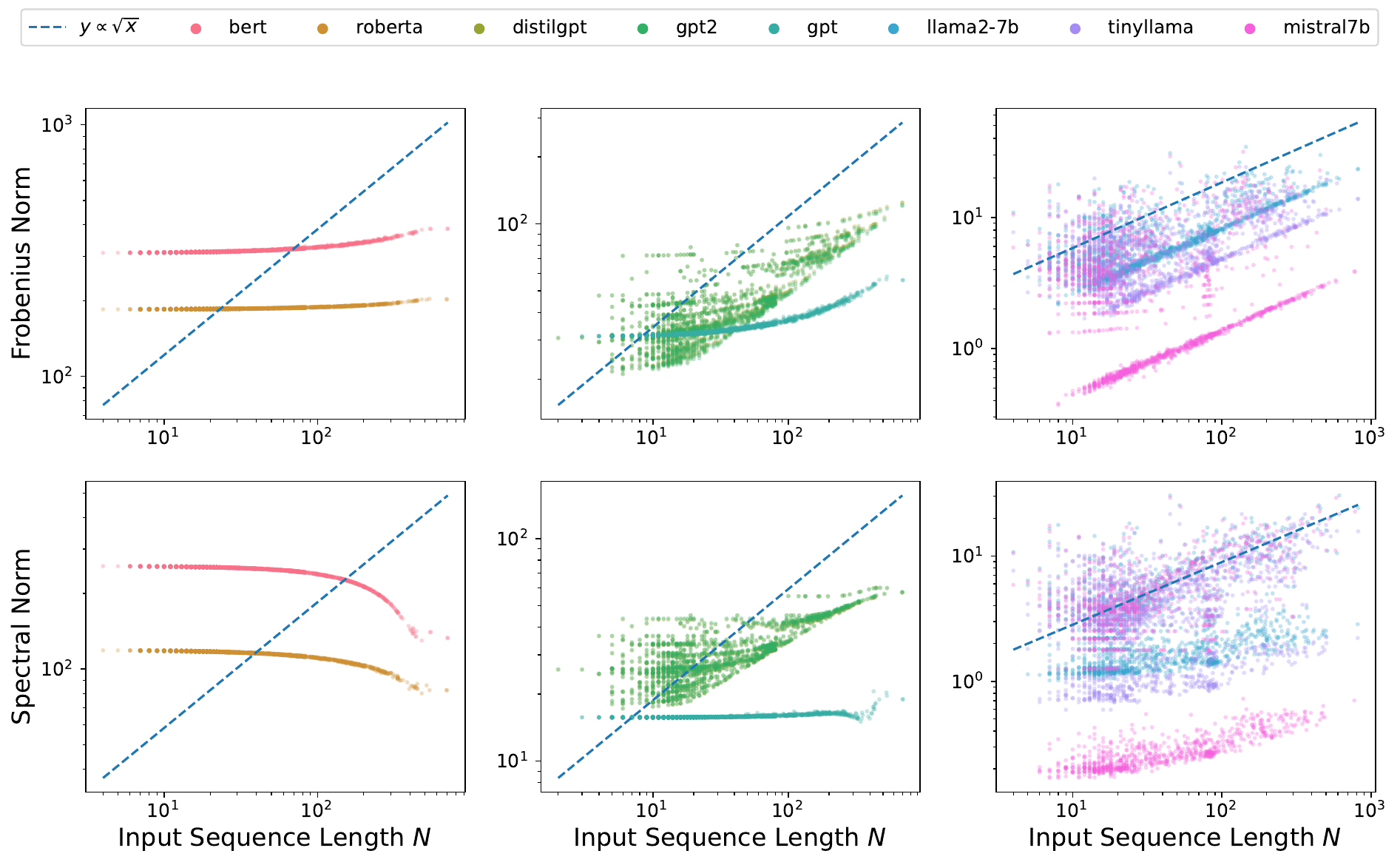}
  \caption{Scaling of the spectral norm $\|S\|$ and the Frobenius norm $\|S\|_{F}$ with $\ell$ for each model, displayed on logarithmic scales for both axes. For reference, the line $y \propto \sqrt{x}$ is also shown. We use tokens in MMLU dataset and convert them to $S$.}
  \label{fig:MMLU_dataset_fro_spe_main}
\end{figure}

We then compute the spectral and Frobenius norms of weight matrices ($W_q,W_k,W_v$) for the large language models.
The result can be seen in Figure~\ref{Fig.LLMs_main}. Many of the LLMs below a dimension $d$ of $10^3$ that we have checked have substantially different norms.
We observe that for larger models such as \textit{Llama2-7b} and
\textit{Mistral-7b}, which are close to the current state-of-the-art open-source models, the norms do not change dramatically.
Therefore, it is reasonable to assume that the spectral norm and the Frobenius norm of the weight matrices are at most $\mathcal{O}(\sqrt{d})$ for advanced LLMs.

\begin{figure}[h]
\centering
\includegraphics[width=\textwidth]{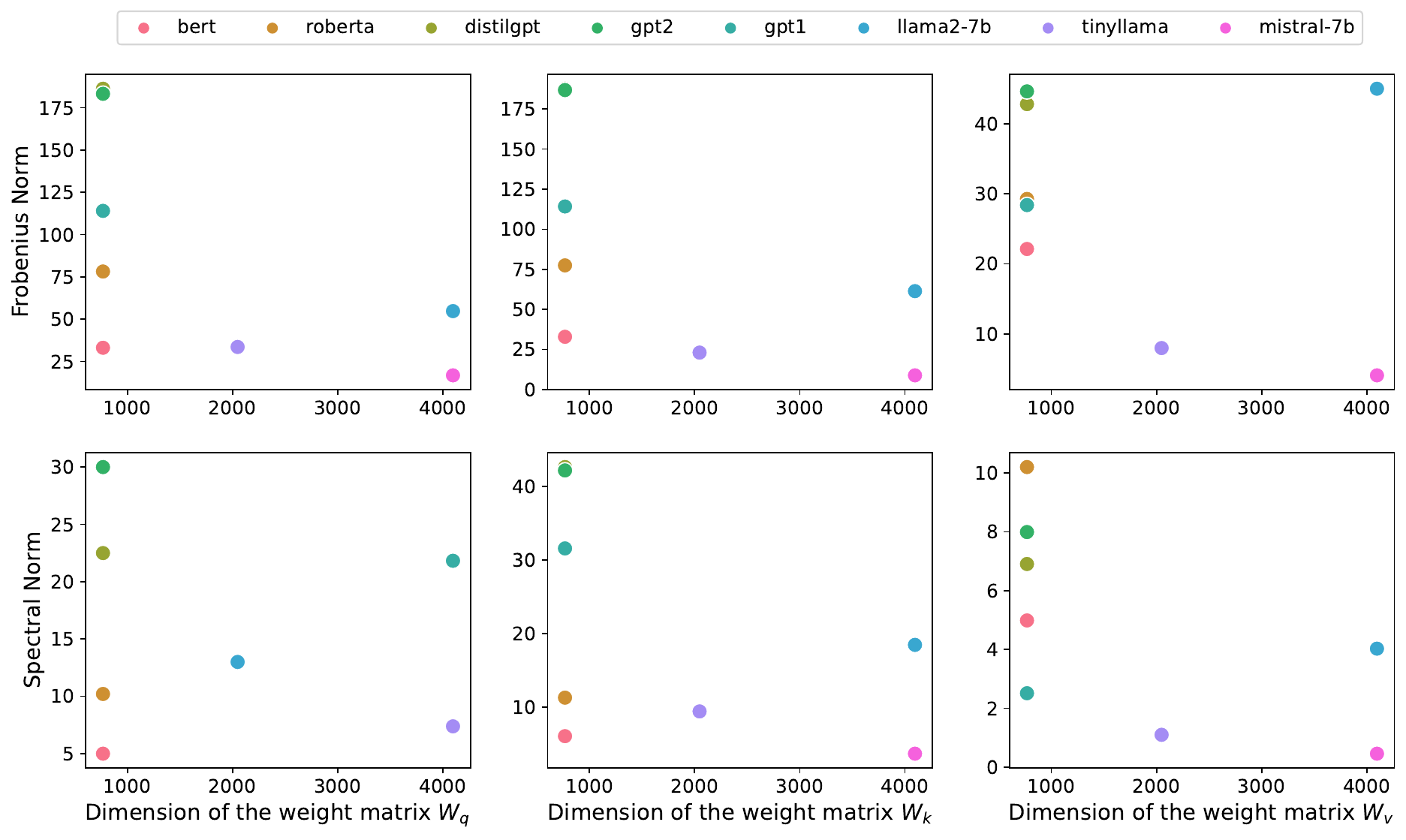}
\caption{\textbf{Norms of weight matrices across open-source LLMs.} The figure shows the maximum Frobenius norm and spectral norm values among weight matrices $W_q, W_k$, and $W_v$ across the eight LLM models.}\label{Fig.LLMs_main}
\end{figure}

Given these numerical experiments it is reasonable to assume that $\alpha_s = \cO(\sqrt{\ell})$ and $\alpha_w = \cO(\sqrt{d})$, and we obtain a query complexity of the quantum transformer in $\widetilde{\mathcal{O}}(d^{\frac{3}{2}} \sqrt{\ell})$.
We continue with a discussion of the possible time complexity. With the QRAM assumption, the input block encodings can be implemented in a polynomially logarithmic time of $\ell$.  Even without a QRAM assumption, there can be cases when the input sequence is generated efficiently, for example when the sequence is generated from a differential equation, see additional discussions in the supplementary material.
In these cases, we demonstrate that a quadratic speedup of the runtime of a single-head and single-block Transformer can be expected.

\section{Code Demonstration}\label{chapt5:section:code-implementation}

We explain how self-attention works with a simple concrete example. Consider a short sequence of three words: ``The cat sleeps''.

First, we convert each word into an embedding vector. For this toy example, we use very small 4-dimensional embeddings:

\begin{lstlisting}[language=Python]
The    = [1, 0, 1, 0]
cat    = [0, 1, 1, 1]
sleeps = [1, 1, 0, 1]
\end{lstlisting}

In self-attention, each word needs to ``attend'' to all other words in the sequence. This happens through three key steps using learned weight matrices ($W_q$, $W_k$, $W_v$) to transform the embeddings into queries, keys, and values. When we multiply our word embeddings by these matrices, we get 

\begin{lstlisting}[language=Python]
import numpy as np

# Input tokens (3 tokens, each with 3 features)
S = np.array([
    [1, 0, 1, 0],  # for "The"
    [0, 1, 1, 1],  # for "cat"
    [1, 1, 0, 1]   # for "sleeps"
])

# Initialize weights for Query, Key, and Value (4x3 matrices)
W_q = np.array(
    [
        [0.2, 0.4, 0.6, 0.8],
        [0.1, 0.3, 0.5, 0.7],
        [0.9, 0.8, 0.7, 0.6],
        [0.5, 0.4, 0.3, 0.2],
    ]
)

W_k = np.array(
    [
        [0.1, 0.3, 0.5, 0.7],
        [0.6, 0.4, 0.2, 0.1],
        [0.8, 0.9, 0.7, 0.6],
        [0.2, 0.1, 0.3, 0.4],
    ]
)

W_v = np.array(
    [
        [0.3, 0.5, 0.7, 0.9],
        [0.6, 0.4, 0.2, 0.1],
        [0.8, 0.9, 0.7, 0.6],
        [0.5, 0.4, 0.3, 0.2],
    ]
)

# Compute Query, Key, and Value matrices
Q = S @ W_q
K = S @ W_k
V = S @ W_v
\end{lstlisting}

Next, we compute attention scores by multiplying $Q$ and $K^{\top}$, then applying softmax.

\begin{lstlisting}[language=Python]
# Compute scaled dot-product attention
d = Q.shape[1]  # Feature dimension
attention_scores = Q @ K.T / np.sqrt(d)

def softmax(x):
    """Compute softmax values for each set of scores in x."""
    return np.exp(x) / np.sum(np.exp(x), axis=1, keepdims=True)

attention_weights = softmax(attention_scores)


\end{lstlisting}
The attention weight would be:
$$
\text{softmax}\left(\frac{Q \cdot K^{\top}}{\sqrt{4}}\right) = \begin{bmatrix}
    0.324 & 0.467 & 0.209\\
    0.305 & 0.515 & 0.180\\
    0.346 & 0.432 & 0.222
\end{bmatrix}.
$$
The final output captures how each word relates to every other word in the sentence. In this case, ``sleeps'' pays most attention to ``cat'' (0.432), some attention to ``The'' (0.346), and less attention to itself (0.222). Finally, we use these scores to create a weighted sum of the values:
\begin{lstlisting}
output = attention_weights @ V
\end{lstlisting}
The final output is
$$
\text{output} = \begin{bmatrix}
1.536 & 1.519 & 1.265 & 1.157\\
1.566 & 1.536 & 1.261 & 1.137\\
1.512 & 1.507 & 1.269 & 1.174
\end{bmatrix}.
$$

\section{Bibliographic Remarks}

Transformer architecture has profoundly revolutionized AI and has broad impacts. 
As classical computing approaches its physical limitations, it is important to ask how we can leverage quantum computers to advance Transformers with better performance and energy efficiency. Besides the fault-tolerant quantum Transformers introduced in Chapter~\ref{chapt6:sec:quantum_transformer}, multiple works are advancing this frontier from various perspectives.

One aspect is to design novel quantum neural network architectures introduced in Chapter~\ref{cha5:qnn} with the intuition from the Transformer, especially the design of the self-attention block. In particular, \citet{li2023quantumselfattentionneuralnetworks} proposes the quantum self-attention neural networks and verifies their effectiveness with the synthetic data sets of quantum natural language processing.
There are several follow-up works along this direction \citep{Cherrat_2024, evans2024learningsasquatchnovelvariational,widdows2024quantumnaturallanguageprocessing}.

Another research direction is exploring how to utilize quantum processors to advance certain parts of the transformer architecture. Specifically, 
\citet{gao2023fastquantumalgorithmattention} considers how to compute the self-attention matrix under sparse assumption and shows a quadratic quantum speedup.
\citet{liu2024quantumcircuitbasedcompressionperspective} harnesses quantum neural networks to generate weight parameters for the classical model. In addition,
\citet{liu2024towards} devises a quantum algorithm for the training process of large-scale neural networks, implying an exponential speedup under certain conditions.
There are several other works that consider machine learning related optimization problems \citep{Yang2023quantumalphatron, zhang2024comparisonsneedoptimizingsmooth, wang2024nearoptimalquantumalgorithmminimizing, rebentrost2018quantumgradientdescentnewtons}.

Despite the progress made, several important questions remain unresolved. Among them, one key challenge is devising efficient methods to encode classical data or parameters onto quantum computers. Currently, most quantum algorithms can only implement one or at most constant layers of Transformers  \citep{guo2024quantumlinear2024,liao2024gptquantumcomputer, khatri2024quixerquantumtransformermodel} without quantum tomography. Are there effective methods that can be generalized to multiple layers, or is achieving this even necessary? Moreover, given the numerous variants of the classical Transformer architecture, can these variants also benefit from the capabilities of quantum computers? Lastly, if one considers training a model directly on a quantum computer, is it possible to do so in a ``quantum-native'' manner—avoiding excessive data read-in and read-out operations?

\chapter{Conclusion}

In this tutorial, we systematically explored the landscape of quantum machine learning, covering foundational principles, the adaptation of classical models to quantum frameworks, and the theoretical underpinnings of quantum algorithms. By providing practical implementations and discussing emerging research directions, we aimed to bridge the gap between classical AI and quantum computing for researchers and practitioners.

The insights gained from this tutorial highlight how quantum machine learning has the potential to revolutionize various domains, from fundamental scientific research to practical applications in industry. As quantum hardware continues to evolve, quantum machine learning will likely play a central role in harnessing the computational advantages of quantum systems. In the meantime, the field of quantum machine learning faces challenges, as discussed at the end of each chapter. Overcoming these barriers will be critical for unlocking the full potential of quantum machine learning.

Moving forward, interdisciplinary collaboration between quantum computing and AI researchers will be essential for addressing these challenges and realizing the transformative potential of QML. All in all, we hope this tutorial serves as a valuable resource for those eager to contribute to this rapidly evolving discipline.

%\begin{acknowledgements}
%The authors are grateful to xxx.
%\end{acknowledgements}

\appendix
\chapter{Notations Summary}\label{App:journalcodes}

\begin{longtable}{  p{.20\textwidth}   p{.80\textwidth} }
\hline
 \textbf{Notation} & \textbf{Concept} \\ \hline
 $a,b, \bm{a}_j, \bm{b}_j, \alpha,\beta$ &  Scalars \\
        $\bm{x}, \bm{y}$ & Vectors \\
        $W, A$ & Matrices \\
  $\mathbb{R}$ & Real Euclidean space \\
  $\mathbb{C}$ & Complex Euclidean space \\
        $\mathbb{N}$ & The set of natural numbers \\
        $[a]$ & The set of integers $\{1,2,\cdots,a\}$ \\
     $\mathbb{E}[\cdot]$ & Expectation value of a random variable \\
        ${\rm Var}[\cdot]$ & Variance of a random variable \\
        $\mathcal{O}$ & Asymptotic upper bound notation \\
        $\Omega$ & Asymptotic lower bound notation \\
        $\top$     & Transpose operation\\
        $a*$       & complex conjugate of the number $a$ \\
        $\dagger$  & Conjugate Transpose operation \\
        $\ket{\cdot}\bra{\cdot}$  & Outer product operation\\
        $A \odot B$ & element-wise multiplication (Hadamard product) of matrices $A$ and $B$ \\
        $f \circ g$ & composition of functions $f$ and $g$ \\
        \hline
$\ket{0}, \ket{1}, \ket{\psi}$ & Pure quantum state in Dirac notation\\
$N$ & Number of qubits \\
        $\ket{0^N}, \ket{0}^{\otimes N}$ & Zero state with $N$-qubits\\
        $\mathbb{I}_{d}$ & Identity matrix with the size $d\times d$\\
        $\rho$ & Quantum state in density matrix representation\\
        $H$ & Hamiltonian \\
        $U, V$ & Unitary operator \\
        $X, Y, Z$ & Pauli operators \\
    $\RX, \RY, \RZ$ & Rotation gates along the $x$, $y$ and $z$ axes, respectively \\
        $CX, CZ$ & Controlled-X gate and controlled-Z gate \\
        $\mathcal{E}$, $\mathcal{N}$ & Quantum channel \\
        $O$ & Observable \\
    $\braket{O}$ & Expectation of observable $O$ \\ 
        $\left\|\cdot\right\|_{op}$ & Operator norm \\ \hline
        $n$ & Number of training examples \\
        $\mathcal{D}$ & Dataset \\
        ${\rm Tr}(O\rho)$ & Expectation value of an observable $O$ \\
        $U(\bm{\theta})$ & Parameterized quantum circuit \\
        $\mathcal{L}(\bm{\theta})$ & Loss function \\
        $\nabla_{\bm{\theta}} \mathcal{L}(\bm{\theta})$ & Gradient of loss function $\mathcal{L}$ w.r.t parameters $\bm{\theta}$ \\
    \hline
    \caption{Notations used in this work.}
    \label{tab:summary_notations_all_contents}
\end{longtable}

\chapter{Concentration Inequality}
In this section, we introduce some of the most common concentration inequalities in statistical learning theory. These inequalities are widely used in deriving generalization error bounds for learning models. In practical scenarios, one often needs to infer properties of an unknown distribution based on finite data samples drawn from that distribution. Concentration inequalities address the deviations of functions of independent random variables from their expectations. They provide tools to analyze the difference between the empirical mean (or some estimate) and the true expectation of random variables that follow a probability distribution.

We begin by recalling some basic tools that will be used throughout this section. For any nonnegative random variable 
$X$ following the probability distribution $p(x)$, its expectation can be written as
\begin{equation}
    \mathbb{E}X=\int_{0}^{\infty} xp(x) \mathrm{d} x.
\end{equation}
This leads directly to a fundamental building block for concentration inequalities, namely Markov's inequality.
\begin{lemma}
    [Markov's inequality] 
    For any nonnegative random variable $X$, and a positive constant $t >0$, we have
    \begin{equation}
        \mathbb{P}\{X\ge t\}  \le \frac{\mathbb{E}X}{t}.
    \end{equation}
\end{lemma}
\begin{proof}
    Employing the definition of cumulative distribution function, we have
    \begin{align}
        \mathbb{P}\{X\ge t\} = & \int_{t}^{\infty} p(x) \mathrm{d} t
        \nonumber \\
        \le & \int_{t}^{\infty} p(x)  \frac{x}{t}\mathrm{d} t
        \nonumber \\
        \le & \frac{\int_{0}^{\infty} p(x)x \mathrm{d}}{t} 
        \nonumber \\
        \le & \frac{\mathbb{E}X}{t}, 
    \end{align}
    where the first inequality follows that $x/t>1$ in the interval $x\in [t,\infty]$, and the second inequality employs the positiveness of integral term $p(x)x/t$.
\end{proof}
Using Markov's inequality, it follows that if $\phi$ is a strictly monotonically increasing, nonnegative function, then for any random variable $X$ and real number $t>0$, we have
\begin{equation}\label{appB:eq:markov_ext}
    \mathbb{P}\{X\ge t\} = \mathbb{P}\{\phi(X) \ge \phi(t)\} \le \frac{\mathbb{E} \phi(X)}{\phi(t)}.
\end{equation}
An application of this result with $\phi(x)=x^2$ leads to the simplest concentration inequality, i.e., Chebyshev's inequality.
\begin{lemma}
    Let $X$ be an arbitrary random variable and the real number $t > 0$, then
    \begin{equation}
        \mathbb{P}\{ |X-\mathbb{E}X| \ge t\} 
        \le \frac{\Var(X)}{t^2}.
    \end{equation}
\end{lemma}
\begin{proof}
    Utilizing the extention of Markov's inequality in Eqn.~\eqref{appB:eq:markov_ext} with setting $\phi(x)=x^2$ yields
    \begin{align}
        \mathbb{P}\{ |X-\mathbb{E}X| \ge t\} 
        = & \mathbb{P}\{ |X-\mathbb{E}X|^2 \ge t^2\}  
        \nonumber \\
        \le & \frac{\mathbb{E} (X-\mathbb{E}X)^2}{t^2}
        \nonumber \\
        = & \frac{\Var(X)}{t^2}.
    \end{align}
\end{proof}

More generally, by taking $\phi(x)=x^q$ ($x\ge 0$), then for any $q>0$, we obtain the following moment-based inequality
\begin{equation}
     \mathbb{P}\{ |X-\mathbb{E}X| \ge t\} 
    \le  \frac{\mathbb{E} (X-\mathbb{E}X)^q}{t^q}.
\end{equation}
Here, the parameter $q$ can be chosen to optimize the upper bound in specific examples. Such moment bounds often provide sharp estimates for tail probabilities. A related idea forms the basis of Chernoff's bounding method. In particular, by setting $\phi(x)=e^{sx}$ for some $s>0$, we can derive a useful upper bound for any random variable $X$ and $t >0$, 
\begin{equation}
    \mathbb{P}\{X\ge t\} = \mathbb{P}\{e^{sX} \ge e^{st}\} \le \frac{\mathbb{E} e^{sX}}{e^{st}}.
\end{equation}
In Chernoff's method, the goal is to choose an appropriate $s > 0$ to minimize the upper bound or make it as small as possible. 

Now, we turn to concentration inequalities for sums of independent random variables. Specifically, we aim to bound probabilities of deviations from the mean, i.e., $\mathbb{P}\{|S_n-\mathbb{E}S_n| \ge t\}$, where $S_n=\sum_{i=1}^n X_i$, and $X_1, \cdots, X_n$ are independent real-valued random variables.

By applying Chebyshev's inequality to $S_n$, we obtain
\begin{equation}
    \mathbb{P}\{|S_n-\mathbb{E}S_n| \ge t\} \le \frac{\Var(S_n)}{t^2} = \frac{\sum_{i=1}^n \Var(X_i)}{t^2}.
\end{equation}
In terms of the sample mean, this can be rewritten as, 
\begin{equation}
    \mathbb{P}\left\{ \left| \frac{1}{n}\sum_{i=1}^n X_i -\mathbb{E} X_i \right| \ge \varepsilon \right\} \le \frac{\sigma^2}{n \varepsilon^2},
\end{equation}
where $\sigma^2 = \frac{1}{n} \sum_{i=1}^n \Var(X_i)$.
Chernoff's bounding method is particularly useful for bounding tail probabilities of sums of independent random variables.  By exploiting the independence property (i.e., the expected value of a product of independent random variables equals the product of their expected values), Chernoff's bound can be expressed as
\begin{align}\label{appB:eq:chernoff_bound}
    \mathbb{P}\left\{  S_n - \mathbb{E}S_n  \ge t \right\} \le &
    e^{-st} \mathbb{E}\left[\exp \left(s\sum_{i=1}^n (X_i-\mathbb{E}X_i)\right) \right]
    \nonumber \\
    = & e^{-st} \prod_{i=1}^n  \mathbb{E}\left[\exp \left(s(X_i-\mathbb{E}X_i)\right) \right] \quad \mbox{(by independence)}.
\end{align}

Now, the challenge then becomes finding a good upper bound for the moment generating function of the random variables $X_i - \mathbb{E}X_i$.
For bounded random variables, one of the most elegant results is Hoeffding's inequality \citet{hoeffding1994probability}.

\begin{lemma}
    [Hoeffding's inequality]
    Let $X$ be a random variable with $\mathbb{E}X=0, a \le X \le b$. Then for $s>0$,
    \begin{equation}
        \mathbb{E}[e^{sx}] \le \exp\left(\frac{s^2(b-a)^2}{8}\right)
    \end{equation}
\end{lemma}

This lemma, combined with Eqn.~\eqref{appB:eq:chernoff_bound} immediately implies Hoeffding's tail inequality \citep{hoeffding1994probability}.

\begin{theorem}
    Let $X_1,\cdots, X_n$ be independent bounded random variables such that $X_i$ falls in the interval $[a_i, b_i]$ with probability one. Then for any real number $t > 0$, we have
    \begin{equation}
        \mathbb{P}\left\{  S_n - \mathbb{E}S_n  \ge t \right\} \le \exp\left(\frac{-2t^2}{\sum_{i=1}^n(b_i-a_i)^2}\right),
    \end{equation}
    and 
    \begin{equation}
        \mathbb{P}\left\{  S_n - \mathbb{E}S_n  \le -t \right\} \le \exp\left(\frac{-2t^2}{\sum_{i=1}^n(b_i-a_i)^2}\right),
    \end{equation}
\end{theorem}
Hoeffding's inequality, first proven for binomial random variables by \citet{chernoff1952measure} and \citet{okamoto1959some}, provides a powerful tool for bounding tail probabilities. However, a limitation is that it does not take into account the variance of the 
$X_i$'s, which can sometimes yield loose bounds.

\chapter{Haar Measure and Unitary t-design}\label{App:Haar_design}

In this section, we introduce some basic knowledge of Haar measure~\citep{haar1933massbegriff} and unitary $t$-design~\citep{dankert2009exact}, which are extensively employed in group representation theory and quantum information~\citep{nagy1993haar,adam2013applications}, especially in the analysis of barren plateaus and the trainability of variational quantum algorithms~\citep{larocca2024review}.

We begin with the Haar measure. Roughly speaking, Haar measure is a unique probability measure that generates the uniform distribution over a compact group. In this chapter, we focus on the Haar measure on the unitary space $\mathcal{U}(d)$ for convenience. Mathematically, the Haar measure is uniform given by invariant properties in Definition~\ref{app_haar_design_def_haar}.

\begin{definition}[Haar measure on $\mathcal{U}(d)$]\label{app_haar_design_def_haar}
A measure $\mu$ is the Haar measure on the unitary space $\mathcal{U}(d)$ if and only if  $\mu$ is
\begin{enumerate}
\item left-invariant, \ie \ $\mu(U\mathcal{S})=\mu(\mathcal{S})$ for any measurable set $\mathcal{S} \subseteq \mathcal{U}(d)$ and any unitary $U \in \mathcal{U}(d)$.
\item right-invariant, \ie \ $\mu(\mathcal{S}U)=\mu(\mathcal{S})$ for any measurable set $\mathcal{S} \subseteq \mathcal{U}(d)$ and any unitary $U \in \mathcal{U}(d)$.
\item a probability measure, \ie \ $\int d\mu (U) =1$.
\end{enumerate}
\end{definition}

As provided in Definition~\ref{app_haar_design_def_haar}, the Haar measure is continuous on the whole space due to left- and right-invariant properties. Therefore, researchers are interested in approximating the Haar measure with the uniform distribution over some finite sets. Specifically, the unitary set whose uniform distribution shares the same $t$-th moment with the Haar measure is defined as the {unitary $t$-design}. 

\begin{definition}[Unitary $2$-design]\label{app_haar_design_def_design}
Let $\mu$ be the Haar measure on the space $\mathcal{U}(d)$. Then, a finite set $\mathcal{S}$ forms a unitary $t$-design if and only if it fulfills one of the following equivalent conditions:
\begin{enumerate}
\item $$ \frac{1}{|\mathcal{S}|} \sum_{U \in \mathcal{S}} U^{\otimes t} \otimes (U^{\dag})^{\otimes t}  = \int_{\mathcal{U}(d)} U^{\otimes t} \otimes (U^{\dag})^{\otimes t} d \mu (U) . $$ 
\item Let $P_{t,t}(U)$ be the polynomial with at most $t$ degrees of elements from $U$ and at most $t$ degrees of elements from $U^\dag$. Then 
\begin{align*}
\frac{1}{|\mathcal{S}|} \sum_{U \in \mathcal{S}} P_{t,t} (U)  ={}& \int_{\mathcal{U}(d)} P_{t,t} (U) d \mu (U) .
\end{align*}
\end{enumerate}
\end{definition}

\begin{corollary}\label{app_haar_design_def_tt_design}
A unitary $t$-design is also a unitary $(t-1)$-design.
\end{corollary}

We remark that invariant properties of the Haar measure lead to several useful formulations of unitary $t$-designs provided in Facts~\ref{app_haar_design_def_ave_1_design} and \ref{app_haar_design_def_ave_2_design}.

\begin{fact}[Average over unitary $1$-design~\citep{puchala2017symbolic}]\label{app_haar_design_def_ave_1_design}
Let $\mathcal{S}$ be a set of unitary $1$-design on $\mathcal{U}(d)$ and $\mu$ be the corresponding Haar measure. Then
\begin{align*}
\frac{1}{|\mathcal{S}|} \sum_{U \in \mathcal{S}} U_{ij} U^{*}_{i' j'} ={}& \int_{\mathcal{U}(d)} U_{ij} U^{*}_{i' j'} d \mu (U) = \frac{1}{d} \delta_{ii'} \delta_{jj'} .
\end{align*}
\end{fact}

\begin{fact}[Average over unitary $2$-design~\citep{puchala2017symbolic}]\label{app_haar_design_def_ave_2_design}
Let $\mathcal{S}$ be a set of unitary $2$-design on $\mathcal{U}(d)$ and $\mu$ be the corresponding Haar measure. Then
\begin{align*}
{}& \frac{1}{|\mathcal{S}|} \sum_{U \in \mathcal{S}} U_{i_1 j_1} U_{i_2 j_2} U^{*}_{i_1' j_1'} U^{*}_{i_2' j_2'} ={} \int_{\mathcal{U}(d)} U_{i_1 j_1} U_{i_2 j_2} U^{*}_{i_1' j_1'} U^{*}_{i_2' j_2'} d \mu (U) \\
={}& \frac{1}{d^2-1} \left( \delta_{i_1 i_1'} \delta_{i_2 i_2'} \delta_{j_1 j_1'} \delta_{j_2 j_2'} + \delta_{i_1 i_2'} \delta_{i_2 i_1'} \delta_{j_1 j_2'} \delta_{j_2 j_1'} \right) \\
-{}& \frac{1}{d(d^2-1)} \left( \delta_{i_1 i_1'} \delta_{i_2 i_2'} \delta_{j_1 j_2'} \delta_{j_2 j_1'} + \delta_{i_1 i_2'} \delta_{i_2 i_1'} \delta_{j_1 j_1'} \delta_{j_2 j_2'} \right) .
\end{align*}
\end{fact}

How big is a unitary $t$-design? \citet{roy2009unitary} have proved that, for instance, the size of unitary $1$-design and unitary $2$-design scale polynomially to the dimension of the unitary space.

\begin{fact}[The size of a unitary $2$-design~\citep{roy2009unitary}]\label{app_haar_design_def_design_size}
A unitary $1$-design on $\mathcal{U}(d)$ has no fewer than $d^2$ elements. A unitary $2$-design on $\mathcal{U}(d)$ has no fewer than $d^4 - 2d^2 + 2$ elements.
\end{fact}

For a system with $N$ qubits, the dimension of the unitary space is $d=2^N$. Therefore, an exact $t$-design could involve exponential numbers of ensembles with increased qubits. Could we obtain an approximation to the unitary $t$-design, which can be generated in polynomial times with less degree of freedom? \citet{haferkamp2022random} has proved that random quantum circuits with linear depths could form an approximate unitary $t$-design. 

\begin{definition}[Approximate unitary designs]\label{app_haar_design_def_approx_design}
Let $\mu$ be the Haar measure on the space $\mathcal{U}(d)$. We denote the moment superoperator $\Phi_{\nu}^{(t)}(A):=\int_{\mathcal{U}(d)} U^{\otimes t} A (U^{\dag})^{\otimes t}$.  
Denote by $M_n(\mathcal{C})$ the $n\times n$ complex matrices.
Denote by $\| \Phi \|_{\diamond}:=\max_{X;\|X\|_1 \leq 1} \| (\Phi \otimes \mathcal{I}_{n} ) X \|_1$ the diamond norm for the linear transformation $\Phi: M_n(\mathcal{C}) \rightarrow M_m(\mathcal{C})$ and $X \in M_{n^2}(\mathcal{C})$. Then, a probability distribution $\nu$ on $\mathcal{U}(d)$ is an $\epsilon$-approximate unitary $t$-design if
\begin{align*}
\left\| \Phi_{\nu}^{(t)} - \Phi_{\mu}^{(t)} \right\|_{\diamond} \leq{}& \frac{\epsilon}{d^t}.
\end{align*}
\end{definition}

\begin{fact}[Random quantum circuits form approximate unitary designs, informal version from \citet{haferkamp2022random}]\label{app_haar_design_def_rqc_approx_design}
For the number of qubits $N \geq \mathcal{O}(\log t)$, alternative layered random quantum circuits with Haar-random unitary gates sampled from $\mathcal{U}(4)$ lead to an $\epsilon$-approximate unitary $t$-design when the circuit depth
\begin{align*}
k \geq{}& \mathcal{O} \left( t^{4+o(1)} \left( Nt + \log \frac{1}{\epsilon} \right) \right) ,
\end{align*}
where the term $o(1) \rightarrow 0$ when $t \rightarrow \infty$.
\end{fact}

\backmatter

\theendnotes

%\bibliographystyle{unsrtnat}
%\bibliography{bookbib}

\end{document}